\documentclass[11pt]{article}

\usepackage[margin=1in]{geometry}
\RequirePackage{amsthm,amsmath,amsfonts,amssymb}
\RequirePackage[authoryear]{natbib}
\RequirePackage[colorlinks,citecolor=blue,urlcolor=blue,linkcolor=blue]{hyperref}
\RequirePackage{graphicx}

\usepackage{authblk}
\usepackage{bbm}
\usepackage{bm}
\usepackage{enumitem}
\usepackage{booktabs}
\usepackage{float}
\usepackage{multirow}
\usepackage{threeparttable}
\usepackage{array}
\usepackage{comment}
\usepackage{mathtools}
\usepackage{hhline}
\usepackage{xcolor}
\usepackage{subcaption}
\usepackage{xparse}
\usepackage{dsfont}
\usepackage{cleveref}
\usepackage{changepage}

\providecommand{\mathbbm}[1]{\mathds{#1}}

\theoremstyle{plain}
\newtheorem{theorem}{Theorem}

\newtheorem{corollary}[theorem]{Corollary}
\newtheorem{proposition}{Proposition}

\theoremstyle{definition}

\newtheorem{assumption}{Assumption}

\newcommand{\indep}{\perp \!\!\! \perp}

\newcommand{\g}{\mathbf{g}}
\newcommand{\gstar}{\mathbf{g}^{*}}
\newcommand{\rma}{\mathrm{a}}

\newcommand{\E}{\mathbb{E}}

\newcommand\mystrut{\rule{0pt}{10pt}}

\newcolumntype{L}[1]{>{\raggedright\let\newline\\arraybackslash\hspace{0pt}}m{#1}}
\newcolumntype{C}[1]{>{\centering\let\newline\\arraybackslash\hspace{0pt}}m{#1}}
\newcolumntype{R}[1]{>{\raggedleft\let\newline\\arraybackslash\hspace{0pt}}m{#1}}

\allowdisplaybreaks

\title{\bf Higher-order Spillover Effects Under Partial Interference}

\author[1]{Qixiang Xu}
\author[1]{Laura Forastiere}
\affil[1]{Department of Biostatistics, Yale University}
\date{}

\begin{document}

\maketitle

\begin{abstract}
Interference, under which a unit's outcome is affected by the treatment of other units through network connections, is often present when units interact on a network.
When the network of interactions is measured, researchers are often interested in the spillover effect from first-order neighbors. When this is the case,
the prevailing approach often involves the neighborhood interference assumption, which is oftentimes overly restrictive. In this paper, we instead rely on a  generalized interference assumption, which allows one's potential outcomes to be influenced by the treatment of units from a wider area of the network, referred to as the `interference set'. For instance, this can be a community detected through a community detection algorithm, or the set of units that can be reached through a finite network path. Under this assumption, we define new causal estimands to quantify spillover effects from first-order neighbors and, in general, from units at a specific network distance $h$. We employ two hypothetical Bernoulli distributions with different probabilities for the $h$-order neighborhood and for the rest of the units in the interference set. We first derive the bias of an approach that relies on a wrong interference set or incorrect exposure mapping function. We then develop new Horvitz-Thompson and Hajek estimators and corresponding weighted regression estimators under the generalized interference assumption. We conduct a series of simulations to assess the bias of OLS estimators---which rely on restrictive interference assumptions and an exposure mapping function---, and the performance of our estimators in different interference scenarios and random graphs. We then apply our estimators to a two-stage randomized trial implemented in Honduras to assess a maternal and child health intervention.
\end{abstract}


\medskip

\section{Introduction}

\subsection{Motivation and Recent Literature}
    Causal inference methods have been extensively developed based on the fundamental assumption of no interference between units. This assumption ensures that the potential outcomes of one unit are not influenced by the treatments assigned to others. Within the potential outcome framework, this is often formalized as part of the Stable Unit Treatment Value Assumption (SUTVA) \citep{rubin1980randomization}. However, in practice, this assumption is often not met due to real-world complexities. For instance, in the spread of infectious diseases, one's likelihood of infection depends not only on their own vaccination status but also on that of their contacts. Similarly, in educational interventions, a student's learning outcomes may be influenced not only by their own participation in a program, but also by their peers' participation, as collaborative learning or peer interactions play a role in shaping individual performance. This phenomenon, known as interference, has become a central focus of recent causal inference studies  \citep[e.g.,][]{benjamin2018spillover, Kim2015SocialNT}. When designing and evaluating interventions, it is essential to consider the potential for interference, as it can significantly impact the effectiveness and efficiency of the intervention. For instance, several studies have documented both direct and spillover effects in economics, health, and education interventions \citep{cai2015social,
    egger2022general, Olese060784, Airoldi2024InductionOS}.

    However, measuring interference can be challenging, since potential outcomes can, in principle, depend on the whole treatment vector in the sample.
    To allow for estimation, the number of potential outcomes must be reduced by making simplifying assumptions on the extent and mechanism of interference.
    One assumption that is commonly made when  units are partitioned into groups, such as schools or villages, is the \textit{partial interference} assumption, which allows for interference within groups but assumes that there is no interference between groups \citep{rosenbaum2007interference}. Under this assumption, \citet{hudgens2008toward} proposed causal estimands and randomization-based estimators for a two-stage randomized experiments, where in the first stage clusters are assigned to a treatment allocation strategy and in the second stage units are assigned to treatment according to their cluster's assignment.
    \citet{liu2014large} derive large-sample properties of randomization-based estimators for causal effects and introduced a sensitivity analysis framework to assess the robustness of these methods.
    In addition to randomized experiments, some literature also considers observational studies under the partial interference assumption. \citet{tchetgen2012causal} and then \citet{liu2016inverse} developed inverse probability weighted (IPW) estimators for causal effects under a hypothetical Bernoulli treatment allocation strategy, and  \citet{perez2014assessing} showed the application of these methods in evaluating the effect of the Cholera vaccine.
    Further extensions include incorporating unit-level covariates and treatment assignment dependencies  \citep[e.g.,][]{papadogeorgou2019causal}, and addressing non-compliance
    \citep[e.g.,][]{ditraglia2023identifying, Imai03042021}.

    However, these methods under the partial interference assumption do not consider the structures of individuals' relationships through which interference may occur.
    When interference is due to outcome diffusion, as is the case with transmission of infectious diseases or peer influence in behaviors, spillover effects can be explained by social interactions in a network,
    represented as graphs, where nodes symbolize units and edges depict connections. As \cite{tchetgen2021auto} shows, interference can occur between any units connected by a long-range path in the network, rather than being limited to directly connected units.
    With networks, it is common to assume that potential outcomes depend on a function of the treatments received by other units \citep{manski2013identification}, referred to as  \textit{exposure mapping function}
    \citep{aronow2017estimating}.
    Such a function serves two purposes: i) restrict the set of units whose treatment may affect one's outcome, known as the \textit{interference set}, and
    ii) specify the mechanism of interference through a summarizing function of the treatments in the interference set--e.g., the proportion or number of treated units.
    A common assumption in network settings is
    the so-called neighborhood interference assumption, which restricts the scope of interference to an individual's direct friends (network neighbors)
    \citep{sussman2017elements, forastiere2021identification, verbitsky2012causal, cortez2022exploiting,leung2022causal, leung2020treatment, ogburn2014causal}.
    In observational studies, under the neighborhood interference assumption, \citet{forastiere2021identification} proposed a parametric generalized propensity score-based estimator, while \citet{lee2023} extended the IPW estimator, with and without the application of an exposure mapping function.

    However, the neighborhood interference assumption is often considered too restrictive, especially with outcome diffusion, where disease transmission or behavioral influence can pass through connections beyond direct neighbors. There is also some literature relaxing the neighborhood interference assumption.
    \citet{leung2022causal}, for instance, defines causal estimands under a K-neighborhood interference assumption, allowing for interference to occur from units up to network distance $K$, but inference is conducted under a weaker approximate neighborhood interference (ANI) assumption, under which interference effects from more distant alters diminish with increasing distance from the ego. A related concept is explored in \citet{fang2025inwardoutwardspillovereffects}, who define estimands based on neighbors’ spillover effects while assuming partial interference.
    \citet{belloni2022neighborhood} propose a data-driven method to identify the interference set provided by a sub-network within a specific radius around each unit,
    and then use a doubly robust estimator to estimate the spillover effect within the identified radius; however, their method requires prior specification of the exposure mapping function and does not estimate spillover effects from different network distances.

Most causal inference methods under interference have focused on the estimation of direct effects of the treatment when fixing the treatment vector or a summary of it among the units in the interference set, and of the spillover effects of altering the exposure to these units' treatments by changing the value of the exposure mapping function or by modifying the allocation mechanism assigning the treatment to the interference set. However, we may be interested in disentangling the spillover effect from different subsets defined by the network distance from the ego. Assessing such effects is crucial for policymakers wishing to understand the extent of interference and potentially leverage its heterogeneous nature across a network by altering the network structure or the treatment allocation.
Applied literature has estimated such spillover effects from higher-order network neighbors using simple parametric estimators and assuming that interference extends up to the network distance of interest \citep[e.g.,][]{cai2015social,banerjee2013diffusion}. Thus, the definition of the interference sets has been driven by the effects of interest, rather than by a formal assessment driven by substantive-matter knowledge or statistical tests.

The misspecification of the exposure mapping function, in general, and of interference sets, in particular, has been the focus of recent work.
When the interference set is defined based on a network, the network structure can be mismeasured due to potentially different types of networks driving interference, or to issues with data collection, including missing links.
Recent advances in this area include methods for recovering the true underlying network structure \citep{bhattacharya2020causal, li2021causal, hardy2019estimating}, approaches for bias correction in egocentric network-based studies when the true network is observed for a subsample \citep{chao2025estimation}, and methods for sensitivity analysis that account for potential interference transmitted through unobserved networks in addition to the observed one \citep{egami2021spillover, Onnela2011SpreadingPI}.
A recent paper by \citet{weinstein2023causal} derives the bias arising from network misspecification, demonstrating that it depends on the divergence between the true and misspecified networks, quantified by the extent of disagreement in induced exposure values across units, and further proposes unbiased estimators that accommodate settings with multiple networks.
\citet{savje2024causal} make a significant contribution by analyzing the implications of misspecifying the exposure mapping, irrespective of the underlying source of misspecification. They emphasize that exposure mappings need not be correctly specified, as their primary role is to define estimands that capture the aspects of treatment and interference deemed substantively relevant, rather than to represent the full causal structure. Under such misspecification, they show that, given relatively mild conditions on the specification errors, standard estimators remain consistent for a generalized estimand—specifically, the contrast in expected potential outcomes conditional on the exposure values of interest.

Current methods for estimating spillover effects in networks present several limitations. First, spillover effects are typically defined and estimated under a pre-specified exposure mapping function, limiting the extent of interference within groups or sets of units at specific network distances from the ego, and assuming homogeneous or heterogeneous mechanisms of interference, through scalar or multivariate summarizing functions.
This approach may be subject to misspecification in multiple ways, including in the underlying network structure, but also in the assumptions related to the interference sets and the mechanisms of interference.
Second, the misspecification of both the interference structure and mechanisms  remains largely unexplored, with methods dealing with network mismeasurement commonly relying on the correct specification of  the exposure mapping function.
    Third, while higher-order interference through network paths is common in practice, as with disease transmission or information diffusion \citep{cai2015social}, no causal inference methods to disentangle the spillover effects from higher-order neighbors have been proposed.

\subsection{Contributions}

In this work, we focus on the estimation of spillover effects from the treatment of units at different network distances under a super-population perspective. We address the issue of potentially misspecified exposure mapping functions. We first derive bias formulas, distinguishing between the misspecification of interference sets and that of interference mechanisms, specified by a summarizing function. We then develop weighted estimators under conservative interference sets and without relying on exposure mapping functions.

Concerning the bias, our approach is to show that the quantity that would be estimated under wrong exposure mapping functions identifies specific weighted averages of potential outcomes under the correct exposure mapping function.
\citet{savje2024causal} shows similar results under a design-based framework and under a general misspecification of the exposure mapping function, without distinguishing between the misspecification sources. In addition, they view the identified quantity under misspecification as a meaningful quantity, rather than as a biased quantity, claiming that it can be interpreted as the effect of intervening to alter the value of the exposure defined on a network and interference set of interest, and an assumed mechanism of interference, which may not be the true one but is considered relevant and possible to intervene on. On the contrary,
similarly to recent works \citep[e.g.,][]{weinstein2023causal, chao2025estimation},
we are interested in estimating the true extent of interference, that is, the effect of altering the value of the exposure defined on the true network and interference set driving interference, as well as by the true interference mechanism, even if these may be unknown.
As opposed to recent literature on the misspecification of networks  \citep[e.g.,][]{weinstein2023causal, hardy2019estimating}, that rely on the correct specification of the interference set and the mechanism of interference, we address the misspecification of both and rely on more conservative assumptions.

We first introduce a generalized interference assumption allowing us to index potential outcomes by the individual treatment and the treatment vector in a conservative interference set, without specifying a summarizing function defining the mechanism of interference.  The common partial interference assumption and neighborhood interference assumption are special cases of this generalized assumption, where interference sets are groups or network neighbors, respectively.
      Under this assumption, we define spillover effects from higher-order neighbors' treatments,
      by hypothesizing a Bernoulli assignment with two probabilities, one for the $h$-order neighborhood and one for remaining units within the interference set, and comparing average potential outcomes when $h$-order neighborhood is assigned to treatment with different probabilities.
      This definition allows us to avoid making assumptions on the mechanism of interference and on a restrictive interference set.

      Under a two-stage design, we develop the nonparametric Horvitz-Thompson and Hajek estimators, and a weighted least squares estimator for marginal structure models. A key advantage of these estimators is that they do not require exposure mapping specifications. For each estimator, under a super-population perspective, we establish theoretical properties including unbiasedness and consistency and derive their variance forms.

      A simulation study is conducted to evaluate estimators' performance across different interference scenarios, including settings with no interference, different orders of interference, and heterogeneous interference patterns. We compare our approaches to ``naive" regression estimators that rely on potentially misspecified exposure mappings. The simulation results validate our theoretical findings regarding unbiasedness and variance, showing the importance of avoiding exposure mappings given the bias introduced by a ``naive" specification of the interference set and interference function.

      Finally, we apply our estimators to a two-stage randomized trial conducted in Honduras, where a maternal and child health intervention was assigned at the household level under varying village dosages. Using baseline social network data, we estimate first- and second-order spillover effects under multiple interference-set specifications, including (i) all individuals reachable through any finite network path, (ii) distance-restricted $h$-neighborhoods, and (iii) communities identified through a network-based community detection algorithm. This application illustrates the empirical patterns of spillover effects at different network distances and different hypothetical scenarios.

    The rest of the paper is organized as follows: Section \ref{sec:Preliminaries} introduces preliminaries, including notations, assumptions, treatment assignment mechanisms, our $h$-order potential outcome, and the causal estimands of interest. Section \ref{sec:Bias} introduces the naive estimator and derives its estimate and bias under different cases of incorrect interference sets or wrong exposure mapping functions, emphasizing the importance of accurate interference set specifications. Section \ref{sec:method} introduces our Horvitz-Thompson estimators, Hajek estimators, and weighted regression estimators for $h$-order potential outcomes and spillover effects. In Section \ref{sec:simulation}, we run simulation studies and evaluate the performance and variance of our estimators across different interference scenarios through simulation. Section \ref{sec:application} performs the real data application for our estimators in Honduras study. Finally, Section \ref{sec:conclusion} concludes this article with some closing remarks. Additional theoretical details, proofs, simulation results, and empirical results are provided in the Appendix.

\section{Preliminaries}
\label{sec:Preliminaries}
  \subsection{Notation}
We consider a sample $\mathcal{N}$ of $I$ non-overlapping clusters, denoted by $i=1, \ldots, I$,
each representing a group of $n_i$ individuals denoted by $j=1, \ldots, n_i$, forming the sub-sample $\mathcal{N}_i=\{ij\}_{j=1}^{n_i}\subset\mathcal{N}$,
with a total of $N=\sum^{I}_{i=1} n_i$ units across all clusters.
Let us consider a binary treatment, such as the maternal and child health intervention assigned to treated households in our motivating application.
We denote by
$A_{i j}$ the treatment of unit $j$ in cluster $i$, with $A_{i j}=1$ if the unit is treated and 0 otherwise, where we assume perfect compliance.
The treatment vector for each cluster $i$ is denoted as $\mathbf{A}_i=[{A}_{i1},\cdots,{A}_{in_{i}}]^T$, and we use $\mathcal{A}(n_i)$ to denote the set of all possible treatment assignment vectors of length $n_i$ in cluster $i$.
The treatment vector for all units in the sample is denoted as $\mathbf{A}=\left \{ \mathbf{A}_{1},\cdots,\mathbf{A}_{I} \right \}$, and we use $\mathcal{A}(N)$ to denote the set of all possible treatment assignments across all clusters. For each unit $j$ in cluster $i$ we also denote by $Y_{ij}$ the observed outcome, and by $\mathbf{Y}_i=\left \{ Y_{i1},\cdots,Y_{in_i} \right \}$ the vector of the observed outcomes in cluster $i$.

Within each cluster $i$, we assume we observe a network of connections among all individuals (e.g., the friendship network).
The network of each cluster $i$ is represented by a graph
$\left(\mathcal{N}_{i}, \mathcal{E}_i\right)$, where $\mathcal{N}_{i}$ denotes the set of nodes (individuals) and $\mathcal{E}_i$ denotes the set of edges (connections). An edge $e_{i j, i k} \in \mathcal{E}_i$ signifies the presence of a connection between individual $j$ and $k$ within cluster $i$. Here, for simplicity, we assume an undirected network such that $e_{i j, i k}=e_{i k, i j}$. An extension of our methods to directed networks is straightforward.
Let $d_{ij,ik}$ denote the shortest distance between unit $j$ and unit $k$ in cluster $i$, which corresponds to the length of the shortest path connecting them. These are standard definitions in graph theory and network analysis \citep{wasserman1994social, deo2016graph}. Note that the distance from a unit to itself is zero, i.e., $d_{i j, i j}=0$. We define the $h^{th}$-order neighborhood of unit $j$ in cluster $i$ as the set $\mathcal{N}_{i j}^{h}=\left\{ik \in \mathcal{N}_i: d_{i j, i k} = h\right\}$, which comprises all units in cluster $i$ whose shortest distance to unit $j$ is equal to $h$, and we denote the size of this set as $|\mathcal{N}^h_{ij}|$. Specifically, the first-order neighborhood is simply the set of direct neighbors of unit $ij$, i.e., $\mathcal{N}_{i j}^{1}=\left\{ik \in \mathcal{N}_i: d_{i j, i k} = 1\right\}=\left\{ik \in \mathcal{N}_i: e_{i j, i k} \in \mathcal{E}_i\right\}$. Conversely, $\mathcal{N}_{i\backslash(j,\mathcal{N}_{i j}^{h})} = \mathcal{N}_{i}\backslash \left\{ij, \mathcal{N}_{i j}^{h}\right\}$ includes all units in cluster $i$ excluding unit $j$ and  its $h^{th}$-order neighborhoods. Therefore, for each unit and each distance $h$, we can partition the set of nodes $\mathcal{N}_{i}$ into three parts: $\left(ij, \mathcal{N}^{h}_{i j}, \mathcal{N}_{i\backslash(j,\mathcal{N}_{i j}^{h})}\right)$.

\subsection{Randomized Experiment}
\label{sec:design}
We consider a randomized experiment with an experimental design parametrized by the parameter or vector of parameters $\Delta$. We assume that the experimental design is such that the following positivity assumption holds.
\begin{assumption}[Positivity] \label{assump:positivity} For each cluster $i=1,\dots, I$ we have
$$
0<P_{\Delta}(\mathbf{A}_i=\mathbf{a}_i)<1 \quad \forall \mathbf{a}_i\in \mathcal{A}(n_i)
\footnote{We can also consider designs where Assumption \ref{assump:positivity} holds only for a subset of $\mathcal{A}(n_i)$.}
$$
\end{assumption}

Throughout the paper, and in particular in the simulation study and the application,
we consider the setting of a two-stage treatment allocation design  \citep{hudgens2008toward}, with a completely randomized design at the cluster level in the first stage and then a second-stage randomization at the individual level. In particular, in the first stage, clusters are assigned to a probability $\delta_m$, with $m=1,\dots,M$. Let $\Delta=\{\delta_1,\cdots, \delta_M \}$, and let $S_i\in \Delta$ be the cluster assignment indicator.
Under a completely randomized assignment for the first stage, $I_m=I/M$ clusters are assigned to each $\delta_m$, such that $P(S_i=\delta_m)=1/M$ \footnote{A different first-stage randomization could be considered.}.

In the second stage, units within cluster $i$ are assigned to treatment according to the strategy $S_i$, that is, the treatment vector $\mathbf{A}_i$ is generated according to some probability distribution
$
P(\mathbf{A}_i=\mathbf{a}_i \mid S_i=\delta_m)
$. The assignment mechanism can operate at the individual level or at the level of subgroups within clusters (e.g., households), depending on the design of the study.

A commonly studied special case, which we employ in our simulation study, is the individual-level Bernoulli design, where conditional on $S_i=\delta_m$, each individual is independently assigned treatment with probability $\delta_m$.
Under this two-stage randomized experiment, the probability of observing  a treatment vector $\mathbf{a}_{i}$ in cluster $i$ can be written as \begin{equation}
\label{eq:twostage}P_{\Delta}\left(\mathbf{A}_{i}=\mathbf{a}_{i}\right)=\sum^{M}_{m=1} P(S_i=\delta_m)P(\mathbf{A}_i=\mathbf{a}_i| S_i=\delta_m)=\frac{1}{M}\sum^{M}_{m=1} \delta_m^{\sum_{j} \rma_{ij}}(1-\delta_m)^{n_i-\sum_{j} \rma_{ij}}
\end{equation}
Other experimental designs can also be considered, and our methods can also be extended to observational studies.

  \subsection{Potential Outcomes and Interference Assumptions}

We denote by $Y_{ij}(\mathbf{a})$ the potential outcome of unit $j$ in cluster $i$ under the treatment vector $\mathbf{A}=\mathbf{a}$. Here, we take a super-population perspective, assuming that potential outcomes for each unit are fixed and their distribution in the sample is due to the sampling mechanism \citep{imbens2015causal}.

Under a randomized assignment of the treatment, the following treatment unconfoundedness assumption holds by design.
\begin{assumption}[Treatment Unconfoundedness] \label{assump:randomization}
$$
Y_{ij}(\mathbf{a})\indep \mathbf{A} \quad \forall ij\in \mathcal{N}, \mathbf{a}\in \mathcal{A}(N)
$$
\end{assumption}

\noindent We also make the following consistency assumption.
\begin{assumption}[Consistency] \label{assump:consistency}
$\forall ij\in \mathcal{N}, \mathbf{a}\in \mathcal{A}(N)$,
if $\mathbf{A}=\mathbf{a}$, then
$Y_{ij}=Y_{ij}(\mathbf{a})$

\end{assumption}
 Assumption \ref{assump:consistency}  ensures that the observed outcome is consistent with the potential outcome determined by the observed treatment vector $\mathbf{A}$.
 The consistency assumption is commonly stated together with the no-interference assumption in the stable unit treatment value assumption (SUTVA) \citep{rubin1980}.
However, the no-interference assumption, which posits that one's potential outcome does not depend on other units' treatment, can be overly restrictive in many real-world scenarios.

Here, we relax the no-interference assumption by allowing a unit's potential outcome to depend on the treatment received by other units belonging to a subset $\mathcal{I}_{ij} \subseteq  \mathcal{N}_i \backslash  \left\{ij\right\}$ that we refer to as
\textit{interference set}. Formally, we make the following assumption.

\begin{assumption}[Generalized Partial Interference]
\label{assump:Generalized}
Let $\mathcal{I}_{ij} \subseteq  \mathcal{N}_i \backslash  \left\{ij\right\} $ be a subset of the sub-sample of cluster $i$.
Then, we have that
$$
\forall \mathbf{a}, \mathbf{a}^{\prime} \in \mathcal{A}(N) \text { s.t. } \mathrm{a}_{ij}=\mathrm{a}_{ij}^{\prime},\mathbf{a}_{\mathcal{I}_{ij}}=\mathbf{a}_{\mathcal{I}_{ij}}^{\prime},\, \text{then} \,\, Y_{i j}(\mathbf{a})=Y_{i j}\left(\mathbf{a}^{\prime}\right)
$$
\end{assumption}
 Assumption \ref{assump:Generalized} rules out interference between clusters, and allows a unit's potential outcome to depend only on its own treatment and the treatments of units within its interference set $\mathcal{I}_{ij} \subseteq  \mathcal{N}_i \backslash  \left\{ij\right\} $. Note that we assume the interference set to be a subset of the unit's cluster sub-sample, ruling out interference between clusters.
 Under Assumption \ref{assump:consistency} and Assumption \ref{assump:Generalized}, we express the potential outcome $Y_{ij}(\mathbf{a})$ as $Y_{ij}(\mathbf{a}_{ij}, \mathbf{a}_{\mathcal{I}_{ij}})$.

 Common interference assumptions in the causal inference literature, such as the partial interference assumption \citep[e.g.,][]{hudgens2008toward} or the neighborhood interference assumption \citep[e.g.,][]{forastiere2021identification}, are special cases of Assumption \ref{assump:Generalized} that depend on the specification of the interference set $\mathcal{I}_{ij}$. We formalize here three special cases of Assumption~\ref{assump:Generalized}.
\begin{enumerate}[label={\textbf{Assumption \ref{assump:Generalized}.\Alph*}}, ref={~\!\!\ref{assump:Generalized}.\Alph*}, leftmargin=3.6cm]
\item \label{assump:partial} (Partial Interference) Assumption \ref{assump:Generalized} holds with $\mathcal{I}_{ij}=\mathcal{N}_i\backslash \{ij\}$.
\item \label{assump:neighborhood} (Neighborhood Interference) Assumption \ref{assump:Generalized} holds with $\mathcal{I}_{ij}=\mathcal{N}^1_{ij}$.
\item \label{assump:h-neighborhood} ($k^{th}$-order Neighborhood Interference) Assumption \ref{assump:Generalized} holds with $\mathcal{I}_{ij}=\left\{i l: d_{i j, i l}\leq k\right\}= \cup_{h=1}^{k}\mathcal{N}^{h}_{ij}$.
\end{enumerate}
Assumption \ref{assump:partial} rules out interference between clusters and assumes that a unit's potential outcome depends at most on the treatments of units within its own cluster. This assumption is typically made in cluster settings when we do not have further information on geographical distances or social interactions. It is plausible when clusters are well separated; however, as cluster size increases, relying on the partial interference assumption may be too conservative. On the contrary, Assumption \ref{assump:neighborhood} limits the scope of interference to one's direct network neighbors. This assumption may be too restrictive, especially when outcome diffusion across the network is involved.
Assumption \ref{assump:h-neighborhood} extends  Assumption \ref{assump:neighborhood} to the $k^{th}$-order neighborhood.
The interference set can also be defined using network community detection algorithms to partition large networks into smaller, more manageable sub-networks, where nodes are densely connected but share few connections with the rest of the network \citep{newman2006finding,clauset2004finding,blondel2008fast}.

Selecting the appropriate interference set is crucial, as incorrect interference sets will lead to biased estimates, as will be shown in the next section.

\section{Misspecified Interference Assumptions}
\label{sec:Bias}
\subsection{Exposure Mapping Function}
Assumptions such as partial interference (Assumption \ref{assump:partial}), neighborhood interference (Assumption \ref{assump:neighborhood}), and $h$-order interference (Assumption \ref{assump:h-neighborhood}) are often introduced to limit the scope of interference. In addition,
to further reduce complexity, in the presence of interference,  a function of the treatment vector in the interference set, $\mathbf{g}(\cdot): \mathcal{A}(|\mathcal{I}_{ij}|) \rightarrow \mathbb{R}\mystrut^{ p}$ is often employed. Formally, the following assumption is put forth.

\begin{assumption}[Generalized Partial Interference and Exposure Mapping]
\label{assump:GeneralizedEM}
Let $\mathcal{I}_{ij} \subseteq  \mathcal{N}_i \backslash  \left\{ij\right\} $ be a subset of the sub-sample of cluster $i$, and let $\mathbf{g}(\cdot): \mathcal{A}(|\mathcal{I}_{ij}|) \rightarrow \mathbb{R}\mystrut^p$.
Then, we have that
$$
\forall \mathbf{a}, \mathbf{a}^{\prime} \in \mathcal{A}(N) \text { s.t. } \mathrm{a}_{ij}=\mathrm{a}_{ij}^{\prime},\mathbf{g}(\mathbf{a}_{\mathcal{I}_{ij}})=\mathbf{g}(\mathbf{a}_{\mathcal{I}_{ij}}^{\prime}),\, \text{then } \,\, Y_{i j}(\mathbf{a})=Y_{i j}\left(\mathbf{a}^{\prime}\right)
$$
\end{assumption}
Under Assumption \ref{assump:GeneralizedEM}, a unit's potential outcome depends only on the individual treatment $A_{ij}$ and a function $\mathbf{g}(\cdot)$ of the treatment vector in the unit's interference set. We refer to $\mathbf{g}(\cdot)$ as the \textit{exposure mapping function} \footnote{Here we distinguish between the definition of the interference set and a function $\mathbf{g}(\cdot)$ applied to the treatment vector in such a set. We call $\mathbf{g}(\cdot)$ the exposure mapping function, although in the literature such term is used for a general function applied to the whole treatment vector $\mathbf{a}$, which can be seen a composite function of one defining the interference set and our function  $\mathbf{g}(\cdot)$ specifying the interference mechanism \citep{aronow2017estimating, savje2024causal}.}.
Define $\mathbf{G}_{ij}=\mathbf{g}(\mathbf{A}_{\mathcal{I}_{ij}})$. Then, under Assumption \ref{assump:GeneralizedEM}, without loss of information, the potential outcome $Y_{ij}(\mathbf{a})$ can instead be written as $Y_{ij}(\rma, \mathbf{g})$, which represents the potential outcome of unit $j$ in cluster $i$ under individual treatment $A_{ij}=\rma$ and treatments among units in the interference set $\mathcal{I}_{ij}$ such that $\mathbf{G}_{ij}=\mathbf{g}(\mathbf{A}_{\mathcal{I}_{ij}})=\mathbf{g}$.
In this way, the unit's potential outcome does not depend on the full treatment vector $\mathbf{A}_{\mathcal{I}_{ij}}$, but only on a summary of it as defined by the exposure mapping function.
$\mathbf{G}_{ij}$ can be, for example, the proportion of treated units in the interference set, or a weighted average of $\mathbf{A}_{\mathcal{I}_{ij}}$, with weights that depend on the network distance or on units' characteristics.
$\mathbf{G}_{ij}$ can also be a multidimensional variable in $\mathbb{R}\mystrut^p$, with $p>1$.
For example, if we assume that interference is limited to the second-order neighborhood, that is, $\mathcal{I}_{ij}=\mathcal{N}^1_{ij}\cup \mathcal{N}^2_{ij}$, then we can have  $\mathbf{G}_{ij}=\mathbf{g}(\mathbf{A}_{\mathcal{I}_{ij}})=[\mathbf{G}_{ij,1},\mathbf{G}_{ij,2}]^T=[g_1(\mathbf{A}_{\mathcal{N}^1_{ij}}),g_2(\mathbf{A}_{\mathcal{N}^2_{ij}})]^T$, where $g_h(\mathbf{A}_{\mathcal{N}^h_{ij}})$ takes the proportion of treated units in the $h^{th}$-order neighborhood. Note that Assumption \ref{assump:Generalized} can be seen as a special case of Assumption \ref{assump:GeneralizedEM}, with $\mathbf{g}(\cdot)$ being the identity function and $\mathbf{G}_{ij}=\mathbf{A}_{\mathcal{I}_{ij}}$.

Given an exposure mapping function, under Assumption \ref{assump:GeneralizedEM}, it is common to define spillover effects by comparing average potential outcomes under two different values of $\mathbf{G}_{ij}$, i.e., $Y_{ij}(\rma,\mathbf{g})$ compared to $Y_{ij}(\rma,\mathbf{g}')$ \citep[e.g.,][]{aronow2017estimating}. For instance, under the previous setting with second-order neighborhood interference and $G_{ij,h}=\frac{1}{|\mathcal{N}^h_{ij}|}\sum_{ik\in \mathcal{N}^h_{ij}} A_{ik}$ with $h=1,2$, the spillover effect from second-order neighbors can be defined by comparing
average potential outcomes changing the proportion of treated second-order neighbors from $G_{ij,2}=\mathbf{g}_2$ to $G_{ij,2}=\mathbf{g}_2'$, while fixing the individual treatment $A_{ij}=\rma$ and the proportion of treated first-order neighbors $G_{ij,1}=\mathbf{g}_1$, i.e., comparing $Y_{ij}(\rma, \mathbf{g}_1, \mathbf{g}_2)$ vs $Y_{ij}(\rma, \mathbf{g}_1, \mathbf{g}_2')$
\citep{cai2015social}.

In order to identify such spillover effects, under Assumption \ref{assump:GeneralizedEM}, even in the presence of a randomized experiment, the exposure mapping function $g(\cdot)$ must satisfy an additional unconfoundedness assumption.
\begin{assumption}[Unconfoundedness of the Exposure Mapping]
\label{assump:unc}
Let $\mathcal{I}_{ij} \subseteq  \mathcal{N}_i \backslash  \left\{ij\right\} $ be a subset of the sub-sample of cluster $i$, and let $\mathbf{g}(\cdot): \mathcal{A}(|\mathcal{I}_{ij}|) \rightarrow \mathbb{R}\mystrut^p$ be an  exposure mapping function. $\mathbf{g}(\cdot)$ is said to be an unconfounded exposure mapping function with respect to $Y_{ij}(\mathbf{a})$, if
$$
Y_{ij}(\mathbf{a}) \indep  \mathbf{g}(\mathbf{A}_{\mathcal{I}_{ij} })|\mathbf{A}=\mathbf{a}, \quad \forall ij\in \mathcal{N}, \mathbf{a}\in \mathcal{A}(N)
$$
\end{assumption}
\noindent If Assumption \ref{assump:unc} holds, $\mathbf{g}(\cdot)$ with $\mathcal{I}_{ij}$ is said to be an unconfounded exposure mapping function with respect to $Y_{ij}(\mathbf{a})$.
Assumption \ref{assump:unc} is a property of the function $\g(\cdot)$ applied to a set $\mathcal{I}_{ij}$ and with respect to a specific potential outcome $Y_{ij}(\mathbf{a})$. If the exposure $\mathbf{G}_{ij}=\mathbf{g}(\mathbf{A}_{\mathcal{I}_{ij}})$ depends on some variable also affecting the outcome, then Assumption \ref{assump:unc} does not hold.
For instance, given $\mathcal{I}_{ij}=\mathcal{N}^1_{ij}$ and $g(\mathbf{A}_{\mathcal{N}^1_{ij}})=\sum_{il\in \mathcal{N}^1_{ij}} A_{il}$, the number of treated first-order neighbors $g(\mathbf{A}_{\mathcal{N}^1_{ij}})$ would not be an unconfounded exposure mapping function if unit $ij$'s number of first-order neighbors (degree) $|\mathcal{N}^1_{ij}|$ was associated with its potential outcome $Y_{ij}(\mathbf{a})$.
Note that Assumption \ref{assump:unc} does not assume the set $\mathcal{I}_{ij}$ and the function $g(\cdot)$ to be such that Assumption \ref{assump:GeneralizedEM} holds.

\begin{proposition}[Identification under Correct Interference Assumptions]
\label{prop:identification}
Let $\mathbf{A}$ be determined by a randomized design such that
Assumptions
\ref{assump:randomization} and \ref{assump:consistency} hold. Let $\mathcal{I}_{ij} \subseteq  \mathcal{N}_i \backslash  \left\{ij\right\} $ be a subset of the sub-sample of cluster $i$, and let $\mathbf{g}(\cdot): \mathcal{A}(|\mathcal{I}_{ij}|) \rightarrow \mathbb{R}\mystrut^p$ such that Assumptions \ref{assump:GeneralizedEM} and \ref{assump:unc} hold, and, given values $a\in\{0,1\}$ and $\mathbf{g}\in \mathbb{R}\mystrut^p$, $0<Pr(A_{ij}=\rma, \mathbf{G}=\mathbf{g})<1$ for all $ij\in \mathcal{N}$. Then, we have the following identification result:
$$
\E[Y_{ij}(\rma,\mathbf{g})]=\mathbb{E}[Y_{ij} \mid A_{ij}=\rma, \mathbf{G}_{ij}=\mathbf{g}]
\footnote{When we have $0<Pr(A_{ij}=\rma, \mathbf{G}=\mathbf{g})<1$ only for a subset of units $V_{\rma,\mathbf{g}}\subset \mathcal{N}$, the expectation on the right hand side should be considered conditional on this subset.}
$$
\end{proposition}
\noindent Proposition \ref{prop:identification} shows that under a randomized experiment and consistency,
given the correct definition of the interference set $\mathcal{I}_{ij}$ and of the exposure mapping function $\mathbf{g}(\cdot)$, such that Assumption \ref{assump:GeneralizedEM} holds, and also such that the exposure mapping function $\mathbf{g}(\cdot)$ is unconfounded (Assumption \ref{assump:unc})
$\E[Y_{ij}(\rma,\mathbf{g})]$ would  be identified by $ \mathbb{E}[Y_{ij} \mid A_{ij}=\rma, \mathbf{G}_{ij}=\mathbf{g}]$. Estimators of this identifying quantity $\mathbb{E}[Y_{ij} \mid A_{ij}=\mathrm{a}, \mathbf{G}_{ij}=\mathbf{g}]$ include the Horvitz-Thompson estimator \citep[e.g.,][]{aronow2017estimating}, i.e., a weighted average of the observed outcomes among those with $ A_{ij}=\mathrm{a}, \mathbf{G}_{ij}=\mathbf{g}$, or parametric estimators, regressing the observed outcome on $A_{ij}$ and $\mathbf{G}_{ij}$ \citep[e.g.,][]{cai2015social}, whose unbiasedness depend on the correct specification of the parametric model.

It is worth noting that when the exposure mapping function is unconfounded only after conditioning on some other variables, such as degree in the case of the number of treated neighbors, identification of  $\E[Y_{ij}(\rma,\mathbf{g})]$ would still be possible after controlling for those variables \citep{forastiere2021identification}. Here, for simplicity, we consider the simpler case where Assumption \ref{assump:unc} holds.

\subsection{Misspecified Exposure Mapping or Interference Set}
\label{sec:bias}
Proposition \ref{prop:identification} outlines the identification of mean potential outcomes when both the
interference set $\mathcal{I}_{ij}$ and of the exposure mapping function $\mathbf{g}(\cdot)$ are correctly specified.

Nevertheless, either the interference set $\mathcal{I}_{ij}$ or the exposure mapping function $\mathbf{g}(\cdot)$, or both, may be misspecified. In fact, the specification of the interference set is usually approximate and it often depends either on the data structure or on the causal effect of interest. For instance, when spillover effects from first-order neighbors are of interest, it is common to limit the extent of interference to the first-order neighborhood \citep[e.g.,][]{sussman2017elements, forastiere2021identification,leung2022causal, leung2020treatment, ogburn2014causal}.
As for the exposure mapping function, specifying the mechanism through which interference occurs and to what it depends on is not straightforward. Oftentimes, the choice of the exposure mapping is made more for simplification purposes than to reflect the true interference mechanism.

We will now investigate the consequences of misspecifying the interference set and the exposure mapping function.
We will focus on the identifying quantity $ \mathbb{E}[Y_{ij} \mid A_{ij}=\mathrm{a}, \mathbf{G}^*_{ij}=\mathbf{g}]$ that would be estimated if we incorrectly assumed that Assumption \ref{assump:GeneralizedEM} were true with $\mathbf{G}^*_{ij}=\mathbf{g}^*(\mathbf{A}_{\mathcal{I}^*_{ij}})$, where $\mathcal{I}^*_{ij}$ and  $\mathbf{g}^*(\cdot)$ are the misspecified interference set and exposure mapping function. In Proposition \ref{prop:incorrect_interference}, we will show what this quantity actually identifies when Assumption \ref{assump:GeneralizedEM} does hold but under a different interference set $\mathcal{I}_{ij}$ and exposure mapping function $\mathbf{g}(\cdot)$. Because typically we tend to make strong assumptions on the interference set, we consider the case where $\mathcal{I}^*_{ij}\subseteq \mathcal{I}_{ij}$.

\begin{proposition}[Incorrect Interference Assumptions]
\label{prop:incorrect_interference}
Let $\mathbf{A}$ be determined by a randomized design such that
Assumptions \ref{assump:randomization} and \ref{assump:consistency} hold.
Let $\mathcal{I}_{ij} \subseteq  \mathcal{N}_i \backslash  \left\{ij\right\} $ be a subset of the sub-sample of cluster $i$, and let $\mathbf{g}(\cdot): \mathcal{A}(|\mathcal{I}_{ij}|) \rightarrow \mathbb{R}\mystrut^p$ such that Assumptions \ref{assump:GeneralizedEM} and \ref{assump:unc} hold.
Let $\mathcal{I}^{*}_{ij}$ be a second   set such that $\mathcal{I}^{*}_{ij} \subseteq \mathcal{I}_{ij}$, and let $\mathbf{g}^{*}(\cdot): \mathcal{A}(|\mathcal{I}^{*}_{ij}|) \rightarrow \mathbb{R}\mystrut^d$ be a second unconfounded exposure mapping function for $Y_{ij}(\mathbf{a})$ (Assumption \ref{assump:unc} holds), and $\mathbf{G}^*_{ij}=\mathbf{g}^*(\mathbf{A}_{\mathcal{I}^*_{ij}})$.
 Then, the following equality holds:
$$
\begin{aligned}
\E[Y_{ij} \mid A_{ij}
=\mathrm{a}, \mathbf{G}^*_{ij}=\mathbf{g}]=
\!\!\!\!\!\!\!\!\!\sum_{\substack{ \mathbf{a}_{\mathcal{I}^{*}_{ij} }
:  \gstar (\mathbf{a}_{\mathcal{I}^{*}_{ij} })=\mathbf{g}
}}
\,\,\,\sum_{\mathbf{a}_{\mathcal{I}_{ij} \setminus \mathcal{I}^{*}_{ij} }
}
&\E\bigg[Y_{ij} \bigg(A_{ij}=\mathrm{a}, \g(\mathbf{A}_{\mathcal{I}_{ij}^{*}}=\mathbf{a}_{\mathcal{I}_{ij}^{*} }, \mathbf{A}_{\mathcal{I}_{ij} \setminus \mathcal{I}_{ij}^{*} }=\mathbf{a}_{\mathcal{I}_{ij} \setminus \mathcal{I}_{ij}^{*} })\bigg)\bigg] \\
&\times P(A_{\mathcal{I}_{ij}^{*}}=\mathbf{a}_{\mathcal{I}_{ij}^{*} }, A_{\mathcal{I}_{ij} \setminus \mathcal{I}_{ij}^{*} }=\mathbf{a}_{\mathcal{I}_{ij} \setminus \mathcal{I}^{*}_{ij} }\mid A_{ij}=\mathrm{a},  \mathbf{G}^*_{ij}=\mathbf{g})\\
\end{aligned}
$$
\end{proposition}

\noindent
\textit{Proof of Proposition \ref{prop:incorrect_interference} is provided in Appendix A.1.}

\noindent Proposition \ref{prop:incorrect_interference} shows that under incorrect assumptions on the interference set and the exposure mapping function, the quantity $\E[Y_{ij} \mid A_{ij}
=\mathrm{a}, \mathbf{G}^*_{ij}=\mathbf{g}]$ that we would target is a weighted average of mean potential outcomes under individual treatment $A_{ij}=\mathrm{a}$ and values of the true exposure mapping function $\g(\cdot)$ applied to the treatment of the true interference set $\mathcal{I}_{ij}$, i.e., $g(\mathbf{A}_{\mathcal{I}_{ij}})$, with $\mathbf{A}_{\mathcal{I}_{ij}}$ drawn from its distribution conditional on $A_{ij}=\mathrm{a}$ with values in the subset $\mathcal{I}_{ij}^{*}$ such that $g^*(\mathbf{A}_{\mathcal{I}^*_{ij}})=\mathbf{g}$.

To facilitate the understanding of Proposition \ref{prop:incorrect_interference}, we present two corollaries when either the interference set is incorrect but the exposure mapping function is well specified, or the exposure mapping function is incorrect but the interference set is well specified.

\begin{corollary}[Misspecified Interference Set]
\label{cor:3A}
Given the conditions of Proposition ~\ref{prop:incorrect_interference},
if function $\g^{*}(\cdot): \mathcal{A}(|\mathcal{I}^{*}_{ij}|) \rightarrow \mathbb{R}\mystrut^d$ is the restriction of function $\g(\cdot)$ to the subvector $\mathbf{A}_{\mathcal{I}^{*}_{ij}}$, with $d<p$, such that $\g(\mathbf{A}_{\mathcal{I}_{ij}})=(\g^*(\mathbf{A}_{\mathcal{I}^*_{ij}}),\g^{\backslash *}(\mathbf{A}_{\mathcal{I}_{ij}\setminus\mathcal{I}^*_{ij}}) )$,
 then we have:
$$
\begin{aligned}
\E[Y_{ij} \mid A_{ij}
=\mathrm{a}, \mathbf{G}^*_{ij}=\mathbf{g}]
=\!\!\!\!\!
\sum_{\substack{ \mathbf{a}_{\mathcal{I}_{ij}^{*}}:
\g^*( \mathbf{a}_{\mathcal{I}_{ij}^{*}})=\mathbf{g}}} \,\,\sum_{\mathbf{a}_{\mathcal{I}_{ij} \setminus \mathcal{I}_{ij}^{*} }
}
&\E\bigg[Y_{ij} \bigg(A_{ij}=\mathrm{a},
\g,
\g^{\backslash *}(A_{\mathcal{I}_{ij} \setminus \mathcal{I}_{ij}^{*} }=\mathbf{a}_{\mathcal{I}_{ij} \setminus \mathcal{I}_{ij}^{*} })\bigg) \bigg]\\
&\times \!P(A_{\mathcal{I}_{ij}^{*}}=\mathbf{a}_{\mathcal{I}_{ij}^{*}}, A_{\mathcal{I}_{ij} \setminus \mathcal{I}_{ij}^{*} }=\mathbf{a}_{\mathcal{I}_{ij} \setminus \mathcal{I}_{ij}^{*} }\mid \!\!A_{ij}=\mathrm{a},  \mathbf{G}^*_{ij}=\mathbf{g})
\end{aligned}
$$
\end{corollary}
A common setting corresponding to that in Corollary~\ref{cor:3A} is when a neighborhood interference assumption (Assumption~\ref{assump:neighborhood}) is made, while instead interference spans the whole cluster (Assumption~\ref{assump:partial}), with $\g(\cdot)$ being the identity function.
Another example is that  when Assumption~\ref{assump:h-neighborhood} holds with $k=2$, that is, the correct interference set is
$\mathcal{I}_{ij}=\mathcal{N}^{1}_{ij}\cup\mathcal{N}^{2}_{ij}$, and $\mathbf{g}_h(\mathbf{A}_{\mathcal{N}^h_{ij}})$ takes the proportion of treated units in the $h^{th}$-order neighborhood with $\mathbf{G}_{ij}=\mathbf{g}(\mathbf{A}_{\mathcal{I}_{ij}})=[G_{ij,1},G_{ij,2}]^T=[\mathbf{g}_1(\mathbf{A}_{\mathcal{N}^1_{ij}}),\mathbf{g}_2(\mathbf{A}_{\mathcal{N}^2_{ij}})]^T$. When the neighborhood interference assumption (Assumption~\ref{assump:neighborhood}) is mistakenly assumed, with the incorrect interference set being $\mathcal{I}^{*}_{ij}=\mathcal{N}^{1}_{ij}$, and the exposure mapping function $\mathbf{g}^{*}(\cdot)$ is restricted to $\mathbf{g}_{1}(\cdot)$ with $\mathbf{G}^{*}_{ij}=G_{ij,1}$, then we have:
$$
\begin{aligned}
&\E[Y_{ij} \mid A_{ij}=\mathrm{a}, G_{ij,1}={g}_1]=
\sum_{g_2=0}^1
\E[Y_{ij}(\mathrm{a},G_{ij,1}={g}_1,G_{ij,2}={g}_{2})]
P(G_{ij,2}={g}_{2}\mid A_{ij}=\mathrm{a},G_{ij,1}={g}_1)
\end{aligned}
$$
If potential outcomes followed a linear model, such as $\E[Y_{ij}(\mathrm{a},G_{ij,1}={g}_1,G_{ij,2}={g}_{2})]=\beta_0+\beta_a \mathrm{a}+\beta_{g_1} g_1+\beta_{g
_2} g_2$, and an OLS estimator was used to estimate the first-order spillover effect $\beta_{g_1}$ omitting the term $g_2$, we would have the common
omitted variable bias
\citep{mccallum1972relative, clarke2005phantom}.
For instance, if we wanted to estimate the spillover effect for the untreated of having $g_1$ treated first-order neighbors vs $g_1'$, by omitting the second-order term we would estimate $\sum_{g_2=0}^1
\E[Y_{ij}(\mathrm{a},G_{ij,1}={g}_1,G_{ij,2}={g}_{2})]P(G_{ij,2}={g}_{2}\mid A_{ij}=\mathrm{a},G_{ij,1}={g}_1)- \sum_{g_2=0}^1
\E[Y_{ij}(\mathrm{a},G_{ij,1}={g}_1',G_{ij,2}={g}_{2})]P(G_{ij,2}={g}_{2}\mid A_{ij}=\mathrm{a},G_{ij,1}={g}_1')$. In principle, if $G_{ij,1}$ and $G_{ij,2}$ were independent given the individual treatment, we would estimate a weighted first-order spillover effect $\sum_{g_2=0}^1
\big(\E[Y_{ij}(\mathrm{a},G_{ij,1}={g}_1,G_{ij,2}={g}_{2})]-
\E[Y_{ij}(\mathrm{a},G_{ij,1}={g}_1',G_{ij,2}={g}_{2})]\big)P(G_{ij,2}={g}_{2}\mid A_{ij}=\mathrm{a})$.
However, typically $G_{ij,1}$ and $G_{ij,2}$ depend on $\mathcal{N}^1_{ij}$ and $\mathcal{N}^2_{ij}$, which are correlated. Even if $G_{ij,1}$ and $G_{ij,2}$ were  uncorrelated by design, in finite samples these two quantities may show empirical correlation
\citep{simon1954spurious}, and will lead to biased estimates.  We will show details of this in Section \ref{sec:simulation} through simulations.

\begin{corollary}[Misspecified Exposure Mapping]\label{cor:3B}
Given the conditions of Proposition \ref{prop:incorrect_interference}, if $\mathcal{I}^{*}_{ij}=\mathcal{I}_{ij}$, then
$$
\begin{aligned}
\E[Y_{ij} \mid A_{ij}
=\mathrm{a}, \mathbf{G}^{*}_{ij}=\mathbf{g}]=\!\!\!\!\!\!\!\!\!\!\!\!\sum_{\substack{ \mathbf{a}_{\mathcal{I}_{ij}} \in  \mathcal{A}_{\mathcal{I}_{ij}}: \g^{*}(\mathbf{a}_{\mathcal{I}_{ij}})=\mathbf{g}}}
\!\!\!\!\!\!\!\!\!\!\!&\E[Y_{ij} (A_{ij}=\mathrm{a}, \g(A_{\mathcal{I}_{ij}}=\mathbf{a}_{\mathcal{I}_{ij}}))] \\
&\times P(A_{\mathcal{I}_{ij}}=\mathbf{a}_{\mathcal{I}_{ij}}\mid \!A_{ij}=\mathrm{a}, \mathbf{G}^{*}_{ij}=\mathbf{g})
\end{aligned}
$$
\end{corollary}
A result similar to that of Corollary \ref{cor:3B} was already shown in \cite{aronow2017estimating} and \cite{savje2024causal}, although in a design-based setting. When the assumed interference set is well specified, i.e., $\mathcal{I}^{*}_{ij}=\mathcal{I}_{ij}$, but the exposure mapping function is not, the identified quantity is a weighted average of mean potential outcomes under the true exposure weighted by the conditional distribution of the treatments in the interference set given the value of the wrong exposure.
We  will further illustrate this setting in Section \ref{sec:simulation} using simulations.

\cite{savje2024causal} argues that the quantity identified under the wrong exposure mapping function (Corollary \ref{cor:3B}) is a meaningful quantity that can be considered as the effect of intervening to alter the value of the exposure of interest,  which may not be a result of the true interference structure but can still be seen as relevant and possible to intervene on. On the contrary, here, our focus is on the assessment of the true extent of spillover effects from different network distances, that is, the ones arising from the true interference structure \citep{weinstein2023causal, chao2025estimation}.

\section{Higher-Order Spillover Effects}
\label{sec:method}

We aim to estimate the spillover effects from sets of units at different network distances from the ego. To this end, a common approach in the empirical literature has been to limit the extent of interference to the network distance of interest and to make use of strong exposure mapping functions.
The previous section demonstrated that misspecification of the interference set or the exposure mapping function leads to identifying quantities that may not represent the true interference structure.
To mitigate such issue,
here we take a conservative approach: rather than assuming a specific exposure mapping or a restrictive interference set, we assume that each unit’s interference set $\mathcal{I}_{ij}$ is sufficiently large to include all possible sources of interference that could plausibly affect the outcome. By not imposing a functional form on how others’ treatments affect a unit, this framework prevents the bias due to misspecification of the exposure mechanism and allows for flexible estimation of spillover effects.

\subsection{Causal Estimands}
 Under
 Assumption \ref{assump:Generalized}, we define new causal estimands to quantify spillover effects from units at a specific network distance $h$.
 Denote by $\mathcal{C}^{h}_{ij} = \mathcal{N}_{i j}^{h} \cap \mathcal{I}_{ij} $ unit $ij$'s  $h^{th}$-order neighborhood within its interference set, which we still refer to as \textit{$h^{th}$-order neighborhood} with a slight abuse of terminology, and denote by $\mathcal{C}_{i\backslash(j, \mathcal{C}^{h}_{ij})} $ its complement in the interference set, i.e., $\mathcal{C}_{i\backslash(j, \mathcal{C}^{h}_{ij})}
 =\mathcal{I}_{ij} \setminus \mathcal{C}^{h}_{ij}$.
 Let $\mathbf{A}_{\mathcal{C}^{h}_{ij}}$ and $\mathbf{A}_{ \mathcal{C}_{i\backslash(j, \mathcal{C}^{h}_{ij})}}$ be the treatment vectors in these subsets and $\mathbf{a}_{\mathcal{C}^{h}_{ij}}$ and $\mathbf{a}_{ \mathcal{C}_{i\backslash(j, \mathcal{C}^{h}_{ij})}}$ their realizations.
 For each unit $ij$ and for a given $h$, given this partition of the interference set, the potential outcome $Y_{i j}(\rma, \mathbf{a}_{ \mathcal{I}_{ij}})$ defined under Assumption \ref{assump:Generalized} can be expressed as
 $Y_{i j}(\rma, \mathbf{a}_{ \mathcal{C}^h_{ij}}, \mathbf{a}_{ \mathcal{C}_{i\backslash(j, \mathcal{C}^{h}_{ij})}})$.

 Consider a hypothetical Bernoulli treatment allocation where a unit's $h$-order neighborhood within its interference set $\mathcal{C}^{h}_{ij}$ is assigned treatment with a specified probability $\alpha_h$, and the remaining units within its interference set, $\mathcal{C}_{i\backslash(j, \mathcal{C}^{h}_{ij})}$
 are assigned treatment with a potentially different probability $\alpha$. We can express the corresponding Bernoulli treatment assignments as follows:
\begin{equation}
\begin{aligned}
\label{eq:hypotheticalp}
P_{\alpha_h}\!\left(\mathbf{A}_{\mathcal{C}^{h}_{ij}}=\mathbf{a}_{\mathcal{C}^{h}_{ij}}\right)&=\prod_{ik \in \mathcal{C}^{h}_{ij}} \alpha_h^{a_{i k}}\left(1-\alpha_h\right)^{1-a_{i k}} \\
P_{\alpha}\!\left(\mathbf{A}_{ \mathcal{C}_{i\backslash(j, \mathcal{C}^{h}_{ij})}}=\mathbf{a}_{ \mathcal{C}_{i\backslash(j, \mathcal{C}^{h}_{ij})}}\right)&=\prod_{ik \in \mathcal{C}_{i\backslash(j, \mathcal{C}^{h}_{ij})}} \alpha^{a_{i k}}\left(1-\alpha\right)^{1-a_{i k}}
\end{aligned}
\end{equation}

With these hypothetical distributions,  we define the individual average potential outcome for each unit $ij$ and for a specific distance $h$, denoted by $\overline{Y}_{i j}^{h}\left(\rma, \alpha, \alpha_{h}\right)$, as follows:
\begin{equation}
\begin{aligned}
\overline{Y}_{i j}^{h}\left(\rma,\alpha, \alpha_{h} \right)=
\sum_{\mathbf{a}_{ \mathcal{C}^h_{ij}}} \sum_{\mathbf{a}_{ \mathcal{C}_{i\backslash(j, \mathcal{C}^{h}_{ij})}}}
Y_{i j}\bigl(\rma, \mathbf{a}_{ \mathcal{C}^h_{ij}}, \mathbf{a}_{ \mathcal{C}_{i\backslash(j, \mathcal{C}^{h}_{ij})}}\bigr)
P_{\alpha_{h}}\left(\mathbf{A}_{\mathcal{C}^{h}_{ij}}=\mathbf{a}_{ \mathcal{C}^h_{ij}}\right) P_{\alpha}\left(\mathbf{A}_{\mathcal{C}_{i\backslash(j, \mathcal{C}^{h}_{ij})}}=\mathbf{a}_{ \mathcal{C}_{i\backslash(j, \mathcal{C}^{h}_{ij})}}\right)
\end{aligned}
\end{equation}
This represents
the average outcome if unit $j$ in cluster $i$ received treatment $\rma$, while, in its interference set $\mathcal{I}_{ij}$, units in the $h^{th}$-order neighborhood were assigned treatment under a Bernoulli allocation strategy with probability $\alpha_h$, and the remaining of the interference set were assigned treatment with probability $\alpha$
\footnote{In settings where Assumption \ref{assump:positivity} holds only for a subset of $\mathcal{A}(n_i)$, the hypothetical treatment assignment used to define causal estimands must satisfy a nested-support condition: any treatment vector with positive probability under the hypothetical assignment must also have positive probability under the actual design. This ensures that the causal estimands are identifiable under the observed design and the inverse probability weights in our estimators remain well-defined.}
.
It is worth noting that if $\alpha_h=\alpha$ and Assumption~\ref{assump:partial} holds, this individual average potential outcome simplifies to the commonly used individual average potential outcome under partial interference \citep[e.g.,][]{tchetgen2012causal}.
Furthermore, in the special case where a unit's interference set coincides with its entire $h^{th}$-order neighborhood (Assumption \ref{assump:h-neighborhood} $\mathcal{I}_{ij}=\mathcal{N}^{h}_{ij}$), we have that $\mathcal{C}^{h}_{ij} = \mathcal{I}_{ij} $ and $\mathcal{C}_{i\backslash(j,\mathcal{C}^{h}_{ij})}=\emptyset$, and thus, the treatment vector $\mathbf{A}_{\mathcal{C}_{i\backslash(j,\mathcal{C}^{h}_{ij})}}$  is empty. In this case, we let $P_{\alpha}(\mathbf{A}_{\mathcal{C}_{i\backslash(j, \mathcal{C}^{h}_{ij})}}=\mathbf{a}_{ \mathcal{C}_{i\backslash(j, \mathcal{C}^{h}_{ij})}})=1$ for any $\alpha$.
 For each cluster $i$ and distance $h$, we define the cluster average potential outcome among units with a non-empty $h^{th}$-order neighborhoods in the interference set. Let $\mathcal{J}^{h}_{i}=\{j:  |\mathcal{C}^{h}_{ij}|\geq 1\}$ be the set of units in cluster $i$ with at least one $h$-order neighbor in the interference set, and let $N^{h}_{i}=|\mathcal{J}^{h}_{i}|$. Then the cluster average potential outcome $\overline{Y}_i^{h}\left(\rma, \alpha, \alpha_{h}\right)$ is defined as
\begin{equation}
    \overline{Y}_i^{h}\left(\rma, \alpha, \alpha_{h}\right)=\frac{1}{N^{h}_{i}}\sum_{j\in \mathcal{J}^{h}_{i}}\overline{Y}_{ij}^{h}\left(\rma, \alpha, \alpha_{h}\right)
\end{equation}
and the population average potential outcome as
\begin{equation}
\bar{Y}^{h}\left(\rma, \alpha, \alpha_{h}\right)=\mathbb{E}_{G_0}\left\{\bar{Y}_{i}^{h}\left(\rma, \alpha, \alpha_{h}\right)\right\}
\end{equation}
where the expectation is taken over the distribution $G_0$ of cluster average potential outcomes  in the super-population of clusters \citep{papadogeorgou2019causal}.
We now define the $h$-order spillover effects by comparing population average potential outcomes associated with different hypothetical treatment assignments within the $h^{th}$-order neighborhood, keeping the individual treatment and the assignment in the rest of the interference set fixed. Specifically, we focus on the comparison between two  hypothetical Bernoulli trials assigning treatment to the $h^{th}$-order neighborhood with two different probabilities, $\alpha_h$ and $\alpha^{\prime}_h$:
\begin{equation}
\label{eq:SE}
\text{SE}^{h}\left(\alpha_{h}, \alpha_{h}^{\prime} ; \rma, \alpha\right)=\bar{Y}^{h}\left(\rma, \alpha, \alpha_{h}\right)-\bar{Y}^{h}\left(\rma, \alpha, \alpha_{h}^{\prime}\right)
\end{equation}
Equation \eqref{eq:SE} can be interpreted as the spillover effect from the treatment of units in the $h^{th}$-order neighborhood. Specifically, it is the average effect on units with individual treatment equal to $\rma$ of altering the distribution of the treatment in a unit's  $h^{th}$-order neighborhood by changing the assignment probability from $\alpha_{h}^{\prime}$ to $\alpha_{h}$.
We let this effect potentially depend on the assignment probability $\alpha$ considered for the rest of the interference set, that is, $\mathcal{C}_{i\backslash(j, \mathcal{C}^{h}_{ij})}$.

\subsection{Estimators}\label{sec:estimators}
\subsubsection{Horvitz-Thompson Estimators}\label{thm:ht}
 We extend the Horvitz-Thompson estimator proposed under partial interference \citep{tchetgen2012causal, papadogeorgou2019causal} to incorporate network information in order  to estimate average potential outcomes under hypothetical Bernoulli assignments on the $h$-order neighborhood under Assumption \ref{assump:Generalized}.

 To estimate the population average potential outcome $\overline{Y}^{h}\left(\rma, \alpha, \alpha_{h}\right)$, the Horvitz-Thompson (HT) estimator is given by:
\begin{equation}
\widehat{Y}^{h}_{\text{HT}}\left(\rma, \alpha, \alpha_{h}\right)=\frac{1}{I} \sum_{i=1}^{I}\frac{1}{N^{h}_{i}} \sum_{j \in \mathcal{J}^{h}_{i}}
w_{ij}^h(\alpha, \alpha_h)\mathbbm{1}\left(A_{ij}=\rma\right) Y_{ij},
\end{equation}
 with weights  $$w_{ij}^h(\alpha, \alpha_h)
=\frac{P_{\alpha_{h}}\left(\mathbf{A}_{\mathcal{C}^{h}_{ij}}\right) P_{\alpha}\left(\mathbf{A}_{\mathcal{C}_{i\backslash(j, \mathcal{C}^{h}_{ij})}}\right)}{P_{\Delta}\left(\mathbf{A}_{ij,\mathcal{I}_{ij}}\right)}$$
The HT estimator is a weighted average of the observed outcomes among units with $h$-order neighbors. The numerator of the weight represents the probability of observing the treatment vector $\mathbf{A}_{\mathcal{I}_{ij}}$ that was actually observed under the combined hypothetical treatment allocation with probabilities $\alpha_h$ for units within the $h$-neighborhood and $\alpha$ for remaining units outside the $h$-neighborhood in the interference set. The denominator is the propensity score of unit $j$ in cluster $i$ and its interference set $P_{\Delta}\left(\mathbf{A}_{ij,\mathcal{I}_{ij}}\right)$,
that is, the probability of observing the realized treatments for unit $j$ in cluster $i$ and its interference set, denoted by $\mathbf{A}_{ij,\mathcal{I}_{ij}}$, under the actual randomization design $\Delta$. Under a two-stage randomized experiment, this is given by Equation \eqref{eq:twostage} restricted to the subset $(ij,\mathcal{I}_{ij})$.
In this way, the HT estimator will give a higher weight to the observations whose interference set has treatments that are more likely to occur under the hypothetical treatment allocation strategy and unlikely under the actual design.
Given the Horvitz-Thompson estimator of $\widehat{Y}^h_{\text{HT}}\left(\mathrm{a}, \alpha, \alpha_{h}\right)$, the $h$-order spillover effects can be estimated as follows:
$$
\widehat{\text{SE}}^{h}_{\text{HT}}\left(\alpha_{h}, \alpha_{h}^{\prime} ; \rma, \alpha\right)=\widehat{Y}^h_{\text{HT}}\left(\mathrm{a}, \alpha, \alpha_{h}\right)-\widehat{Y}^h_{\text{HT}}\left(\mathrm{a}, \alpha, \alpha^{\prime}_{h}\right)
$$
It can be shown that the HT estimator is unbiased (see Appendix A.3 for the proof).
However, it may suffer from high variance, especially when dealing with large interference sets and extreme probabilities that lead to large weights.

\subsubsection{Hajek Estimators}
The Hajek estimator offers a refinement of the Horvitz-Thompson estimator, providing a reduced variance in many practical settings while preserving consistency \citep{liu2016inverse, aronow2017estimating}.
Under Assumption \ref{assump:Generalized}, the Hajek estimator is defined as follows:
$$
\begin{aligned}
&
\widehat{Y}^h_{\text{Hajek}}\left(\mathrm{a}, \alpha, \alpha_{h}\right)=\frac{ \sum_{i=1}^{I} \sum_{j \in \mathcal{J}^{h}_{i}}
w_{ij}^h(\alpha, \alpha_h)
\mathbbm{1}\left(A_{i j}=\mathrm{a}\right) Y_{i j}}{ \sum_{i=1}^{I} \sum_{j \in \mathcal{J}^{h}_{i}} w_{ij}^h(\alpha, \alpha_h)\mathbbm{1}\left(A_{i j}=\mathrm{a}\right)}
\end{aligned}
$$
 with weights  $$w_{ij}^h(\alpha, \alpha_h)
=\frac{P_{\alpha_{h}}\left(\mathbf{A}_{\mathcal{C}^{h}_{ij}}\right) P_{\alpha}\left(\mathbf{A}_{\mathcal{C}_{i\backslash(j, \mathcal{C}^{h}_{ij})}}\right)}{P_{\Delta}\left(\mathbf{A}_{ij,\mathcal{I}_{ij}}\right)}$$
The key difference between the Hajek estimator and the Horvitz-Thompson estimator lies in the normalization approach. Specifically, the Hajek estimator can be viewed as a ratio estimator, a normalized version of an individual-weighted Horvitz-Thompson estimator, where the sum of the weights replaces $\sum_i N^h_i$ in the denominator. It is important to note that our Horvitz-Thompson estimator presented earlier in section \ref{thm:ht} is cluster-weighted, which is more efficient than the individual-weighted version \citep{papadogeorgou2019causal}. The Hajek estimator normalizes the weights by dividing by their sum rather than by the number of units, which helps to stabilize the estimator, making it more robust in scenarios with extreme probabilities, particularly for large interference sets where the Horvitz-Thompson estimator may exhibit high variance.
While the Hajek estimator is not unbiased in finite samples, it is consistent and asymptotically unbiased, and often provides a more efficient approach to estimating spillover effects in the presence of interference for large samples \citep{aronow2017estimating, papadogeorgou2019causal}.

Given the estimated average potential outcomes using the Hajek estimator, the $h$-order spillover effects can be estimated as follows:
$$
\widehat{\text{SE}}^{h}_{\text{Hajek}}\left(\alpha_{h}, \alpha_{h}^{\prime} ; \rma, \alpha\right)=\widehat{Y}^h_{\text{Hajek}}\left(\rma, \alpha, \alpha_{h}\right)-\widehat{Y}^h_{\text{Hajek}}\left(\rma, \alpha, \alpha^{\prime}_{h}\right)
$$

\subsubsection{Weighted Least Squares Estimator for Marginal Structural Model}
\label{sec:wls_estimator}
The Horvitz-Thompson and Hajek estimators are both non-parametric estimators. We also propose a weighted least squares estimator (WLS) that relies on Horvitz-Thompson weights while improving its efficiency by modeling the relationship between the outcome and $\alpha_h$.
We assume a marginal structural model (MSM) for the potential outcome given the individual treatment $\rma$, and the hypothetical Bernoulli distribution parameters $\alpha$ and $\alpha_h$:
Here, we show the estimator under a linear model, but more flexible non-linear models can also be considered. Specifically, we posit the following model:
\begin{equation}
\label{eq:msm}
\E[\overline{Y}_{i j}^{h}(\mathrm{a},\alpha,\alpha_h)]=\gamma^{h,\alpha}_0+\gamma^{h,\alpha}_1 \mathrm{a}+\gamma^{h,\alpha}_2 \alpha_h + \gamma^{h,\alpha}_3 \mathrm{a}\cdot \alpha_h
\end{equation}
where we assume linearity with $\alpha_h$
but leave the coefficient $\alpha$-specific. Stronger parametric assumptions on $\alpha$ could also be imposed. Let
$\boldsymbol{\gamma}^{h,\alpha}=(\gamma^{h,\alpha}_0,\gamma^{h,\alpha}_1, \gamma^{h,\alpha}_2, \gamma^{h,\alpha}_3)^T$.
Under this model, the $h$-order spillover effects are given by:
$$
\text{SE}^h(\alpha_h, \alpha_h^{\prime};\rma, \alpha)= \gamma^{h,\alpha}_2 (\alpha_h -\alpha_h^{\prime}) + \gamma^{h,\alpha}_3 \rma (\alpha_h -\alpha_h^{\prime})
$$
 To estimate the model parameters, we use a data expansion approach where each observation is replicated with different hypothetical values of $\alpha_h$ with its associated inverse probability weight \citep{lange2012simple, Haneuse2013EstimationOT}. This allows us to fit, for each $h$ and $\alpha$ of interest, a single weighted regression that captures how the potential outcomes vary with $\alpha_h$.
Given $h$ and $\alpha$, the parameter estimation proceeds through the following steps:
\begin{enumerate}
\item Data expansion: Given $K$ values of $\alpha_h\in \mathcal{A}_h$, with $K=|\mathcal{A}_h|$, say $\mathcal{A}_h=\{0.1, 0.2, 0.3, ..., 0.9\}$, we create an extended dataset composed
of $K\times N$ records, with $K$ replicates for each unit in the sample. Replicate $k$ of unit $ij$ contains copies
of the variables $A_{ij}$ and $Y_{ij}$, and a pseudo covariate $S_{ijk}$ being the k-th value of $\mathcal{A}_h$, i.e., $S_{ijk}=\mathcal{A}_h[k]$.
Each replicate then represents the hypothetical Bernoulli assignment on the $h$-order neighborhood with probability $S_{ijk}$.
\item Weight calculation: For each record in the expanded dataset, compute the weight $w_{i j k}^{h,\alpha}$, using the Horvitz-Thompson weights defined in Section \ref{thm:ht}, but with an indicator for having an $h$-order neighborhood:
$$ w^{h, \alpha}_{ijk}=w^{h}_{ij}(\alpha, S_{ijk})\mathbbm{1}(|\mathcal{C}^{h}_{ij}|\geq 1)=\frac{P_{\alpha}\left(\mathbf{A}_{\mathcal{C}_{i\backslash(j, \mathcal{C}^{h}_{ij})}}\right) P_{S_{ijk}}\left(\mathbf{A}_{ \mathcal{C
}_{i j}^{h}}\right)}{P_{\Delta}\left(\mathbf{A}_{ij,\mathcal{I}_{ij}}\right)}\mathbbm{1}(|\mathcal{C}^{h}_{ij}|\geq 1)$$
\item Estimating Equation: Define $\mathbf{O}_i = \{\mathbf{Y}_i, \mathbf{A}_i, \mathbf{S}_i, \mathcal{C}^h_{ij}\}_{j=1}^{n_i}$ and $\mathbf{S}_i = (S_{i1k})_{j=1,\ldots,n_i; k=1,\ldots,K}$. We obtain the estimator $\widehat{\boldsymbol{\gamma}}^{h,\alpha}$ by solving the following estimating equation corresponding to the weighted least squares minimization:
$$
\sum_{i=1}^{I} \boldsymbol{\psi}_{\text{WLS}, i}(\mathbf{O}_i;\widehat{\boldsymbol{\gamma}}^{h,\alpha}) = \mathbf{0}
$$
where the score function for cluster $i$ is given by
$$
\boldsymbol{\psi}_{\text{WLS}, i}(\mathbf{O}_i; \boldsymbol{\gamma}^{h,\alpha}) = \sum_{k=1}^{K} \sum_{j=1}^{n_i} \left( Y_{ij} - \mathbf{C}_{ijk}^T \boldsymbol{\gamma}^{h,\alpha} \right) \mathbf{C}_{ijk} w^{h, \alpha}_{ijk},
$$
with $\mathbf{C}_{ijk} = (1, A_{ij}, S_{ijk}, A_{ij} \cdot S_{ijk})^T$.

\item Spillover effect estimation:  The estimated $h$-order spillover effect is given by:
$$
\widehat{\text{SE}}^h_{\text{WLS}}(\alpha_h, \alpha_h^{\prime}; \rma, \alpha)= \widehat{\gamma}^{h,\alpha}_2 (\alpha_h -\alpha_h^{\prime}) + \widehat{\gamma}^{h,\alpha}_3 \rma (\alpha_h -\alpha_h^{\prime}) $$
\end{enumerate}
Compared to the HT and Hajek estimators, under the parametric or semi-parametric  assumptions of the MSM, this weighted least squares estimator can provide more precise estimates when the assumed functional form is correct or closely approximates the true relationship between the variables.
As opposed to common parametric or semi-parametric estimators for outcome models  under exposure mapping functions, which make assumptions on the relationship between the outcome and all the variables resulting from the exposure mapping function--e.g., $G_{ij,1}$ and $G_{ij,2}$, as in the example below Corollary~\ref{cor:3A}--,
here the proposed WLS estimator just relies on the assumed relationship between the outcome and $\alpha_h$ for the $h$ of interest, while the treatments in the rest of the interference set are used only for the calculation of the weights. In this way, the WLS estimator
avoids bias from both misspecified exposure mappings and restrictive assumptions on interference sets, which are commonly needed when their treatments  enter as independent variables in the outcome model rather than in the weights.

\subsubsection{Asymptotic Properties of the Estimators}

We derive the asymptotic properties of the estimators as the number of clusters grows using M-estimation theory \citep[e.g.,][]{van2000asymptotic, stefanski2002mestimation}

\begin{proposition}[Asymptotic Properties of Estimators]
\label{prop:general_asymptotics}
Let $\ell \in \{\text{HT, Hajek, WLS}\}$ denote the estimator type.  Let $\mathbf{O}_i$ denote the observed data for cluster $i$, and let ${\boldsymbol{\theta}}_{\ell}$ be the parameter of interest for estimator $\ell$. Let $\widehat{\boldsymbol{\theta}}_{\ell}$ be the estimator satisfying the estimating equations $\sum_{i=1}^I \boldsymbol{\psi}_{\ell, i}(\mathbf{O}_i; \widehat{\boldsymbol{\theta}}_{\ell}) = \mathbf{0}$. Under Assumptions \ref{assump:positivity}-\ref{assump:Generalized} and standard regularity conditions, as $I \rightarrow \infty$:
$$
\sqrt{I} \left( \widehat{\boldsymbol{\theta}}_{\ell} - \boldsymbol{\theta}_{\ell} \right) \xrightarrow{d} N\left(\mathbf{0}, \bm{\Sigma}_{\ell} \right)
$$
where the asymptotic variance is given by the sandwich formula $\bm{\Sigma}_{\ell} = \mathbf{A}_{\ell}^{-1} \mathbf{B}_{\ell} (\mathbf{A}_{\ell}^{-1})^T$, with $\mathbf{A}_{\ell} = -\mathbb{E}_{G_0}\left[\frac{\partial}{\partial \boldsymbol{\theta}_{\ell}^T} \boldsymbol{\psi}_{\ell, i}(\mathbf{O}_i; \boldsymbol{\theta}_{\ell})\right]$ and $\mathbf{B}_{\ell} = \mathbb{E}_{G_0}[\boldsymbol{\psi}_{\ell, i}(\mathbf{O}_i; \boldsymbol{\theta}_{\ell}) \boldsymbol{\psi}_{\ell, i}(\mathbf{O}_i; \boldsymbol{\theta}_{\ell})^T]$. The estimating functions $\boldsymbol{\psi}_{\ell, i}$ for each estimator are defined as follows:

\begin{enumerate}
    \item Horvitz-Thompson ($\ell=\text{HT}$): The observed data for cluster $i$ is $\mathbf{O}_i = \{\mathbf{Y}_i, \mathbf{A}_i, \mathcal{C}^h_{ij}\}_{j=1}^{n_i}$, the parameter is $\boldsymbol{\theta}_{\text{HT}} = \left(\overline{Y}^h(\rma, \alpha, \alpha_h), \overline{Y}^h(\rma, \alpha, \alpha'_h)\right)^T$, and the estimating function is:
    $$
    \boldsymbol{\psi}_{\text{HT}, i}(\mathbf{O}_i;\boldsymbol{\theta}_{\text{HT}}) =
    \begin{pmatrix}
    \widehat{Y}^h_{\text{HT}, i}(\rma, \alpha, \alpha_h) - \theta_{1,\text{HT}}  \\
    \widehat{Y}^h_{\text{HT}, i}(\rma, \alpha, \alpha'_h) - \theta_{2,\text{HT}}
    \end{pmatrix}
    $$
    where $\widehat{Y}^h_{\text{HT}, i}(\rma, \alpha, \alpha_h) = \frac{1}{N^{h}_{i}} \sum_{j \in \mathcal{J}^{h}_{i}} w_{ij}^h(\alpha, \alpha_h) \mathbbm{1}(A_{ij}=\rma) Y_{ij}$.
    \item Hajek ($\ell=\text{Hajek}$): The observed data for cluster $i$ is $\mathbf{O}_i = \{\mathbf{Y}_i, \mathbf{A}_i, \mathcal{C}^h_{ij}\}_{j=1}^{n_i}$, the parameter is $\boldsymbol{\theta}_{\text{Hajek}} = \left(\overline{Y}^h(\rma, \alpha, \alpha_h), \overline{Y}^h(\rma, \alpha, \alpha'_h)\right)^T$, and the estimating function is:
    $$
    \boldsymbol{\psi}_{\text{Hajek}, i}(\mathbf{O}_i;\boldsymbol{\theta}_{\text{Hajek}}) =
    \begin{pmatrix}
    \widehat{W}^h_{i}(\rma, \alpha, \alpha_h) - \widehat{N}^h_{i}(\rma, \alpha, \alpha_h) \theta_{1, \text{Hajek}} \\
    \widehat{W}^h_{i}(\rma, \alpha, \alpha'_h) - \widehat{N}^h_{i}(\rma, \alpha, \alpha'_h) \theta_{2,\text{Hajek} }
    \end{pmatrix}
    $$
    where $\widehat{W}^h_{i} = \sum_{j \in \mathcal{J}^{h}_{i}} w_{ij}^h(\alpha, \alpha_h) \mathbbm{1}(A_{ij}=\rma) Y_{ij}$ and $\widehat{N}^h_{i} = \sum_{j \in \mathcal{J}^{h}_{i}} w_{ij}^h(\alpha, \alpha_h) \mathbbm{1}(A_{ij}=\rma)$.

    \item Weighted Least Squares ($\ell=\text{WLS}$):  The observed data for cluster $i$ is $\mathbf{O}_i = \{\mathbf{Y}_i, \mathbf{A}_i, \mathbf{S}_i, \mathcal{C}^h_{ij}\}_{j=1}^{n_i}$,
    the parameter is the vector of coefficients $\boldsymbol{\theta}_{\text{WLS}} = \boldsymbol{\gamma}^{h,\alpha}$, and the estimating function is:
    $$
    \boldsymbol{\psi}_{\text{WLS}, i}(\mathbf{O}_i;\boldsymbol{\theta}_{\text{WLS}}) = \sum_{k=1}^{K} \sum_{j=1}^{n_i} \left( Y_{ijk} - \mathbf{C}_{ijk}^T \boldsymbol{\theta}_{\text{WLS}} \right) \mathbf{C}_{ijk} w^{h, \alpha}_{ijk}
    $$
    where definitions of $\mathbf{S}_i$ and $w^{h, \alpha}_{ijk}$ follow Section \ref{sec:wls_estimator}.
\end{enumerate}
\end{proposition}
\begin{proof}
We derive the asymptotic properties separately for each estimator in the Appendix. See Sections A.4, A.5, and A.6 for the proofs of the Horvitz-Thompson, Hajek, and WLS estimators, respectively.
\end{proof}

We estimate the asymptotic variance $\bm \Sigma_{\ell} $
by taking empirical expectations,
i.e.,
replacing $\E_{G_0}\{\cdot\}$ with the sample average and $\boldsymbol{\theta}$ with the estimator $\widehat{\boldsymbol{\theta}}_{\ell}$.
This estimator is consistent under standard regularity conditions \citep{iverson1989convergence}.

For the spillover effect estimator, the asymptotic variance follows from the delta method applied to the joint asymptotic distribution of
$\widehat{Y}^{h}_{\ell}(\rma,\alpha,\alpha_h)$ and
$\widehat{Y}^{h}_{\ell}(\rma,\alpha,\alpha_h')$.
A consistent plug-in estimator is obtained from the empirical sandwich variance. We provide the plug-in variance formulas for the HT, Hajek, and WLS spillover effect estimators in the Appendix A.4,  A.5, and  A.6.

\section{Simulation Study}\label{sec:simulation}
    We conduct a Monte-Carlo simulation study to i) show the bias of restrictive assumptions on the interference set and on the exposure mapping function, as well as to ii) evaluate the performance of our three weighted estimators, Horvitz-Thompson, Hajek, and WLS, under various interference scenarios and random graph structures.
     \subsection{Data Generation}\label{sec:simulation_1}
     Our simulation study implements a two-stage randomized experimental design, as described in Section \ref{sec:design}, over 5000 simulations. Each simulation contains $I=100$ clusters, with $n_i=10$ units per cluster. The first stage assigns clusters to 2 probabilities $\Delta=\{\delta_1, \delta_2 \}$, with clusters equally distributed between the two conditions ($I_1=I_2=50$), and the second stage assigns treatment to units according to a Bernoulli randomization with the first-stage probability.
     Outcomes are generated according to the following model:
    \begin{equation}
    \label{eq:outcome_model}
    Y_{ij}=\beta_{1}+\beta_{2} A_{ij}+\beta_{3} G_{ij,1}+\beta_{4} G_{ij,2}+\beta_{5} G_{ij,1}G_{ij,2}+\epsilon_{ij}
\end{equation}
where $\epsilon_{ij}\sim N(0, 1)$, and $G_{ij,1}$ and $G_{ij,2}$ are the proportion of treated first- and second-order neighbors, respectively.

    We consider four interference scenarios, each defined by different choices of the coefficients in Equation \eqref{eq:outcome_model}, ranging from no interference to interference with higher-order interactions.
    \begin{itemize}
    \setlength{\itemindent}{3em}
    \item[{Scenario 1}:] No interference. A unit's outcome depends only on individual's treatment assignment $A_{ij}$, with $\beta_{1}=2$, $\beta_{2}=5$, and $\beta_{3}=\beta_{4}=\beta_{5}=0$.
    \item[{Scenario 2}:] First-order neighborhood interference. A unit's outcome depends on the individual's treatment $A_{ij}$ and the proportion of treated first-order neighbors,  $G_{ij,1}$ (Assumption \ref{assump:neighborhood}). Here, we set $\beta_{1}=2$, $\beta_{2}=5$, $\beta_{3}=10$, and $\beta_{4}=\beta_{5}=0$.
    \item[{Scenario 3:}] Second-order neighborhood interference. A unit's outcome depends on the individual's treatment $A_{ij}$, the proportion of treated first-order neighbors $G_{ij,1}$, and the proportion of treated second-order neighbors $G_{ij,2}$ (Assumption \ref{assump:h-neighborhood} with $k=2$). Here, we set $\beta_{1}=2$, $\beta_{2}=5$, $\beta_{3}=10$, $\beta_{4}=3$, and $\beta_{5}=0$.
    \item[{Scenario 4:}] Second-order neighborhood interference with interaction. This scenario extends Scenario 3, including in Equation \eqref{eq:outcome_model} the interaction between $G_{ij,1}$ and $G_{ij,2}$, with $\beta_{5}=5$.
    \end{itemize}
    In all scenarios, to evaluate the estimators' performance, we vary $\Delta=\{\delta_1, \delta_2 \}\in\{(0.2, 0.8)$,\\$ (0.3, 0.7), (0.4, 0.6)\}$.
    Regarding the network structures, for Scenarios 1, 2, and 4, we generate networks using the regular graph model \citep{wormald1999models} with a fixed number of connections equal to $n_ip$, with a fixed link probability $p=0.2$. To assess the estimators' sensitivity to the network density, in Scenario 3 we also vary both the graph generation method, comparing the Erd\H{o}s--R\'enyi random graph model \citep{erdos1959random} and the regular graph model, and the link probability $p$, considering values in the set $\{0.2,0.3,0.4,0.5,0.6,0.7,0.8 \}$.

    Under Equation \eqref{eq:outcome_model}, for each scenario, we calculate the theoretical $h$-order spillover effects, $\text{SE}^{h}\left(\alpha_{h}, 0.5;\rma, \alpha\right)$, for $h=1, 2$, where we fix $\alpha_{h}^{\prime}=0.5$ and vary $\alpha_h\in\{0.3, 0.4, 0.6, 0.7\}$ and $\alpha\in\{0.3, 0.4, 0.5, 0.6, 0.7\}$ . For Scenario 4, the first-order spillover effects incorporate the probability of lacking second-order neighbors, which we estimated via Monte Carlo simulation consistent with our graph generation algorithm (see Appendix B.1.4 for details).
    To evaluate the performance of the estimators across different values of $\alpha_h$ on a common scale, we focus on the normalized $h$-order spillover effects
$\frac{\text{SE}^{h}\left(\alpha_{h}, 0.5 ; \rma, \alpha\right)}{(\alpha_h-0.5)},$
  representing the marginal spillover effect per unit increase in $\alpha_h$ relative to the baseline of 0.5. Accordingly, we report
  the bias and variance of this normalized quantity.

For each simulation, we estimate the first- and second-order spillover effects using the proposed HT, Hajek, and WLS estimators. To evaluate the impact of misspecified interference sets, we also employ Ordinary Least Squares (OLS) estimators under different specifications. Across all scenarios, when estimating the first-order spillover effect ($h=1$), we use two standard OLS estimators:
(1) $\text{OLS}(G_1, G_2)$, which regresses the outcome on $A_{i j}, G_{i j, 1}$, and $G_{i j, 2}$ (representing the correct model in Scenario 3); and
(2) $\text{OLS}(G_1)$, which regresses the outcome on $A_{i j}$ and $G_{i j, 1}$ only. The estimator $\text{OLS}(G_1)$ is included to illustrate omitted variable bias due to an incorrectly specified interference set, as described in Corollary~\ref{cor:3A}.
When using OLS estimators, the estimated spillover effect from the
$h$-order neighborhood is obtained as the coefficient on the corresponding exposure mapping variable.
It is important to note that while $\text{OLS}(G_1, G_2)$ is correctly specified for Scenario 3, it represents a misspecified exposure mapping function for Scenario 4 (due to interaction terms).
In Scenario 3, we further illustrate bias due to misspecified exposure mapping (Corollary \ref{cor:3B}) by introducing a naive binary exposure estimator, denoted by $\text{OLS}(G^*_1, G_2)$, which assumes a binary exposure mapping $G^*_{i j, 1}=\mathbbm{1}\left(G_{i j, 1} \geq 0.5\right)$ and regresses the outcome on $A_{ij}, G^*_{i j, 1}$ and $G_{i j, 2}$. This estimator is only used when targeting first-order spillover effects ($h=1$).

We conduct 5,000 Monte Carlo simulations per scenario to obtain
Monte Carlo estimates of the normalized bias and variance across these simulations.
For each
simulation, we also draw 500 bootstrap samples and compute the bootstrap variance of each estimator.
\begin{figure*}[t]
    \centering
    \begin{subfigure}{0.5\linewidth}
        \includegraphics[height=5.5cm]{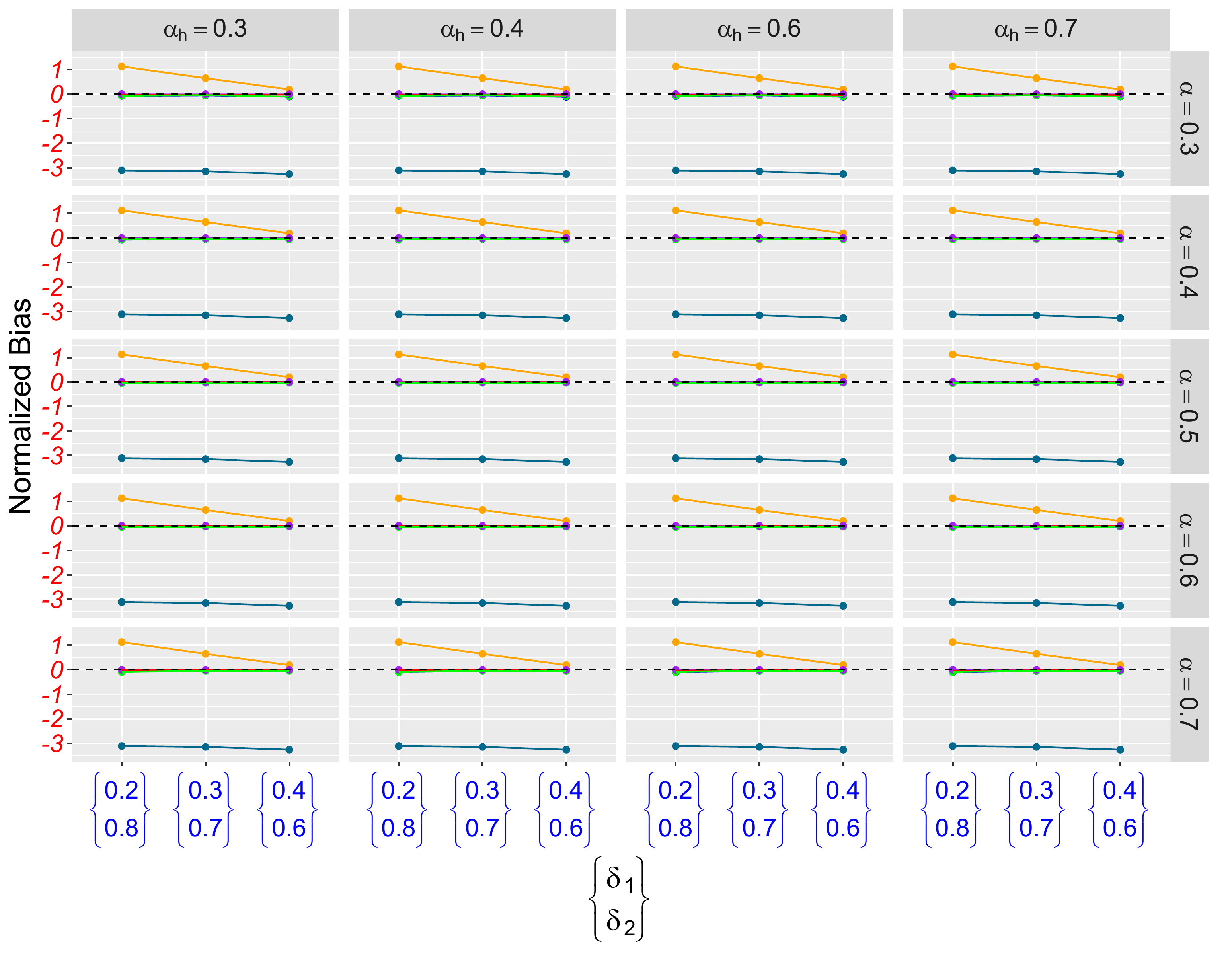}
         \caption{$\rma=0$, $h=1$}
     \label{img1}
    \end{subfigure}
    \begin{subfigure}{0.45\linewidth}
        \includegraphics[height=5.5cm]{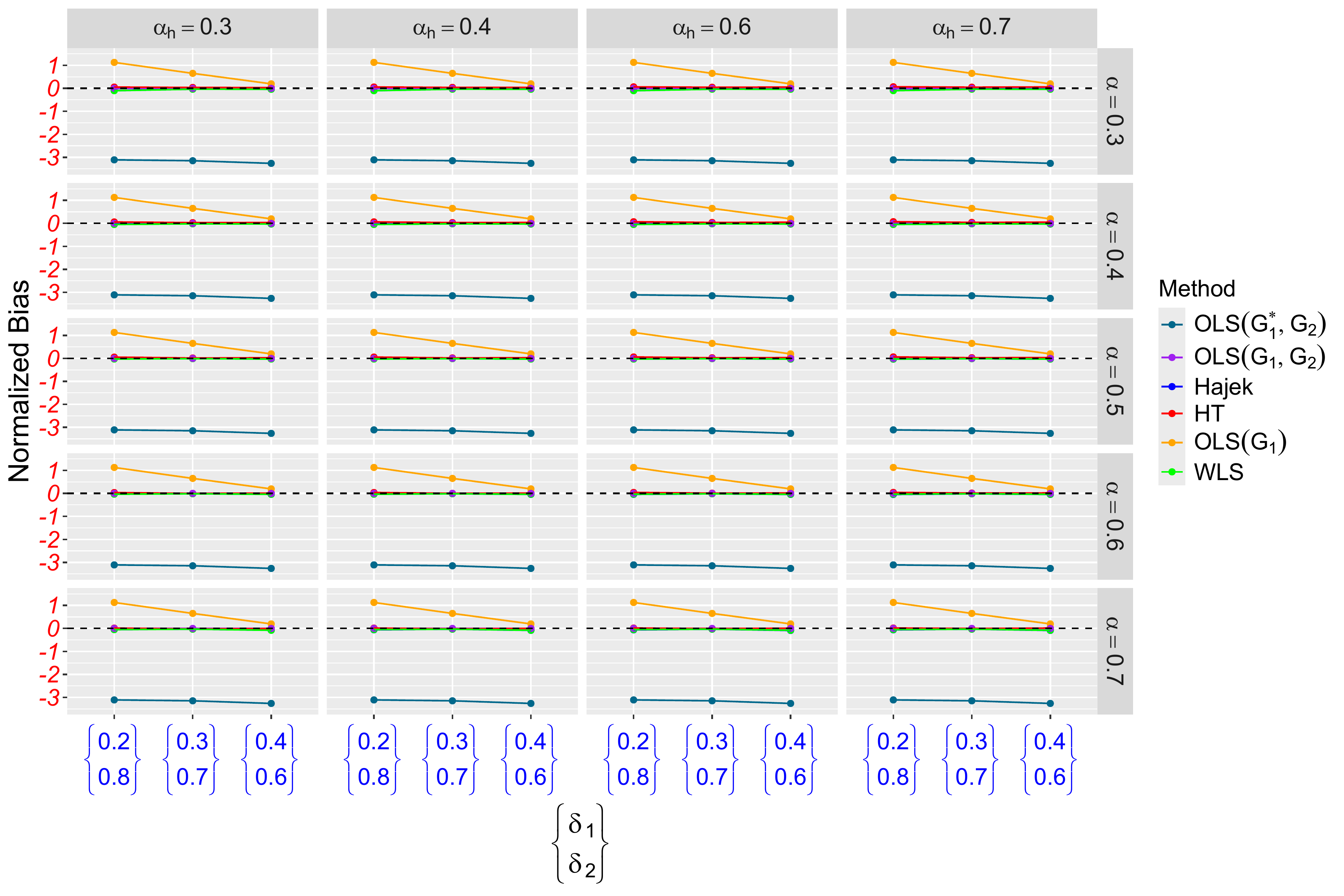}
        \caption{$\rma=1$, $h=1$}
    \label{img2}
    \end{subfigure}
\caption{Scenario 3 (second-order neighborhood interference). Normalized bias of the HT (red), Hajek (blue), and the WLS (green) estimators, along with three OLS estimators: $\text{OLS}(G_1)$ (orange),  $\text{OLS}(G_1, G_2)$ (purple), and $\text{OLS}(G^*_1, G_2)$ (deep blue), where $G^*_{ij,1}=\mathbbm{1}(G_{ij,1}\geq 0.5)$, for the first-order spillover effects $\text{SE}^{1}\left(\alpha_{h}, 0.5;\rma, \alpha\right)$, with $\alpha_h\in\{0.3, 0.4, 0.6, 0.7\}$ and $\alpha\in\{0.3, 0.4, 0.5, 0.6, 0.7\}$, under Scenario 3, a two-stage assignment with $\{\delta_1, \delta_2 \}\in\{(0.2, 0.8), (0.3, 0.7), (0.4, 0.6)\}$, and regular networks with $p=0.2$. }
\label{fig:scenario3}
\end{figure*}
\subsection{Results}
\subsubsection{Bias}
Here, we study the bias of all the estimators when one's outcome depends on the treatments of both first- and second-order neighborhoods. Specifically, we report the normalized bias for the first-order spillover effects under Scenario 3 (Figure \ref{fig:scenario3}), and second-order spillover effects under Scenario 4 (Figure \ref{fig:scenario4}), both using a regular graph model with $p=0.2$.
Bias results for other scenarios are detailed in the Appendix B.

\begin{figure*}[t]
    \centering
    \begin{subfigure}{0.5\linewidth}
        \includegraphics[height=5.5cm]{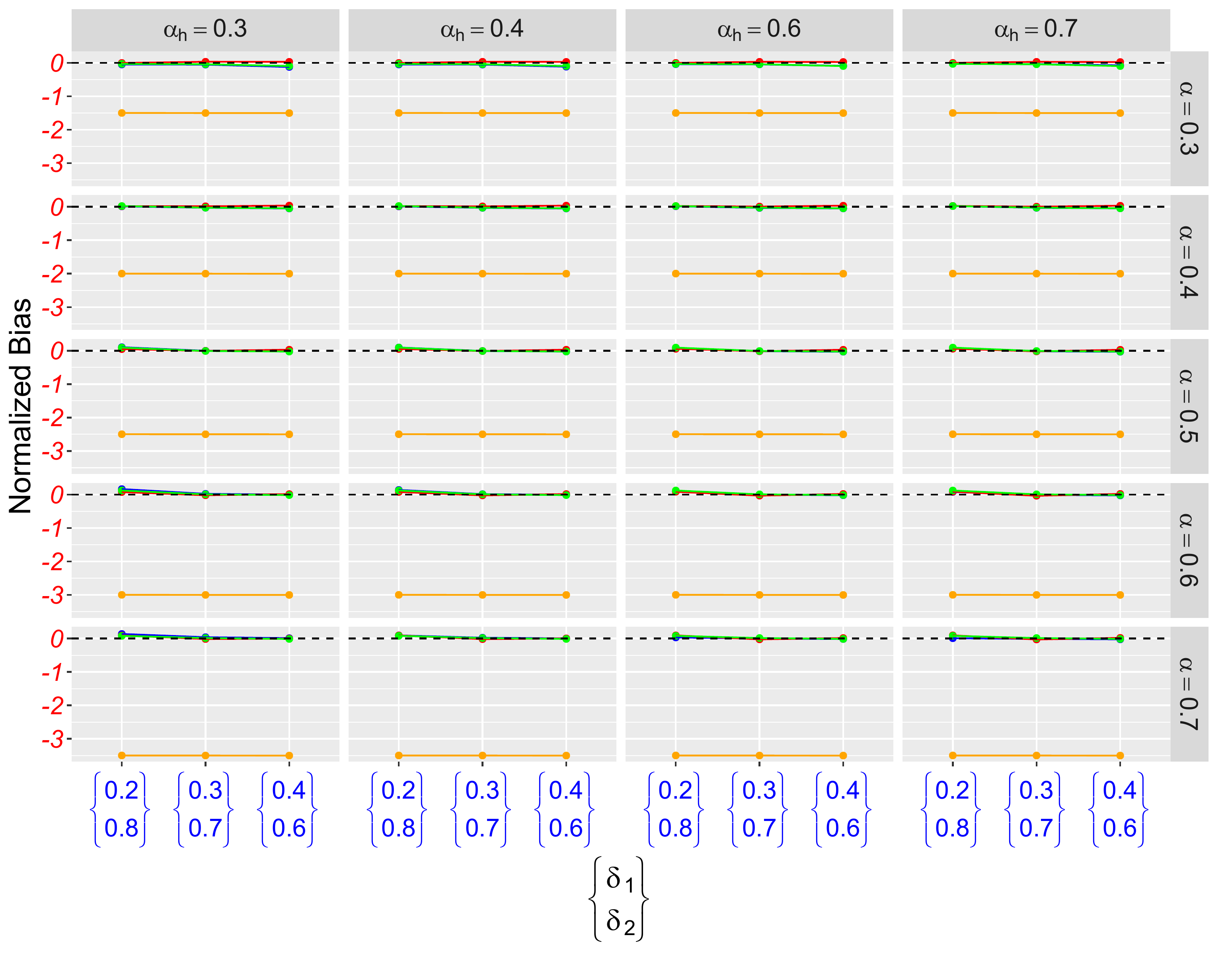}
         \caption{$\rma=0$, $h=2$}
    \label{img3}
    \end{subfigure}
    \begin{subfigure}{0.45\linewidth}
        \includegraphics[height=5.5cm]{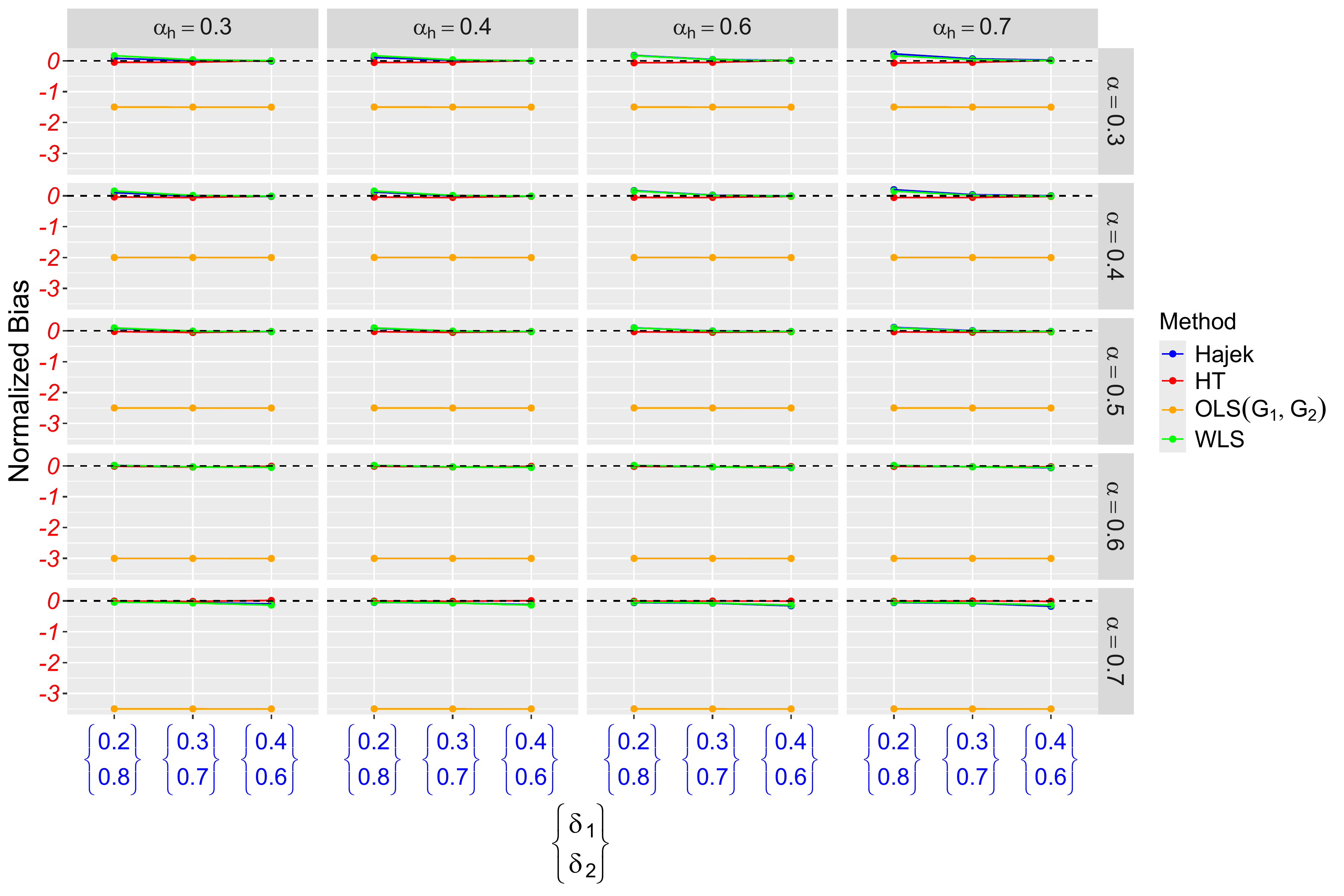}
        \caption{$\rma=1$, $h=2$}
       \label{img4}
    \end{subfigure}
\caption{Scenario 4 (second-order neighborhood interference with interaction). Normalized bias of the HT (red), Hajek (blue), and WLS (green) estimators, along with the OLS estimator with both $G_{ij,1}$ and $G_{ij,2}$, i.e., $\text{OLS}(G_1, G_2)$ (orange), for the second-order spillover effects $\text{SE}^{2}\left(\alpha_{h}, 0.5;\rma, \alpha\right)$ among untreated units ($\rma=0$, left) and treated units ($\rma=1$, right). Parameters include $\alpha_h\in\{0.3, 0.4, 0.6, 0.7\}$ and $\alpha\in\{0.3, 0.4, 0.5, 0.6, 0.7\}$, under Scenario 4, with a two-stage assignment $\{\delta_1, \delta_2 \}\in\{(0.2, 0.8), (0.3, 0.7), (0.4, 0.6)\}$ and regular networks ($p=0.2$).}
\label{fig:scenario4}
\end{figure*}

In Figure \ref{fig:scenario3} we see that when estimating first-order spillover effects ($h=1$), the estimator $\text{OLS}(G_1)$ exhibits bias. Although the treatment is assigned independently in the second stage, making  $G_{ij,1}$ and $G_{ij,2}$ independent within clusters conditional on the cluster's assigned probability, the first-stage randomization induces a correlation between them across clusters, as the distribution of $G_{ij,1}$ and $G_{ij,2}$ depends on the first-stage assignment.
 Specifically, clusters assigned a higher first-stage probability ($\delta_2$) systematically exhibit higher expected values for both $G_{ij,1}$ and $G_{ij,2}$, while those assigned $\delta_1$ exhibit lower values for both. When data is pooled across these distinct clusters, a positive association emerges, a phenomenon akin to Simpson's Paradox (see Appendix B.3).
Therefore, when $G_{ij,2}$ is omitted from the regression, the coefficient on $G_{ij,1}$ absorbs part of the effect of the omitted variable $G_{ij,2}$, leading to omitted variable bias as described in Corollary \ref{cor:3A}.
The magnitude of this bias decreases as the first-stage probabilities $(\delta_1, \delta_2)$ become more similar, because the correlation weakens (see Appendix B.3).
Additionally, the estimator $\text{OLS}(G^*_1, G_2)$ shows significant bias. While it includes the correct interference set through $G_{ij,2}$, it uses a misspecified binary exposure mapping $G^*_{ij,1}=\mathbbm{1}(G_{ij,1}\ge 0.5)$ instead of the true continuous exposure $G_{ij,1}$, leading to bias as described in Corollary~\ref{cor:3B}.
In contrast, the $\text{OLS}(G_1, G_2)$ estimator
shows a negligible bias, because in Scenario 3 both the interference set and the exposure mapping functions are correctly specified.

 However, in Scenario 4, where the outcome depends on the interaction between the proportions of treated first- and second-order neighbors ($G_{ij,1} G_{ij,2}$), the $\text{OLS}(G_1, G_2)$ estimator becomes misspecified. By regressing $Y_{i j}$ only on the main effects $A_{i j}, G_{ij,1}$, and $G_{ij,2}$ without the interaction term, this estimator relies on an incorrect exposure mapping function, as discussed in Corollary \ref{cor:3B}, and leads to the bias observed in Figure \ref{fig:scenario4}.

In contrast, our proposed Horvitz-Thompson, Hajek, and WLS estimators remain unbiased across all scenarios and settings
because they do not rely on exposure mapping specifications or restrictions on the interference set.

\subsubsection{Variance estimators}
To validate our analytical variance estimators derived in Section~\ref{sec:estimators}, we compare them with bootstrap and Monte Carlo variance estimates. The results, presented in Appendix B.2.1, confirm that our analytical variance estimates are consistent with both the bootstrap and Monte Carlo estimates across all scenarios.

\subsubsection{Comparison of the variance of different estimators}
Here, we compare the analytical variance of Horvitz-Thompson, Hajek, and WLS estimators, along with the OLS estimators, when estimating the normalized spillover effects for $h=1$ and $h=2$ for all values of $\alpha_h$ (against 0.5) and all values of $\alpha$, under all combinations of $\delta_1$ and $\delta_2$. Figure \ref{fig:ana_variance} shows the results for untreated units ($\rma=0$)  under Scenario 3 with regular graphs ($p=0.2$).
For other scenarios, the analytical variance of the estimators is reported in the Appendix B.2.2.
\begin{figure*}[htbp]
    \centering
    \begin{subfigure}{0.5\linewidth}
        \includegraphics[height=5.5cm]{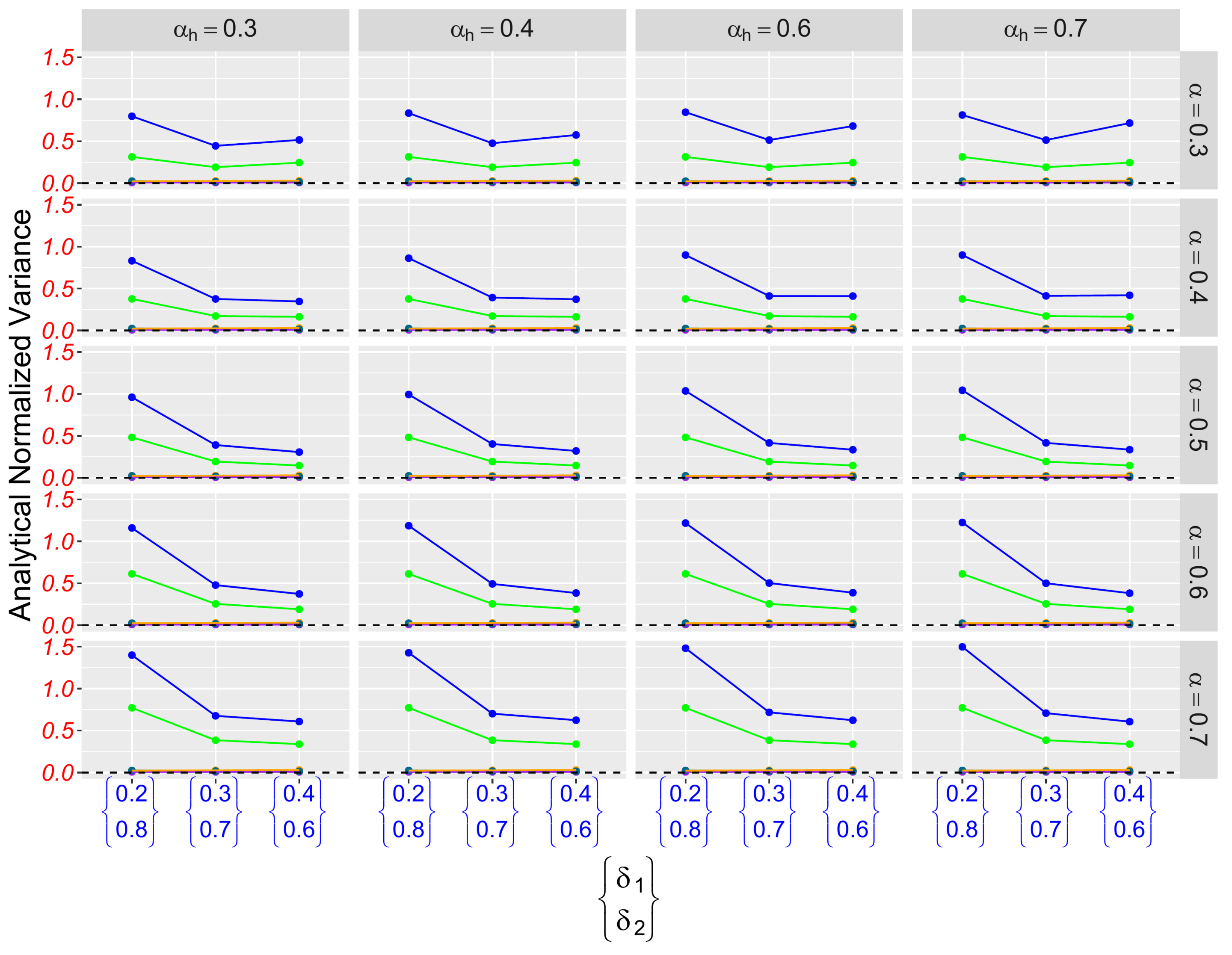}
        \caption{Without HT, $\rma=0$, $h=1$}
        \label{img7a}
    \end{subfigure}
    \begin{subfigure}{0.45\linewidth}
        \includegraphics[height=5.5cm]{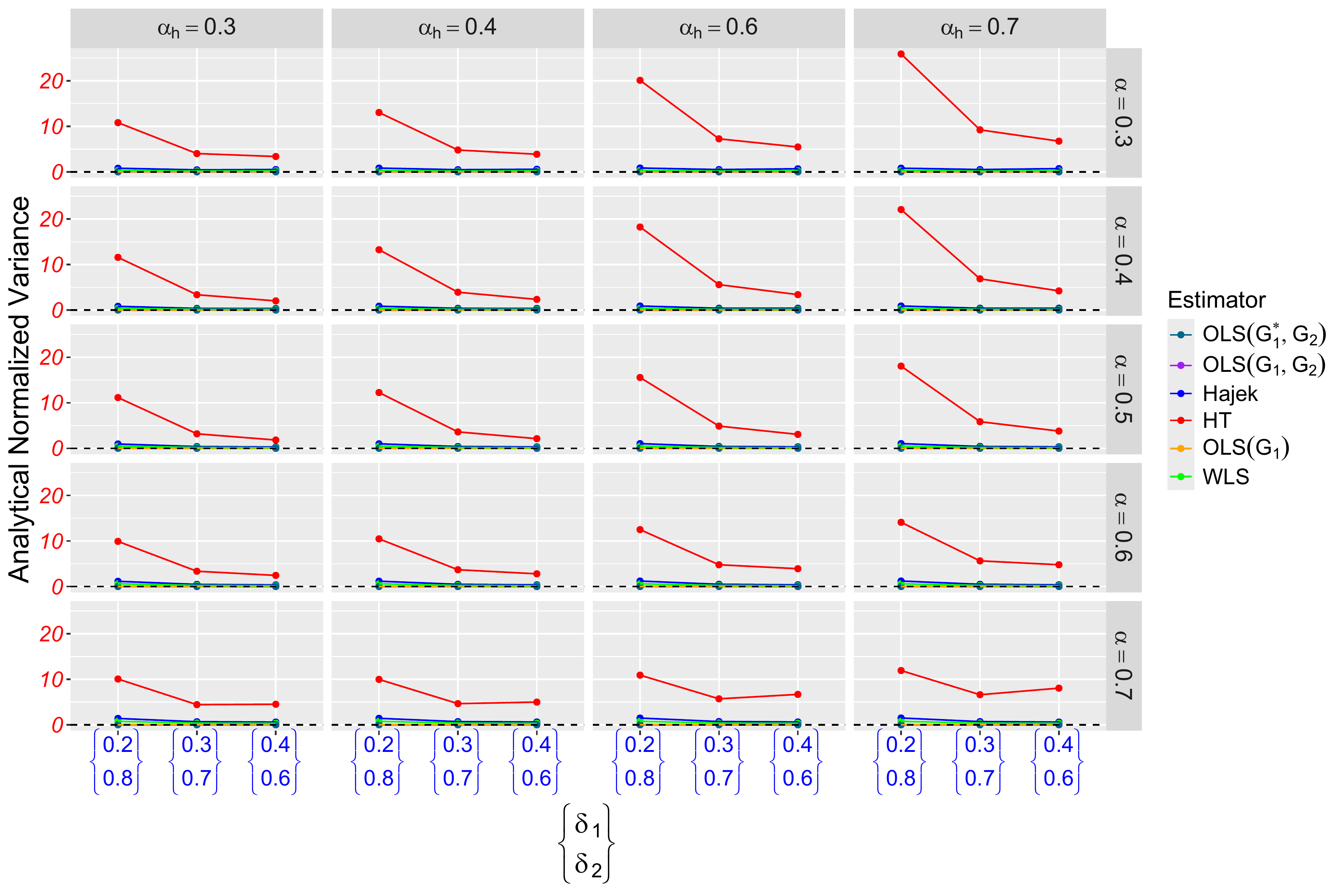}
         \caption{With HT, $\rma=0$, $h=1$}
       \label{img7b}
    \end{subfigure}
    \vspace{0.5cm}
    \begin{subfigure}{0.5\linewidth}
        \includegraphics[height=5.5cm]{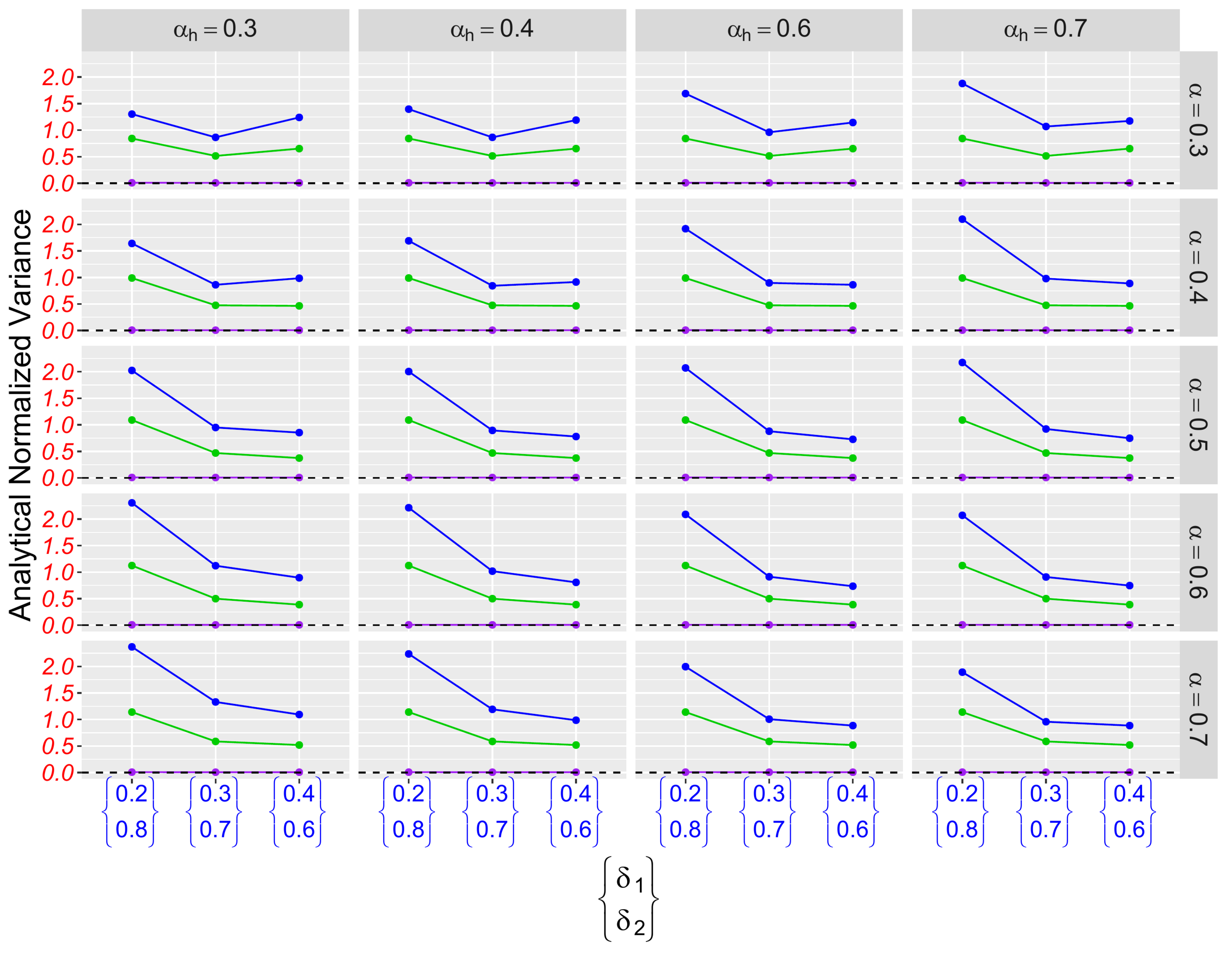}
        \caption{Without HT, $\rma=0$, $h=2$}
        \label{img8a}
    \end{subfigure}
    \begin{subfigure}{0.45\linewidth}
        \includegraphics[height=5.5cm]{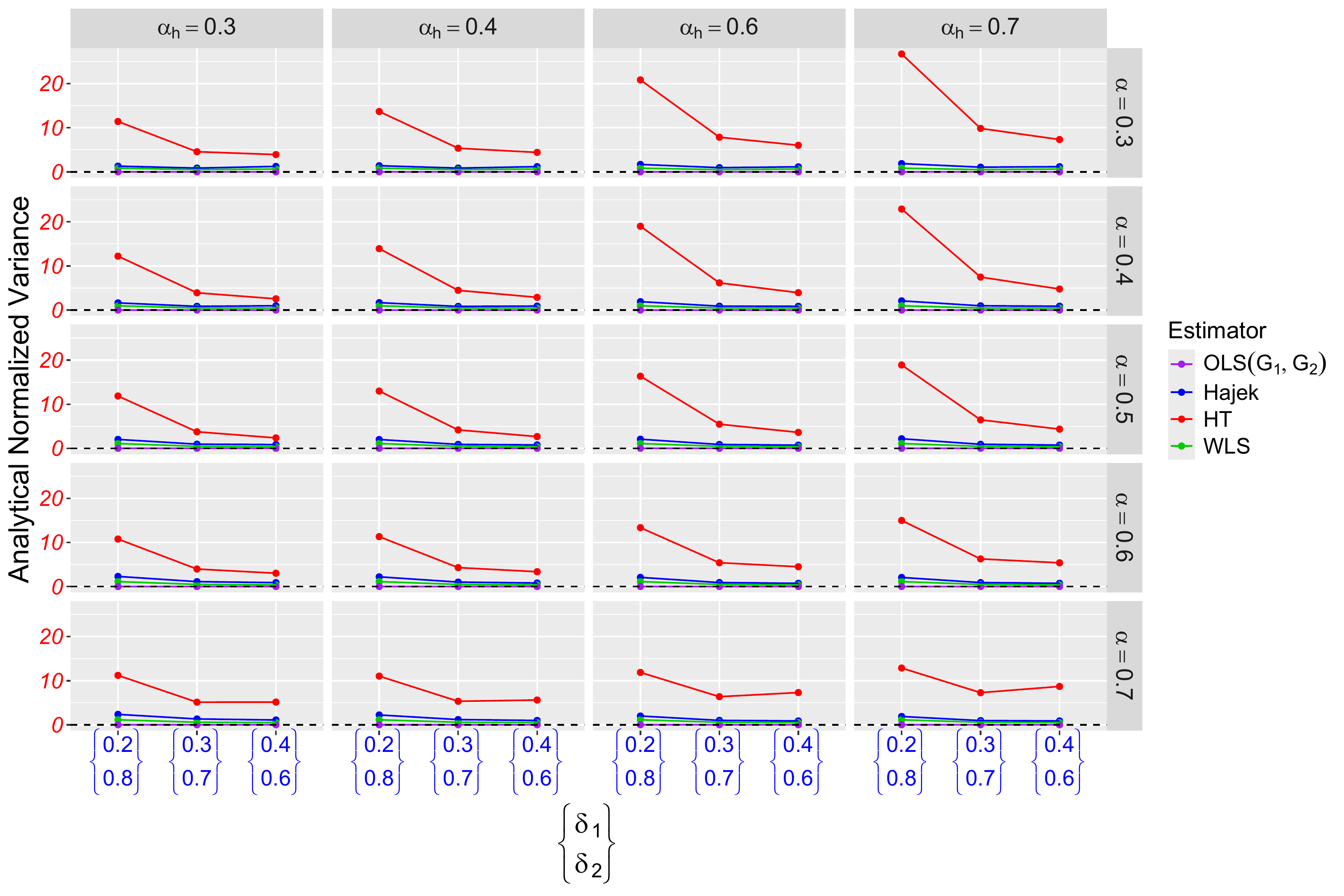}
         \caption{With HT, $\rma=0$, $h=2$}
       \label{img8b}
    \end{subfigure}
     \caption{Scenario 3 (second-order neighborhood interference). Normalized variance of the HT estimator (red), Hajek estimator (blue), WLS estimator (green), $\text{OLS}(G_1)$ estimator (orange),$\text{OLS}(G^*_1, G_2)$ (deep blue), and $\text{OLS}(G_1, G_2)$ estimator (purple), estimated using M-estimation, for the first-order spillover effects $\text{SE}^{1}\left(\alpha_{h}, 0.5;\rma = 0, \alpha\right)$ and second-order spillover effects $\text{SE}^{2}\left(\alpha_{h}, 0.5;\rma = 0, \alpha\right)$ for untreated units, with $\alpha_h\in\{0.3, 0.4, 0.5, 0.6, 0.7\}$ and $\alpha\in\{0.3, 0.4, 0.5, 0.6, 0.7\}$, under Scenario 3, a two-stage assignment with $\{\delta_1, \delta_2 \}\in\{(0.2, 0.8), (0.3, 0.7), (0.4, 0.6)\}$, and regular networks with $p=0.2$. The left panel shows the results without the HT estimator, to better highlight the differences between the Hajek, WLS, and OLS estimators.}
\label{fig:ana_variance}
\end{figure*}

 The Horvitz-Thompson estimator shows significantly higher variance compared to all other estimators. This can be explained by the extreme weights, particularly when there is a mismatch between the design probabilities and the hypothetical probabilities. In contrast, the Hajek and WLS estimators normalize the weights, resulting in smaller variance.

 Comparing these estimators, the Hajek estimator generally shows slightly higher variance than the WLS estimator, which gains efficiency through model-based smoothing. However, this gap narrows when the hypothetical parameters are closer to the design probabilities $(\delta_1, \delta_2)$, specifically when the design provides a $\delta$ that closely mimics the combination of $\alpha$ and $\alpha_h$. In such scenarios, the improved overlap leads to lower variance for both estimators.

 Regarding the impact of the first-stage allocation probabilities $(\delta_1, \delta_2)$, the variance trends are consistent across both first-order ($h=1$) and second-order ($h=2$), but  differ by estimator type.
 For the Horvitz-Thompson estimator, the variance generally decreases as the design probabilities converge from $(0.2, 0.8)$ to $(0.4, 0.6)$ when $\alpha$ is small to moderate. However, at large values of $\alpha$ (e.g., $\alpha = 0.7$), the patterns diverge. For untreated units ($\rma=0$), the trend follows a U-shape, reaching its minimum at the intermediate design $\{0.3, 0.7\}$ rather than the most balanced design where the two probabilities are closest, while for treated units ($\rma=1$), the variance increases monotonically. In contrast, the Hajek and WLS estimators exhibit greater stability. Their variance generally decreases or remains relatively constant as the design converges across most scenarios.

Finally, we observe that the OLS estimators
exhibit the lowest variance among all methods. This is expected, as parametric least squares estimators generally yield lower variance than weighting-based non-parametric or semi-parametric estimators.
However, lower variance comes at the cost of potentially significant bias due to misspecification of the interference set or the exposure mapping, as shown in Figure \ref{fig:scenario3}.
Note that to show the bias due to such misspecifications, HT or Hajek estimators with exposure mappings could have been used \citep{aronow2017estimating}, and would have shown higher variance compared to the OLS estimators, which were employed here for simplicity of illustration.

 \begin{figure*}[htbp]
    \centering
    \begin{subfigure}{0.5\linewidth}
      \hspace*{-0.5cm}
        \includegraphics[height=5.3cm]{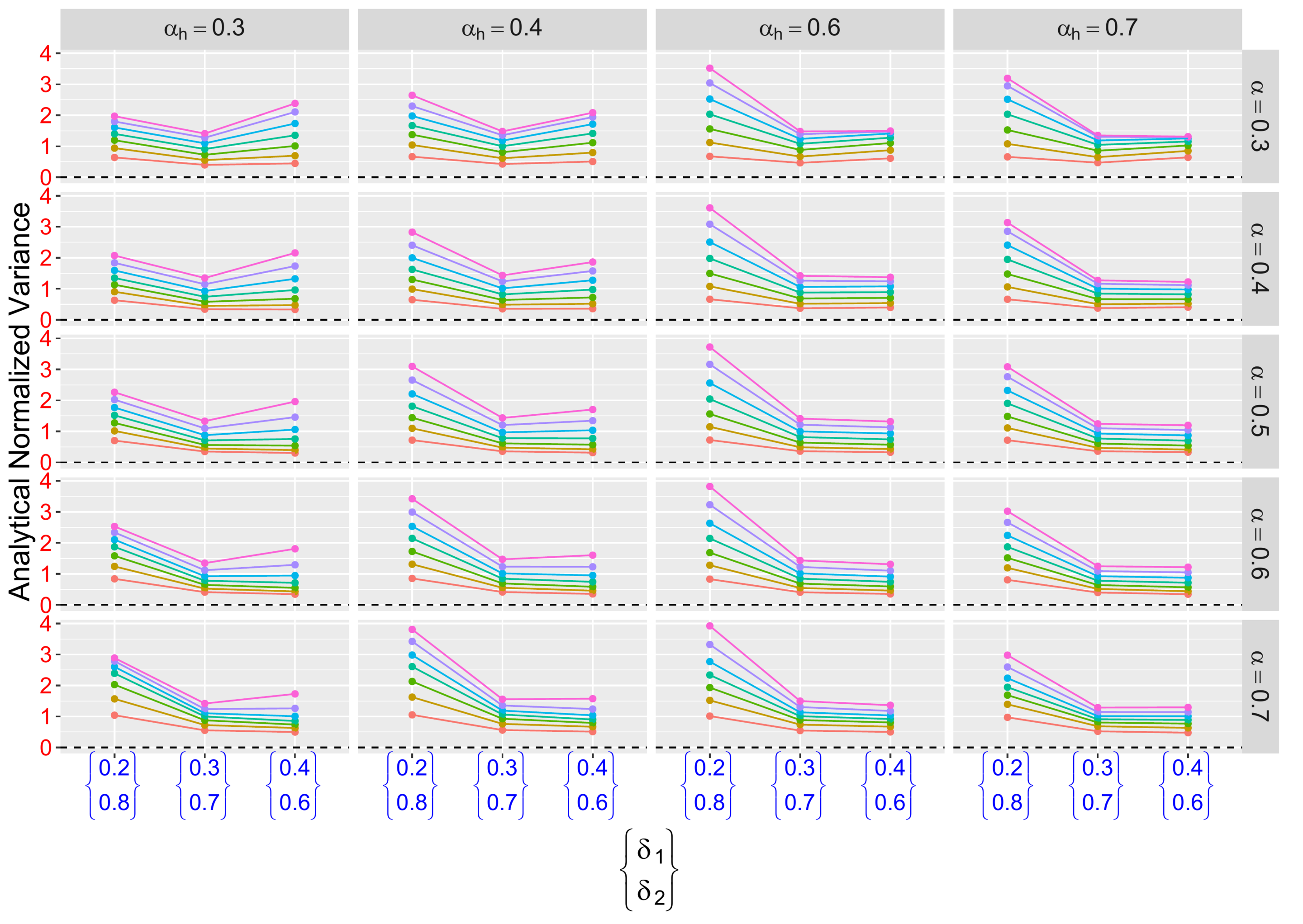}
         \caption{ Erd\H{o}s--R\'enyi random graph, $\rma=0$, $h=1$}
        \label{img9}
    \end{subfigure}
    \begin{subfigure}{0.48\linewidth}
     \hspace*{0.2cm}
        \includegraphics[height=5.3cm]{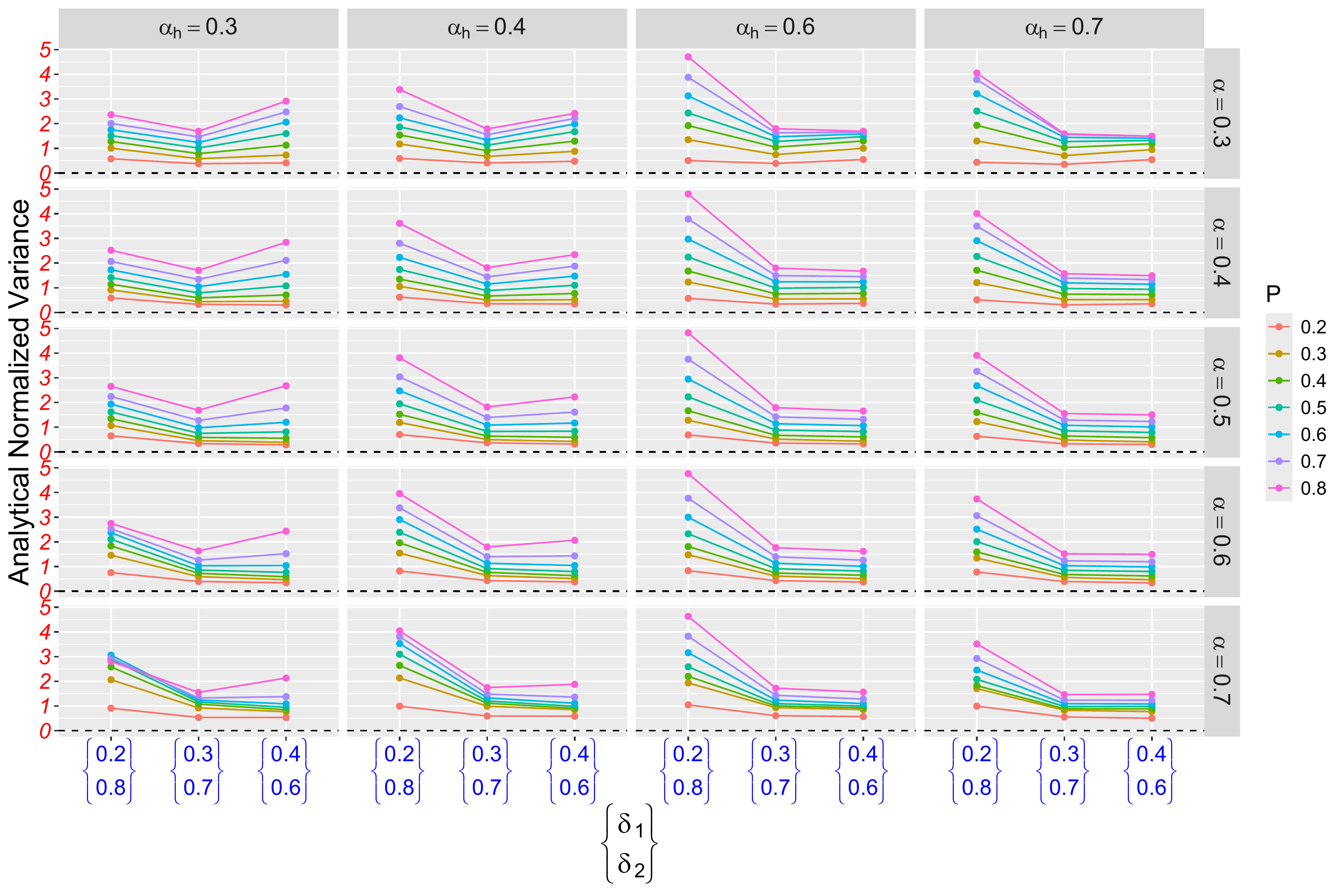}
        \caption{Regular graph, $\rma=0$, $h=1$}
        \label{img10}
    \end{subfigure}

    \begin{subfigure}{0.5\linewidth}
      \hspace*{-0.5cm}
        \includegraphics[height=5.3cm]{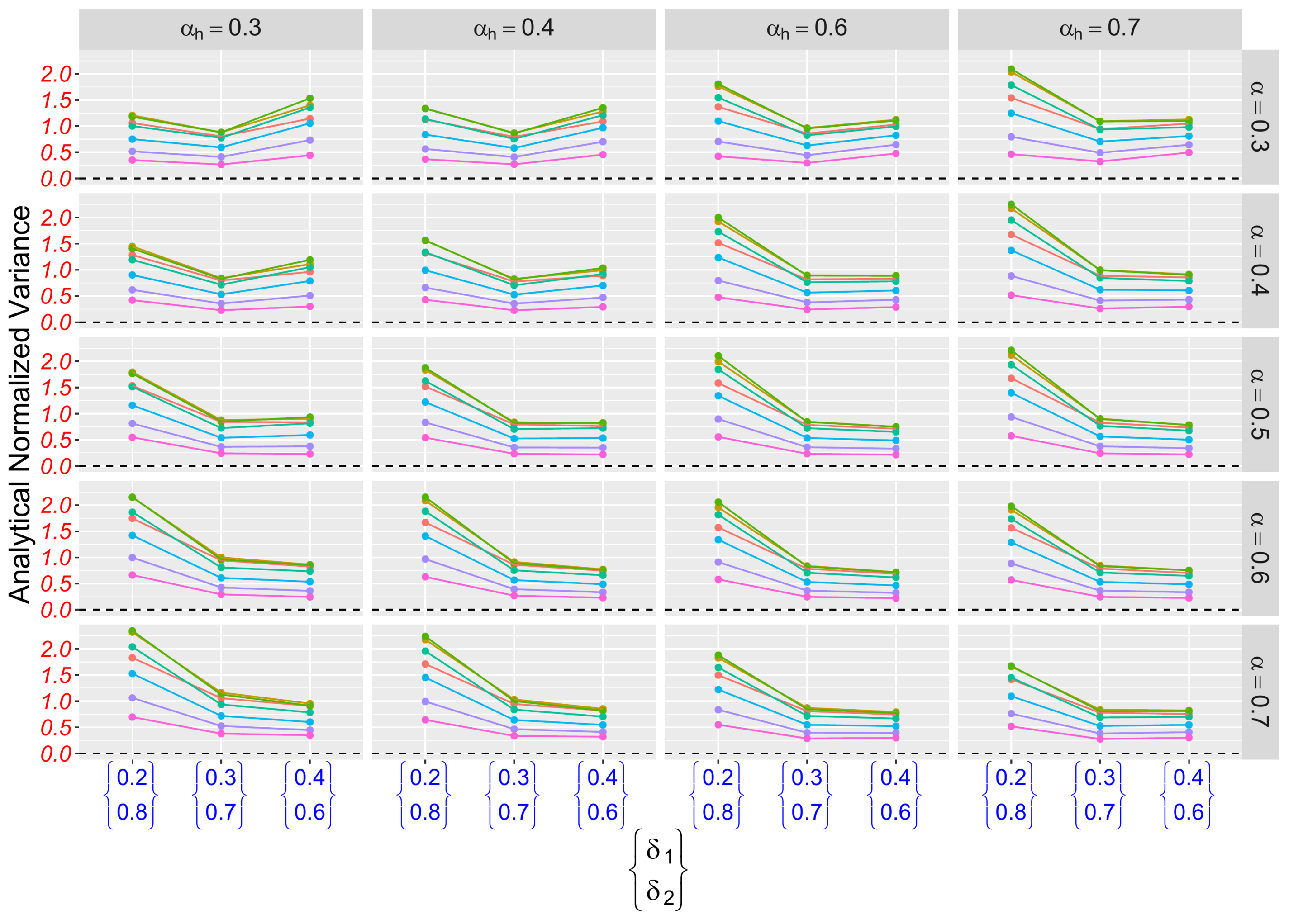}
        \caption{Erd\H{o}s--R\'enyi random graph, $\rma=0$, $h=2$}
        \label{img11}
    \end{subfigure}
    \begin{subfigure}{0.48\linewidth}
    \hspace*{0.2cm}
        \includegraphics[height=5.3cm]{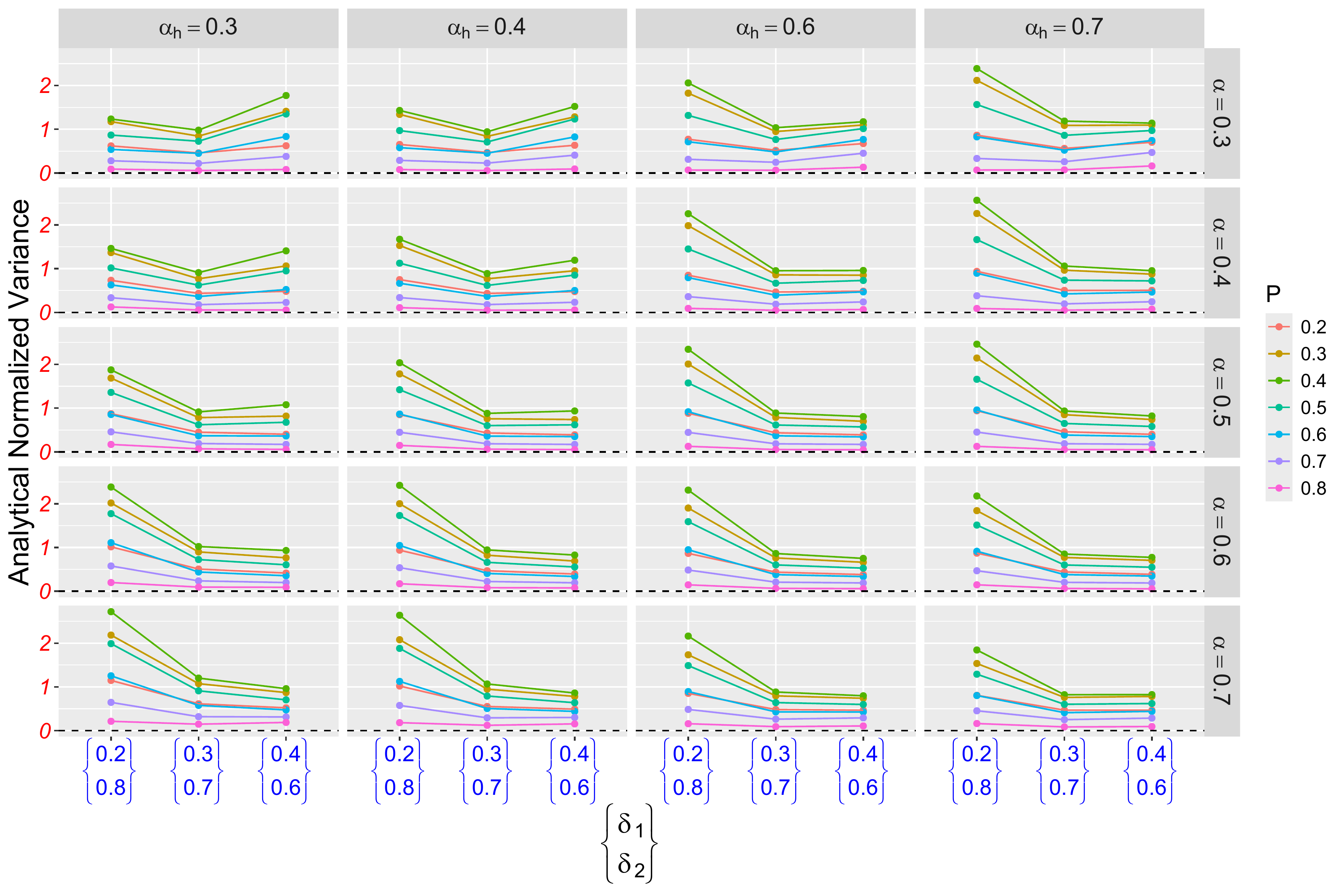}
        \caption{Regular graph, $\rma=0$, $h=2$}
       \label{img12}
    \end{subfigure}
\caption{Scenario 3 (second-order neighborhood interference). Normalized variance of the Hajek estimator estimated using M-estimation,   for the first-order spillover effects $\text{SE}^{1}\left(\alpha_{h}, 0.5;\rma = 0, \alpha\right)$ (top), and second-order spillover effects $\text{SE}^{2}\left(\alpha_{h}, 0.5;\rma =0, \alpha\right)$ (bottom), with $\alpha_h\in\{0.3, 0.4, 0.5, 0.6, 0.7\}$ and $\alpha\in\{0.3, 0.4, 0.5, 0.6, 0.7\}$, under Scenario 3, a two-stage assignment with $\{\delta_1, \delta_2 \}\in\{(0.2, 0.8), (0.3, 0.7), (0.4, 0.6)\}$, and Erd\H{o}s--R\'enyi networks (left) or regular networks (right) with different values of $p$, represented with different colors.}
    \label{fig:images_2x2}
\end{figure*}
\subsubsection{Study of the variance of the Hajek estimator across values of network density}
Figure \ref{fig:images_2x2} depicts the variance of the Hajek estimator across different values of network density parameter $p$ in Erd\H{o}s--R\'enyi random graph and regular graph.
Both Erd\H{o}s--R\'enyi random graphs and regular graphs show similar trends,
with the latter always showing a lower variance:
for $h=1$, a higher $p$ generally leads to higher variance, while for $h=2$, a higher value of $p$ generally leads to lower variance.
This comes from how network density affects the size of the $h$-order neighborhood.
For $h=1$, increasing $p$ expands the size of the first-order neighborhood ($|\mathcal{C}^1_{ij}|$),
resulting in more extreme  weights.
Conversely, for $h=2$, as $p$ increases, the set of second-order neighbors ($|\mathcal{C}^2_{ij}|$) shrinks, leading to a lower variance.

\section{Empirical Application}
\label{sec:application}
\subsection{Study Setting}
We apply our estimators to investigate first and second-order spillover effects of a maternal and child health intervention that was randomly assigned to rural households in Honduras \citep{Shakyae012996, Airoldi2024InductionOS}.
The study was a two-stage randomized experiment, encompassing 176 villages in Honduras. The first stage was a factorial randomization where villages were assigned to two factors: 1) the targeting strategy, random or based on friendship nomination,  and 2) the proportion
of households receiving the intervention, taking values 0\%, 5\%, 10\%, 20\%, 30\%, 50\%, 75\%, and 100\%.
The second stage assigned the intervention to households depending on the targeting strategy and proportion assigned to the village in the first stage. Individuals within the same household received the same treatment.
In the `random targeting' arm,  households were randomly selected to receive the intervention according to the specified proportions.
In the
`friendship nomination' arm, intervention recipients were selected based on the friendship nomination technique.
Our analysis focused on 110 villages from two categories: 88 villages in the `random targeting' arm with proportions 0\% to 100\%, and 22 villages in the `friendship targeting' arm, assigned to proportions 0\% or 100\%.

In the baseline survey, carried out in 2015, in addition to socio-demographic characteristics,  social networks of village residents aged 12 and older were collected. Social network data were collected using a series of sociometric name generators covering diverse relational domains, including friendship, financial exchange,  health advice, and negative ties.
For our analysis, we excluded negative ties and aggregated the remaining positive domains into a single undirected network. Specifically, a link is defined to exist between two individuals if they reported a positive connection in any of the surveyed domains.

Households assigned to the intervention received, on average, 14 household visits where health workers delivered
15 educational modules covering prenatal and postnatal care, newborn health (e.g., breastfeeding, thermal care), and disease prevention (e.g., diarrhea and respiratory illness management) \citep{Shakyae012996, Airoldi2024InductionOS}.

A follow-up survey
to assess the intervention's effectiveness
comprised 83 questions assessing knowledge and attitudes related to maternal and child health. In our analysis, we selected only those individuals who completed both the baseline and follow-up surveys, resulting in 8330 individuals from 4508 households. An aggregate outcome was calculated by taking the proportion of questions answered correctly,
yielding a score between 0 and 1.

We define the treatment $A_{ij}$ based on assignment to the intervention, with $A_{ij}=1$ if the individual's household was assigned to the intervention and $0$ otherwise, resulting in intent-to-treat causal effects\footnote{Some members of the targeted households did not attend all the intervention visits, resulting in non-compliance to the assignment. Here, we do not deal with this issue and focus on spillover effects of the assignment.}.
For each individual, we define the interference set as the subset of individuals that can be reachable through a finite path, i.e., $\mathcal{I}_{ij}=\{ik\in \mathcal{N}:d_{ij,ik}<\infty\}$.
In the Appendix, we also report results under different interference assumptions, including second- and third-order neighborhood interference (Assumption \ref{assump:h-neighborhood}), and a definition of the interference sets based on community detection algorithms.

Given the assignment at the household level,
every individual in the same household receives the same treatment.
Then, the calculation of the interference set's propensity score $P_{\Delta}\left(\mathbf{A}_{ij,\mathcal{I}_{ij}}\right)$ and the probabilities under the hypothetical treatment allocations must take into account the household-level assignment and can be redefined as the joint probability of treatment assignment for unit $ij$'s household and
the unique households represented in the unit's interference set.
Detailed formal definitions of the household-level interference sets and the corresponding probability formulas are provided in the Appendix C.1
\footnote{In our empirical application, the second-stage assignment operates at the household level rather than the individual level, and follows a completely randomized design rather than a Bernoulli design. Specifically, within each village, a fixed number of households are selected to receive the intervention through complete randomization, with all individuals within the same household receiving the same treatment. This design differs from Bernoulli randomization. However, when the number of assignment units (households) is large, the hypergeometric distribution under complete randomization is well approximated by the corresponding binomial distribution with the same marginal treatment probability \citep{imbens2015causal}. Given the large number of households in the villages of our study, for simplicity of calculation, we  approximate the second-stage assignment mechanism as household-level Bernoulli randomization.
}
.

\subsection{Spillover Effects}
We analyze the aggregate outcome using the Horvitz-Thompson, Hajek, and WLS estimators to estimate the spillover effect from households at different network distances.
In particular,
we estimate the spillover effects  $\text{SE}^{h}\left(\alpha_h, \alpha_{h}^{\prime}; \rma, \alpha\right)$, for $h=1,2$, $\alpha=0.2, 0.3, 0.4, 0.5, 0.6,0.7,0.8$, and $\alpha_h=0.2, 0.4, 0.6, 0.8$. We set $\alpha_{h}^{\prime}=0.3$, i.e., the mean of the dosages assigned in the first stage.
Here, we only report the results from the Hajek and WLS estimators. The rest of the results can be found in the Appendix C.
\begin{figure*}[t]
    \centering
    \begin{subfigure}{0.42\linewidth}
      \hspace*{-0.5cm}
        \includegraphics[height=5.5cm]{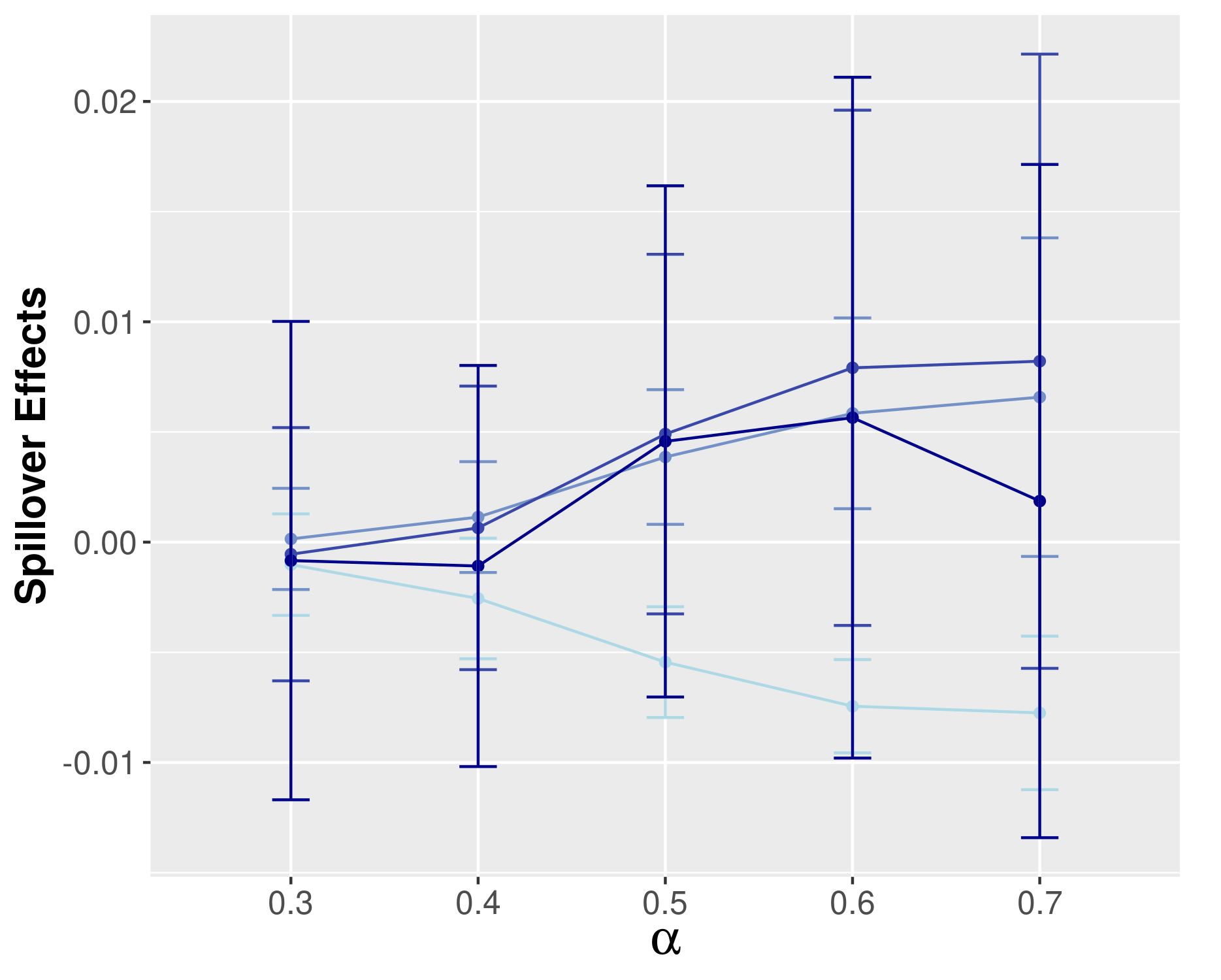}
         \caption{$\rma=0$, $h=1$}
         \label{img13}
    \end{subfigure}
    \hfill
    \begin{subfigure}{0.52\linewidth}
          \includegraphics[height=5.5cm]{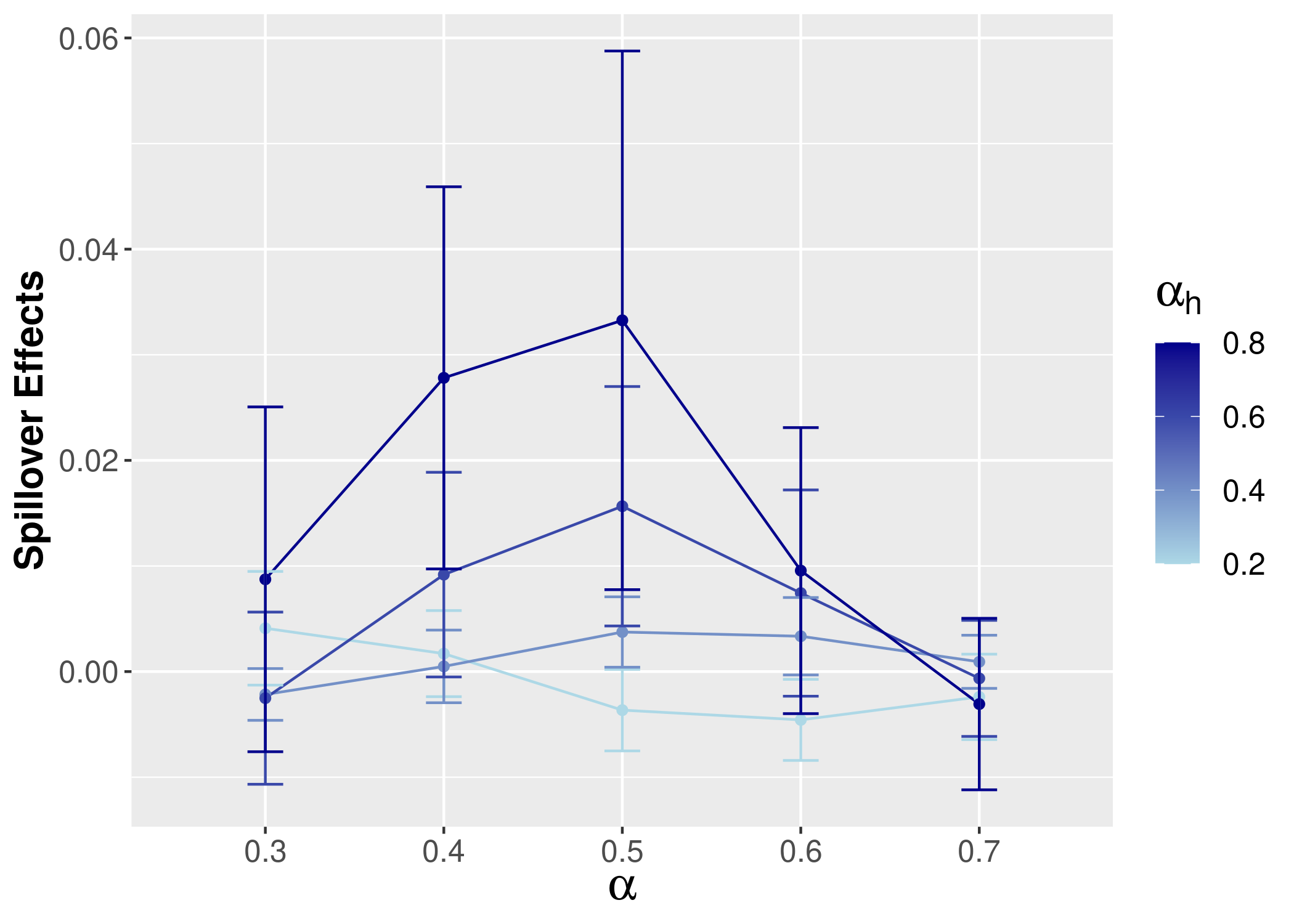}
        \caption{$\rma=1$, $h=1$}
        \label{img14}
    \end{subfigure}
        \caption{First-order spillover effects of the maternal and child health intervention estimated with the Hajek estimator. Point estimates and   95\% confidence intervals are reported for the untreated (left) and for the treated (right), and for different values of $\alpha$ (x-axis) and $\alpha_h$ (colors), assuming the interference set as the subset of individuals reachable through a finite path. }
        \label{fig:app1}
\end{figure*}

Figure \ref{fig:app1}  shows the results for the first-order spillover effects estimated using the Hajek estimator.
For untreated individuals ($\rma=0$), spillover effects from first-order neighbors are consistently small in magnitude and lack statistical significance across all combinations of $\alpha$ and $\alpha_h$. For treated individuals ($\rma=1$), we observe an inverted U-shaped relationship between spillover effects and $\alpha$, with highest effects occurring at $\alpha=0.5$. At the highest $\alpha_h$ value (0.8), the spillover effect $\text{SE}^{1}\left(0.8, 0.3; 1, 0.5 \right)$ reaches approximately 2.7\%
and is statistically significant,  representing a 2.7 percentage point increase in the proportion of correct answers.
This different pattern between treated and untreated individuals can be explained through several interpretations. A plausible explanation is that, treated individuals who receive the health training gain both the foundational knowledge and the cognitive framework necessary to effectively absorb and integrate information.
When these individuals encounter health-related information through discussions with their contacts,
it reinforces and extends their existing knowledge base, producing a positive spillover effect. In contrast, untreated individuals lack this foundational preparation and struggle to extract actionable insights from peer discussions. The knowledge gap between untreated individuals and their trained contacts may further limit effective communication,
as trained individuals might use terminology or assume baseline understanding that untreated individuals do not possess.
Another possible explanation comes from the self-efficacy theory \citep{bandura1977self}. Treated individuals are likely to have higher self-efficacy regarding health topics, leading to greater engagement in related discussions and interactions, thereby enhancing their learning through these activities. Conversely, untreated individuals may lack confidence in the discussion or even feel isolated, resulting in less engagement in these activities. They might even distrust knowledge shared by treated units due to their lack of direct experience, therefore leading to a lower spillover effect.
Finally, the intervention may rewire the network and reduce ties between treated and untreated individuals, contributing to a decrease in spillovers \citep{Papamichalis2025EducationalIR}.

Regarding the pattern with respect to $\alpha$,
the inverted U-shaped relationship observed for the treated can be explained by information saturation: when $\alpha$ is high, i.e., the intervention is assigned with high probability to the whole network excluding an individual's first-order neighbors, frequent collective discussions may arise in the community, and the individual may already receive abundant health information from their broader network, through these collective discussions or even through diffusion of information through direct contacts, treated or untreated (who at higher $\alpha$ tend to be more receptive). Therefore, in this saturated environment, the marginal effect of additional treated first-order contacts becomes smaller and even negligible, as the information they provide is largely redundant with what individuals already know or have heard from other sources.

We also observe distinct patterns in the variance of the first-order spillover effect estimates. Since for $h=1$, the number of individuals in the rest of the interference set significantly exceeds that of the first-order neighborhood, the weights are primarily dominated by $\alpha$.
For untreated individuals, the variance increases monotonically as $\alpha$ increases, because untreated individuals are concentrated in villages with low assignment proportions.
In contrast, for treated individuals ($\rma=1$), the variance peaks around $\alpha=0.5$ and narrows at $\alpha=0.7$. Since treated units are naturally abundant in high-proportion arms, there is substantial data support for estimating effects when $\alpha$ is high, leading to small variance.  When $\alpha$ is small, although the proportion of treated individuals within each cluster is small, the dense concentration of design arms (5\%, 10\%, 20\%, 30\%) ensures that a sufficient number of units can contribute to the estimation, leading to large data support.
However, the design support around $\alpha=0.5$ is relatively sparse,  resulting in a larger variance.

\begin{figure*}[t]
    \centering
    \begin{subfigure}{0.42\linewidth}
      \hspace*{-0.5cm}
        \centering
        \includegraphics[height=5.5cm]{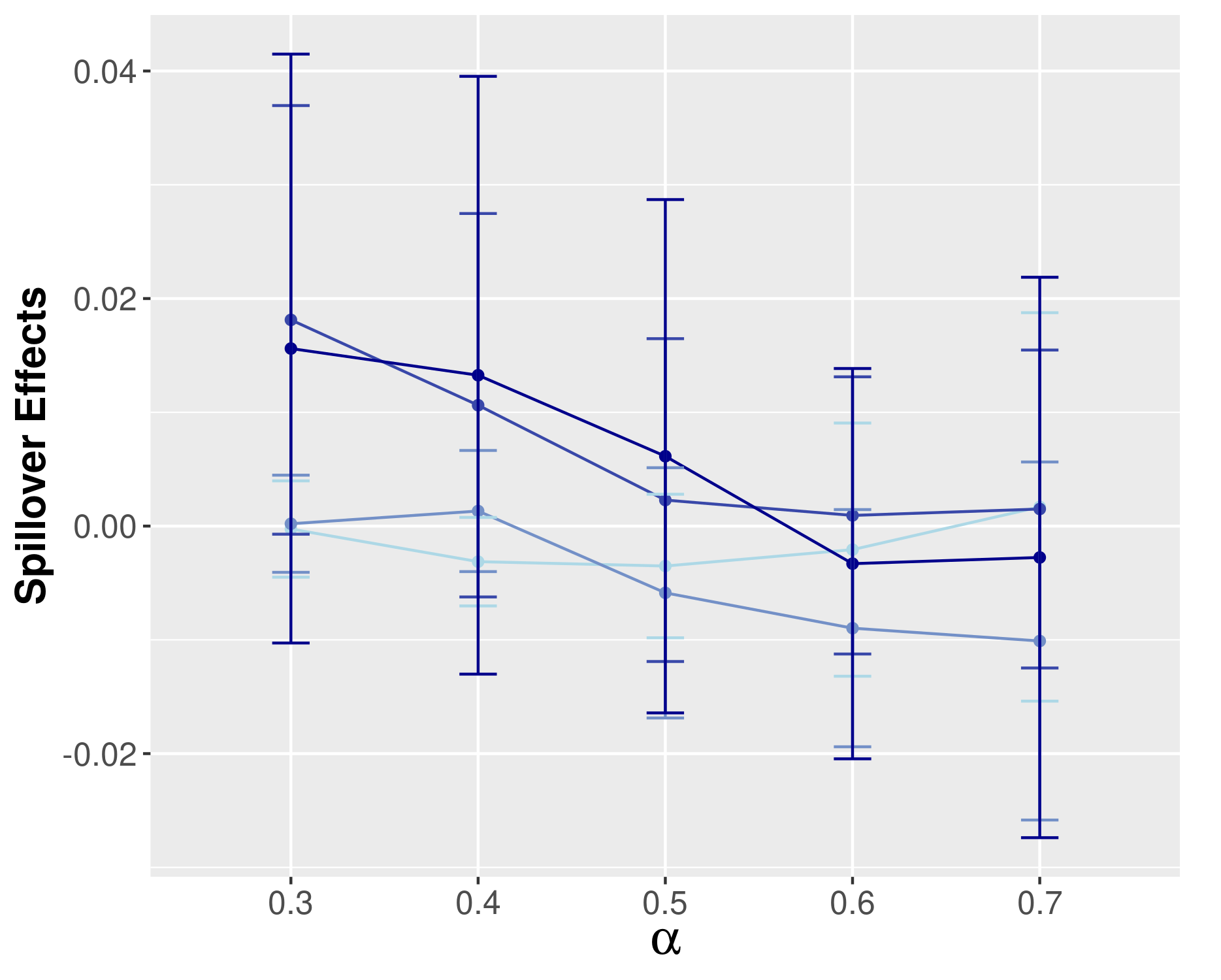}
         \caption{$\rma=0$, $h=2$}
         \label{img15}
    \end{subfigure}
    \hfill
    \begin{subfigure}{0.52\linewidth}
        \centering
          \includegraphics[height=5.5cm]{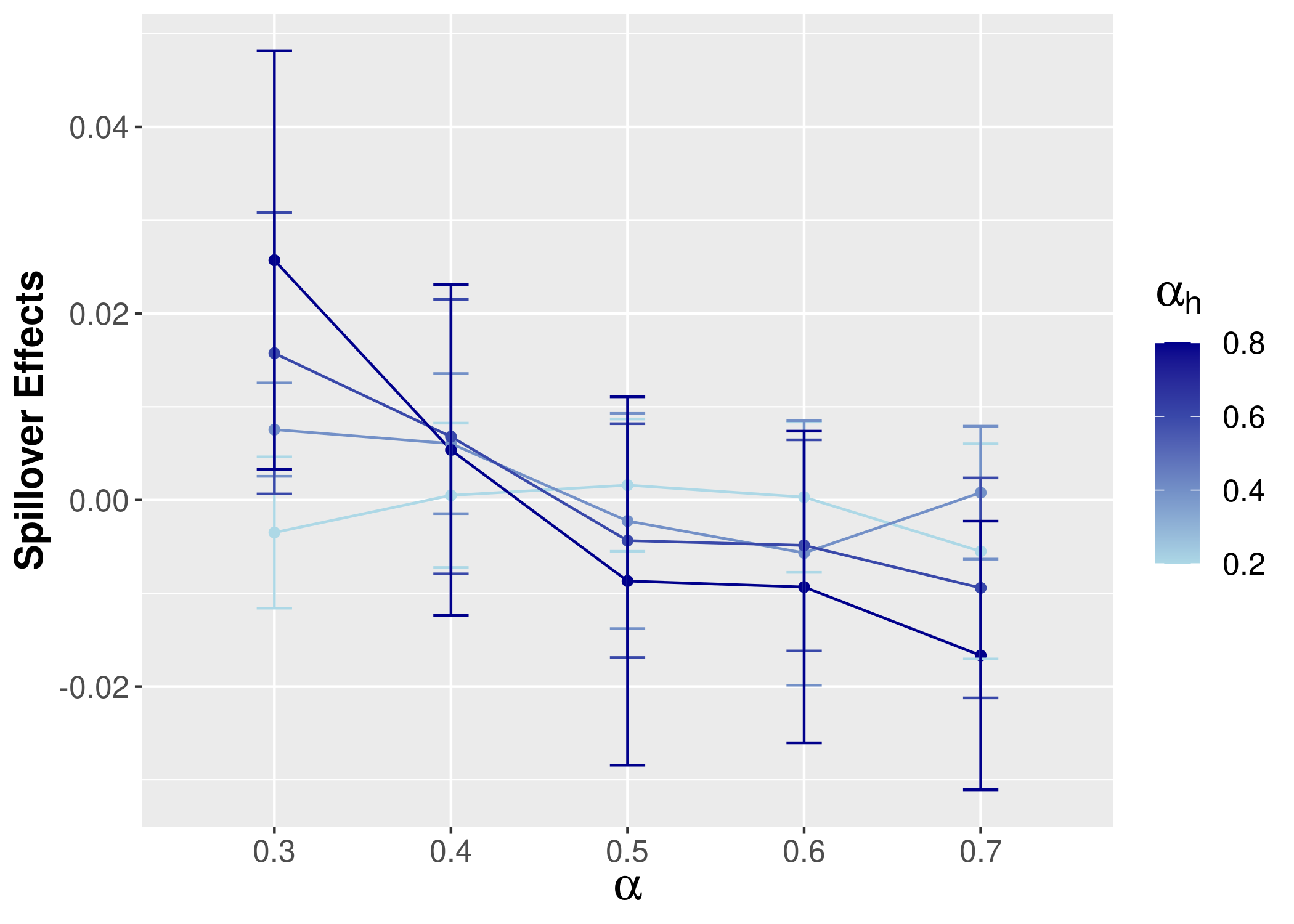}
        \caption{$\rma=1$, $h=2$}
        \label{img16}
    \end{subfigure}
        \caption{Second-order spillover effects of the maternal and child health intervention estimated with the Hajek estimator. Point estimates and   95\% confidence intervals are reported for the untreated (left) and for the treated (right), and for different values of $\alpha$ (x-axis) and $\alpha_h$ (colors), assuming the interference set as the subset of individuals reachable through a finite path. }
        \label{fig:app2}
\end{figure*}
Figure \ref{fig:app2} shows results for second-order spillover effects estimated with the Hajek estimator. Notably, these effects are consistently  larger for treated individuals compared to untreated, maintaining the pattern observed for first-order effects. However, spillover patterns with respect to $\alpha$ differ:
for both treated and untreated individuals, second-order spillover effects are strongest at low $\alpha$ levels and decline as $\alpha$ increases.
 At high $\alpha_h$ (0.8) and low $\alpha$ (0.3), spillover effects reach approximately 2\% for untreated individuals and 2.7\% for treated individuals. However, these effects diminish substantially as $\alpha$ increases, approaching zero for untreated individuals and even becoming negative for treated individuals.
 At lower $\alpha_h$ values (0.2, 0.4), effects remain close to zero across all $\alpha$ levels for both groups.
 This declining pattern can be explained through a substitution mechanism:
 with low $\alpha$
 individuals have few treated first-order contacts, creating a knowledge gap that second-order neighbors can fill.
 However, as $\alpha$ increases, there are more treated first-order neighbors who can provide direct information, as well as more treated individuals in the rest of the community giving rise to collective discussions,  rendering second-order spillovers redundant.

 Surprisingly, for treated individuals at low $\alpha$ levels (0.3-0.4), second-order spillover effects are comparable to or exceed first-order effects, which appears counterintuitive. This counterintuitive finding may be explained by the information transmission process. When direct friends discuss health topics with an individual, they may share comprehensive but relatively unfiltered information, including both essential and non-essential elements. In contrast, knowledge transmitted from second-order friends undergoes an additional filtering stage: the intermediate neighbor processes and distills the information before passing it along, resulting in more condensed and refined knowledge that focuses on key concepts. Under low to medium $\alpha$, where treated individuals are relatively scarce among friends, this filtering mechanism is particularly valuable, and the curated knowledge from second-order connections becomes especially beneficial. This aligns with the cognitive load theory \citep{sweller1988cognitive}, which suggests that by receiving a more streamlined version of the information, people may experience less cognitive overload, allowing them to perform better in the survey. Also, Figure \ref{fig:app2} shows that the confidence intervals are generally widest at low levels of $\alpha$ for second-order spillover effects, indicating greater uncertainty when the intervention is sparse in the rest of the network, and tend to become narrower at higher $\alpha$ levels, particularly for treated individuals ($\rma=1$).

\begin{figure*}[t]
    \centering
    \begin{subfigure}{0.42\linewidth}
      \hspace*{-0.5cm}
        \centering
        \includegraphics[height=5.5cm]{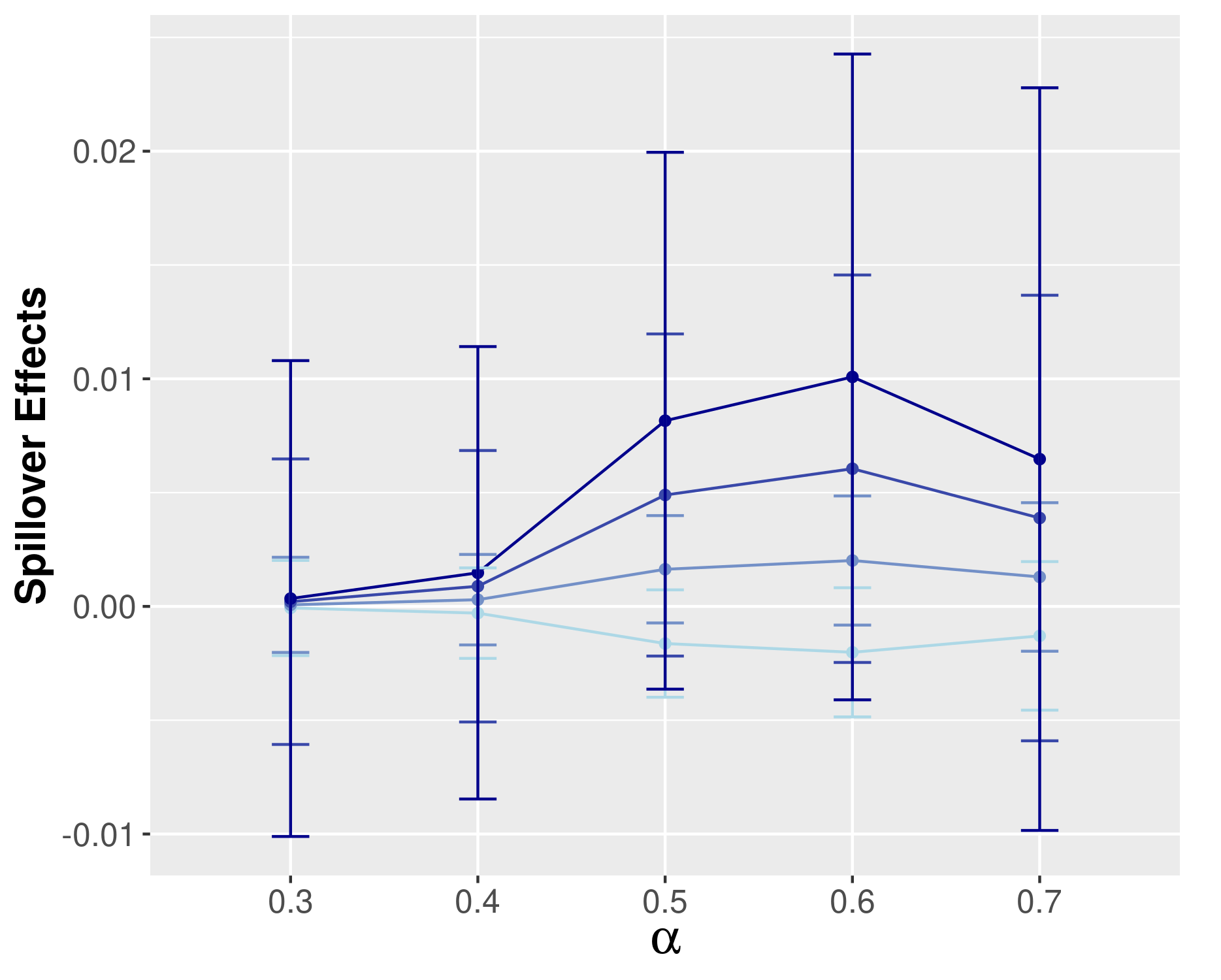}
         \caption{$\rma=0$, $h=1$}
    \end{subfigure}
    \hfill
    \begin{subfigure}{0.52\linewidth}
        \centering
          \includegraphics[height=5.5cm]{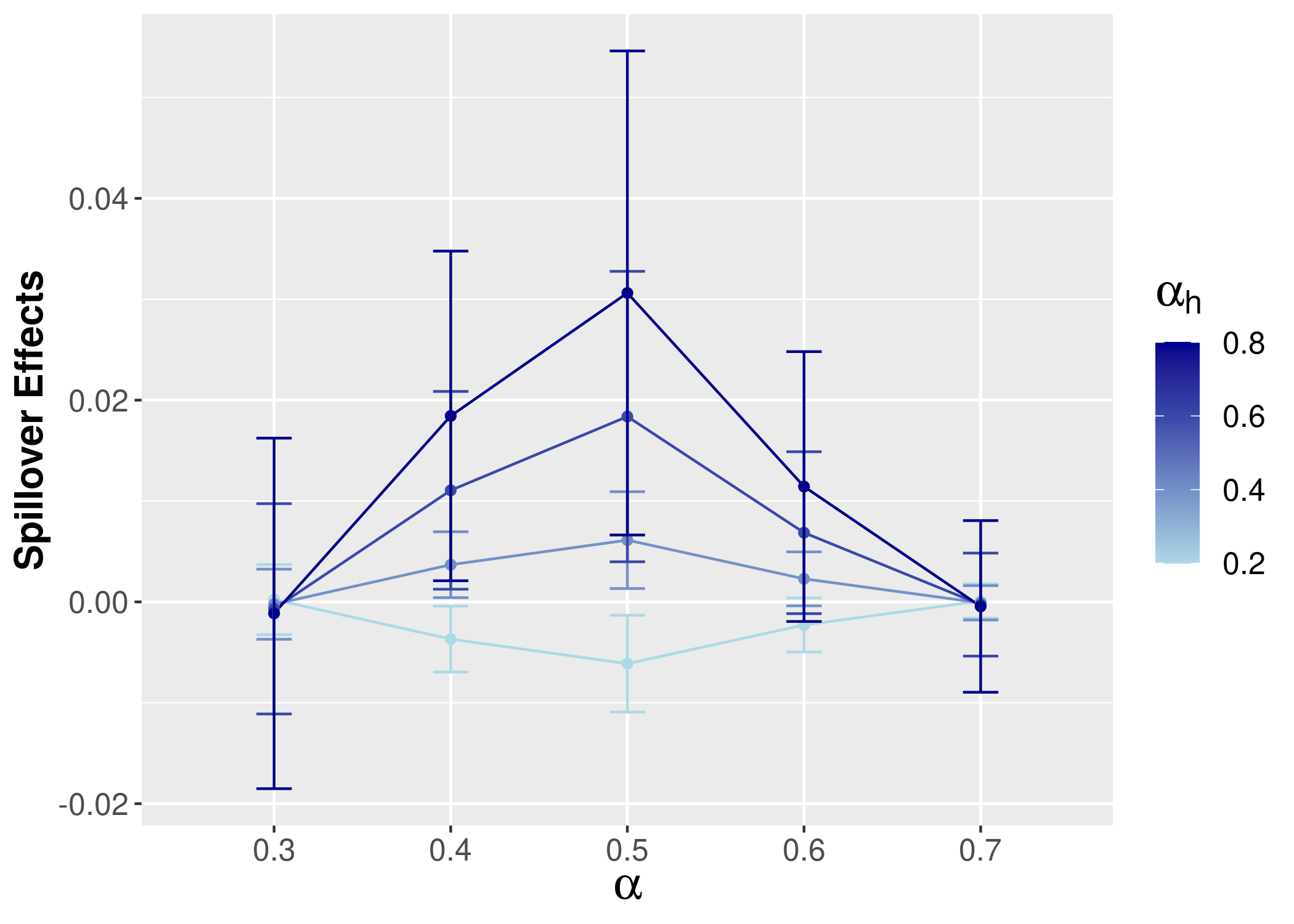}
        \caption{$\rma=1$, $h=1$}
    \end{subfigure}
        \caption{First-order spillover effects of the maternal and child health intervention estimated with the WLS estimator. Point estimates and 95\% confidence intervals are reported for the untreated (left) and for the treated (right), and for different values of $\alpha$ (x-axis) and $\alpha_h$ (colors), assuming the interference set as the subset of individuals reachable through a finite path. }
       \label{fig:WLS1}
\end{figure*}

\begin{figure*}[t]
    \centering
    \begin{subfigure}{0.42\linewidth}
      \hspace*{-0.5cm}
        \centering
        \includegraphics[height=5.5cm]{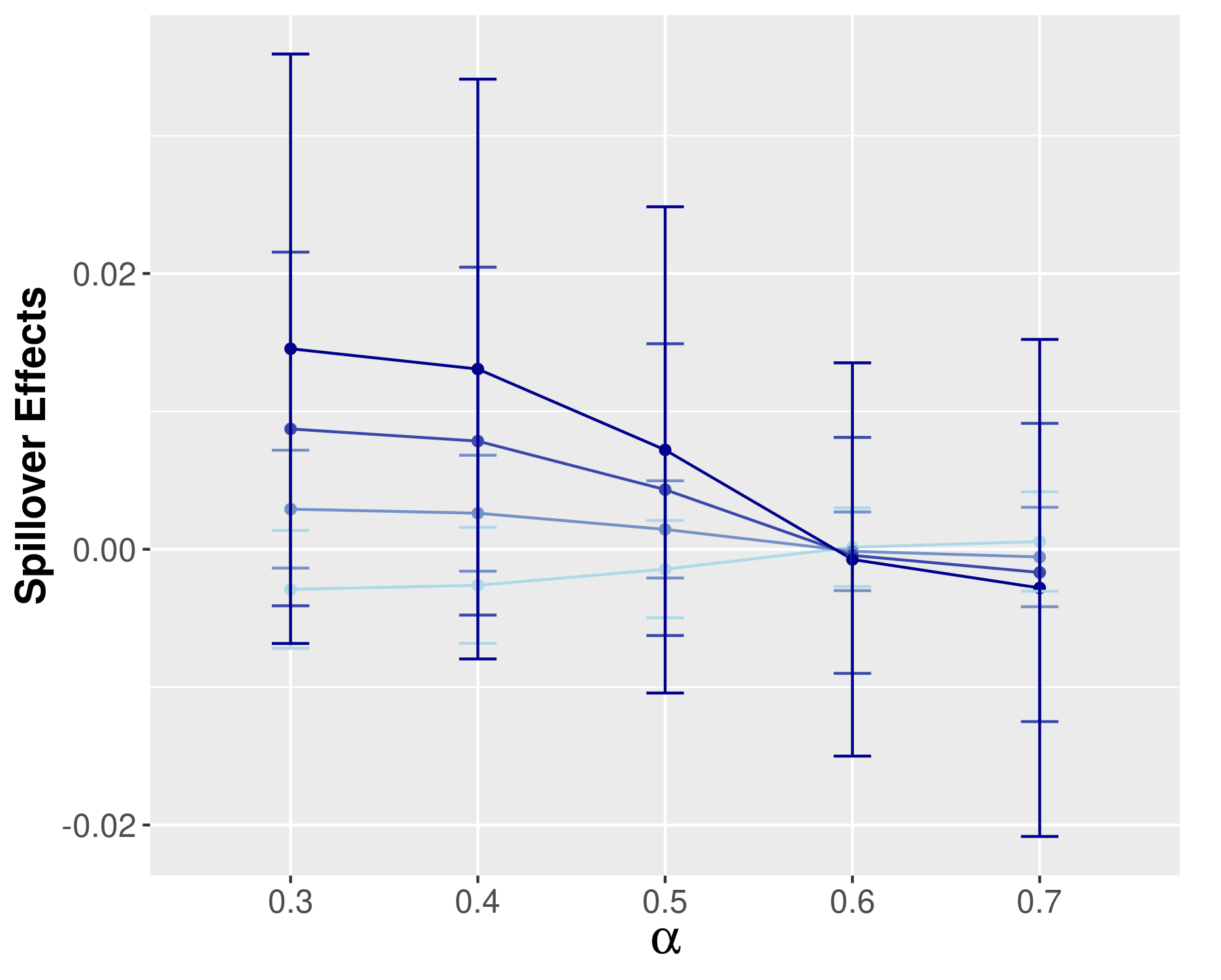}
         \caption{$\rma=0$, $h=2$}
    \end{subfigure}
    \hfill
    \begin{subfigure}{0.52\linewidth}
        \centering
          \includegraphics[height=5.5cm]{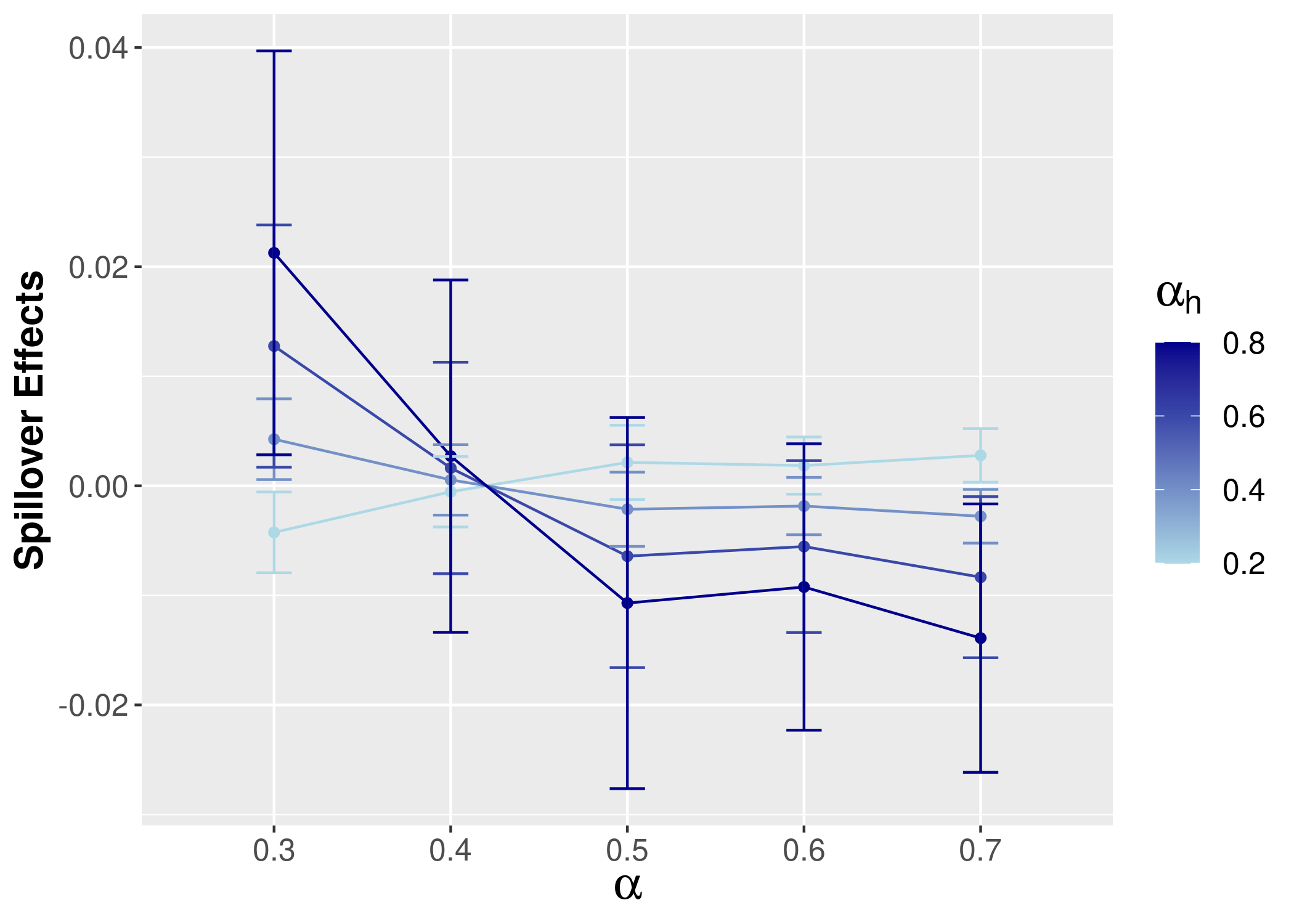}
        \caption{$\rma=1$, $h=2$}
    \end{subfigure}
        \caption{Second-order spillover effects of the maternal and child health intervention estimated with the WLS estimator. Point estimates and   95\% confidence intervals are reported for the untreated (left) and for the treated (right), and for different values of $\alpha$ (x-axis) and $\alpha_h$ (colors),  assuming the interference set as the subset of individuals reachable through a finite path. }
       \label{fig:WLS2}
\end{figure*}

Figure \ref{fig:WLS1} and \ref{fig:WLS2} show results from  the WLS estimator for first- and second-order spillover effects. A key distinction from the Hajek estimator is that the WLS model imposes a linear relationship with $\alpha_h$, forcing spillover effects to scale linearly with the difference $(\alpha_h - \alpha'_h)$ rather than varying freely. Despite this structural constraint, the overall patterns align closely with the Hajek results, showing the robustness of our conclusions.
For first-order spillover effects ($h=1$, Figure \ref{fig:WLS1}),
untreated individuals show relatively negligible effects close to zero across $\alpha$ levels, while treated individuals exhibit an inverted U-shaped pattern, increasing at low $\alpha$ level and then declining at higher levels.
For second-order spillover effects ($h=2$, Figure \ref{fig:WLS2}),
both groups show declining patterns as $\alpha$ increases, with treated individuals declining more sharply.

\section{Conclusion and Discussion}
\label{sec:conclusion}
In this article, we have addressed the problem of estimation of higher-order spillover effects from neighbors at different network distances under a generalized partial interference assumption, where each unit's interference set is a subset of their cluster.
We have introduced a new causal estimand, the $h$-order spillover effects for units at a network distance $h$, which relies on two hypothetical treatment allocations: one for the $h$-order neighborhood, and the other for the rest of the interference set. The $h$-order spillover effect is then defined by changing the parameter of the assignment to the $h$-order neighborhood, while keeping everything else fixed.
This estimand does not depend on the usage of exposure mapping functions and allows a conservative choice of the interference sets.

We have shown that the quantity that is identified under incorrect assumptions on the interference sets or exposure mapping function is the expectation of potential outcomes conditional on the exposure value of interest under the misspecified interference sets or exposure mapping functions.
Because here we are interested in estimating
spillover effects under the true interference structure, even if unknown, we see
the difference between this quantity and the expected potential outcomes under the true interference structure as a bias,
leading to biased estimators,
also confirmed by the simulation study.

To estimate the $h$-order spillover effects, we have developed Horvitz-Thompson, Hajek and weighted regression estimators, and shown their asymptotic properties. For large sample settings, Hajek and weighted regression estimators provide more stable results compared to the Horvitz-Thompson estimator, whose performance can be significantly impacted by extreme weights. In addition, thanks to parametric or semi-parametric assumptions, the weighted regression estimator benefits from smaller variance, although at the cost of potential misspecification of the relationship between the outcome and $\alpha_h$. We have also shown how the variance of these estimators depends on the combination between the estimand parameter $\alpha_h$ and $\alpha$ and the design parameters $\delta_1$ and $\delta_2$.

Our proposed estimators are designed to operate under relatively conservative assumptions on the interference set and do not require specifying an exposure mapping function. This flexibility makes them broadly applicable, particularly in settings where the structure of interference is complex or only partially understood. However, this generality typically comes at a cost: identification regions may be wider, and efficiency may be lower, reflecting the weaker assumptions and the limited structure imposed on potential outcomes. In contrast, estimators that rely on more restrictive assumptions about the interference set and exposure mapping function can often achieve improved efficiency, since they exploit stronger structural information about how treatments propagate and affect outcomes.
Given this bias–variance trade-off, we recommend adopting our conservative approach
when robustness to misspecification is prioritized over efficiency, and when the complexity of real-world social processes makes it imprudent to commit to restrictive exposure mapping or a narrowly defined interference set. This is especially true when spillover effects are due to complex heterogeneous diffusion mechanisms, whose structure is not well understood and difficult to characterize,
and that may affect the scope of interference in various ways, also depending on the timing of data collection.
In addition, when the measured network may not fully capture the channels through which interference operates, researchers can leverage the cluster structure to adopt a conservative specification of the interference set—one that does not depend on a potentially misspecified network—while still defining and estimating well-posed $h$-order spillover effects relative to the network information at hand. When multiple networks are available and it is unclear which one drives interference, the interference set can be formed as the union of these networks, enabling the estimation of $h$-order spillover effects for each layer separately. This approach complements the perspective of \citet{wasserman1994social} by avoiding reliance on a specific exposure mapping function and allowing spillover effects to be studied across multiple network structures that may influence outcomes in different ways.

Our empirical application further demonstrates the practical relevance of the proposed framework by uncovering meaningful first- and second-order spillover effects in the Honduras maternal and child health intervention. Using our proposed estimators, we find that treated individuals benefit substantially more from peer spillovers than untreated individuals, reflecting stronger absorptive capacity and higher baseline preparedness. First-order spillovers exhibit an inverted U-shape, peaking at moderate treatment saturation and declining at higher dosages due to information saturation within the community. Second-order spillovers are strongest when treatment is scarce—where indirect contacts fill informational gaps—and diminish as direct treated neighbors become more prevalent and information flowing from second-order neighbors becomes redundant.
Together, these results highlight the importance of modeling interference flexibly and capturing higher-order network pathways, as both the magnitude and sign of spillover effects depend critically on community-wide treatment intensity and network distance.

Future methodological work can further expand the flexibility and applicability of our framework by developing estimators that more efficiently incorporate structural information about networks while preserving robustness to misspecification. In particular, extending the approach to accommodate network uncertainty, dynamic or time-varying interference, and multivariate or continuous treatments would substantially broaden its relevance for modern field experiments. Advances in semi-parametric and machine-learning–assisted estimation may also offer opportunities to improve efficiency without imposing restrictive exposure mappings. Moreover, identifying optimal strategies for selecting or learning interference sets, especially in large-scale or partially observed networks, remains an open challenge with both theoretical and practical implications.

\section*{Funding}
This work was supported by NIH grants R01MH134715 and R01AG081814.

\appendix
\setcounter{equation}{0}
\setcounter{figure}{0}
\setcounter{table}{0}
\renewcommand{\theequation}{\thesection.\arabic{equation}}
\renewcommand{\thefigure}{\thesection.\arabic{figure}}
\renewcommand{\thetable}{\thesection.\arabic{table}}


\section{Theoretical Derivations and Proofs}
\label{app:theory}

\setcounter{equation}{0}
\setcounter{figure}{0}
\setcounter{table}{0}

\subsection{Identification under correct and incorrect interference assumptions: proofs and additional results}
\label{app:identification_proofs}

\begin{proof}[Proof of Proposition 1]\label{proof:prop_1}
Assumption 2 (treatment unconfoundedness) and Assumption 6 (unconfounded exposure mapping) imply $Y_{ij}(\rma,\mathbf{g}) \indep (A_{ij}, \mathbf{G}_{ij})$.
Then, under this independence, and Assumptions 3 and 5, we have:
\begin{align*}
\E[Y_{ij}(\rma,\mathbf{g})] &= \E[\E[Y_{ij}(\rma,\mathbf{g})|A_{ij}=\rma, \mathbf{G}_{ij}=\mathbf{g}]] \quad (Y_{ij}(\rma,\mathbf{g}) \indep (A_{ij}, \mathbf{G}_{ij}))\\
&= \E[\E[Y_{ij}|A_{ij}=\rma, \mathbf{G}_{ij}=\mathbf{g}]]  \quad \text{(Assumptions 3 and 5)}\\
&= \E[Y_{ij}|A_{ij}=\rma, \mathbf{G}_{ij}=\mathbf{g}]
\end{align*}
\end{proof}

\begin{proof}[Proof of Proposition 2]\label{proof:prop_2}
By the law of total expectation:
{\footnotesize
$$
\begin{aligned}
&E[Y_{ij} \mid A_{ij} =\rma, G^{*}_{ij}=\mathbf{g}]\\
&=\sum_{\mathbf{a}_{\mathcal{I}^{*}_{ij} } \in  \mathcal{A}_{\mathcal{I}^{*}_{ij} }  } \sum_{\mathbf{a}_{\mathcal{I}_{ij} \setminus \mathcal{I}^{*}_{ij} } \in  \mathcal{A}_{\mathcal{I}_{ij} \setminus \mathcal{I}_{ij}^{*} }}E(Y_{ij} \mid  A_{ij}=a, G^{*}_{ij}=\mathbf{g}, A_{\mathcal{I}^{*}_{ij}}=\mathbf{a}_{\mathcal{I}^{*}_{ij} },A_{\mathcal{I}_{ij} \setminus \mathcal{I}_{ij}^{*} }=\mathbf{a}_{\mathcal{I}_{ij} \setminus \mathcal{I}_{ij}^{*} }) \\
&\times P(A_{\mathcal{I}_{ij}^{*}}=\mathbf{a}_{\mathcal{I}_{ij}^{*} }, A_{\mathcal{I}_{ij} \setminus \mathcal{I}_{ij}^{*} }=\mathbf{a}_{\mathcal{I}_{ij} \setminus \mathcal{I}^{*}_{ij} }\mid A_{ij}=\rma, G^{*}_{ij}=\mathbf{g})\\
&=\sum_{\mathbf{a}_{\mathcal{I}^{*}_{ij} } \in  \mathcal{A}_{\mathcal{I}^{*}_{ij} }  } \sum_{\mathbf{a}_{\mathcal{I}_{ij} \setminus \mathcal{I}^{*}_{ij} } \in  \mathcal{A}_{\mathcal{I}_{ij} \setminus \mathcal{I}_{ij}^{*} }}E(Y_{ij} \mid  A_{ij}=a, \mathbf{g}^{*}(\mathbf{a}_{\mathcal{I}^{*}_{ij}})=\mathbf{g}, A_{\mathcal{I}^{*}_{ij}}=\mathbf{a}_{\mathcal{I}^{*}_{ij} },A_{\mathcal{I}_{ij} \setminus \mathcal{I}_{ij}^{*} }=\mathbf{a}_{\mathcal{I}_{ij} \setminus \mathcal{I}_{ij}^{*} }) \\
&\times P(A_{\mathcal{I}_{ij}^{*}}=\mathbf{a}_{\mathcal{I}_{ij}^{*} }, A_{\mathcal{I}_{ij} \setminus \mathcal{I}_{ij}^{*} }=\mathbf{a}_{\mathcal{I}_{ij} \setminus \mathcal{I}^{*}_{ij} }\mid A_{ij}=\rma, \mathbf{g}^{*}(\mathbf{a}_{\mathcal{I}^{*}_{ij}})=\mathbf{g})\quad \text{(using the definition of $G^{*}_{ij}=g^{*}(\cdot)$)}\\
&=
\sum_{\substack{ \mathbf{a}_{\mathcal{I}^{*}_{ij} }
:  \gstar (\mathbf{a}_{\mathcal{I}^{*}_{ij} })=\mathbf{g}
}}
\,\,\,\sum_{\mathbf{a}_{\mathcal{I}_{ij} \setminus \mathcal{I}^{*}_{ij} }
}
\E\bigg[Y_{ij} \bigg(A_{ij}=\rma, \g(\mathbf{A}_{\mathcal{I}_{ij}^{*}}=\mathbf{a}_{\mathcal{I}_{ij}^{*} }, \mathbf{A}_{\mathcal{I}_{ij} \setminus \mathcal{I}_{ij}^{*} }=\mathbf{a}_{\mathcal{I}_{ij} \setminus \mathcal{I}_{ij}^{*} })\bigg)\bigg] \\
&\times P(A_{\mathcal{I}_{ij}^{*}}=\mathbf{a}_{\mathcal{I}_{ij}^{*} }, A_{\mathcal{I}_{ij} \setminus \mathcal{I}_{ij}^{*} }=\mathbf{a}_{\mathcal{I}_{ij} \setminus \mathcal{I}^{*}_{ij} }\mid A_{ij}=\rma,  \mathbf{G}^*_{ij}=\mathbf{g})
\quad \text{(Assumptions 3 and 6)}
\end{aligned}
$$
}
\end{proof}
Corollary 1 and Corollary 2 are special cases of Proposition 2. Therefore, we do not report the specific proofs here.

 \subsection{Misspecified Interference Sets}
  \label{app:overspecification}
We now make the result of Proposition 2 in the main text more general, allowing for the misspecified interference set $\mathcal{I}^{*}_{ij}$ to be larger than the true set $\mathcal{I}_{ij}$ or only partially overlap.
\begin{proposition}[Misspecified Interference Assumptions]
Let $\mathbf{A}$ be determined by a randomized design such that
Assumptions 2 and 3 hold.
Let $\mathcal{I}_{ij} \subseteq  \mathcal{N}_i \backslash  \left\{ij\right\} $ be a subset of the sub-sample of cluster $i$, and let $\mathbf{g}(\cdot): \mathcal{A}(|\mathcal{I}_{ij}|) \rightarrow \mathbb{R}\mystrut^p$ such that Assumptions 5 and 6 hold.
Let $\mathcal{I}^{*}_{ij}$ be a second   set such that $\mathcal{I}^{*}_{ij} \cap \mathcal{I}_{ij}\neq \emptyset$, and let $\mathbf{g}^{*}(\cdot): \mathcal{A}(|\mathcal{I}^{*}_{ij}|) \rightarrow \mathbb{R}\mystrut^d$ be a second unconfounded exposure mapping function for $Y_{ij}(\mathbf{a})$ (Assumption 6 holds), and $\mathbf{G}^*_{ij}=\mathbf{g}^*(\mathbf{A}_{\mathcal{I}^*_{ij}})$.
 Then, the following equality holds:
$$
\begin{aligned}
\E[Y_{ij} \mid A_{ij}
=\rma, \mathbf{G}^*_{ij}=\mathbf{g}]=
\sum_{ \mathbf{a}_{\mathcal{I}_{ij} }
}
\,\,\,\sum_{\mathbf{a}_{\mathcal{I}^{*}_{ij} \setminus \mathcal{I}_{ij} }
}
&\E\bigg[Y_{ij} \bigg(A_{ij}=\rma, \g(\mathbf{A}_{\mathcal{I}_{ij}}=\mathbf{a}_{\mathcal{I}_{ij} } )\bigg)\bigg] \\
&\times P(A_{\mathcal{I}_{ij}}=\mathbf{a}_{\mathcal{I}_{ij} }, A_{\mathcal{I}^{*}_{ij} \setminus \mathcal{I}_{ij}}=\mathbf{a}_{\mathcal{I}^{*}_{ij} \setminus \mathcal{I}_{ij} }\mid A_{ij}=\rma,  \mathbf{G}^*_{ij}=\mathbf{g})\\
\end{aligned}
$$
\end{proposition}
For simplicity, let us consider the case of correct exposure mapping function, i.e., $\g^{*}(\cdot)=\g(\cdot)$, such as the proportion of treated units. Then, the identifying quantity is an average of the potential outcomes where the subset of $\mathcal{I}_{ij}$ overlapping with $\mathcal{I}^{*}_{ij}$ is determined by the value $\mathbf{g}$ that must hold in $\mathcal{I}^{*}_{ij}$. For instance, given $\mathbf{g}=1$, then the treatments in the subset $\mathcal{I}_{ij} \setminus \mathcal{I}^{*}_{ij}$ must be all ones,  and the distribution of the treatment in the complement set depends on the treatment assignment mechanism and its dependency.

\subsection{Proof of the unbiasedness of the Horvitz-Thompson estimator}
\label{app:ht_unbiasedness}
\begin{proof}
Recall the weight defined in Section 4.2:
$$w_{ij}^h(\alpha, \alpha_h) = \frac{P_{\alpha_{h}}\left(\mathbf{A}_{\mathcal{C}^{h}_{ij}}\right) P_{\alpha}\left(\mathbf{A}_{\mathcal{C}_{i\backslash(j, \mathcal{C}^{h}_{ij})}}\right)}{P_{\Delta}\left(\mathbf{A}_{ij,\mathcal{I}_{ij}}\right)}$$

We now show that the Horvitz-Thompson estimator is unbiased for $\bar{Y}^{h}\left(\rma, \alpha, \alpha_{h}\right)$:

{\allowdisplaybreaks
\begin{align*}
&E_{G_0,\mathbf{A}} \left[ \widehat{Y}^{h}_{\text{HT}}\left(\rma, \alpha, \alpha_{h}\right) \right]\\
&= E_{G_0,\mathbf{A}} \left[ \frac{1}{I} \sum_{i=1}^{I}\frac{1}{N^{h}_{i}} \sum_{j \in \mathcal{J}^{h}_{i}} w_{ij}^h(\alpha, \alpha_h)\mathbbm{1}\left(A_{ij}=\rma\right) Y_{ij} \right] \\
&= \frac{1}{I} \sum_{i=1}^{I} E_{G_0} \left[ \frac{1}{N^{h}_{i}} \sum_{j \in \mathcal{J}^{h}_{i}} E_{\mathbf{A}} \left[ w_{ij}^h(\alpha, \alpha_h) \mathbbm{1}\left(A_{ij}=\rma\right) Y_{ij} \right] \right]
\quad \text{(law of iterated expectations)}\\
&= \frac{1}{I} \sum_{i=1}^{I} E_{G_0} \left[ \frac{1}{N^{h}_{i}} \sum_{j \in \mathcal{J}^{h}_{i}} E_{\mathbf{A}} \left[ w_{ij}^h(\alpha, \alpha_h)\mathbbm{1}\left(A_{ij}=\rma\right) Y_{ij}\left(\mathbf{A}_{ij,\mathcal{I}_{ij}}\right) \right] \right]
\quad \text{(Assumption 3)}
\end{align*}

For the inner expectation over $\mathbf{A}$, consider unit $j$ in cluster $i$ where $j \in \mathcal{J}^{h}_{i}$:

\begin{align*}
&E_{\mathbf{A}} \left[ w_{ij}^h(\alpha, \alpha_h)\mathbbm{1}\left(A_{ij}=\rma\right) Y_{ij}\left(\mathbf{A}_{ij,\mathcal{I}_{ij}}\right) \right]\\
&= \sum_{\mathbf{a}_{ij,\mathcal{I}_{ij}}} \frac{P_{\alpha_{h}}\left(\mathbf{a}_{\mathcal{C}^{h}_{ij}}\right) P_{\alpha}\left(\mathbf{a}_{\mathcal{C}_{i\backslash(j, \mathcal{C}^{h}_{ij})}}\right)}{P_{\Delta}\left(\mathbf{a}_{ij,\mathcal{I}_{ij}}\right)} \mathbbm{1}\left(a_{ij}=\rma\right) Y_{ij}\left(\mathbf{a}_{ij,\mathcal{I}_{ij}}\right) P_{\Delta}\left(\mathbf{A}_{ij,\mathcal{I}_{ij}}=\mathbf{a}_{ij,\mathcal{I}_{ij}}\right)\\
&= \sum_{\mathbf{a}_{ \mathcal{C}^h_{ij}}\in \mathcal{A}_{ \mathcal{C}^h_{ij}}} \sum_{\mathbf{a}_{ \mathcal{C}_{i\backslash(j, \mathcal{C}^{h}_{ij})}}\in \mathcal{A}_{ \mathcal{C}_{i\backslash(j, \mathcal{C}^{h}_{ij})}}} P_{\alpha_{h}}\left(\mathbf{a}_{\mathcal{C}^{h}_{ij}}\right) P_{\alpha}\left(\mathbf{a}_{\mathcal{C}_{i\backslash(j, \mathcal{C}^{h}_{ij})}}\right) Y_{ij}\left(\rma, \mathbf{a}_{\mathcal{C}^{h}_{ij}}, \mathbf{a}_{\mathcal{C}_{i\backslash(j, \mathcal{C}^{h}_{ij})}}\right)\\
&= \overline{Y}_{ij}^{h}\left(\rma, \alpha, \alpha_{h}\right)
\end{align*}

Finally, we have:
\begin{align*}
&E_{G_0,\mathbf{A}} \left[ \widehat{Y}^{h}_{\text{HT}}\left(\rma, \alpha, \alpha_{h}\right) \right]\\
&= \frac{1}{I} \sum_{i=1}^{I} E_{G_0} \left[ \frac{1}{N^{h}_{i}} \sum_{j \in \mathcal{J}^{h}_{i}} \overline{Y}_{ij}^{h}\left(\rma, \alpha, \alpha_{h}\right) \right]\\
&= \frac{1}{I} \sum_{i=1}^{I} E_{G_0} \left[ \overline{Y}_i^{h}\left(\rma, \alpha, \alpha_{h}\right) \right]\\
&= E_{G_0} \left[ \overline{Y}_i^{h}\left(\rma, \alpha, \alpha_{h}\right) \right]
\quad \text{(Assumption 5, clusters are i.i.d.)}\\
&= \bar{Y}^{h}\left(\rma, \alpha, \alpha_{h}\right)
\end{align*}
}

\end{proof}

\subsection{Asymptotic properties of the Horvitz-Thompson estimator}
\label{app:ht_variance}
 Let $\boldsymbol{\theta}=(\theta_1, \theta_2)=(\bar{Y}^{h}(\rma,\alpha_h,\alpha), \bar{Y}^{h}(\rma,\alpha'_h,\alpha))$ be the vector of two population average potential outcomes with different $\alpha_h$. Let $\mathbf{O}_i = \{\mathbf{Y}_i, \mathbf{A}_i, \mathcal{C}^h_{ij}\}_{j=1}^{n_i}$ denote the observed data for cluster $i$. We can then construct estimating equations $\sum^{I}_i\boldsymbol{\psi}_{\text{HT}, i}(\mathbf{O}_i; \boldsymbol{\theta})=0$ as follows:

$$
\boldsymbol{\psi}_{\text{HT}, i}(\mathbf{O}_i; \boldsymbol{\theta})=\left(\begin{array}{c}
\psi_{\text{HT}, i,\alpha_h}(\mathbf{O}_i; \boldsymbol{\theta}) \\
\psi_{\text{HT}, i,\alpha'_h}(\mathbf{O}_i; \boldsymbol{\theta})
\end{array}\right)
$$
where
$$
\psi_{\text{HT}, i,\alpha_h}(\mathbf{O}_i; \boldsymbol{\theta}) = \widehat{Y}^{h}_{\text{HT},i}(\alpha_h, \alpha, \rma) - \theta_1,\quad  \psi_{\text{HT}, i,\alpha'_h}(\mathbf{O}_i; \boldsymbol{\theta}) = \widehat{Y}^{h}_{\text{HT},i}(\alpha'_h, \alpha, \rma) - \theta_2
$$

with
$$
\widehat{Y}^{h}_{\text{HT},i}(\alpha_h, \alpha, \rma) = \frac{1}{N^{h}_{i}} \sum_{j \in \mathcal{J}^{h}_{i}} w_{ij}^h(\alpha, \alpha_h)  \mathbbm{1}\left(A_{ij}=\rma\right) Y_{ij}
$$

The solution to $E_{G_o}[ \boldsymbol{\psi}_{\text{HT}, i}(\mathbf{O}_i; \boldsymbol{\theta})]= 0$
is
$$\boldsymbol{\theta} = \left(\begin{array}{c}
E_{G_o}[\frac{1}{N^{h}_{i}} \sum_{j \in \mathcal{J}^{h}_{i}} w_{ij}^h(\alpha, \alpha_h)  \mathbbm{1}\left(A_{ij}=\rma\right) Y_{ij}] \\
E_{G_o}[\frac{1}{N^{h}_{i}} \sum_{j \in \mathcal{J}^{h}_{i}} w_{ij}^h(\alpha, \alpha^{\prime}_h) \mathbbm{1}\left(A_{ij}=\rma\right) Y_{ij}]
\end{array}\right) = \left(\begin{array}{c}
\bar{Y}^{h}\left(\rma, \alpha, \alpha_{h}\right) \\
\bar{Y}^{h}\left(\rma, \alpha, \alpha^{\prime}_{h}\right)
\end{array}\right)$$

Notice that the solution to the first equation $\sum_{i=1}^{I}\psi_{\text{HT}, i,\alpha_h}(\mathbf{O}_i; \boldsymbol{\theta})=0$ is our Horvitz-Thompson estimator given treatment $\rma$ and hypothetical parameter $\alpha$ and $\alpha_h$:
$$
\hat{\theta}_1 = \widehat{Y}^{h}_{\text{HT}}(\rma, \alpha_{h}, \alpha)=\frac{1}{I} \sum_{i=1}^{I}\frac{1}{N^{h}_{i}} \sum_{j \in \mathcal{J}^{h}_{i}} w_{ij}^h(\alpha, \alpha_h) \mathbbm{1}\left(A_{ij}=\rma\right) Y_{ij}
$$

Similarly, the solution to the second equation is $\hat{\theta}_2 = \widehat{Y}^{h}_{\text{HT}}(\alpha'_{h}, \alpha, \rma)$.

Let $\hat{\boldsymbol{\theta}}=(\hat{\theta}_1,\hat{\theta}_2)$.
Since $\boldsymbol{\psi}_{\text{HT}, i}(\mathbf{O}_i, \boldsymbol{\theta})$ is monotone in $\boldsymbol{\theta}$, both $\sum^{I}_i \boldsymbol{\psi}_{\text{HT}, i}(\mathbf{O}_i, \boldsymbol{\theta})$ and $\int \boldsymbol{\psi}_{\text{HT}, i}(\mathbf{O}_i; \boldsymbol{\theta})$ are monotone in $\boldsymbol{\theta}$, which implies uniqueness of the roots and establishes $
\hat{\boldsymbol{\theta}} \xrightarrow{p} \boldsymbol{\theta}
$.

In our setting, we have $I$ independent, non-overlapping clusters. Therefore, under suitable regularity conditions and due to the unbiased estimating equations, as $I \rightarrow \infty$, we obtain the consistent and asymptotically normal sandwich-type estimator of the variance below.

From Theorem 5.4.1 of \citet{van2000asymptotic}, using Slutsky’s Theorem and Delta method when $I \rightarrow \infty$, we have that
$$
\sqrt{I}(\hat{\boldsymbol{\theta}}-\boldsymbol{\theta}) \xrightarrow{d} N(0, \bm{\Sigma}_{\text{HT}})
$$
where
$$
\boldsymbol{\Sigma}_{\text{HT}}=A_{\text{HT}}\left(\boldsymbol{\theta} \right)^{-1} B_{\text{HT}}\left(\boldsymbol{\theta} \right)\left[A_{\text{HT}}\left(\boldsymbol{\theta}\right)^{-1}\right]^T
$$
$$
A_{\text{HT}}(\boldsymbol{\theta}) = -E\left[\frac{\partial}{\partial \boldsymbol{\theta}^T} \boldsymbol{\psi}_{\text{HT}, i}(\mathbf{O}_i; \boldsymbol{\theta})\right] = \left(\begin{array}{cc}
1 & 0 \\
0 & 1
\end{array}\right)
$$
$$
B_{\text{HT}}(\boldsymbol{\theta}) = E[\boldsymbol{\psi}_{\text{HT}, i}(\mathbf{O}_i; \boldsymbol{\theta}) \boldsymbol{\psi}_{\text{HT}, i}(\mathbf{O}_i; \boldsymbol{\theta})^T]
$$

Since $A_{\text{HT}}(\boldsymbol{\theta})$ is an identity matrix, $\bm{\Sigma}_{\text{HT}} = B_{\text{HT}}(\boldsymbol{\theta})$. We can estimate it by taking the mean across different clusters, using the empirical average:
$$
\widehat{\boldsymbol{\Sigma}}_{\text{HT}}  = \widehat{B}_{\text{HT}} = \frac{1}{I} \sum_{i=1}^I \boldsymbol{\psi}_{\text{HT}, i}(\mathbf{O}_i; \hat{\boldsymbol{\theta}}) \boldsymbol{\psi}_{\text{HT}, i}(\mathbf{O}_i; \hat{\boldsymbol{\theta}})^T
$$

 The variance of the spillover effect $\widehat{\text{SE}}^{h}_{\text{HT}} = \hat{\theta}_1 - \hat{\theta}_2$ can then be derived as:
$$
\begin{aligned}
\widehat{\text{Var}}(\widehat{\text{SE}}^{h}_{\text{HT}}) &= \widehat{\text{Var}}(\hat{\theta}_1) + \widehat{\text{Var}}(\hat{\theta}_2) - 2\widehat{\text{Cov}}(\hat{\theta}_1, \hat{\theta}_2) \\
&= \frac{1}{I} \left[ \widehat{B}_{\text{HT},11} + \widehat{B}_{\text{HT},22} - 2\widehat{B}_{\text{HT},12} \right]
\end{aligned}
$$

where

$$
\widehat{B}_{\text{HT},11} = \frac{1}{I} \sum_{i=1}^I \left[\widehat{Y}^{h}_{\text{HT},i}(\alpha_h, \alpha, \rma)- \hat{\theta}_1\right]^2
$$

$$
\widehat{B}_{\text{HT},22} = \frac{1}{I} \sum_{i=1}^I \left[\widehat{Y}^{h}_{\text{HT},i}(\alpha^\prime_h, \alpha, \rma)- \hat{\theta}_2\right]^2
$$
$$
\widehat{B}_{\text{HT},12} = \frac{1}{I} \sum_{i=1}^I \left[\widehat{Y}^{h}_{\text{HT},i}(\alpha_h, \alpha, \rma) - \hat{\theta}_1\right]\left[\widehat{Y}^{h}_{\text{HT},i}(\alpha^\prime_h, \alpha, \rma) - \hat{\theta}_2\right]
$$

\subsection{Asymptotic properties  of the Hajek estimator}
\label{app:hajek_variance}
Similar to the proof for the Horvitz-Thompson estimator, let $\boldsymbol{\theta}=(\theta_1, \theta_2)=(\bar{Y}^{h}(\rma,\alpha,\alpha_h), \bar{Y}^{h}(\rma,\alpha,\alpha'_h))$ be the vector of two population average potential outcomes with different $\alpha_h$. Let $\mathbf{O}_i = \{\mathbf{Y}_i, \mathbf{A}_i, \mathcal{C}^h_{ij}\}_{j=1}^{n_i}$ denote the observed data for cluster $i$. We can then construct estimating equations $\sum^{I}_i\boldsymbol{\psi}_{\text{Hajek}, i}(\mathbf{O}_i; \boldsymbol{\theta})=0$ as follows:
$$
\boldsymbol{\psi}_{\text{Hajek}, i}(\mathbf{O}_i; \boldsymbol{\theta})=\left(\begin{array}{c}
\psi_{\text{Hajek}, i,\alpha_h}(\mathbf{O}_i; \boldsymbol{\theta}) \\
\psi_{\text{Hajek}, i,\alpha'_h}(\mathbf{O}_i; \boldsymbol{\theta})
\end{array}\right)
$$

where
$$
\psi_{\text{Hajek}, i,\alpha_h}(\mathbf{O}_i; \boldsymbol{\theta}) = \widehat{W}^{h}_{i}(\alpha_h,\alpha,\rma) - \widehat{N}^{h}_{i}(\alpha_h, \alpha,\rma)\theta_1
$$
$$
\psi_{\text{Hajek}, i,\alpha'_h}(\mathbf{O}_i; \boldsymbol{\theta}) = \widehat{W}^{h}_{i }(\alpha'_h,\alpha, \rma) - \widehat{N}^{h}_{i}(\alpha'_h,\alpha, \rma)\theta_2
$$

with
$$
\widehat{W}^{h}_{i }(\alpha_h,\alpha, \rma) = \sum_{j \in \mathcal{J}^{h}_{i}} \frac{P_{\alpha_{h}}\left(\mathbf{A}_{\mathcal{C}^{h}_{ij}}\right) P_{\alpha}\left(\mathbf{A}_{\mathcal{C}_{i\backslash(j, \mathcal{C}^{h}_{ij})}}\right)}{P_{\Delta}\left(\mathbf{A}_{ij,\mathcal{I}_{ij}}\right)} \mathbbm{1}\left(A_{ij}=\rma\right) Y_{ij}
$$
$$
\widehat{N}^{h}_{i}(\alpha_h,\alpha, \rma) = \sum_{j \in \mathcal{J}^{h}_{i}} \frac{P_{\alpha_{h}}\left(\mathbf{A}_{\mathcal{C}^{h}_{ij}}\right) P_{\alpha}\left(\mathbf{A}_{\mathcal{C}_{i\backslash(j, \mathcal{C}^{h}_{ij})}}\right)}{P_{\Delta}\left(\mathbf{A}_{ij,\mathcal{I}_{ij}}\right)} \mathbbm{1}\left(A_{ij}=\rma\right)
$$

Since $E_{G_0}[\widehat{N}^{h}_{i }(\alpha_h,\alpha, \rma)] = N^h_i$, true number of units with $h$-neighbors in cluster $i$, then the solution to $\int \boldsymbol{\psi}_{\text{Hajek}, i}(\mathbf{O}_i; \boldsymbol{\theta}) dG_0= 0$ is:
$$
\boldsymbol{\theta} = \left(\begin{array}{c}
\frac{E_{G_0}[\widehat{W}^{h}_{i }(\alpha_h,\alpha, \rma)]}{E_{G_0}[\widehat{N}^{h}_{i }(\alpha_h,\alpha, \rma)]} \\
\frac{E_{G_0}[\widehat{W}^{h}_{i }(\alpha'_h,\alpha, \rma)]}{E_{G_0}[\widehat{N}^{h}_{i }(\alpha'_h,\alpha, \rma)]}
\end{array}\right) = \left(\begin{array}{c}
\frac{E_{G_0}[\sum_{j \in \mathcal{J}^{h}_{i}} w_{ij}^h(\alpha, \alpha_h)\mathbbm{1}\left(A_{ij}=\rma\right) Y_{ij}]}{N^h_i} \\
\frac{E_{G_0}[\sum_{j \in \mathcal{J}^{h}_{i}} w_{ij}^h(\alpha, \alpha^{\prime}_h)\mathbbm{1}\left(A_{ij}=\rma\right) Y_{ij}]}{N^h_i}
\end{array}\right) =\left(\begin{array}{c}
\bar{Y}^{h}\left(\rma, \alpha, \alpha_{h}\right) \\
\bar{Y}^{h}\left(\rma, \alpha, \alpha^{\prime}_{h}\right)
\end{array}\right)
$$

Notice that the solution to the first equation $\sum_{i=1}^{I}\psi_{\text{Hajek},i,\alpha_h}(\mathbf{O}_i; \boldsymbol{\theta})=0$ is our Hajek estimator given treatment $\rma$ and hypothetical parameters $\alpha$ and $\alpha_h$:
$$
\hat{\theta}_1 = \widehat{Y}^{h}_{\text{Hajek}}(\rma, \alpha, \alpha_{h})=\frac{\sum_{i=1}^{I} \widehat{W}^{h}_{i}(\alpha_h,\alpha, \rma)}{\sum_{i=1}^{I} \widehat{N}^{h}_{i}(\alpha_h,\alpha, \rma)}
$$

Similarly, the solution to the second equation $\sum_{i=1}^{I}\psi_{Hajek, i,\alpha'_h}(\mathbf{O}_i; \boldsymbol{\theta})=0$ is the Hajek estimator given treatment $\rma$ and hypothetical parameters $\alpha$ and $\alpha'_h$: $\hat{\theta}_2 = \widehat{Y}^{h}_{\text{Hajek}}(\rma, \alpha, \alpha'_{h})$. This limit $\boldsymbol{\theta}$ corresponds to the individual-weighted average potential outcome, which is the estimand defined in \citet{liu2016inverse}. As discussed in the Appendix of \citet{papadogeorgou2019causal}, it coincides with the cluster-weighted estimand $\bar{Y}^h$ when the cluster size $N_i^h$ is uncorrelated with the cluster-level average outcome.

Let $\hat{\boldsymbol{\theta}}=(\hat{\theta}_1
,\hat{\theta}_2)$. Under suitable regularity conditions, by Theorem 5.4.2 of \citet{van2000asymptotic}, $
\hat{\boldsymbol{\theta}} \xrightarrow{p} \boldsymbol{\theta}
$.  From Theorem 5.4.1 of \citet{van2000asymptotic}, using Slutsky's Theorem and the Delta method as $I \rightarrow \infty$, we have:
$$
\sqrt{I}(\hat{\boldsymbol{\theta}}-\boldsymbol{\theta}) \xrightarrow{d} N(0, \bm{\Sigma}_{\text{Hajek}})
$$

where $\boldsymbol{\Sigma}_{\text{Hajek}}=A_{\text{Hajek}}(\boldsymbol{\theta})^{-1} B_{\text{Hajek}}(\boldsymbol{\theta})[A_{\text{Hajek}}(\boldsymbol{\theta})^{-1}]^T$, with:

$$
A_{\text{Hajek}}(\boldsymbol{\theta}) = -E\left[\frac{\partial}{\partial \boldsymbol{\theta}^T} \boldsymbol{\psi}_{\text{Hajek}, i}(\mathbf{O}_i; \boldsymbol{\theta})\right] = \left(\begin{array}{cc}
E[\widehat{N}^{h}_{i}(\alpha_h,\alpha, \rma)] & 0 \\
0 & E[\widehat{N}^{h}_{i}(\alpha'_h,\alpha, \rma)]
\end{array}\right)
$$

$$
B_{\text{Hajek}}(\boldsymbol{\theta}) = E[\boldsymbol{\psi}_{\text{Hajek}, i}(\mathbf{O}_i; \boldsymbol{\theta}) \boldsymbol{\psi}_{\text{Hajek}, i}(\mathbf{O}_i; \boldsymbol{\theta})^T]
$$

We can estimate $\bm{\Sigma}_{\text{Hajek}}$ using the empirical counterparts $\widehat{A}_{\text{Hajek}}(\hat{\boldsymbol{\theta}})$ and $B_I(\hat{\boldsymbol{\theta}})$:

$$
\widehat{A}_{\text{Hajek}} = \left(\begin{array}{cc}
\widehat{A}_{\text{Hajek},11} & 0 \\
0 & \widehat{A}_{\text{Hajek},22}
\end{array}\right)
$$

where:
$$
\widehat{A}_{\text{Hajek},11}= \frac{1}{I}\sum_{i=1}^{I} \widehat{N}^{h}_{i}(\alpha_h,\alpha, \rma), \quad \widehat{A}_{\text{Hajek},22} = \frac{1}{I}\sum_{i=1}^{I} \widehat{N}^{h}_{i}(\alpha'_h,\alpha, \rma)
$$

$$
\widehat{B}_{\text{Hajek}} = \frac{1}{I} \sum_{i=1}^I \boldsymbol{\psi}_{\text{Hajek}, i}(\mathbf{O}_i; \hat{\boldsymbol{\theta}}) \boldsymbol{\psi}_{\text{Hajek}, i}(\mathbf{O}_i; \hat{\boldsymbol{\theta}})^T = \left(\begin{array}{cc}
\widehat{B}_{\text{Hajek},11} & \widehat{B}_{\text{Hajek},12} \\
\widehat{B}_{\text{Hajek},21} & \widehat{B}_{\text{Hajek},22}
\end{array}\right)
$$

where:
$$
\widehat{B}_{\text{Hajek},11} = \frac{1}{I} \sum_{i=1}^{I} \left[\widehat{W}^{h}_{i}(\alpha_h,\alpha, \rma) - \widehat{N}^{h}_{i}(\alpha_h,\alpha, \rma)\hat{\theta}_1\right]^2
$$
$$
\widehat{B}_{\text{Hajek},22} = \frac{1}{I} \sum_{i=1}^{I} \left[\widehat{W}^{h}_{i}(\alpha'_h,\alpha, \rma) - \widehat{N}^{h}_{i}(\alpha'_h,\alpha, \rma)\hat{\theta}_2\right]^2
$$
$$
\widehat{B}_{\text{Hajek},12} = \widehat{B}_{\text{Hajek},21} = \frac{1}{I} \sum_{i=1}^{I} \left[\widehat{W}^{h}_{i}(\alpha_h,\alpha, \rma) - \widehat{N}^{h}_{i}(\alpha_h,\alpha, \rma)\hat{\theta}_1\right]\left[\widehat{W}^{h}_{i}(\alpha'_h,\alpha, \rma) - \widehat{N}^{h}_{i}(\alpha'_h,\alpha, \rma)\hat{\theta}_2\right]
$$

The asymptotic covariance matrix is estimated by $\boldsymbol{\Sigma}_{\text{Hajek}} = \widehat{A}_{\text{Hajek}}^{-1} \widehat{B}_{\text{Hajek}} \widehat{A}^{-T}_{\text{Hajek}}$. Since $\widehat{{A}}_{\text{Hajek}}$ is diagonal, its inverse is $\widehat{\mathbf{A}}_{\text{Hajek}}^{-1} = \text{diag}(1/\widehat{A}_{Hajek,11}, 1/\widehat{A}_{Hajek,22})$. Expanding the sandwich estimator:
$$
\widehat{\boldsymbol{\Sigma}}_{\text{Hajek}} =
\begin{pmatrix}
\frac{\widehat{B}_{\text{Hajek},11}}{\widehat{A}_{\text{Hajek},11}^2} & \frac{\widehat{B}_{Hajek,12}}{\widehat{A}_{Hajek,11}\widehat{A}_{Hajek,22}} \\
\frac{\widehat{B}_{Hajek,21}}{\widehat{A}_{Hajek,22}\widehat{A}_{Hajek,11}} & \frac{\widehat{B}_{Hajek,22}}{\widehat{A}_{Hajek,22}^2}
\end{pmatrix}
$$
The spillover effect is the difference $\widehat{\text{SE}}^{h} = \hat{\theta}_1 - \hat{\theta}_2$. By the Delta method, its variance is the sum of the marginal variances minus twice the covariance:
$$
\widehat{\text{Var}}(\widehat{\text{SE}}^{h}_{\text{Hajek}}) = \widehat{\boldsymbol{\Sigma}}_{\text{Hajek},11} + \widehat{\Sigma}_{Hajek,22} - 2\widehat{\Sigma}_{Hajek,12} = \frac{1}{I} \left[ \frac{\widehat{B}_{Hajek,11}}{\widehat{A}_{Hajek,11}^{2}} + \frac{\widehat{B}_{Hajek,22}}{\widehat{A}_{Hajek,22}^{2}} - 2 \frac{\widehat{B}_{Hajek,12}}{\widehat{A}_{Hajek,11}\widehat{A}_{Hajek,22}} \right]
$$

\subsection{Asymptotic properties of the weighted regression estimator }
\label{app:wls_variance}
Similarly to the previous proof, we adopt an M-estimation approach to derive the asymptotic variance of the weighted regression estimator. Let $\boldsymbol{\gamma}^{\alpha} = (\gamma^{\alpha}_0, \gamma^{\alpha}_1, \gamma^{\alpha}_2, \gamma^{\alpha}_3)^T$ be the true coefficients in the marginal structural model:
$$E[Y_{ij}(\rma,\alpha,\alpha_h)] = \gamma^{\alpha}_0 + \gamma^{\alpha}_1 \rma + \gamma^{\alpha}_2 \alpha_h + \gamma^{\alpha}_3 \rma \cdot \alpha_h$$
Let $\mathbf{O}_i = \{\mathbf{Y}_i, \mathbf{A}_i, \mathbf{S}_i, \mathcal{C}^h_{ij}\}_{j=1}^{n_i}$ denote the observed data for cluster $i$, where $\mathbf{S}_i = (S_{ijk})_{j=1,\ldots,n_i; k=1,\ldots,K}$ contains the pseudo-covariates from the data expansion. Let $C_{ijk} = (1, A_{ij}, S_{ijk}, A_{ij} \cdot S_{ijk})^T$ where $S_{ijk}$ represents the $k$-th value of $\alpha_h$. Let $w^{h, \alpha}_{ijk}$ correspond to the weight defined in Section 4.2 for the $k$-th replicate. The weighted least squares estimator solves:
$$\hat{\boldsymbol{\gamma}}^{\alpha} = \arg \min_{\boldsymbol{\gamma}^{\alpha}} \sum_{i=1}^{I} \sum_{k=1}^{K} \sum_{j=1}^{n_i} (Y_{ij} - C_{ijk}^T \boldsymbol{\gamma}^{\alpha})^2 w^{h, \alpha}_{ijk}$$
Taking the derivative and setting it to zero yields the estimating equation $$\sum_{i=1}^{I} \boldsymbol{\psi}_{WLS, i}(\mathbf{O}_i; \boldsymbol{\gamma}^{\alpha}) = 0$$, where:
$$\boldsymbol{\psi}_{WLS, i}(\mathbf{O}_i; \boldsymbol{\gamma}^{\alpha}) = \sum_{k=1}^{K} \sum_{j=1}^{n_i} (Y_{ij} - C_{ijk}^T \boldsymbol{\gamma}^{\alpha}) C_{ijk} w^{h, \alpha}_{ijk}$$

The solution to $\int \boldsymbol{\psi}_{WLS, i}(\mathbf{O}_i; \boldsymbol{\gamma}^{\alpha}) dG_0 = 0$ is the true parameter $\boldsymbol{\gamma}^{\alpha}$, since:
$$E_{G_0}[\boldsymbol{\psi}_{WLS, i}(\mathbf{O}_i; \boldsymbol{\gamma}^{\alpha})] = E_{G_0}\left[\sum_{k=1}^{K} \sum_{j=1}^{n_i} (Y_{ij} - C_{ijk}^T \boldsymbol{\gamma}^{\alpha}) C_{ijk} w^{h, \alpha}_{ijk}\right] = 0$$
when $\boldsymbol{\gamma}^{\alpha}$ satisfies the marginal structural model.

Similar to before, under suitable regularity conditions, by Theorem 5.4.2 of \citet{van2000asymptotic}, $\hat{\boldsymbol{\gamma}}^{\alpha} \xrightarrow{p} \boldsymbol{\gamma}^{\alpha}$. From Theorem 5.4.1 of \citet{van2000asymptotic}, using Slutsky's Theorem and the Delta method as $I \rightarrow \infty$, we have:
$$\sqrt{I}(\hat{\boldsymbol{\gamma}}^{\alpha} - \boldsymbol{\gamma}^{\alpha}) \xrightarrow{d} N(0, \bm{\Sigma}_{\text{WLS}})$$

where $\bm{\Sigma}_{\text{WLS}} = A_{\text{WLS}}(\boldsymbol{\gamma}^{\alpha})^{-1} B_{\text{WLS}}(\boldsymbol{\gamma}^{\alpha})[A_{\text{WLS}}(\boldsymbol{\gamma}^{\alpha})^{-1}]^T$, with:
$$A_{\text{WLS}}(\boldsymbol{\gamma}^{\alpha}) = -E\left[\frac{\partial}{\partial (\boldsymbol{\gamma}^{\alpha})^T} \boldsymbol{\psi}_{WLS, i}(\mathbf{O}_i; \boldsymbol{\gamma}^{\alpha})\right] = E\left[\sum_{k=1}^{K} \sum_{j=1}^{n_i} C_{ijk} C_{ijk}^T w^{h, \alpha}_{ijk}\right]$$
$$B_{\text{WLS}}(\boldsymbol{\gamma}^{\alpha}) = E[\boldsymbol{\psi}_{WLS, i}(\mathbf{O}_i; \boldsymbol{\gamma}^{\alpha}) \boldsymbol{\psi}_{WLS, i}(\mathbf{O}_i; \boldsymbol{\gamma}^{\alpha})^T]$$

We can estimate $\bm{\Sigma}_{\text{WLS}}$ using the empirical counterparts $A_{\text{WLS}}(\hat{\boldsymbol{\gamma}}^{\alpha})$ and $B_{\text{WLS}}(\hat{\boldsymbol{\gamma}}^{\alpha})$:
$$A_{\text{WLS}}(\hat{\boldsymbol{\gamma}}^{\alpha}) = \frac{1}{I}\sum_{i=1}^{I} \sum_{k=1}^{K} \sum_{j=1}^{n_i} C_{ijk} C_{ijk}^T w^{h, \alpha}_{ijk}$$
$$
\begin{aligned}
B_{\text{WLS}}(\hat{\boldsymbol{\gamma}}^{\alpha}) &= \frac{1}{I} \sum_{i=1}^{I} \boldsymbol{\psi}_{WLS, i}(\mathbf{O}_i; \hat{\boldsymbol{\gamma}}^{\alpha}) \boldsymbol{\psi}_{WLS, i}(\mathbf{O}_i; \hat{\boldsymbol{\gamma}}^{\alpha})^T\\
&=\frac{1}{I} \sum_{i=1}^{I} \left[\sum_{k=1}^{K} \sum_{j=1}^{n_i} (Y_{ij} - C_{ijk}^T \hat{\boldsymbol{\gamma}}^{\alpha}) C_{ijk} w^{h, \alpha}_{ijk}\right]\left[\sum_{k=1}^{K} \sum_{j=1}^{n_i} (Y_{ij} - C_{ijk}^T \hat{\boldsymbol{\gamma}}^{\alpha}) C_{ijk} w^{h, \alpha}_{ijk}\right]^T
\end{aligned}
$$

Therefore:
$$\hat{\boldsymbol{\Sigma}}_{\text{WLS}} = \frac{1}{I} A_{\text{WLS}}(\hat{\boldsymbol{\gamma}}^{\alpha})^{-1} B_{\text{WLS}}(\hat{\boldsymbol{\gamma}}^{\alpha}) [A_{\text{WLS}}(\hat{\boldsymbol{\gamma}}^{\alpha})^{-1}]^T$$

For the estimated $h$-order spillover effect: $
\widehat{\text{SE}}_{\text{WLS}}(\rma, \alpha, \alpha_h, \alpha_h^{\prime},\hat{\boldsymbol{\gamma}}^{\alpha})= \hat{\gamma}_2 (\alpha_h -\alpha_h^{\prime}) + \hat{\gamma}_3 \rma (\alpha_h -\alpha_h^{\prime}) $
We can get its estimated variance as
$$
\begin{aligned}
\widehat{Var}(\widehat{\text{SE}}_{\text{WLS}}(a, \alpha, \alpha_h, \alpha_h^{\prime},\hat{\boldsymbol{\gamma}}^{\alpha}))&= (\alpha_h-\alpha^{'}_h)^2 \left[ \widehat{Var}(\hat{\gamma^{\alpha}_2} )+  \rma^2 \widehat{Var}(\hat{\gamma^{\alpha}_3} )+  2\rma  \widehat{Cov}(\hat{\gamma^{\alpha}_2}, \hat{\gamma^{\alpha}_3}) \right]\\
&= \frac{1}{I}(\alpha_h - \alpha'_h)^2 \left[\widehat{\boldsymbol{\Sigma}}_{WLS, 22} + \rma^2 \widehat{\boldsymbol{\Sigma}}_{WLS,33} + 2\rma \widehat{\boldsymbol{\Sigma}}_{WLS,23}\right]
\end{aligned}
$$
where $\widehat{\boldsymbol{\Sigma}}_{\mathrm{WLS},kl}$ denotes the $(k,l)$th element of $\widehat{\boldsymbol{\Sigma}}_{\mathrm{WLS}}$.

\section{Supplemental Materials for Simulation Study}
\setcounter{equation}{0}
\setcounter{figure}{0}
\setcounter{table}{0}
\label{app:relative_bias}
In this section, we provide the complete set of results from the Monte-Carlo simulation study described in Section 5. The results are organized as follows: Appendix \ref{app:bias} presents all figures for the relative bias of the estimators. Appendix \ref{app:variance} contains all figures related to the variance of the estimators, including the validation of analytical variance, comparisons across estimators, and the study of network density effects. Appendix \ref{app:correlation} provides supplementary analysis on the empirical correlation between exposure mappings.

\subsection{Bias of Spillover Effect Estimators} \label{app:bias}
This section contains the full results for the bias of all the estimators across the four scenarios, where the networks are generated under a regular graph model with a fixed number of connections equal to  $n_i\times p=10\times0.2$. We vary the second-stage assignment probabilities $(\delta_1, \delta_2)$, as well as the parameters $\alpha$ and $\alpha_h$ defining the spillover effects that we estimate. Each  figure, corresponding to one scenario, is further broken down by the order of the spillover effect ($h \in 1,2$) and the individual's treatment status ($a \in 0,1$).

\subsubsection{Scenario 1: No interference}

As shown in Figure~\ref{fig:scenario1_bias}, all estimators exhibit near-zero normalized bias under Scenario 1 (no interference).

\begin{figure*}[!htbp]
    \centering
    \begin{subfigure}{0.5\linewidth}
        \centering
        \includegraphics[height=5.5cm]{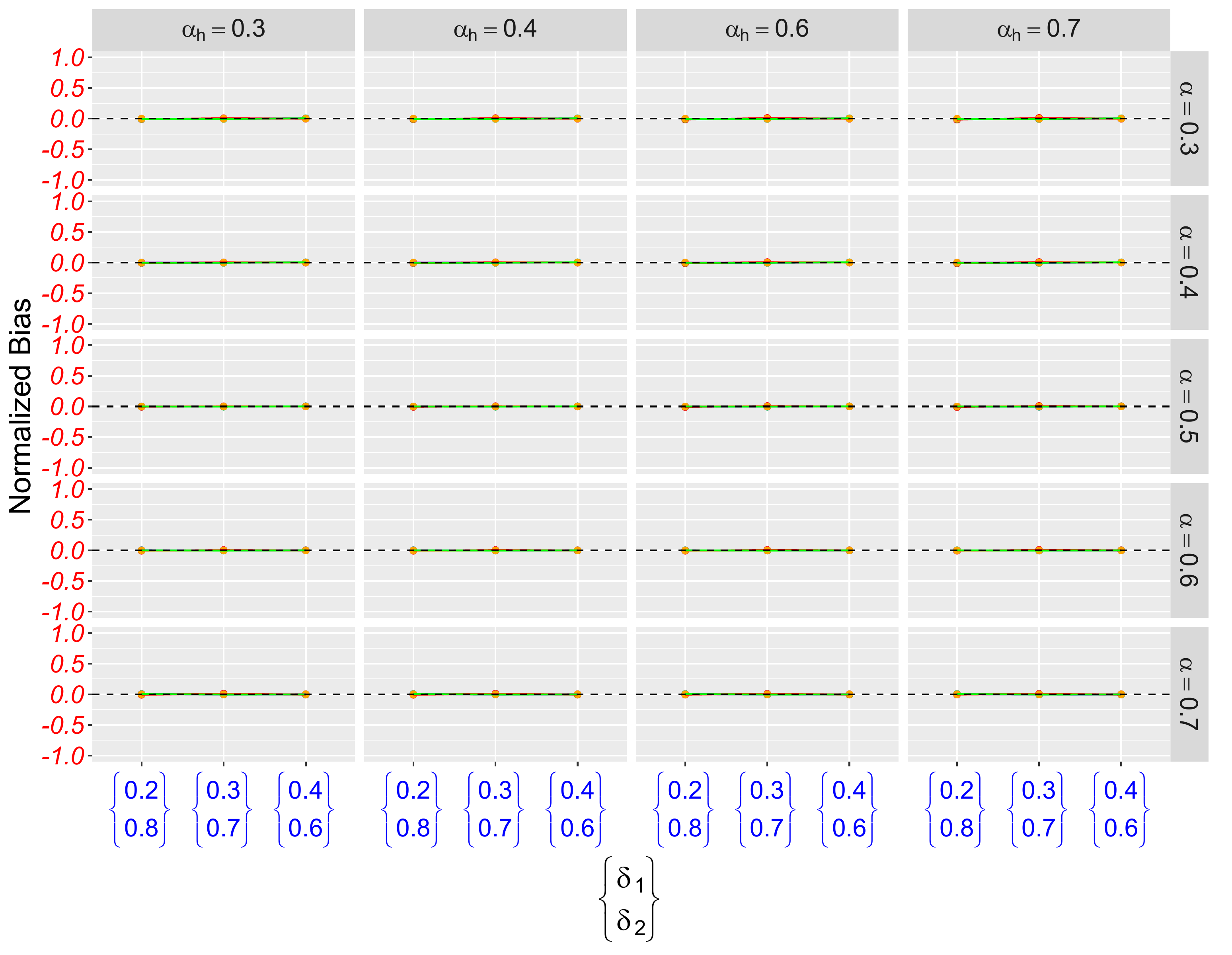}
        \caption{$\rma=0, h=1$}
        \label{fig:scen1_a0h1}
    \end{subfigure}
    \begin{subfigure}{0.45\linewidth}
        \centering
        \includegraphics[height=5.5cm]{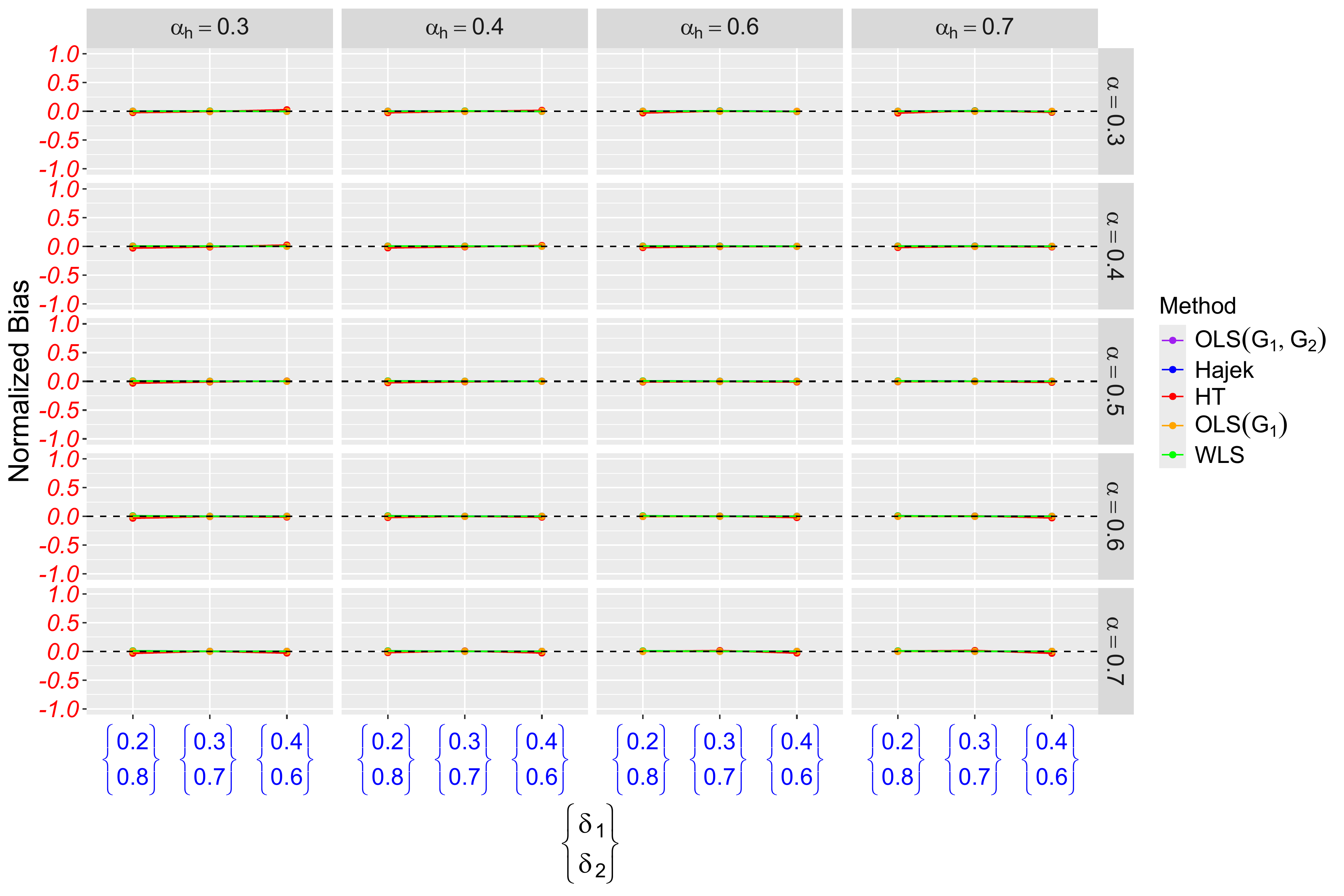}
        \caption{$\rma=1, h=1$}
        \label{fig:scen1_a1h1}
    \end{subfigure}

    \vspace{0.5cm}

    \begin{subfigure}{0.5\linewidth}
        \centering
        \includegraphics[height=5.5cm]{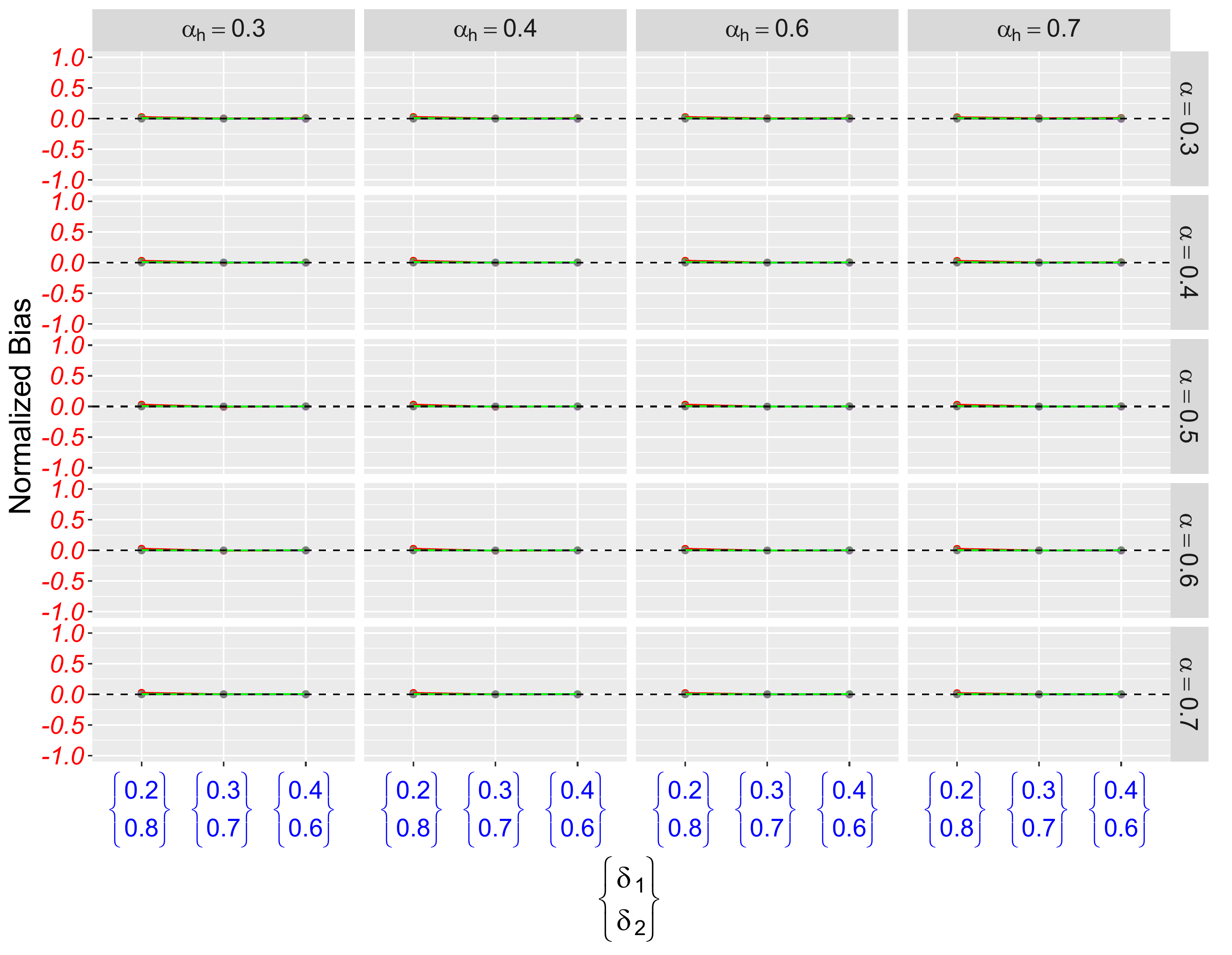}
        \caption{$\rma=0, h=2$}
        \label{fig:scen1_a0h2}
    \end{subfigure}
    \begin{subfigure}{0.45\linewidth}
        \centering
        \includegraphics[height=5.5cm]{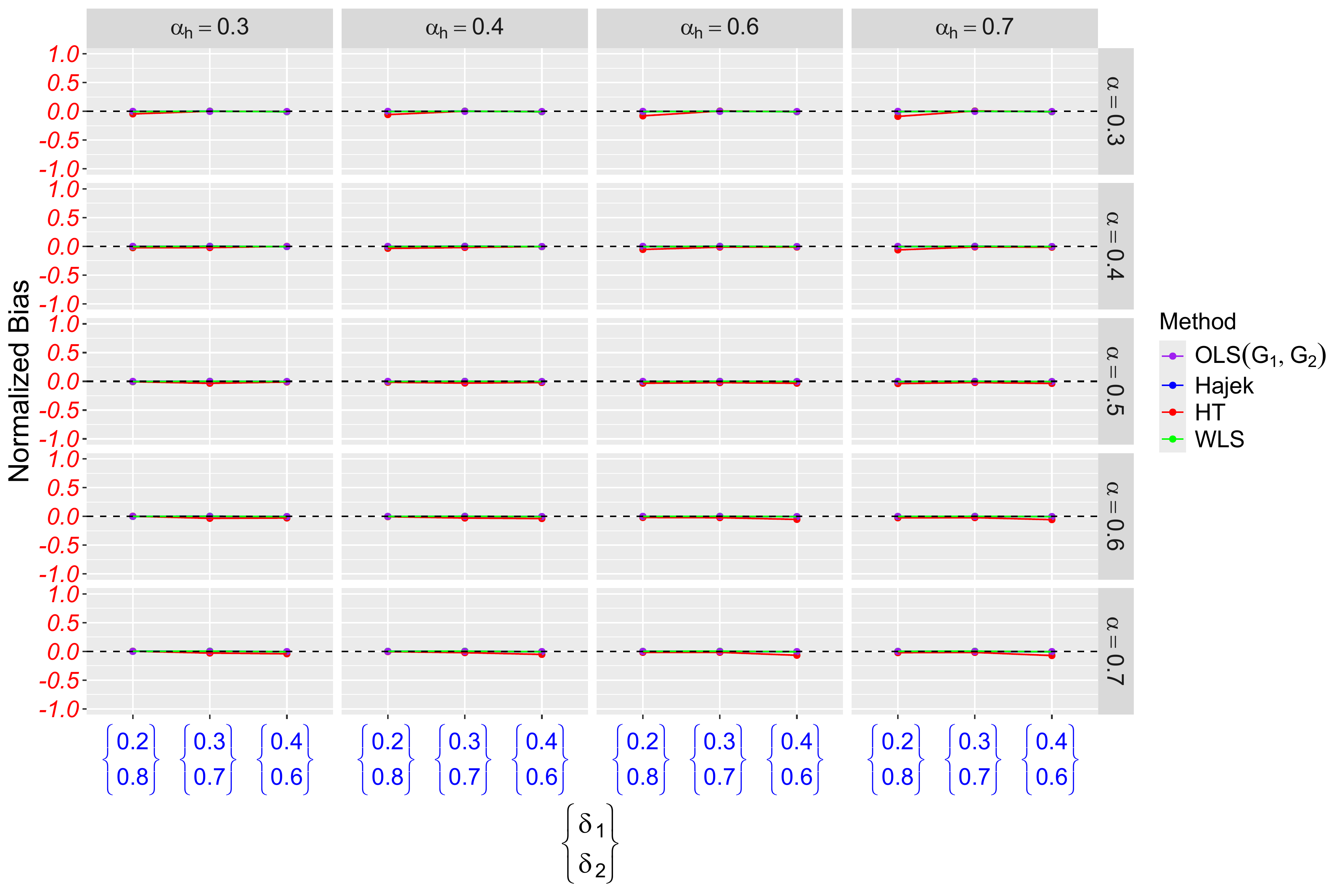}
        \caption{$\rma=1, h=2$}
        \label{fig:scen1_a1h2}
    \end{subfigure}

    \caption{Scenario 1 (No Interference). Normalized bias of the HT (red), Hajek (blue), and WLS (green) estimators. We also include the $\text{OLS}(G_1, G_2)$ estimator (purple) in all panels, and the $\text{OLS}(G_1)$ estimator (orange) specifically for the first-order spillover effects ($h=1$). Results are shown for the spillover effects $\text{SE}^{h}\left(\alpha_{h}, 0.5;\rma, \alpha\right)$, with $\alpha_h\in\{0.3, 0.4, 0.6, 0.7\}$ and $\alpha\in\{0.3, 0.4, 0.5, 0.6, 0.7\}$, under Scenario 1, a two-stage assignment with $\{\delta_1, \delta_2 \}\in\{(0.2, 0.8), (0.3, 0.7), (0.4, 0.6)\}$, and regular networks with $p=0.2$.}
    \label{fig:scenario1_bias}
\end{figure*}

\subsubsection{Scenario 2: First-order neighborhood interference}
As shown in Figure~\ref{fig:scenario2_bias}, most estimators exhibit near-zero normalized bias under Scenario 2 (first-order interference).
For $h=2$, the Hajek and WLS estimators display a small bias, which is likely driven by finite-sample effects due to the limited number of clusters ($I=100$). As further discussed in Scenario 3 (see Figure~\ref{fig:scenario3_bias_I500}), this bias vanishes when increasing the number of clusters.

\begin{figure*}[!htbp]
    \centering
    \begin{subfigure}{0.5\linewidth}
        \includegraphics[height=5.5cm]{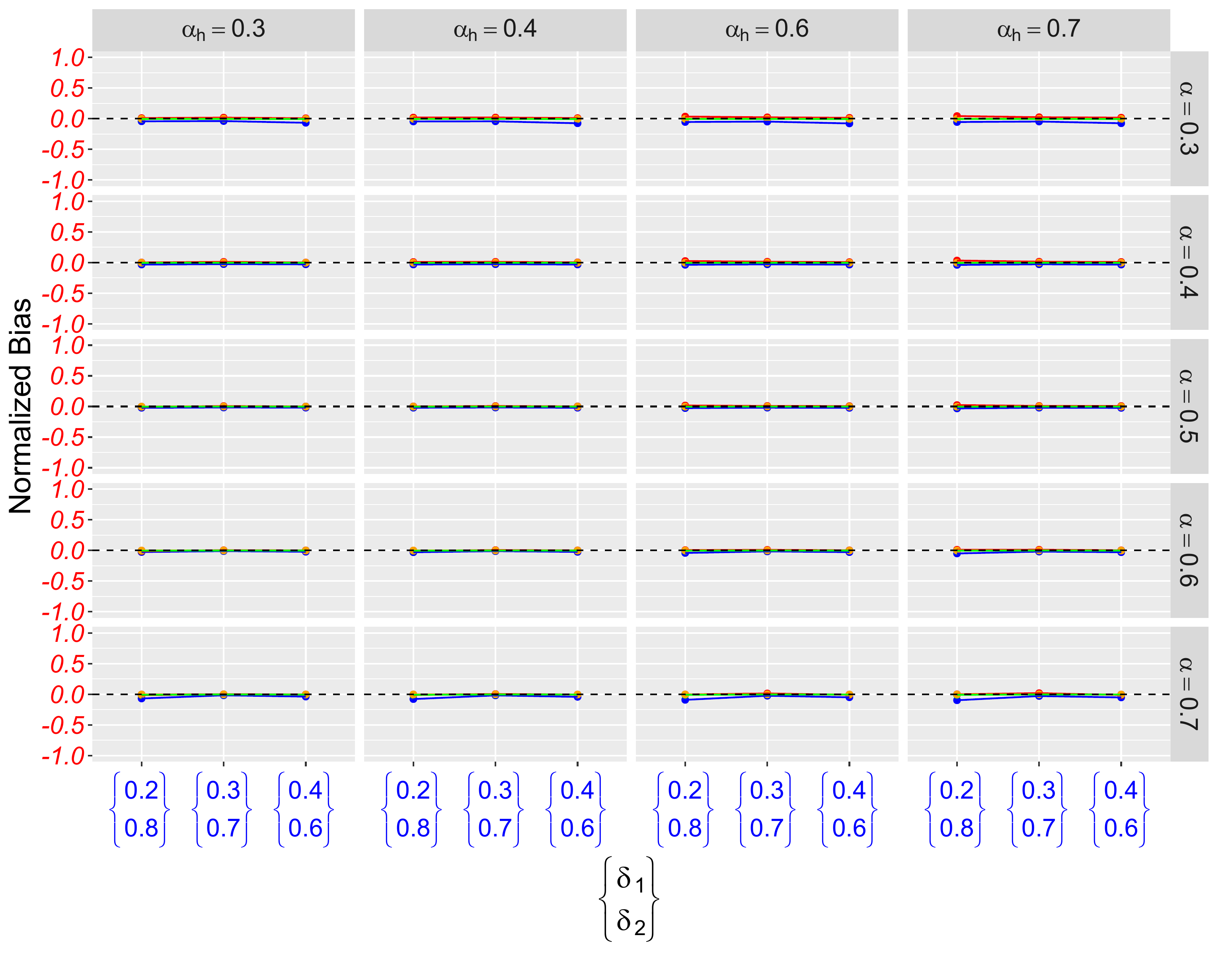}
        \caption{$\rma=0, h=1$}
        \label{fig:scen2_a0h1}
    \end{subfigure}
    \begin{subfigure}{0.45\linewidth}
        \includegraphics[height=5.5cm]{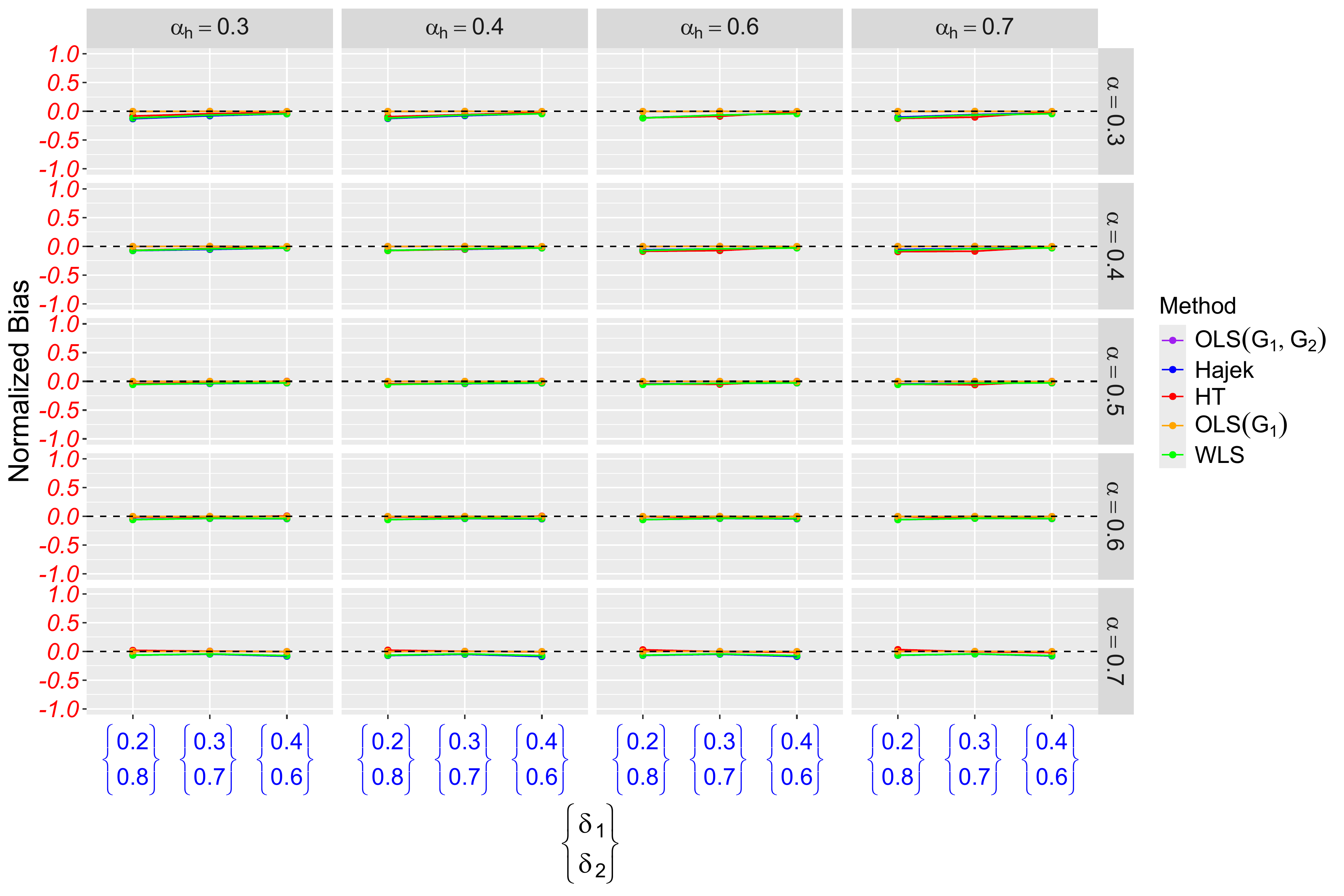}
        \caption{$\rma=1, h=1$}
        \label{fig:scen2_a1h1}
    \end{subfigure}

    \vspace{0.5cm}

    \begin{subfigure}{0.5\linewidth}
        \includegraphics[height=5.5cm]{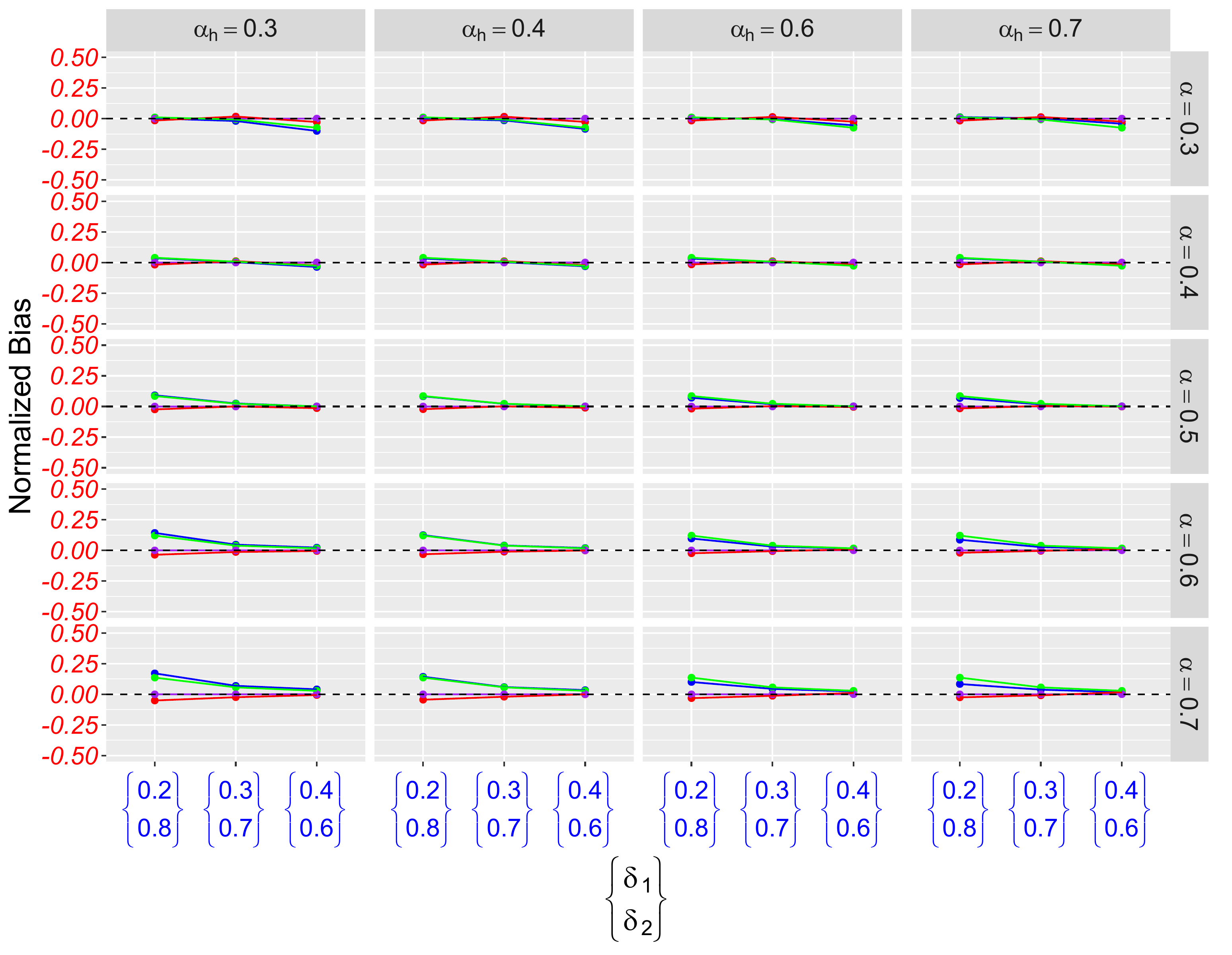}
        \caption{$\rma=0, h=2$}
        \label{fig:scen2_a0h2}
    \end{subfigure}
    \begin{subfigure}{0.45\linewidth}
        \includegraphics[height=5.5cm]{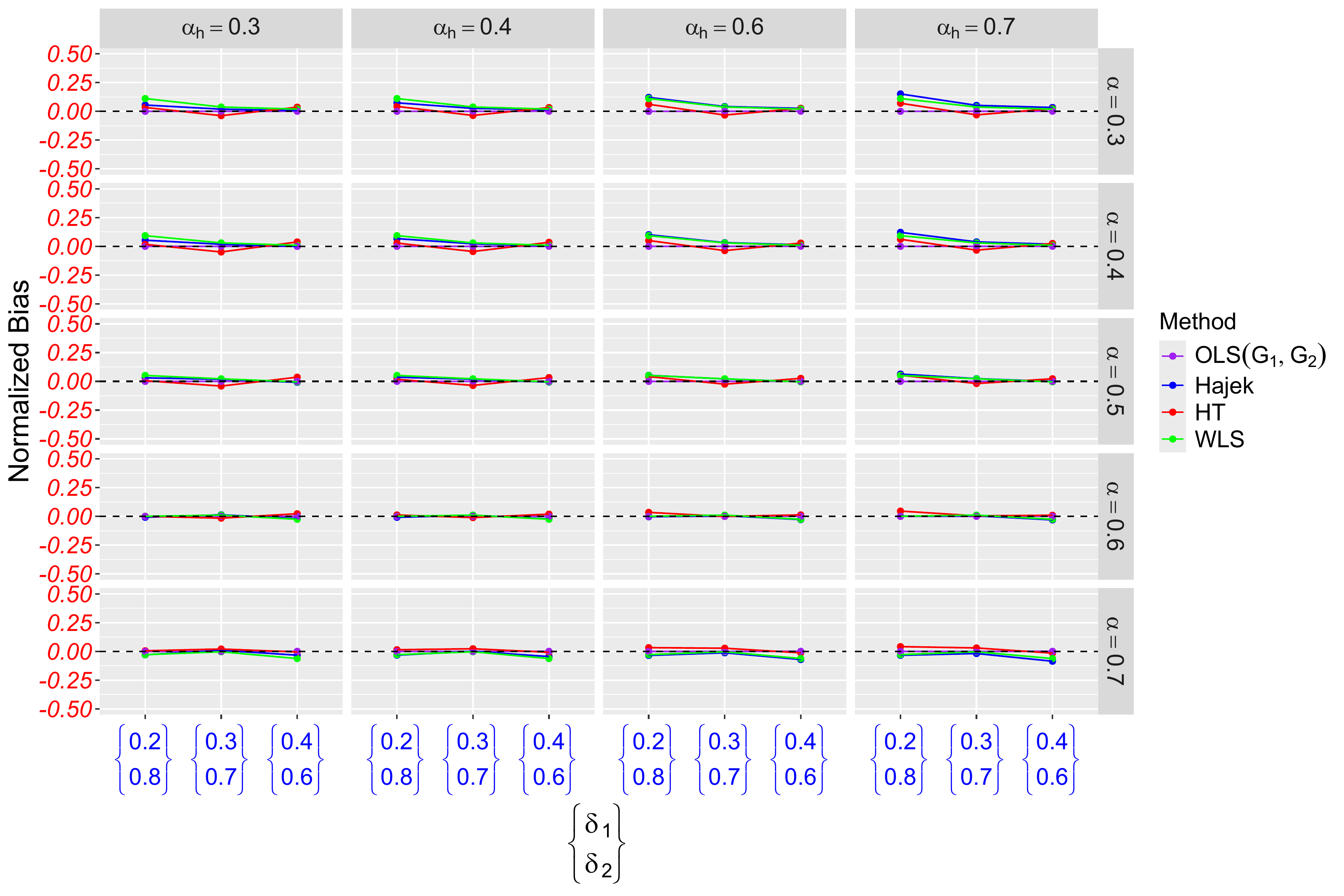}
        \caption{$\rma=1, h=2$}
        \label{fig:scen2_a1h2}
    \end{subfigure}
    \caption{Scenario 2 (First-order neighborhood interference). Normalized bias of the HT (red), Hajek (blue), and WLS (green) estimators. We also include the $\text{OLS}(G_1, G_2)$ estimator (purple) in all panels, and the $\text{OLS}(G_1)$ estimator (orange) specifically for the first-order spillover effects ($h=1$). Results are shown for the spillover effects $\text{SE}^{h}\left(\alpha_{h}, 0.5;\rma, \alpha\right)$, with $\alpha_h\in\{0.3, 0.4, 0.6, 0.7\}$ and $\alpha\in\{0.3, 0.4, 0.5, 0.6, 0.7\}$, under Scenario 2, a two-stage assignment with $\{\delta_1, \delta_2 \}\in\{(0.2, 0.8), (0.3, 0.7), (0.4, 0.6)\}$, and regular networks with $p=0.2$.}
    \label{fig:scenario2_bias}
\end{figure*}
\subsubsection{Scenario 3: Second-order neighborhood interference}

For estimating first-order spillover effects ($h=1$), we employ two OLS estimators: $\text{OLS}(G_1, G_2)$, which regresses the outcome on $A_{ij}$, $G_{ij,1}$, and $G_{ij,2}$, and $\text{OLS}(G_1)$, which regresses the outcome on $A_{ij}$ and $G_{ij,1}$ only. We include $\text{OLS}(G_1)$ to illustrate the bias arising from an incorrect interference set as described in Corollary 1. For estimating second-order spillover effects ($h=2$), we only employ the $\text{OLS}(G_1, G_2)$ estimator. Results are shown in Figure~\ref{fig:scenario3_bias}.

\begin{figure*}[!htbp]
    \centering
    \begin{subfigure}{0.5\linewidth}
        \includegraphics[height=5.5cm]{1.bias_p=0.2/left/left_scenario3_with_naive_Bias_a0h1_p=0.2.png}
        \caption{$\rma=0, h=1$}
        \label{fig:scen3_a0h1}
    \end{subfigure}
    \begin{subfigure}{0.45\linewidth}
        \includegraphics[height=5.5cm]{1.bias_p=0.2/right/scenario3_with_naive_Bias_a1h1_p=0.2.png}
        \caption{$\rma=1, h=1$}
        \label{fig:scen3_a1h1}
    \end{subfigure}

    \vspace{0.5cm}

    \begin{subfigure}{0.5\linewidth}
        \includegraphics[height=5.5cm]{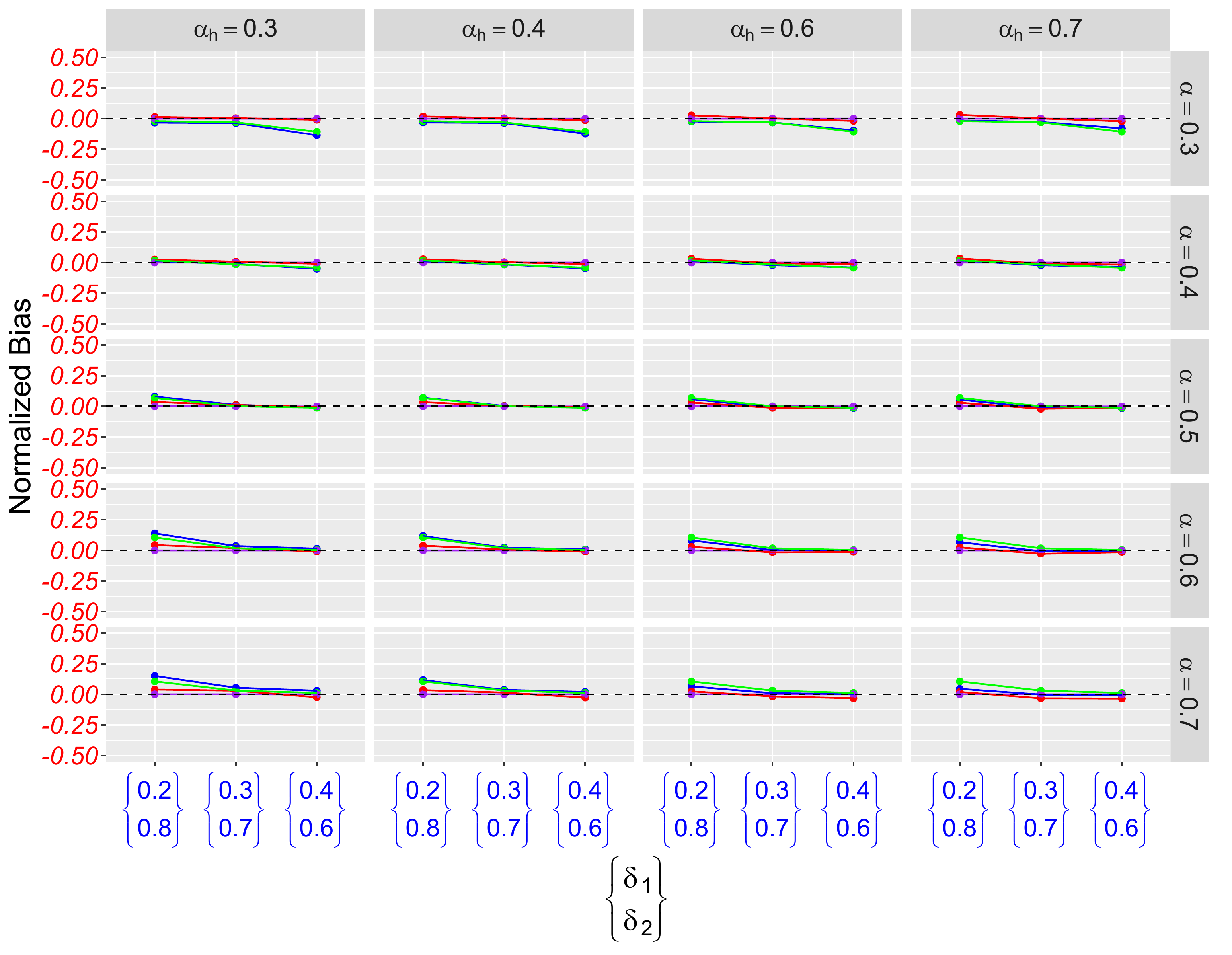}
        \caption{$\rma=0, h=2$}
        \label{fig:scen3_a0h2}
    \end{subfigure}
    \begin{subfigure}{0.45\linewidth}
        \includegraphics[height=5.5cm]{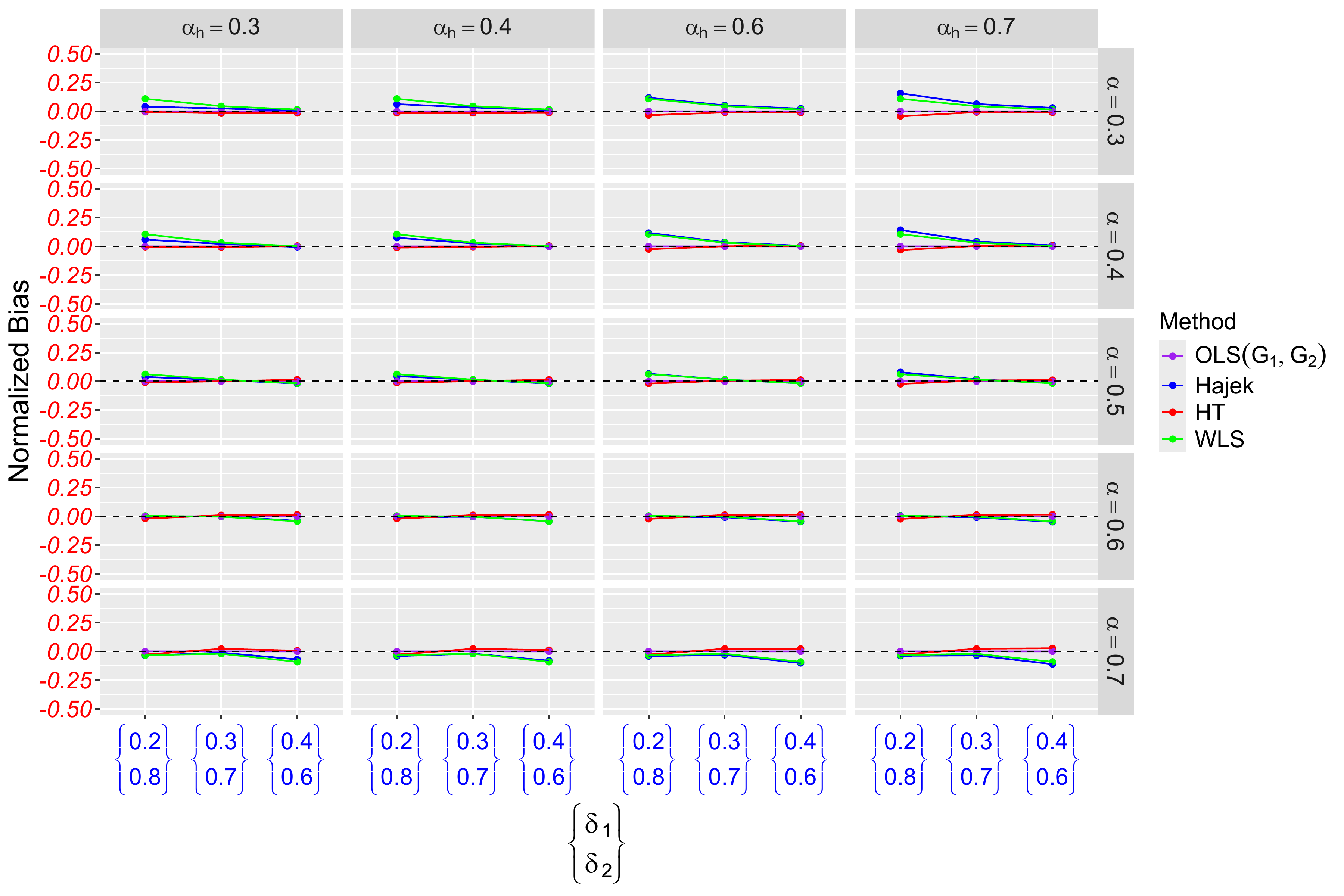}
        \caption{$\rma=1, h=2$}
        \label{fig:scen3_a1h2}
    \end{subfigure}
\caption{Scenario 3 (Second-order neighborhood interference). Normalized bias of the HT (red), Hajek (blue), and WLS (green) estimators. We also include the $\text{OLS}(G_1, G_2)$ estimator (purple) in all panels, and the $\text{OLS}(G_1)$ estimator (orange), $\text{OLS}(G^*_1, G_2)$ (deep blue) specifically for the first-order spillover effects ($h=1$), where $G^*_{ij,1}=\mathbbm{1}(G_{ij,1}\geq 0.5)$. Results are shown for the spillover effects $\text{SE}^{h}\left(\alpha_{h}, 0.5;\rma, \alpha\right)$, with $\alpha_h\in\{0.3, 0.4, 0.6, 0.7\}$ and $\alpha\in\{0.3, 0.4, 0.5, 0.6, 0.7\}$, under Scenario 3, a two-stage assignment with $\{\delta_1, \delta_2 \}\in\{(0.2, 0.8), (0.3, 0.7), (0.4, 0.6)\}$, and regular networks with $p=0.2$.}
    \label{fig:scenario3_bias}
\end{figure*}

Figure~\ref{fig:scenario3_bias} presents the results. For the first-order effects ($h=1$), the misspecified $\text{OLS}(G_1)$ estimator exhibits significant bias. Notably, this bias decreases as the first-stage assignment probabilities $(\delta_1,\delta_2)$ become closer, moving from $\{0.2, 0.8\}$ to $\{0.4, 0.6\}$. This bias stems from the empirical correlation between $G_{ij,1}$ and $G_{ij,2}$ induced by the two-stage design, which leads to omitted variable bias when $G_{ij,2}$ is excluded; a detailed analysis of this correlation mechanism is provided in Appendix \ref{app:correlation}.
Additionally, the estimator $\text{OLS}(G^*_1, G_2)$, which uses the correct interference set but a misspecified binary exposure mapping, exhibits significant bias, as described in Corollary 2. Conversely, for second-order effects ($h=2$), the correctly specified $\text{OLS}(G_1, G_2)$ and our proposed estimators are generally unbiased.

It is important to note that the Hajek and WLS estimators exhibit a small bias for the second-order spillover effects ($h=2$) in Scenarios 2 and 3 when the number of clusters is $I=100$. This is expected, as these estimators are consistent but not unbiased in finite samples. To verify this, we conducted a supplementary simulation with a larger number of clusters ($I=500$) over 500 Monte Carlo replications for scenario 3.  Figure \ref{fig:scenario3_bias_I500} shows the results for this supplementary simulation. As anticipated, the bias in the Hajek and WLS estimators vanishes, empirically confirming their consistency. However, the Horvitz-Thompson estimator shows bias in the larger sample size because of the presence of extreme weights and the high variance.

\begin{figure*}[!htbp]
    \centering
    \begin{subfigure}{0.5\linewidth}
        \includegraphics[height=5.5cm]{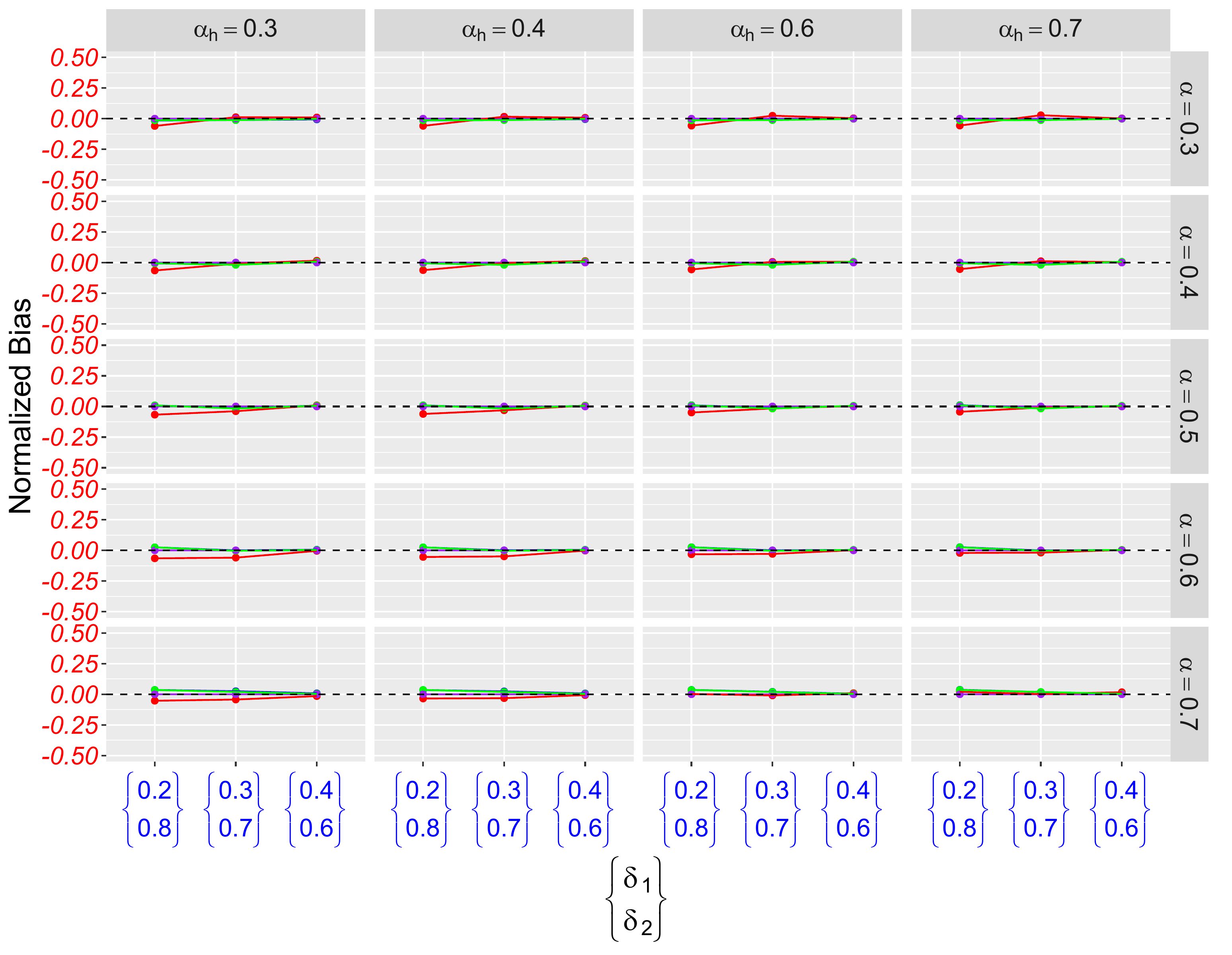}
        \caption{$\rma=0, h=2$}
        \label{fig:scen3_I500_a0h2}
    \end{subfigure}
    \begin{subfigure}{0.45\linewidth}
              \includegraphics[height=5.5cm]{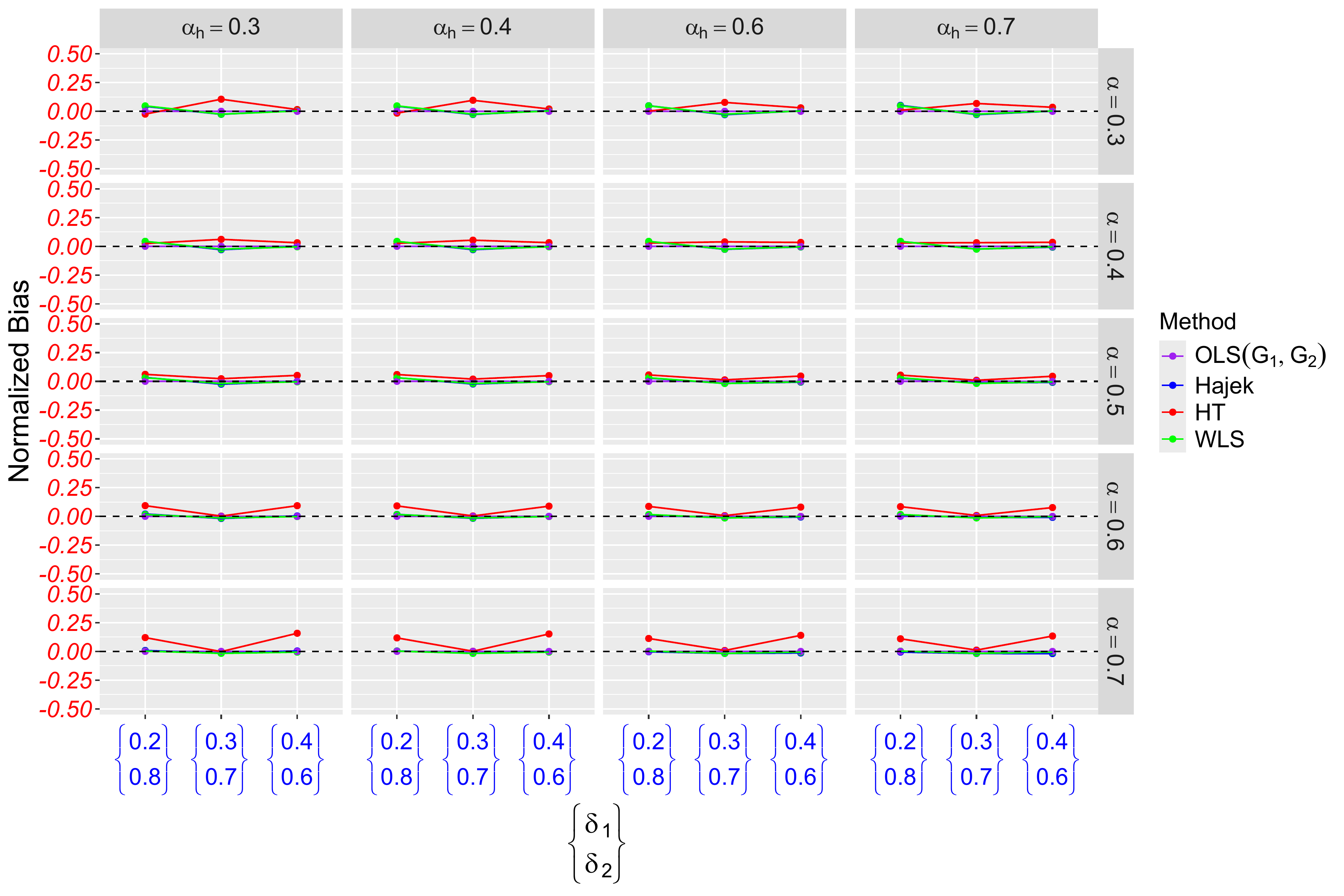}
        \caption{$\rma=1, h=2$}
        \label{fig:scen3_I500_a1h2}
    \end{subfigure}

    \caption{Consistency Check for $h=2$ (Scenario 3 with $I=500$). Normalized bias of the HT (red), Hajek (blue), and WLS (green) estimators for spillover effects $\text{SE}^{h}\left(\alpha_{h}, 0.5;\rma, \alpha\right)$, under Scenario 3 with a larger sample size ($I=500$ clusters, 500 simulations) and regular networks ($p=0.2$). Compared to the $I=100$ case, the small finite-sample bias in Hajek and WLS estimators for $h=2$ vanishes, confirming their consistency. The HT estimator may still exhibit deviations due to extreme weights.}
    \label{fig:scenario3_bias_I500}
\end{figure*}

\subsubsection{Scenario 4: Second-order neighborhood interference with interaction}
\label{app:scenario4_detail}
In Scenario 4, the outcome depends on the interaction between the proportions of treated first- and second-order neighbors ($G_{ij,1} G_{ij,2}$). When estimating first-order spillover effects ($h=1$), we employ two OLS estimators: $\text{OLS}(G_1, G_2)$, which regresses $Y_{ij}$ on $A_{ij}$, $G_{ij,1}$, and $G_{ij,2}$ (omitting the interaction term), and $\text{OLS}(G_1)$, which regresses on $A_{ij}$ and $G_{ij,1}$ only. For second-order effects ($h=2$), we only use $\text{OLS}(G_1, G_2)$. Note that $\text{OLS}(G_1, G_2)$ is misspecified in this scenario due to the missing interaction term $G_{ij,1}G_{ij,2}$, corresponding to the bias from Corollary 2.

Given the outcome model
   $$
    Y_{ij} \sim N\left(2+5 A_{ij}+10 G_{ij}^{1}+3 G_{ij}^{2} + 5\cdot G_{ij}^{1}\cdot  G_{ij}^{2}, \sigma^{2}=1\right)
    $$
To compute the theoretical value of first-order spillover effects,
we need to consider the probability that $\mathcal{N}^2_{ij}$ is empty  (i.e., the unit has no distance-2 neighbors), and, in turn, $G^2_{ij}$ is not defined. Let
$$
    \pi_2 = \Pr\bigl(\mathcal{N}^2_{ij}=\emptyset \bigr),
$$
be the probability that a randomly chosen unit has no second neighbors. Then,
the true first-order spillover effects are given by
$$
    \text{SE}^{1}(\alpha_h,\alpha_h'; \rma,\alpha)
    = \bigl[ 10(\alpha_h - \alpha_h') + 5 \alpha (\alpha_h - \alpha_h') \bigr] (1 - \pi_2)
    = (\alpha_h - \alpha_h') ( 10 + 5\alpha ) (1 - \pi_2).
    \label{eq:true-h1-effect}
$$
We generate the networks within clusters using 2-regular graphs, where every unit has exactly two direct neighbors. In this graph with $n_i=10$ labeled vertices (i.e., a disjoint union of cycles), a unit has no  second-order neighbors if and only if it lies on a 3-cycle. Thus,
$$\pi_2=\operatorname{Pr} \text{(vertex is in a 3-cycle)}$$
Under the uniform distribution over labeled simple 2-regular graphs, the possible cycle decompositions are $\{10\}, \{7+3\}, \{6+4\}, \{5+5\}$, and $\{4+3+3\}$, where $\{m_k\}$ denotes the partition with $m_k$ cycles of length $k$. The number of graphs with each decomposition is
$$
N\left(\left\{m_k\right\}\right)=\frac{10!}{\prod_k(2 k)^{m_k} m_{k}!},
$$
yielding $N_{\{10\}}=181{,}440$, $N_{\{7+3\}}=43{,}200$, $N_{\{6+4\}}=37{,}800$, $N_{\{5+5\}}=18{,}144$, $N_{\{4+3+3\}}=6{,}300$, and $N_{\text{tot}}=286{,}884$. In the cycle decompositions $\{7+3\}$ and $\{4+3+3\}$, a uniformly chosen vertex lies in a 3-cycle with probabilities $3/10$ and $6/10$, respectively; otherwise the probability is 0. Therefore:
$$
\pi_2=\frac{N_{7+3} \cdot(3 / 10)+N_{4+3+3} \cdot(6 / 10)}{N_{\text {tot }}}=\frac{16,740}{286,884} \approx 0.05835 .
$$
However, our simulations use the \texttt{sample\_k\_regular} function from the \textbf{igraph} package in R to generate networks. This function draws regular graphs using a configuration-model-with-rejection scheme and does not produce a perfectly uniform distribution over simple graphs for small $n$. Therefore, to determine the ``true value" of the first-order spillover effects for simulations that use \texttt{sample\_k\_regular}, we estimate $\pi_2$ via Monte Carlo by generating 100000 graphs using \texttt{sample\_k\_regular} and computing the fraction of vertices in 3-cycles. For $n_i=10$, we obtain
$\hat{\pi}_2=0.104$.

Figure~\ref{fig:scenario4_bias} shows that while the proposed estimators (HT, Hajek, WLS) remain unbiased under interaction-based interference, the misspecified OLS estimators ($\text{OLS}(G_1),\text{OLS}(G_1,G_2)$) exhibit significant bias due to the omission of the interaction term and/or second-order neighbors.

\begin{figure*}[!htbp]
    \centering
    \begin{subfigure}{0.5\linewidth}
        \includegraphics[height=5.5cm]{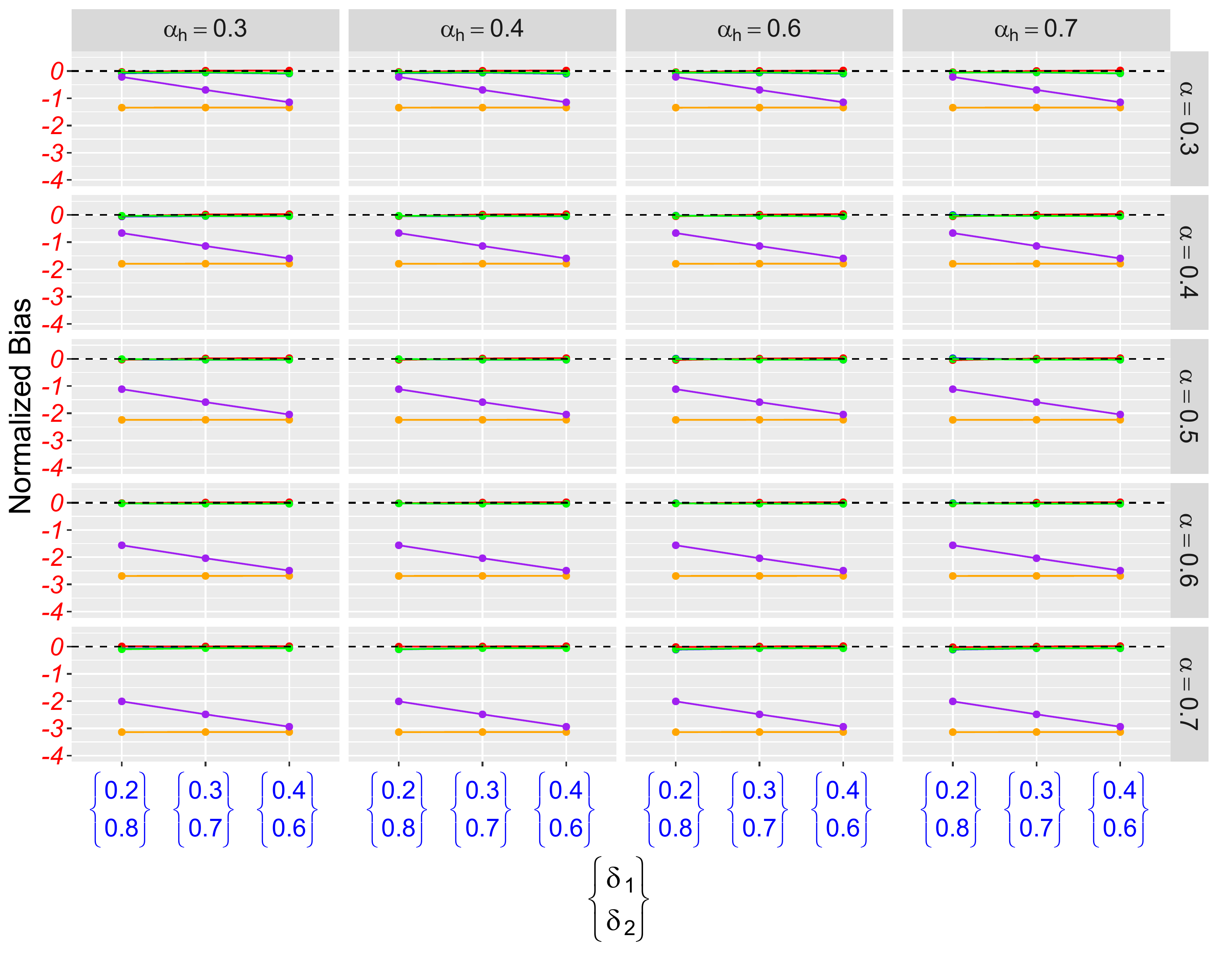}
        \caption{$\rma=0, h=1$}
        \label{fig:scen4_a0h1}
    \end{subfigure}
    \begin{subfigure}{0.45\linewidth}
        \includegraphics[height=5.5cm]{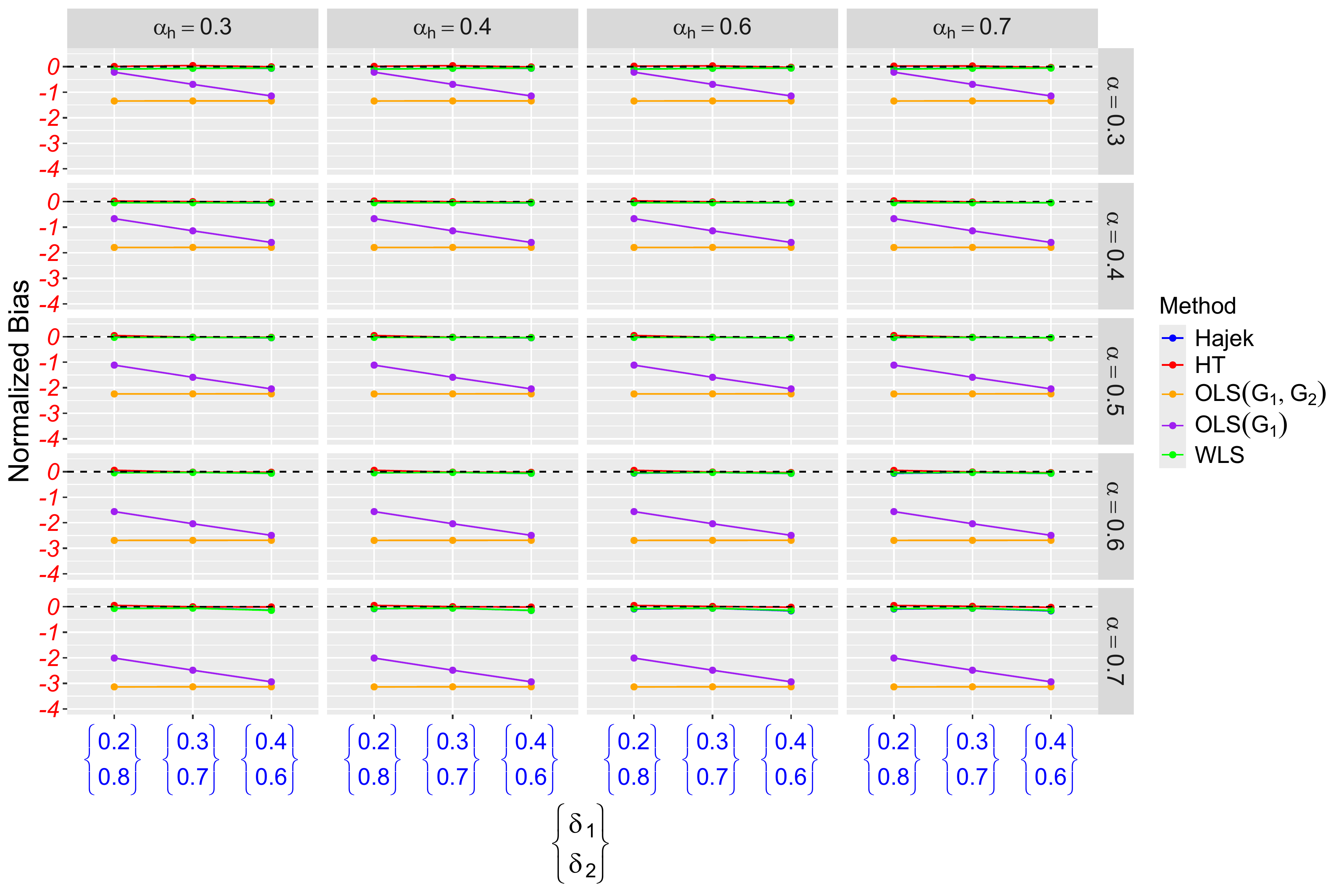}
        \caption{$\rma=1, h=1$}
        \label{fig:scen4_a1h1}
    \end{subfigure}

    \vspace{0.5cm}

    \begin{subfigure}{0.5\linewidth}
        \includegraphics[height=5.5cm]{1.bias_p=0.2/left/left_scenario4_with_naive_Bias_a0h2_p=0.2.png}
        \caption{$\rma=0, h=2$}
        \label{fig:scen4_a0h2}
    \end{subfigure}
    \begin{subfigure}{0.45\linewidth}
        \includegraphics[height=5.5cm]{1.bias_p=0.2/right/scenario4_with_naive_Bias_a1h2_p=0.2.png}
        \caption{$\rma=1, h=2$}
        \label{fig:scen4_a1h2}
    \end{subfigure}
    \caption{Scenario 4 (Second-order neighborhood interference with interaction). Normalized bias of the HT (red), Hajek (blue), and WLS (green) estimators. We also include the misspecified OLS estimators: $\text{OLS}(G_1, G_2)$ (purple) is included in all panels (omitting the interaction term), while $\text{OLS}(G_1)$ (orange) is included only for $h=1$ (omitting both the interaction term and second-order neighbors). Results are shown for the spillover effects $\text{SE}^{h}\left(\alpha_{h}, 0.5;\rma, \alpha\right)$, with $\alpha_h\in\{0.3, 0.4, 0.6, 0.7\}$ and $\alpha\in\{0.3, 0.4, 0.5, 0.6, 0.7\}$, under Scenario 4, a two-stage assignment with $\{\delta_1, \delta_2 \}\in\{(0.2, 0.8), (0.3, 0.7), (0.4, 0.6)\}$, and regular networks with $p=0.2$.}
    \label{fig:scenario4_bias}
\end{figure*}

\subsection{Variance of Spillover Effect Estimators}\label{app:variance}

This section presents the full results for the variance analysis.

\subsubsection{Validation of analytical variance estimators}
\label{app:variance_validation}
Here we show the comparison between the analytical variance, the Monte Carlo variance, and the one obtained through bootstrap, for the Horvitz-Thompson, Hajek, and WLS estimators, for Scenario 3 under a regular graph model with link probability $p=0.2$. The bootstrap variance was averaged over 50 simulation runs, with each run based on 500 bootstrap replications, while the Monte Carlo and analytical variance was computed from 5,000 simulation replications.

Figures~\ref{fig:var_HT}--\ref{fig:var_WLS} show that the analytical variance estimates derived via M-estimation (red) closely match the empirical Monte Carlo variances (green) across all estimators and parameter settings, validating our theoretical derivations. We note that the bootstrap variance occasionally exhibits slight deviations, which is expected given that the bootstrap is computed using a limited number of replications due to computational constraints.

\begin{figure*}[!htbp]
    \centering
    \begin{subfigure}{0.5\linewidth}
        \includegraphics[height=5.5cm]{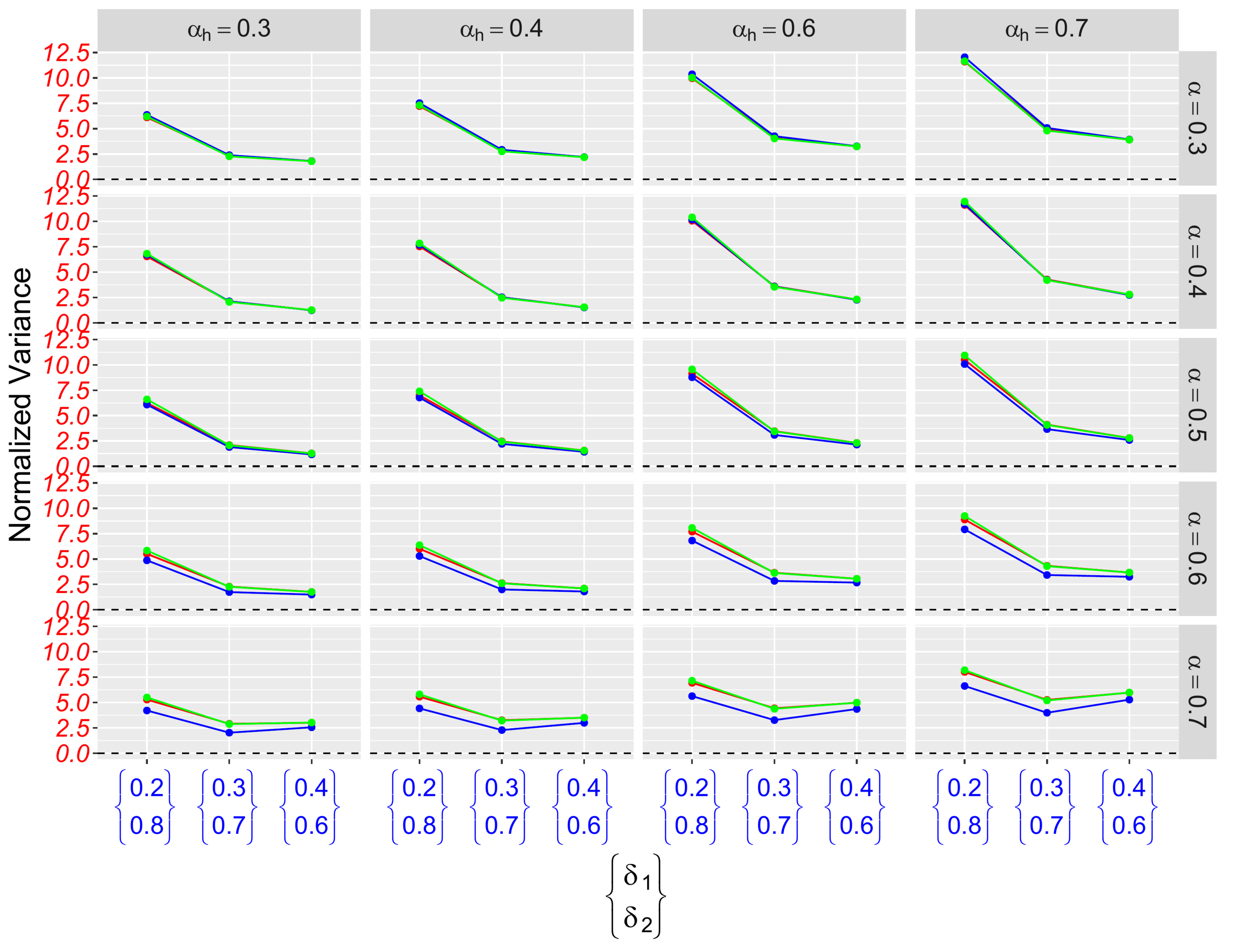}
        \caption{$\rma=0, h=1$}
    \end{subfigure}
    \begin{subfigure}{0.45\linewidth}
        \includegraphics[height=5.5cm]{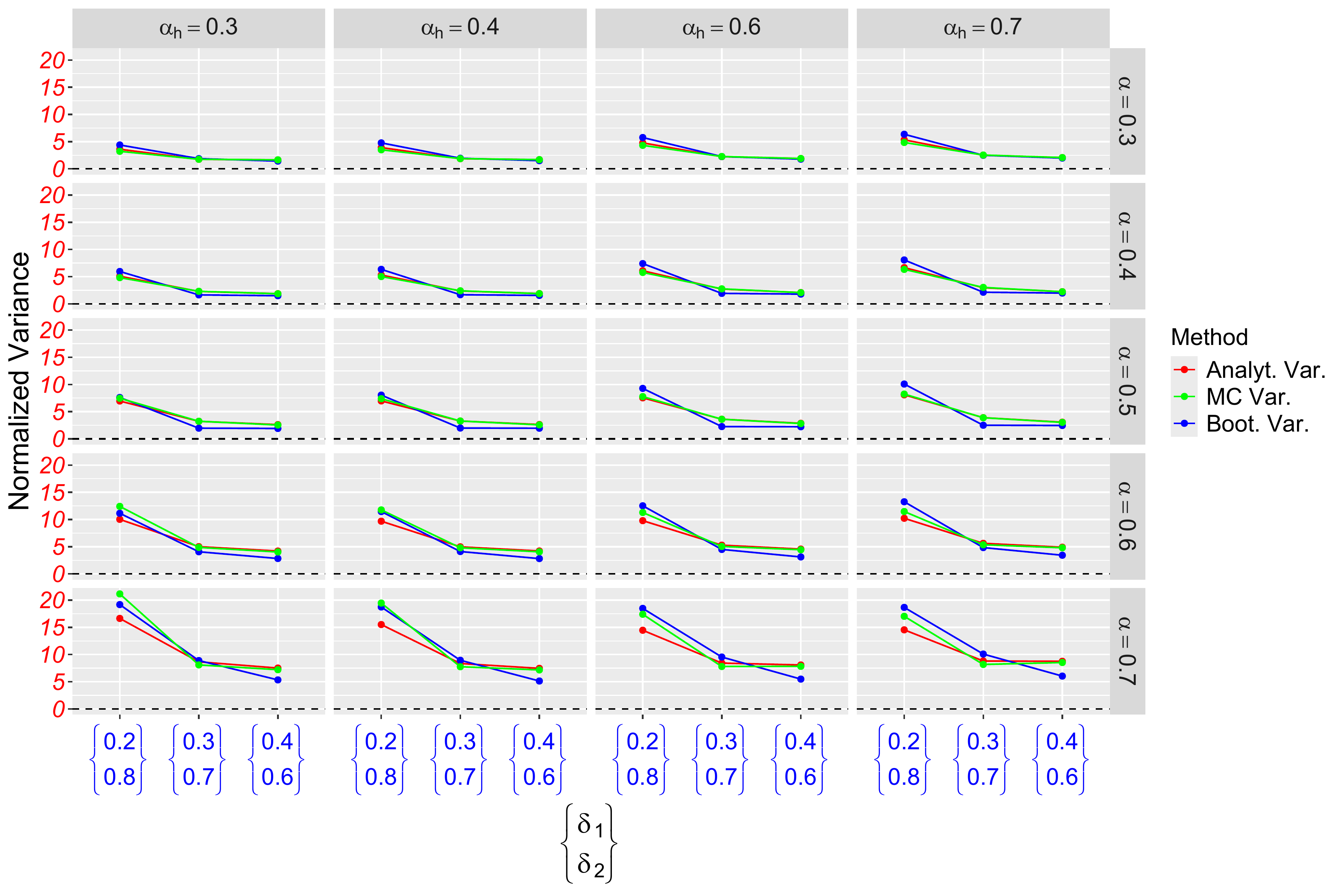}
        \caption{$\rma=0, h=2$}
    \end{subfigure}

    \vspace{0.5cm}

    \begin{subfigure}{0.5\linewidth}
        \includegraphics[height=5.5cm]{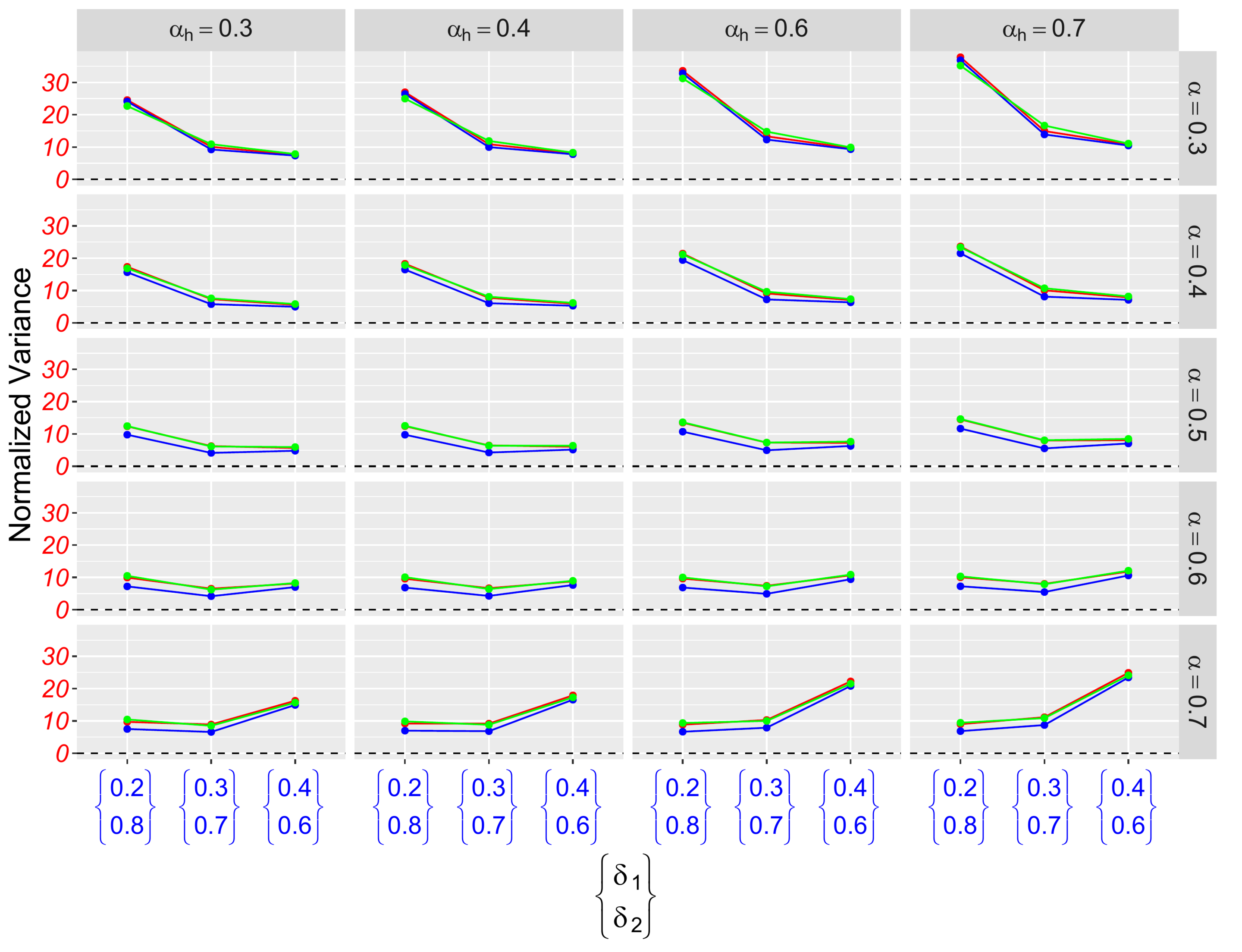}
        \caption{$\rma=1, h=1$}
    \end{subfigure}
    \begin{subfigure}{0.45\linewidth}
        \includegraphics[height=5.5cm]{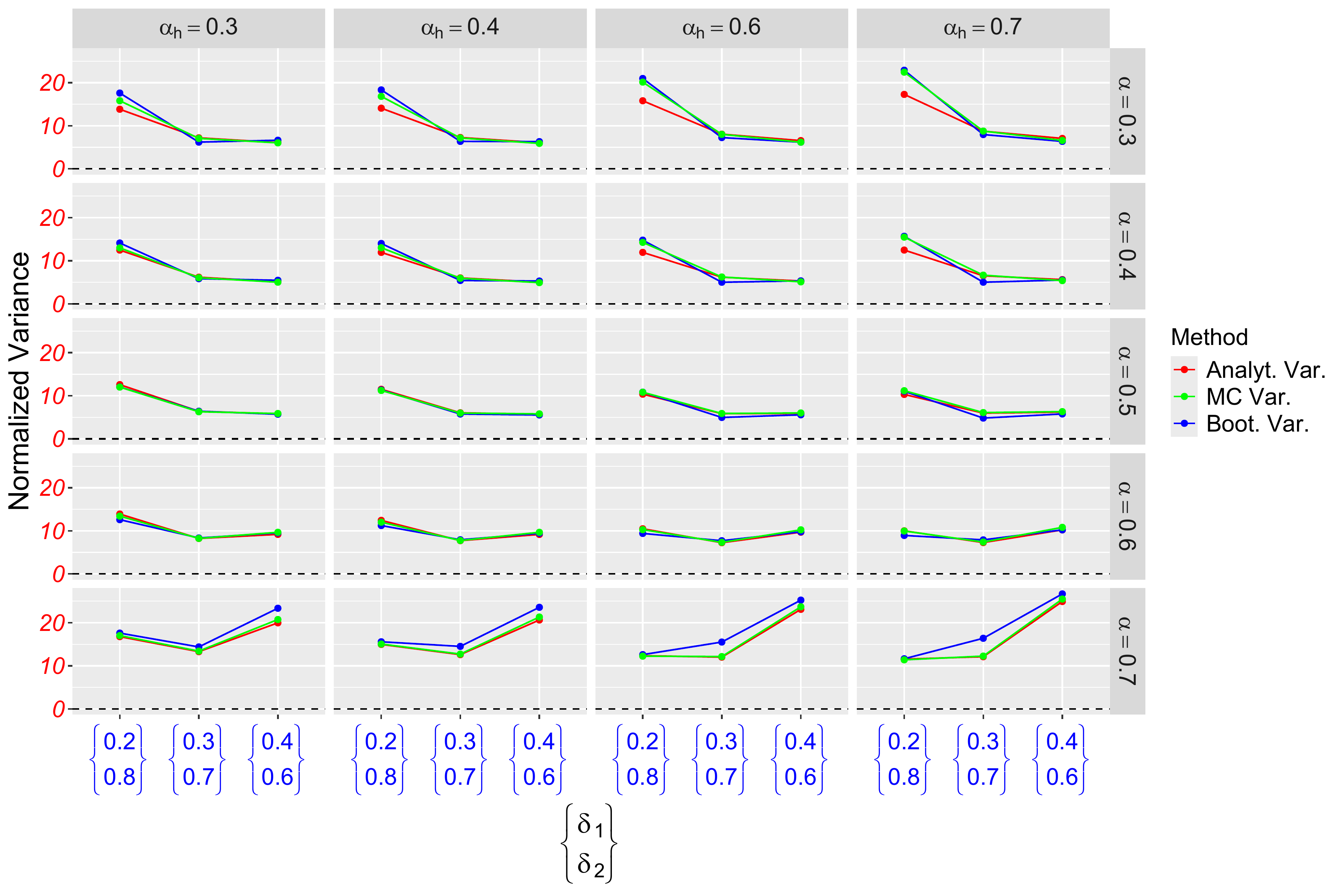}
        \caption{$\rma=1, h=2$}
    \end{subfigure}
    \caption{Validation of the Horvitz-Thompson variance estimator. Normalized variance estimated using analytical M-estimation (red), bootstrap (blue), and Monte Carlo simulations (green) for the spillover effects $\text{SE}^{h}\left(\alpha_{h}, 0.5;\rma, \alpha\right)$, with $\alpha_h\in\{0.3, 0.4, 0.5, 0.6, 0.7\}$ and $\alpha\in\{0.3, 0.4, 0.5, 0.6, 0.7\}$, under Scenario 3, a two-stage assignment with $\{\delta_1, \delta_2 \}\in\{(0.2, 0.8), (0.3, 0.7), (0.4, 0.6)\}$, and regular networks with $p=0.2$.}
    \label{fig:var_HT}
\end{figure*}

\begin{figure*}[!htbp]
    \centering
    \begin{subfigure}{0.5\linewidth}
        \includegraphics[height=5.5cm]{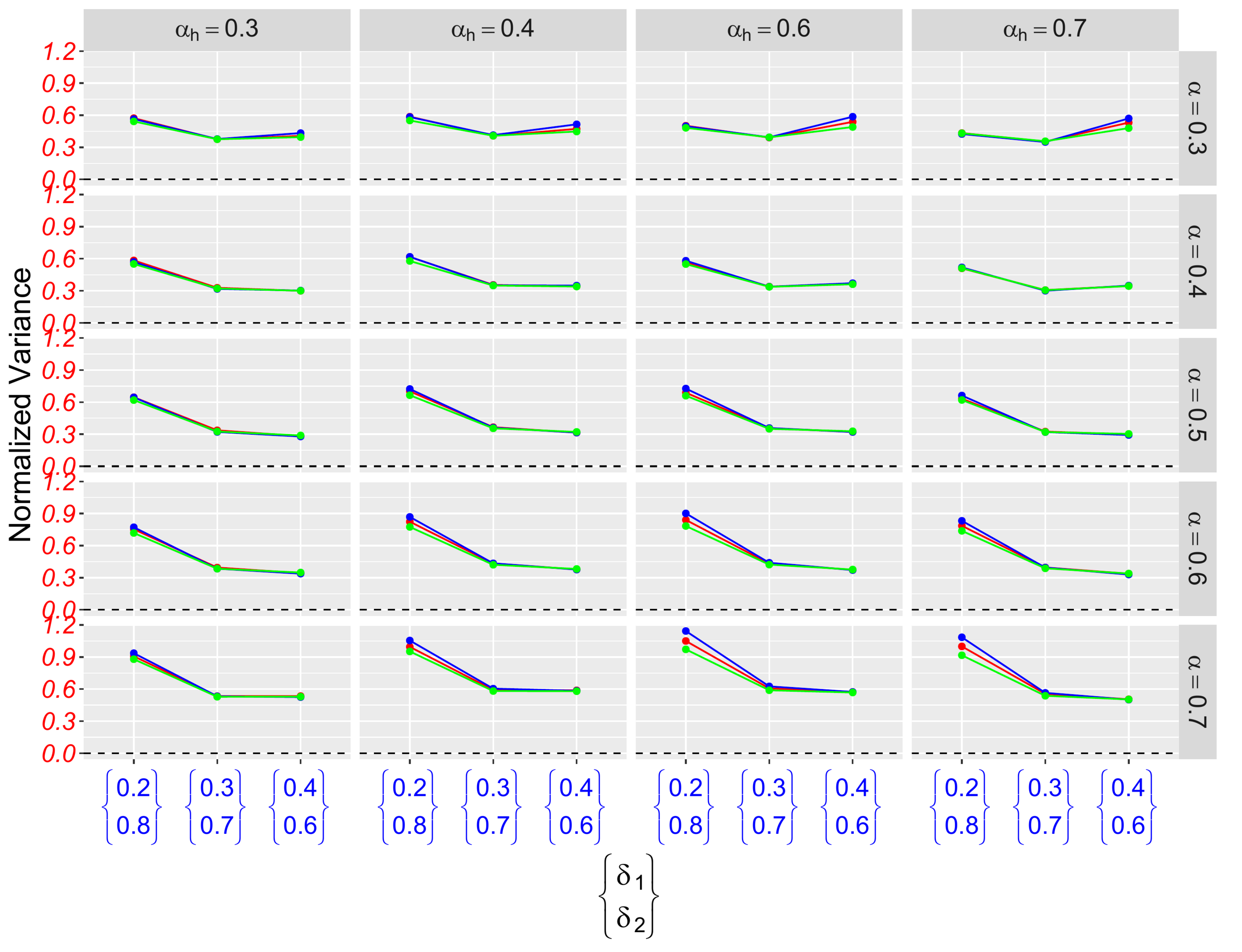}
        \caption{$\rma=0, h=1$}
    \end{subfigure}
    \begin{subfigure}{0.45\linewidth}
        \includegraphics[height=5.5cm]{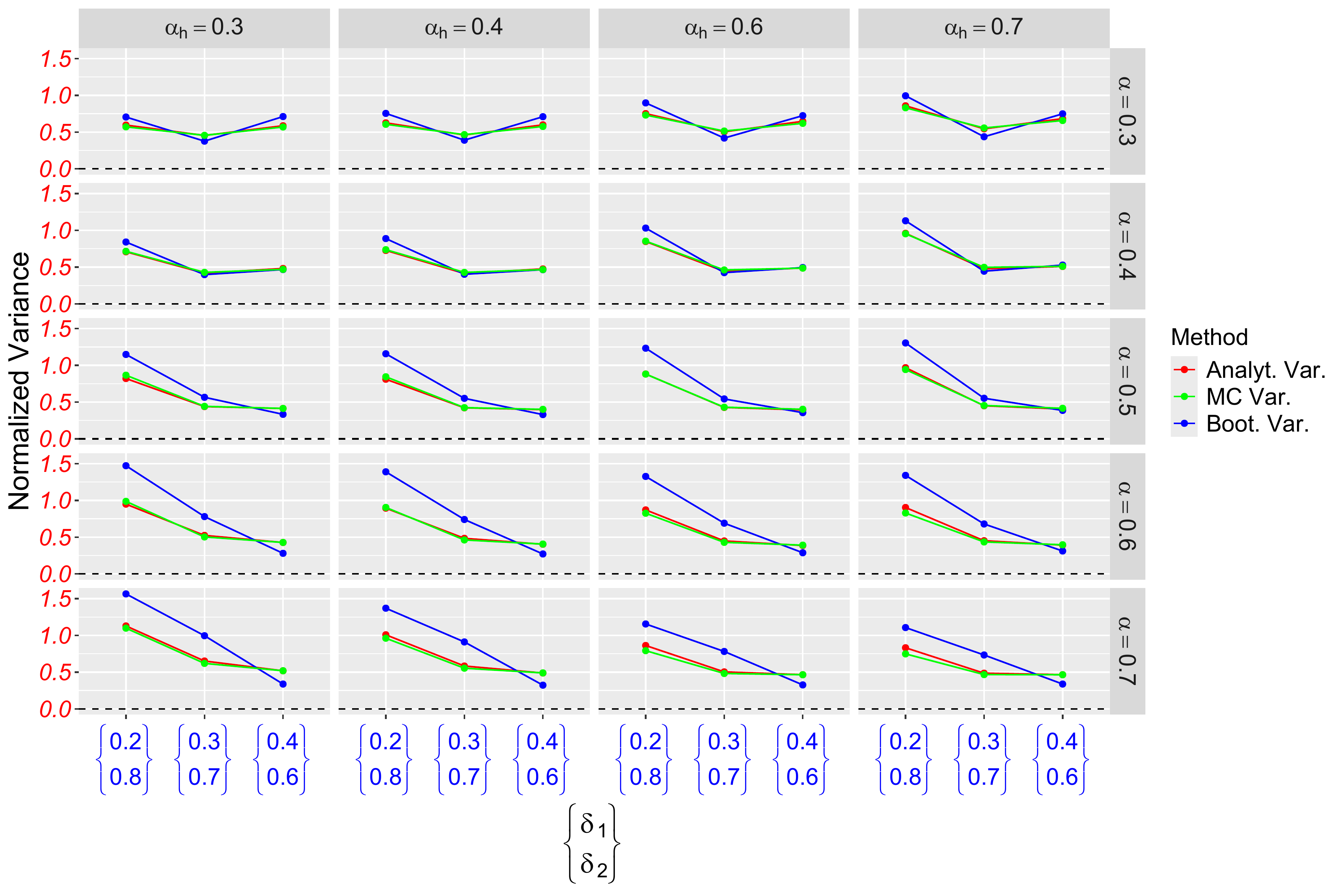}
        \caption{$\rma=0, h=2$}
    \end{subfigure}

    \vspace{0.5cm}

    \begin{subfigure}{0.5\linewidth}
        \includegraphics[height=5.5cm]{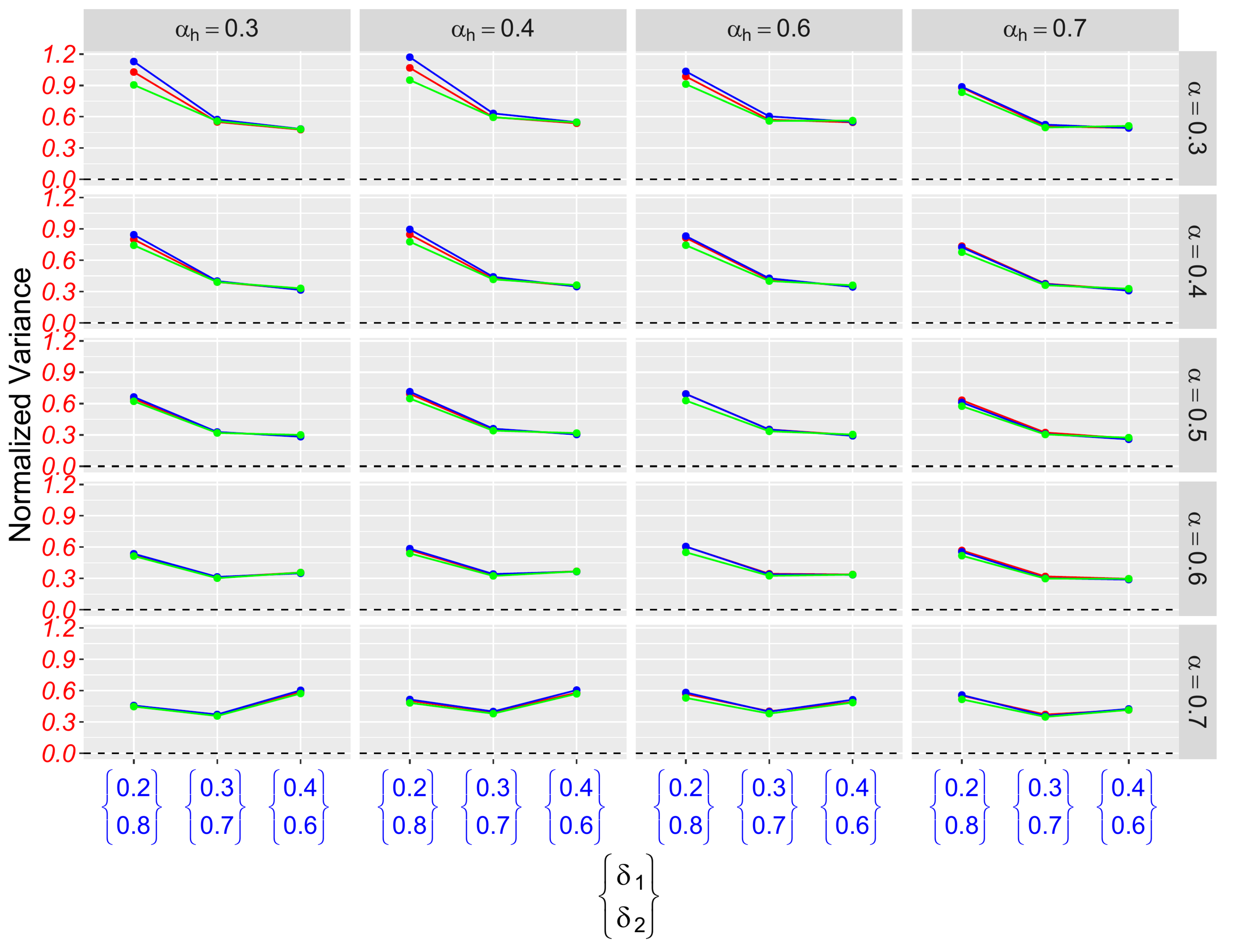}
        \caption{$\rma=1, h=1$}
    \end{subfigure}
    \begin{subfigure}{0.45\linewidth}
        \includegraphics[height=5.5cm]{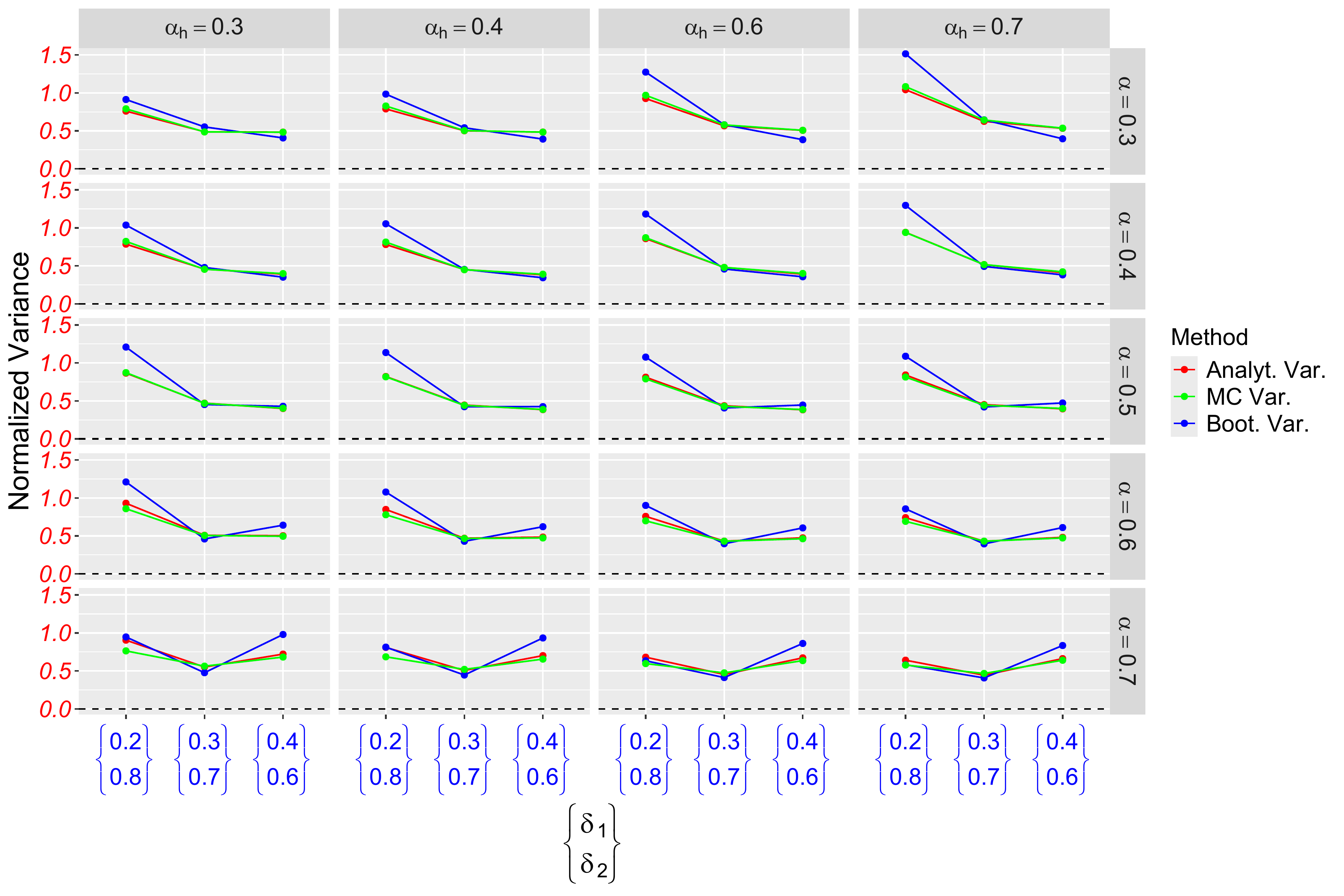}
        \caption{$\rma=1, h=2$}
    \end{subfigure}
    \caption{Validation of the Hajek variance estimator. Normalized variance estimated using analytical M-estimation (red), bootstrap (blue), and Monte Carlo simulations (green) for the spillover effects $\text{SE}^{h}\left(\alpha_{h}, 0.5;\rma, \alpha\right)$, with $\alpha_h\in\{0.3, 0.4, 0.6, 0.7\}$ and $\alpha\in\{0.3, 0.4, 0.5, 0.6, 0.7\}$, under Scenario 3, a two-stage assignment with $\{\delta_1, \delta_2 \}\in\{(0.2, 0.8), (0.3, 0.7), (0.4, 0.6)\}$, and regular networks with $p=0.2$.}
    \label{fig:var_Hajek}
\end{figure*}

\begin{figure*}[t]
    \centering
    \includegraphics[width=0.9\linewidth]{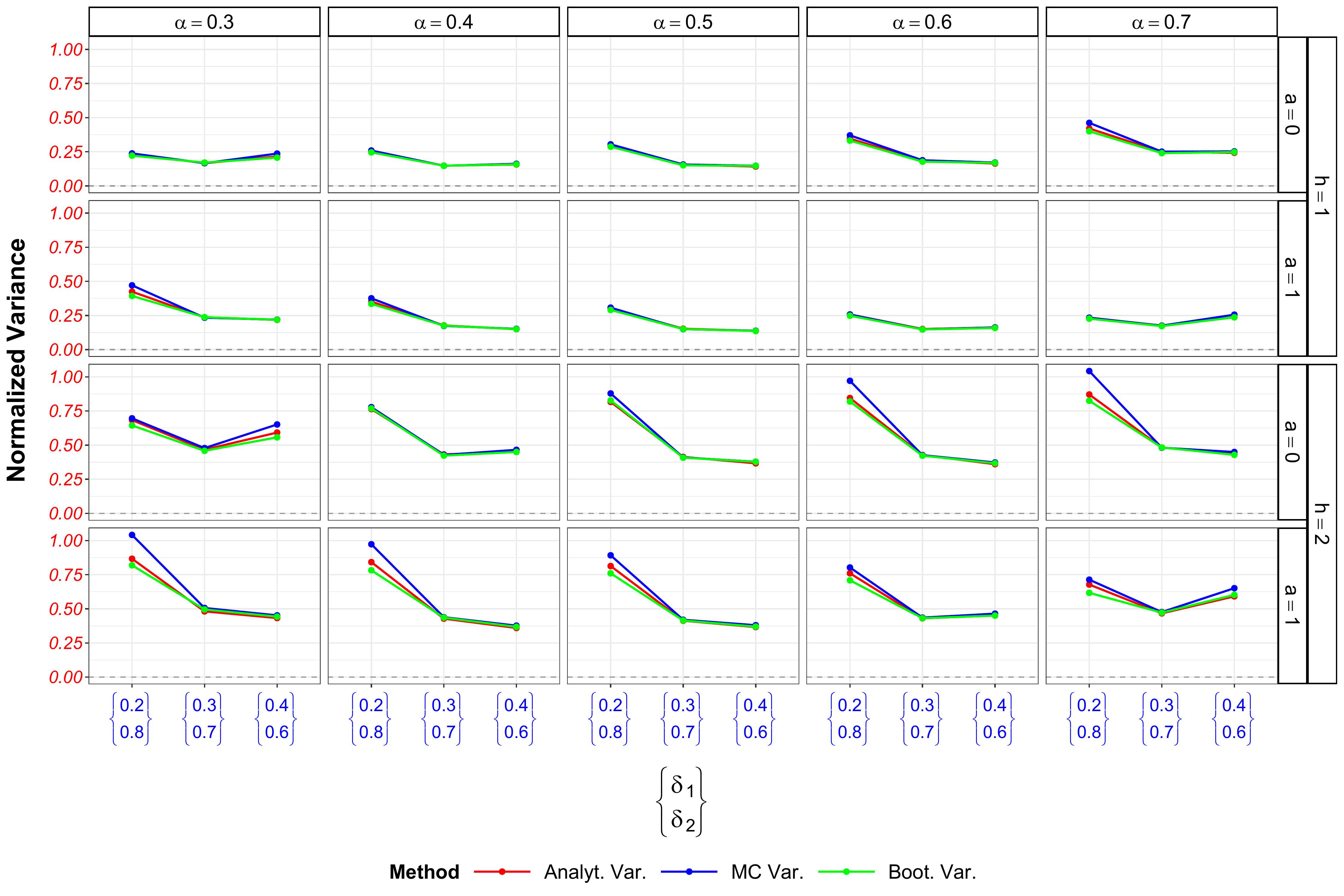}

    \caption{Validation of the WLS variance estimator. Normalized variance estimated using analytical M-estimation (red), Monte Carlo simulations (blue), and bootstrap (green) for the spillover effects $\text{SE}^{h}\left(\alpha_{h}, 0.5;\rma, \alpha\right)$, with $\alpha\in\{0.3, 0.4, 0.5, 0.6, 0.7\}$, under Scenario 3, a two-stage assignment with $\{\delta_1, \delta_2 \}\in\{(0.2, 0.8), (0.3, 0.7), (0.4, 0.6)\}$, and regular networks with $p=0.2$. The rows of the grid correspond to combinations of the treatment assignment $\rma \in \{0, 1\}$ and the network distance $h \in \{1, 2\}$.}
    \label{fig:var_WLS}
\end{figure*}

\subsubsection{Comparison of variance across different estimators}
\label{app:variance_comp}
Here we include all comparison plots for the analytical variance of the Horvitz-Thompson, Hajek, and WLS estimators, for Scenario 3 under a regular graph model with link  probability $p=0.2$

Figure~\ref{fig:ana_variance_with_ht} shows that the Horvitz-Thompson estimator (red) exhibits significantly higher variance than all other methods. This is expected, as Horvitz-Thompson estimator is highly sensitive to extreme weights, particularly when the hypothetical probabilities $(\alpha, \alpha_h)$ diverge from the design probabilities $(\delta_1, \delta_2)$.

Figure~\ref{fig:ana_variance_without_ht} excludes the Horvitz-Thompson estimator to compare the more efficient estimators. We observe that the WLS estimator (green) generally yields lower variance than the Hajek estimator (blue), because the WLS estimator exploits the parametric structure of the marginal structural model. As expected, the OLS estimators (orange, purple and deep blue) show the lowest variance due to their parametric nature, though this comes at the cost of bias when the model is misspecified (as seen in Section~\ref{app:bias}).

\begin{figure*}[!htbp]
    \centering
    \begin{subfigure}{0.5\linewidth}
      \includegraphics[height=5.5cm]{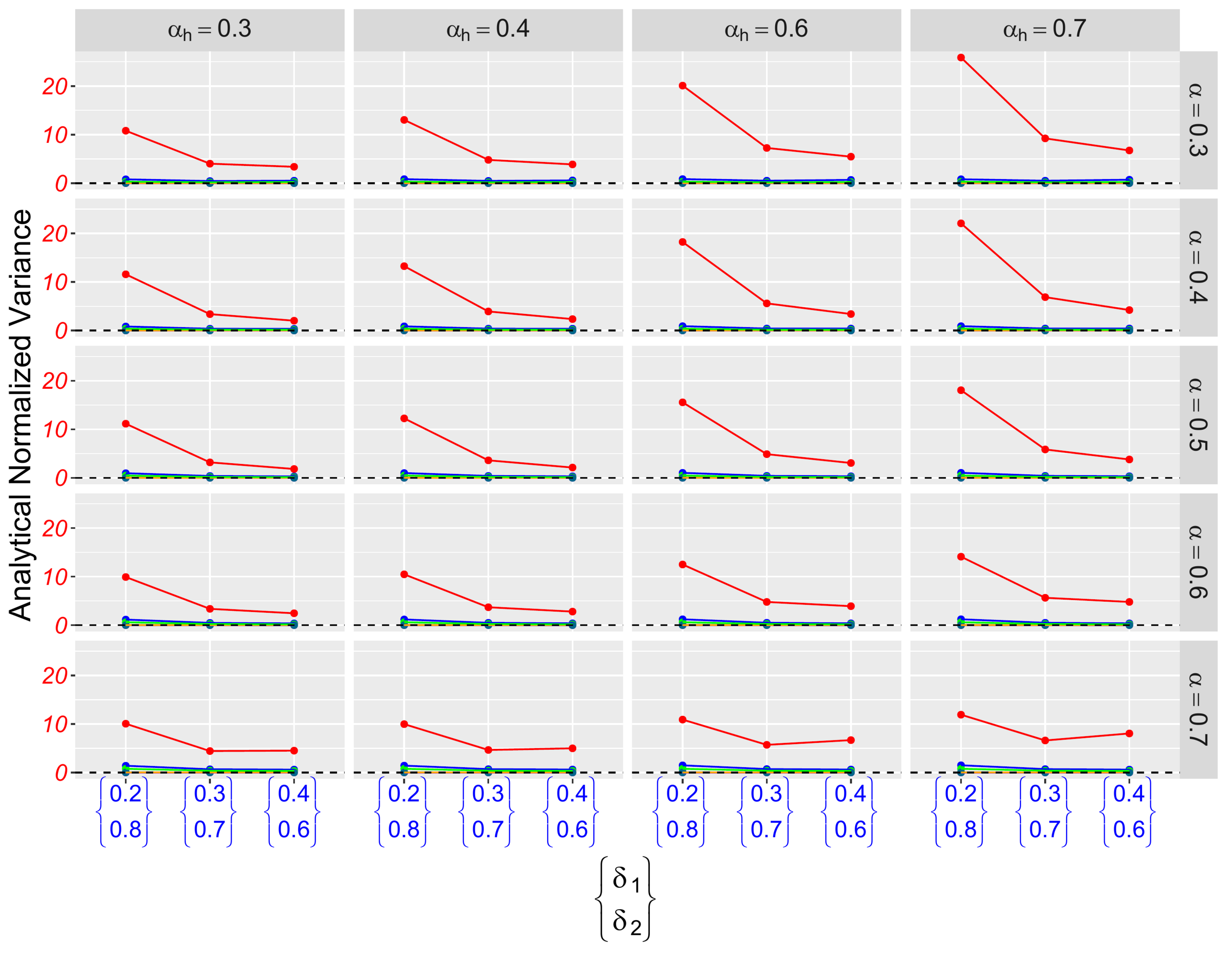}
      \caption{$\rma=0, h=1$}
      \label{fig:ana_wht_a0h1}
    \end{subfigure}
    \begin{subfigure}{0.45\linewidth}
      \includegraphics[height=5.5cm]{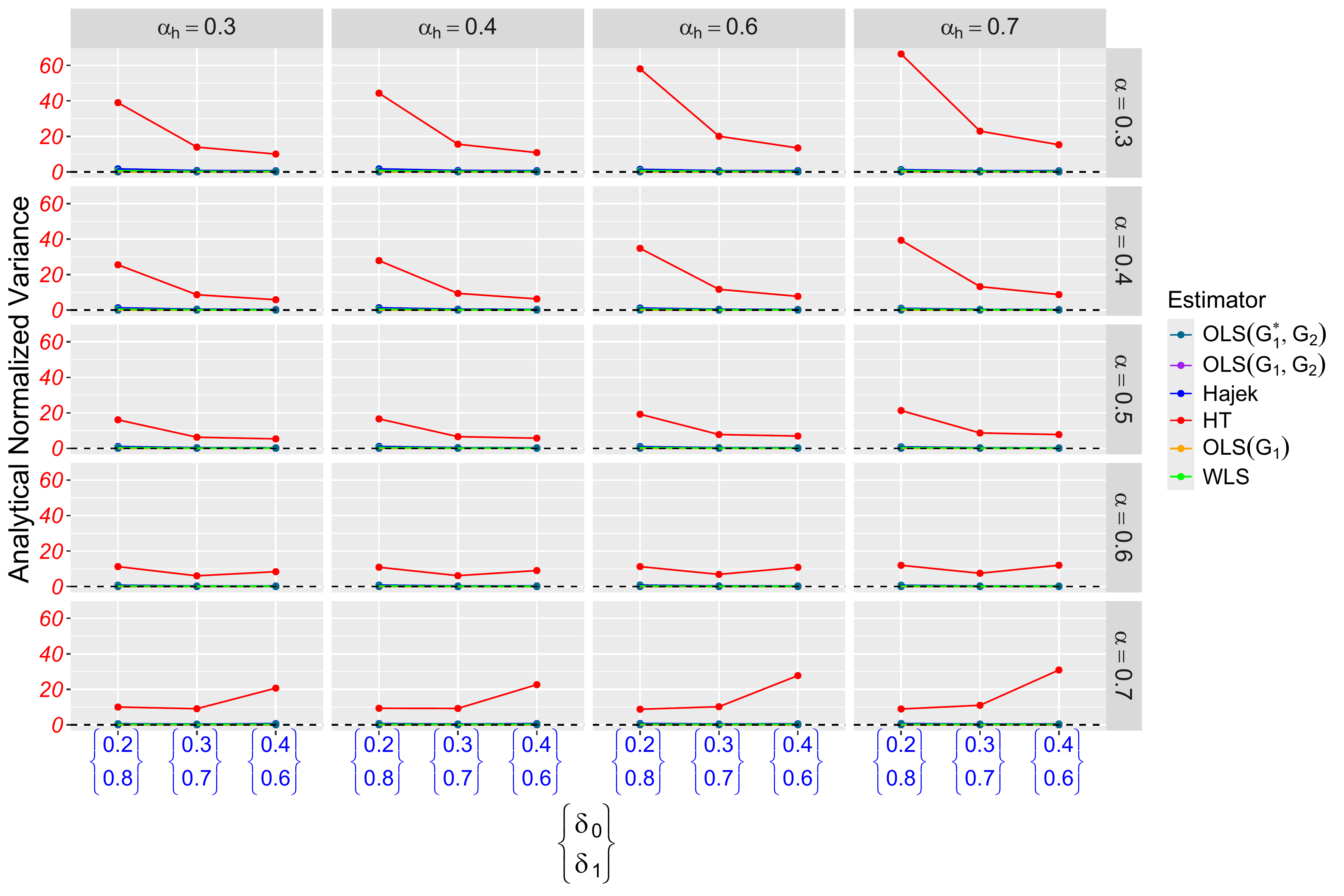}
      \caption{$\rma=1, h=1$}
      \label{fig:ana_wht_a1h1}
    \end{subfigure}

    \vspace{0.5cm}
    \begin{subfigure}{0.5\linewidth}
      \includegraphics[height=5.5cm]{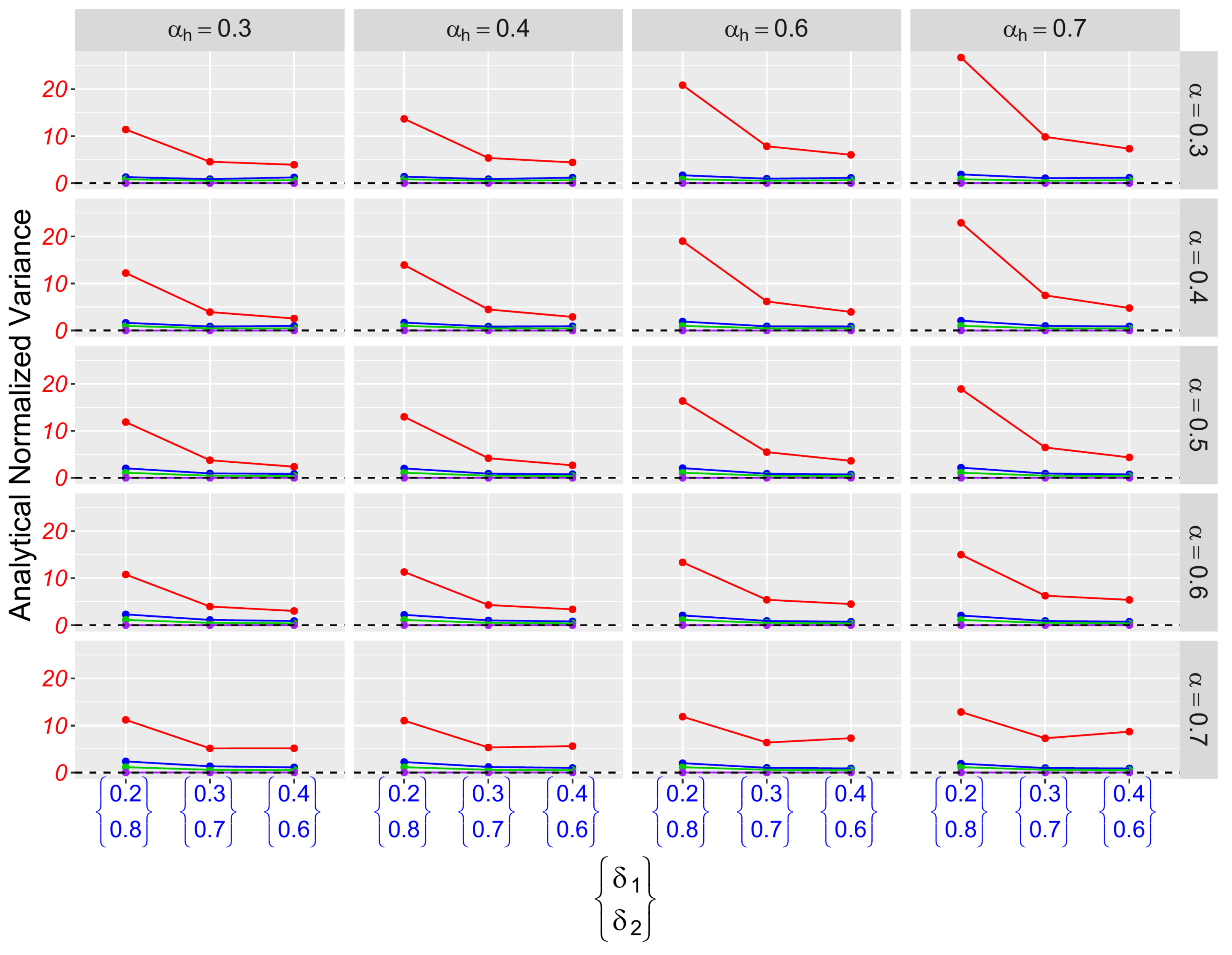}
      \caption{$\rma=0, h=2$}
      \label{fig:ana_wht_a0h2}
    \end{subfigure}
    \begin{subfigure}{0.45\linewidth}
      \includegraphics[height=5.5cm]{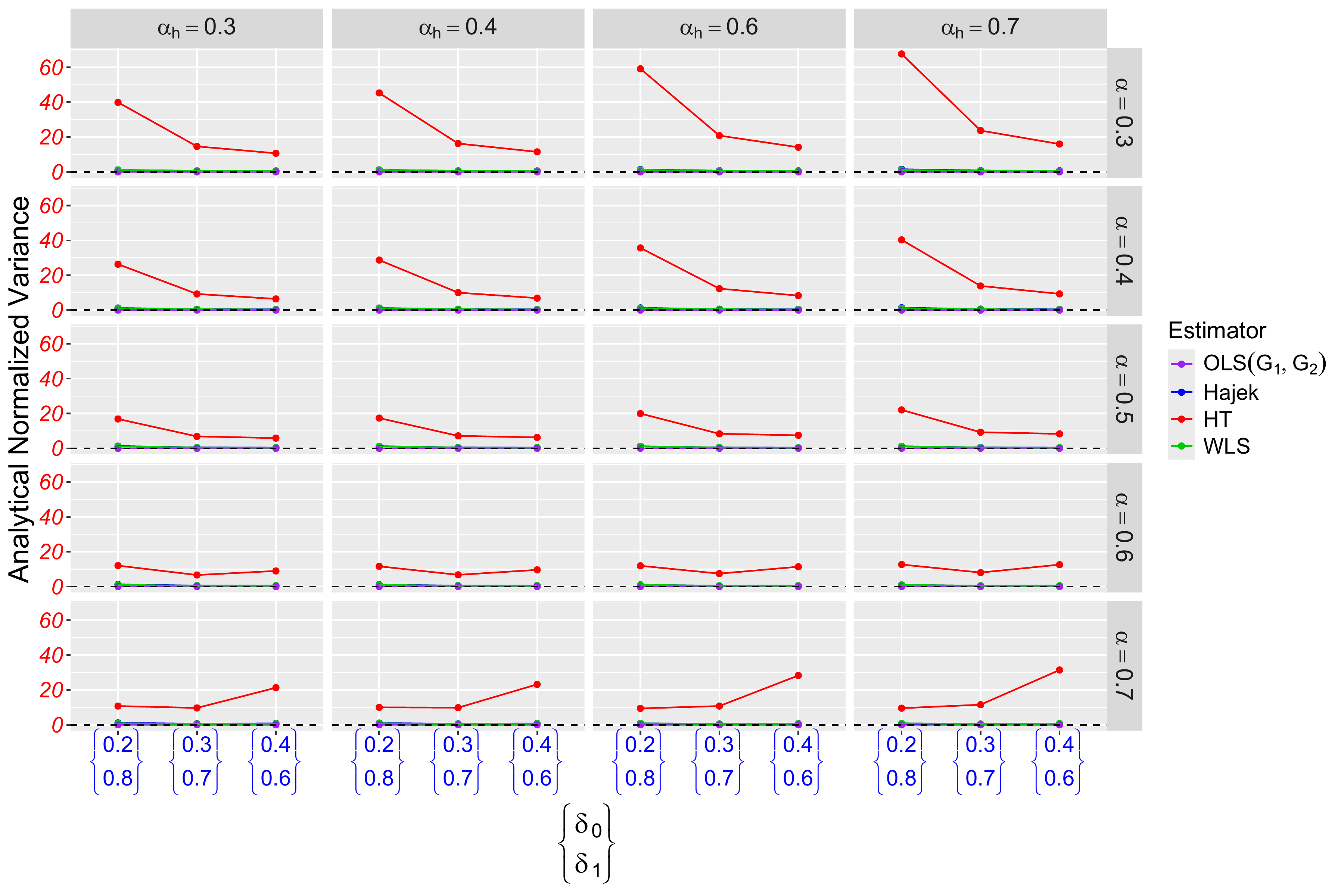}
      \caption{$\rma=1, h=2$}
      \label{fig:ana_wht_a1h2}
    \end{subfigure}
    \caption{Comparison of analytical normalized variance. Normalized variance of the HT (red), Hajek (blue), and WLS (green) estimators
    for the spillover effects $\text{SE}^{h}\left(\alpha_{h}, 0.5;\rma, \alpha\right)$, with $h=1, 2$, $\alpha_h\in\{0.3, 0.4, 0.6, 0.7\}$ and $\alpha\in\{0.3, 0.4, 0.5, 0.6, 0.7\}$, under Scenario 3, a two-stage assignment with $\{\delta_1, \delta_2 \}\in\{(0.2, 0.8), (0.3, 0.7), (0.4, 0.6)\}$, and regular networks with $p=0.2$.
    We also include the variance of the $\text{OLS}(G_1, G_2)$ estimator (purple) in all panels, and the $\text{OLS}(G_1)$ estimator (orange), $\text{OLS}(G^*_1, G_2)$ estimator (deep blue) specifically for $h=1$.
    }

    \label{fig:ana_variance_with_ht}
\end{figure*}

\begin{figure*}[!htbp]
    \centering
    \begin{subfigure}{0.5\linewidth}
      \includegraphics[height=5.5cm]{3._Study_of_different_Variance/left/left_Ana_a0h1_regular_Without_HT_Var_p=0.2.png}
      \caption{$\rma=0, h=1$}
      \label{fig:ana_woht_a0h1}
    \end{subfigure}
    \begin{subfigure}{0.45\linewidth}
      \includegraphics[height=5.5cm]{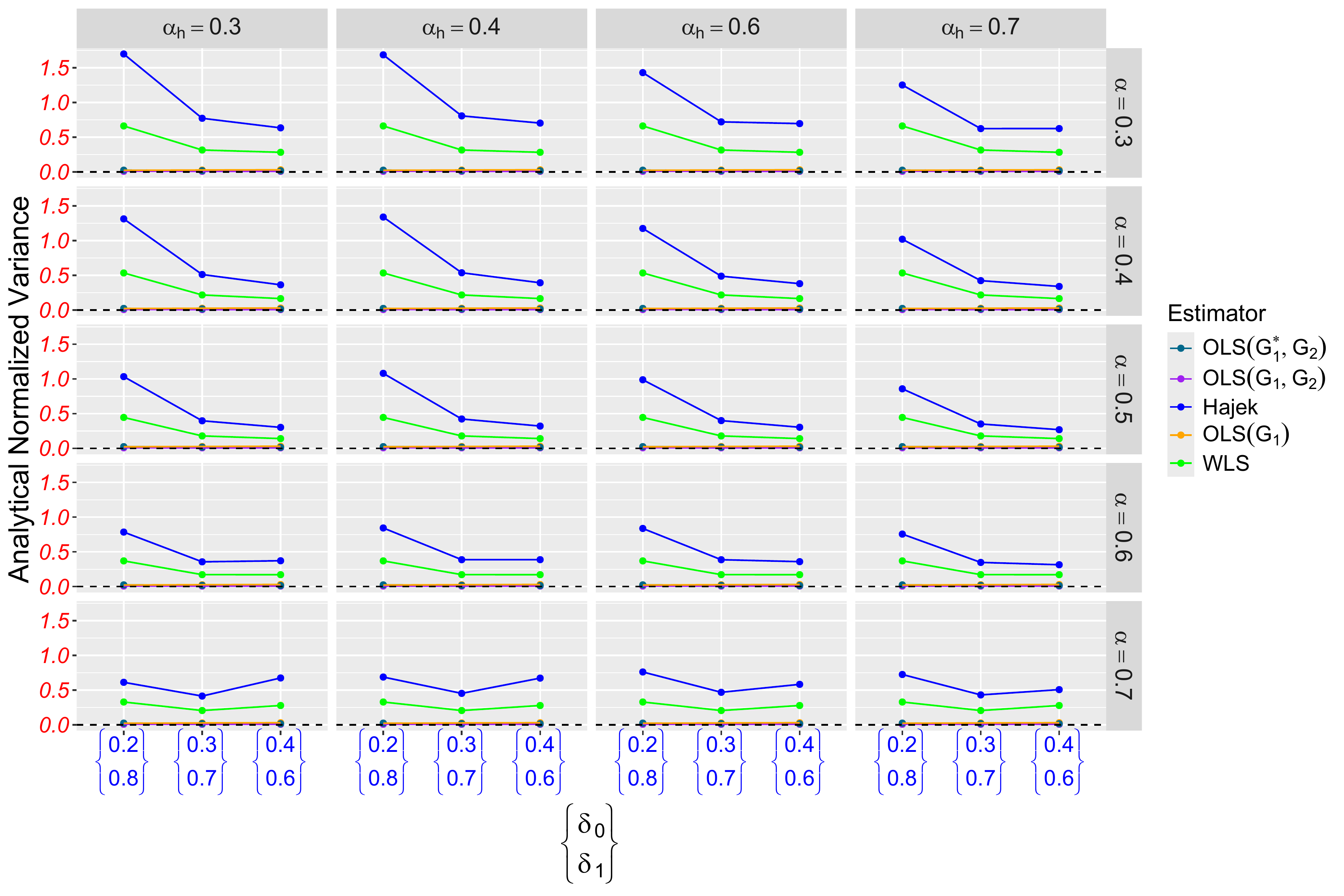}
      \caption{$\rma=1, h=1$}
      \label{fig:ana_woht_a1h1}
    \end{subfigure}

    \vspace{0.5cm}

    \begin{subfigure}{0.5\linewidth}
      \includegraphics[height=5.5cm]{3._Study_of_different_Variance/left/left_Ana_a0h2_regular_Without_HT_Var_p=0.2.png}
      \caption{$\rma=0, h=2$}
      \label{fig:ana_woht_a0h2}
    \end{subfigure}
    \begin{subfigure}{0.45\linewidth}
      \includegraphics[height=5.5cm]{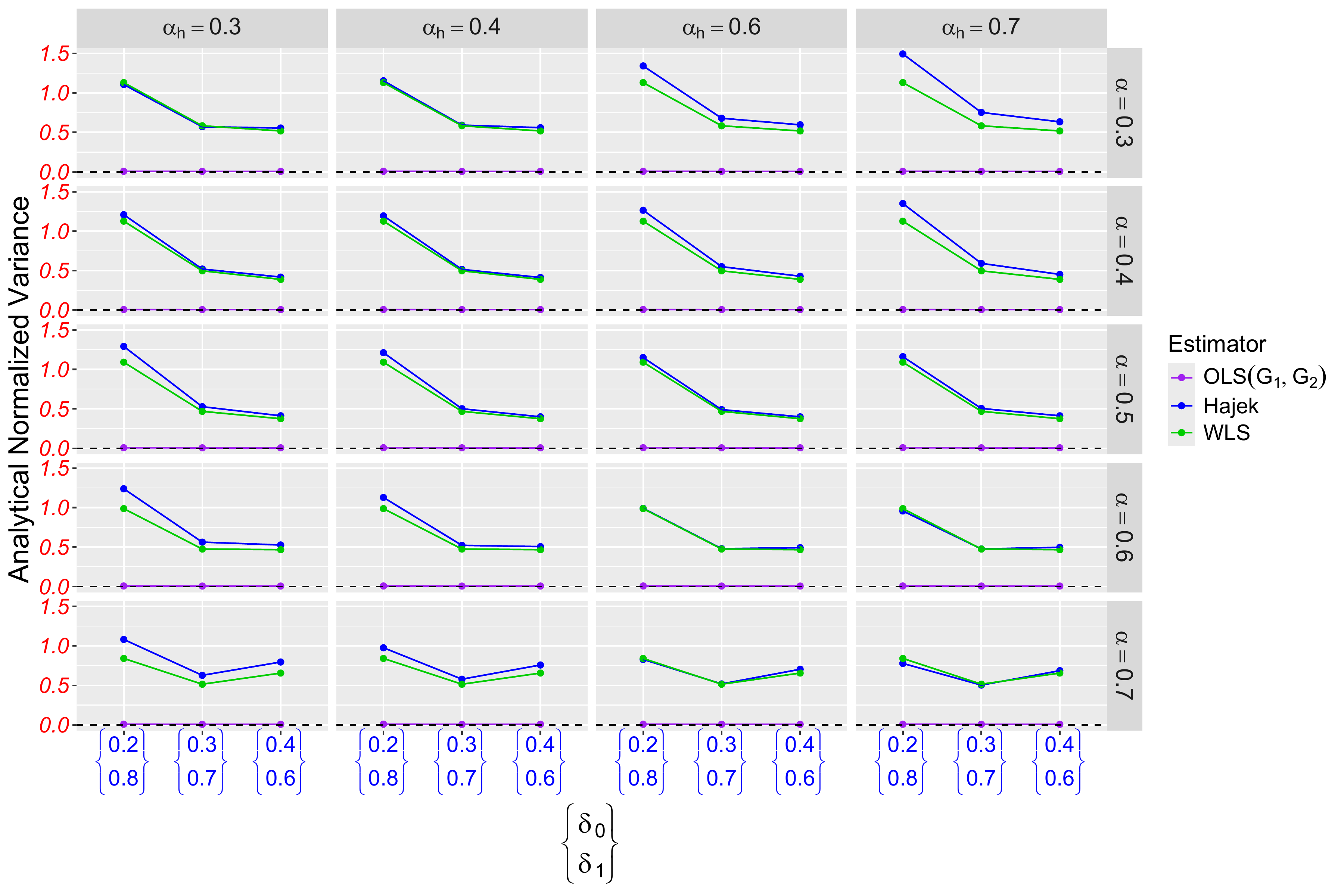}
      \caption{$\rma=1, h=2$}
      \label{fig:ana_woht_a1h2}
    \end{subfigure}
    \caption{Comparison of analytical variance. Normalized variance of the Hajek (blue) and WLS (green) estimators for the spillover effects $\text{SE}^{h}\left(\alpha_{h}, 0.5;\rma, \alpha\right)$, with $h=1, 2$, $\alpha_h\in\{0.3, 0.4, 0.6, 0.7\}$ and $\alpha\in\{0.3, 0.4, 0.5, 0.6, 0.7\}$, under Scenario 3, a two-stage assignment with $\{\delta_1, \delta_2 \}\in\{(0.2, 0.8), (0.3, 0.7), (0.4, 0.6)\}$, and regular networks with $p=0.2$. We also include the variance of the $\text{OLS}(G_1, G_2)$ estimator (purple) in all panels, and the $\text{OLS}(G_1)$ estimator (orange),$\text{OLS}(G^*_1, G_2)$ estimator (deep blue) specifically for $h=1$. This figure excludes the Horvitz-Thompson estimator to highlight the differences between Hajek and WLS.}
    \label{fig:ana_variance_without_ht}
\end{figure*}

\subsubsection{Study of Hajek estimator variance across network densities}
Here in figure~\ref{fig:hajek_var_p_a0}, we study how network density $p$ affects the variance of the Hajek estimator under Scenario~3, for both Erd\H{o}s--R\'enyi and regular graph models.

Overall, the variance patterns differ by the order of interference ($h$), reflecting how the size of the corresponding $h$-neighborhoods changes as $p$ increases.

For first-order effects ( $h=1$ ), the variance increases with $p$. Higher density leads to a larger number of first-order neighbors $\left(\left|\mathcal{N}_{i j}^1\right|\right)$, leading to more extreme propensity scores and higher variance.

For second-order effects ( $h=2$ ), the variance peaks around $p \approx 0.3-0.4$. Initially, increasing $p$ improves network connectivity, increasing the number of second-order neighbors $\left(\left|\mathcal{N}_{i j}^2\right|\right)$. However, as $p$ increases further in small clusters ( $n_i=10$ ), the network becomes saturated. Many units that were previously at distance 2 become directly connected, thereby shrinking the size of $\left(\left|\mathcal{N}_{i j}^2\right|\right)$ and reducing the variance.

Furthermore, we observe that these trends are consistent across network topologies, with Erdős-Rényi and regular graphs displaying very similar patterns.

\begin{figure*}[!htbp]

    \centering
    \begin{subfigure}{0.5\linewidth}
      \includegraphics[height=5cm]{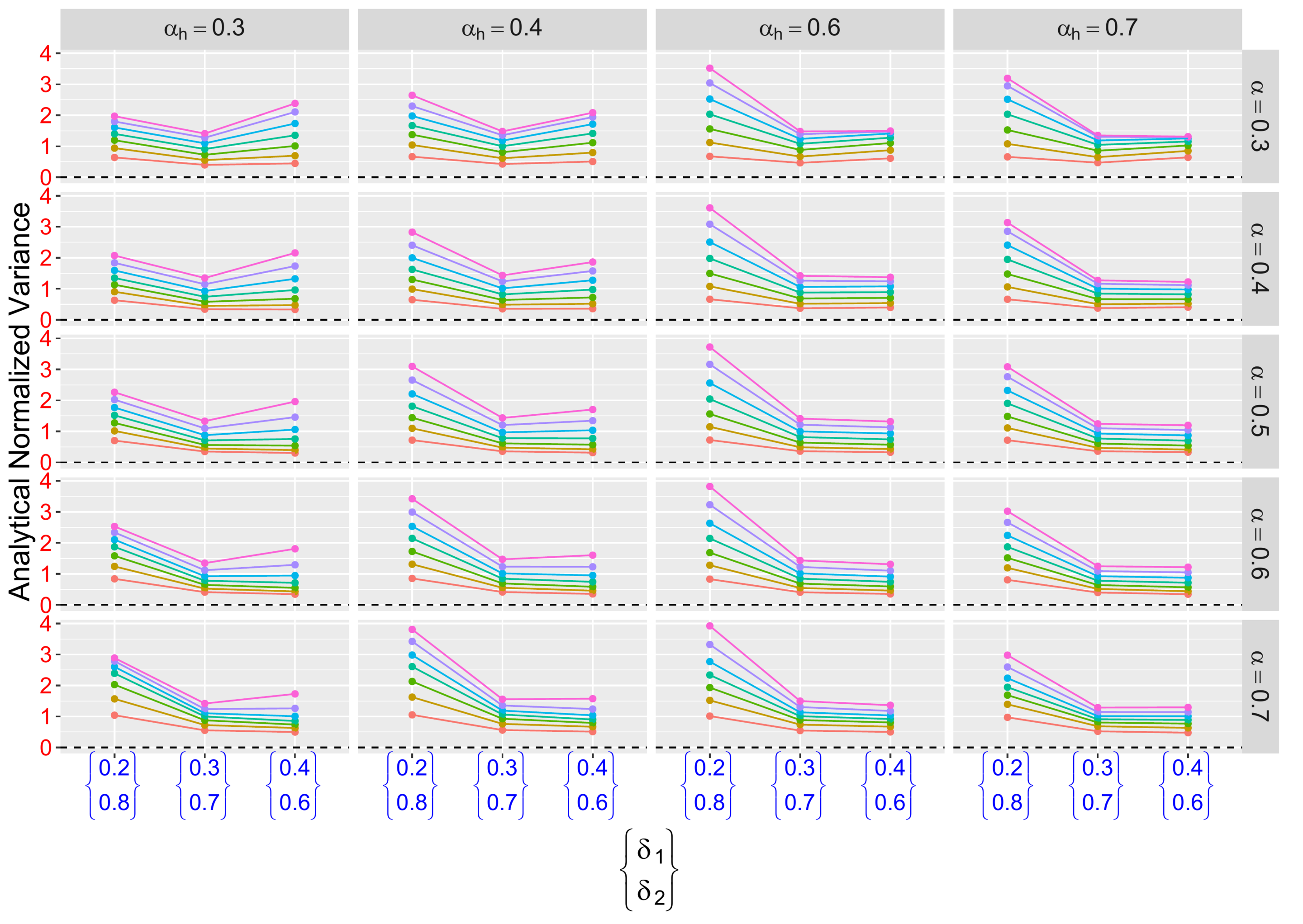}
      \caption{Erd\H{o}s--R\'enyi, $\rma=0, h=1$}
      \label{fig:hajek_p_ER_a0h1}
    \end{subfigure}
    \begin{subfigure}{0.45\linewidth}
      \includegraphics[height=5cm]{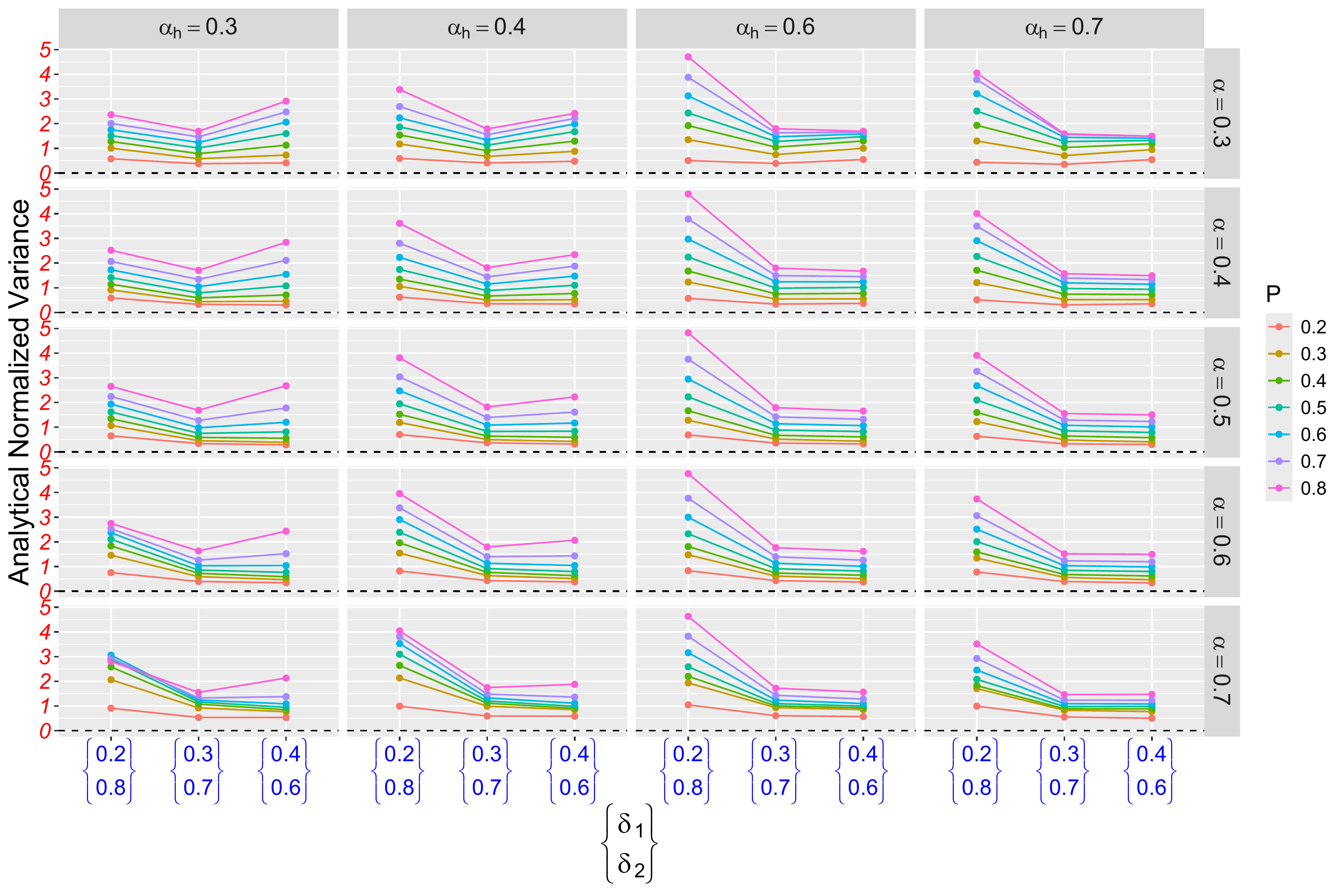}
      \caption{Regular, $\rma=0, h=1$}
      \label{fig:hajek_p_Reg_a0h1}
    \end{subfigure}

    \vspace{0.5cm}

    \begin{subfigure}{0.5\linewidth}
      \includegraphics[height=5cm]{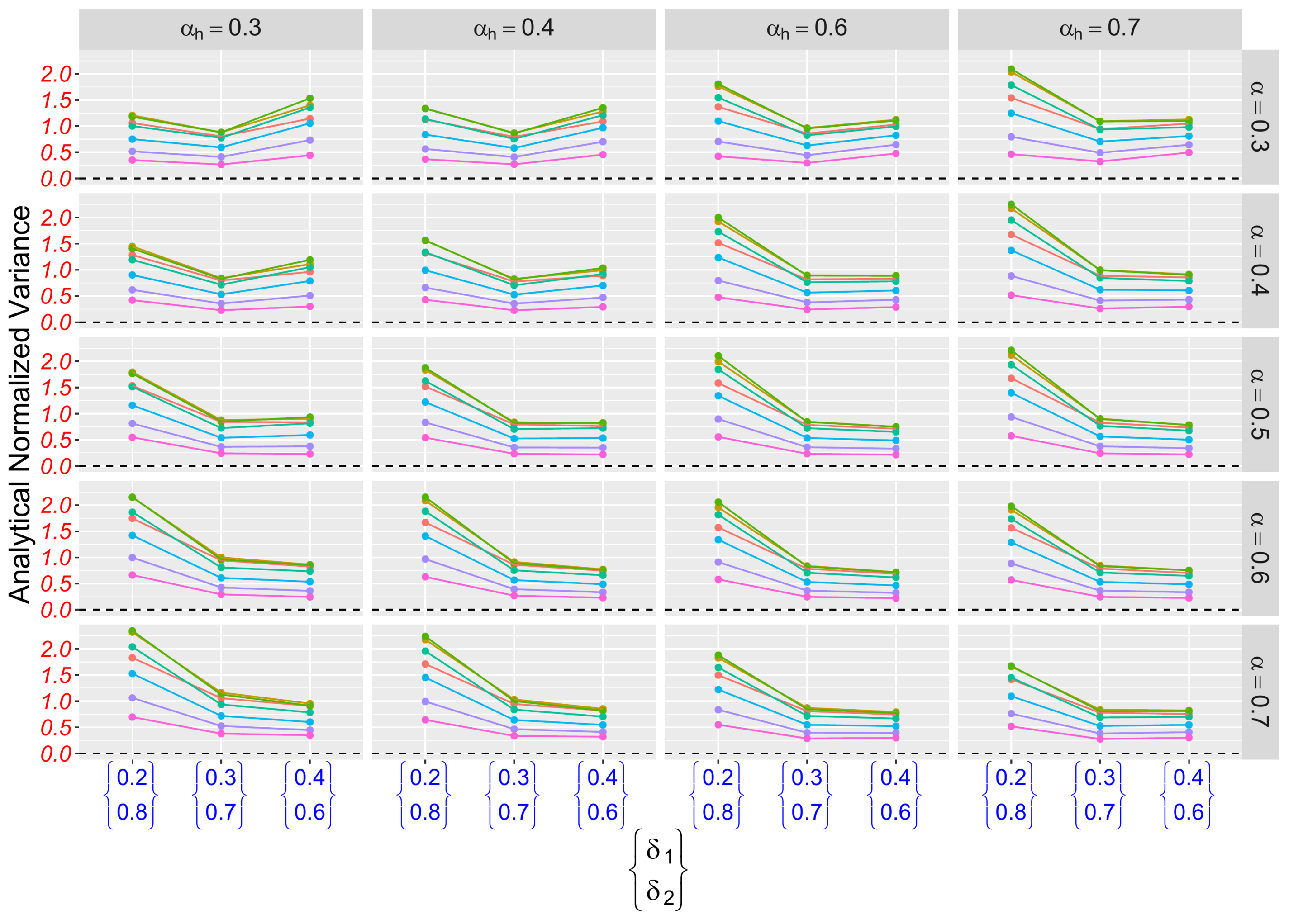}
      \caption{Erd\H{o}s--R\'enyi, $\rma=0, h=2$}
      \label{fig:hajek_p_ER_a0h2}
    \end{subfigure}
    \begin{subfigure}{0.45\linewidth}
      \includegraphics[height=5cm]{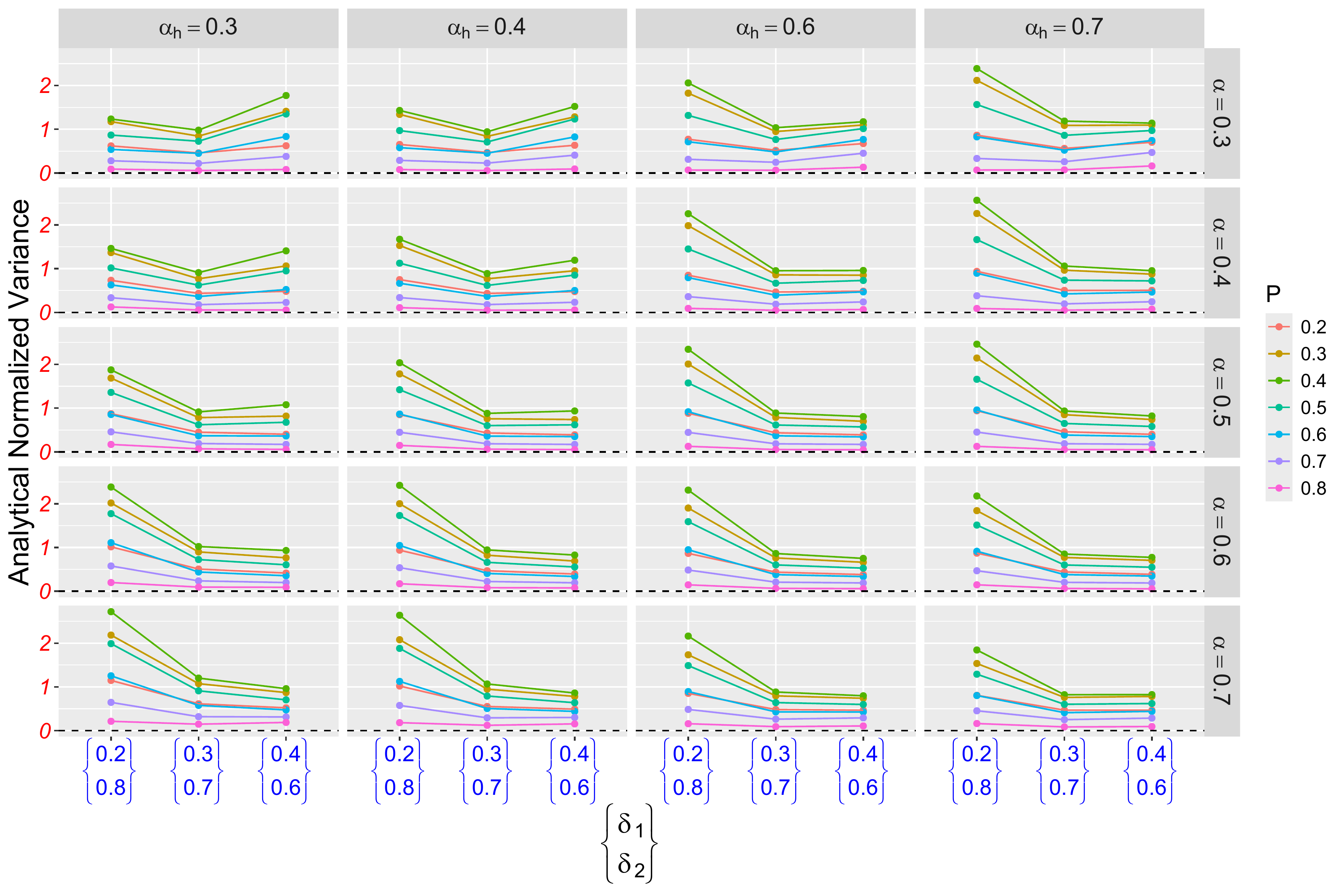}
      \caption{Regular, $\rma=0, h=2$}
      \label{fig:hajek_p_Reg_a0h2}
    \end{subfigure}
    \caption{Study of Hajek estimator variance across network densities. Analytical normalized variance of the Hajek estimator for the spillover effects $\text{SE}^{h}\left(\alpha_{h}, 0.5;\rma, \alpha\right)$ across different values of $p \in \{0.2, \dots, 0.8\}$ (represented by different colors), with $\alpha_h\in\{0.3, 0.4, 0.6, 0.7\}$ and $\alpha\in\{0.3, 0.4, 0.5, 0.6, 0.7\}$, under Scenario 3, a two-stage assignment with $\{\delta_1, \delta_2 \}\in\{(0.2, 0.8), (0.3, 0.7), (0.4, 0.6)\}$. The rows correspond to $h=1$ (top) and $h=2$ (bottom) with $\rma=0$. Left column: Erd\H{o}s--R\'enyi graphs; Right column: Regular graphs.}
    \label{fig:hajek_var_p_a0}
\end{figure*}

\subsection{Empirical Correlation Between Exposure Mappings} \label{app:correlation}
This section provides an explanation for the correlation seen between the exposure mappings $G^{1}_{ij}$ and $G^{2}_{ij}$ induced by the two-stage experiment. We report the empirical correlation between $G^{1}_{ij}$ and $G^{2}_{ij}$ across different first-stage probability settings ( $\delta_1, \delta_2$ ) to explain the bias of the $\text{OLS}(G_1)$ estimator when estimating the first-order spillover effect in Scenario 3 (see Section 5).

\begin{figure}[H]

    \centering
        \includegraphics[height=6cm]{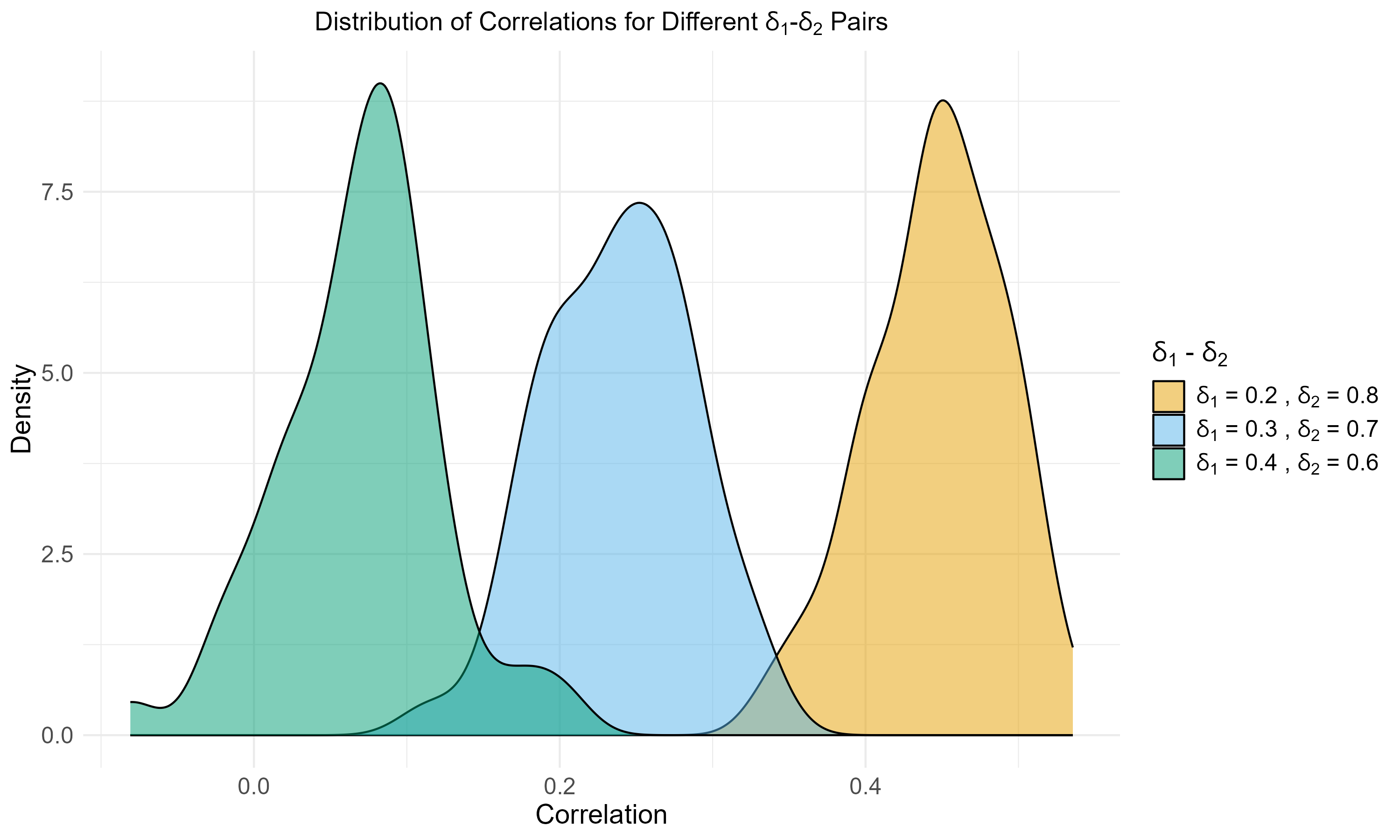}
         \caption{Correlation distributions for different \protect$\delta_1$ and \protect$\delta_2$}
    \label{correlation_Delta}
\end{figure}
As shown in Figure~\ref{correlation_Delta}, the correlation between $G^1_{ij}$ and $G^2_{ij}$ is in fact induced by the first-stage treatment randomization.
Within any cluster assigned to a fixed treatment probability, individual treatment assignments are independent. Therefore, conditional on the first-stage assignment, the proportion of treated first-order neighbors, $G^1_{ij}$,  and treated second-order neighbors, $G^2_{ij}$, are uncorrelated. However, when data from all clusters are pooled for analysis, a strong marginal correlation emerges. Clusters assigned to the low probability ( $\delta_1$ ) will systematically exhibit lower average values for both $G^1_{ij}$ and $G^2_{ij}$. Conversely, clusters assigned the high probability ($\delta_2$) will exhibit higher average values for both measures. Pooling these two distinct groups creates a positive association. This phenomenon is a classic example of Simpson's Paradox.

The magnitude of this correlation depends directly on the separation between $\delta_1$ and $\delta_2$. We illustrate this by comparing three scenarios: a high-separation case ( $\delta_1, \delta_2=0.2,0.8$ ), a medium-separation case ($\delta_1, \delta_2=0.4,0.6$ ), and a low-separation case ($\delta_1, \delta_2=0.47,0.53$ ). Note that for these illustrative figures, we set the cluster size to $n_i=100$ to clearly visualize the distributional separation and the resulting trends.

\begin{figure*}[!htbp]
    \centering
    \begin{subfigure}{\linewidth}
        \centering
        \includegraphics[height=6cm]{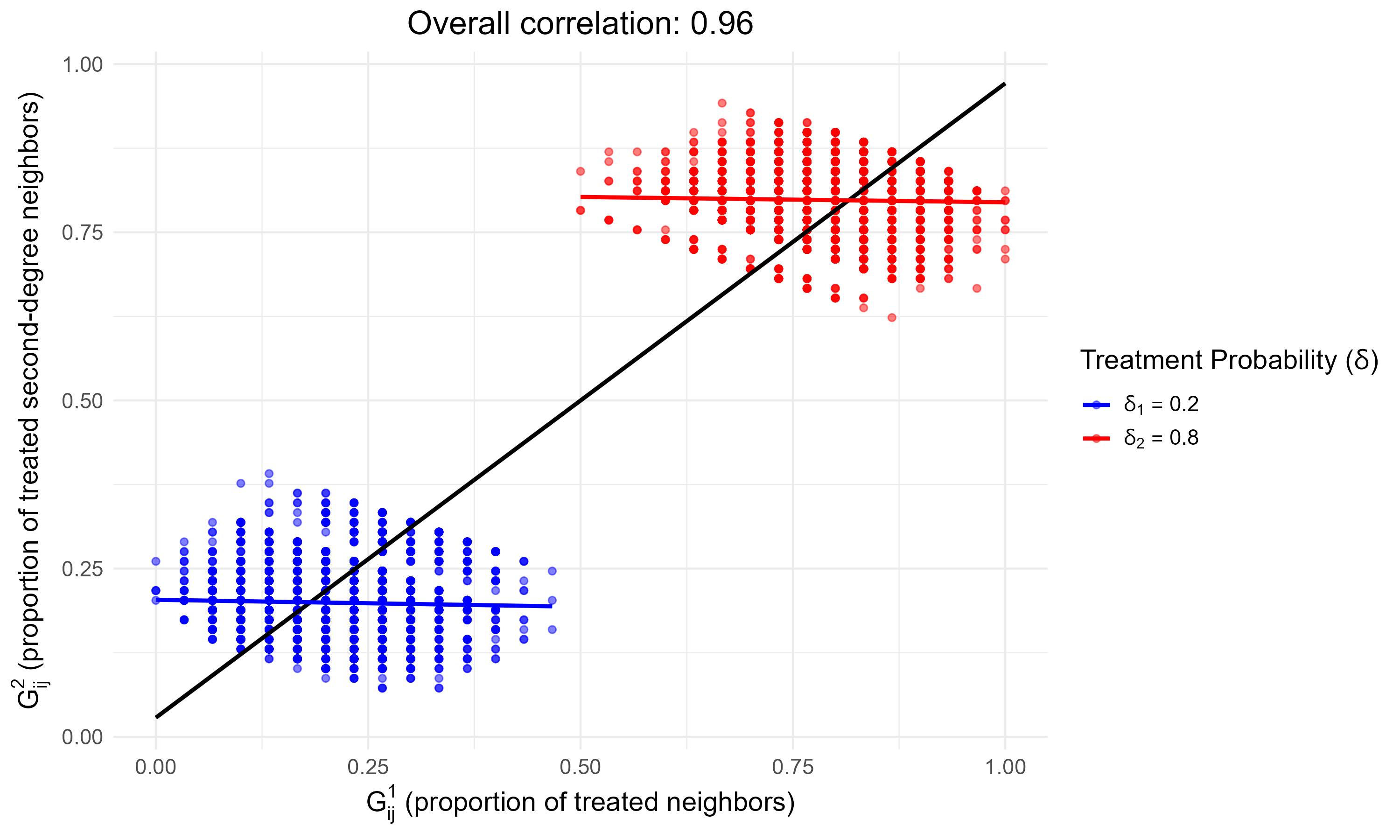}
        \caption{high-separation case: $\{\delta_1, \delta_2\} = \{0.2, 0.8\}$}
        \label{fig:simpson_02}
    \end{subfigure}
    \vspace{0.3cm}

    \begin{subfigure}{\linewidth}
         \centering
         \includegraphics[height=6cm]{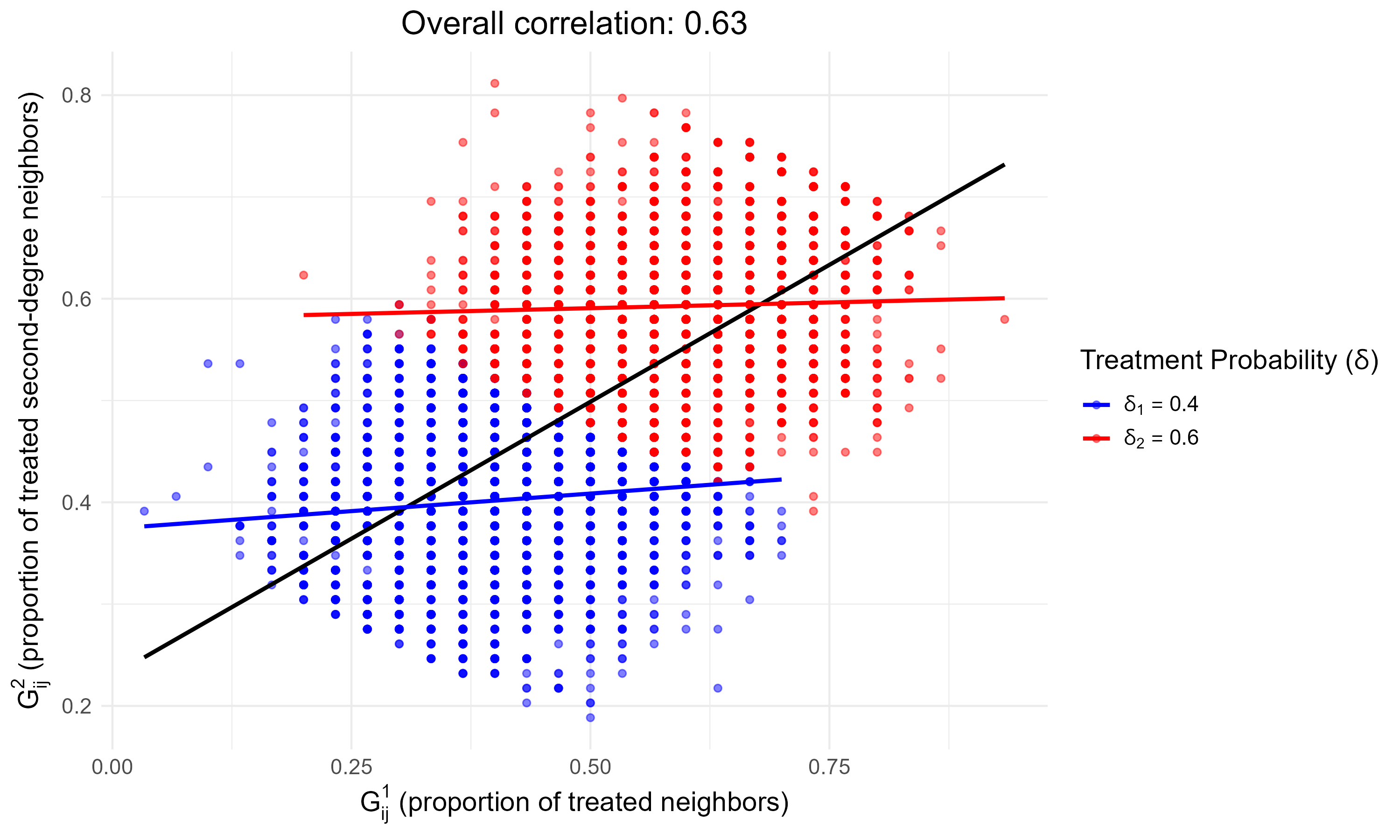}
         \caption{medium-separation case: $\{\delta_1, \delta_2\} = \{0.4, 0.6\}$}
         \label{fig:simpson_04}
    \end{subfigure}
    \vspace{0.3cm}

    \begin{subfigure}{\linewidth}
        \centering
        \includegraphics[height=6cm]{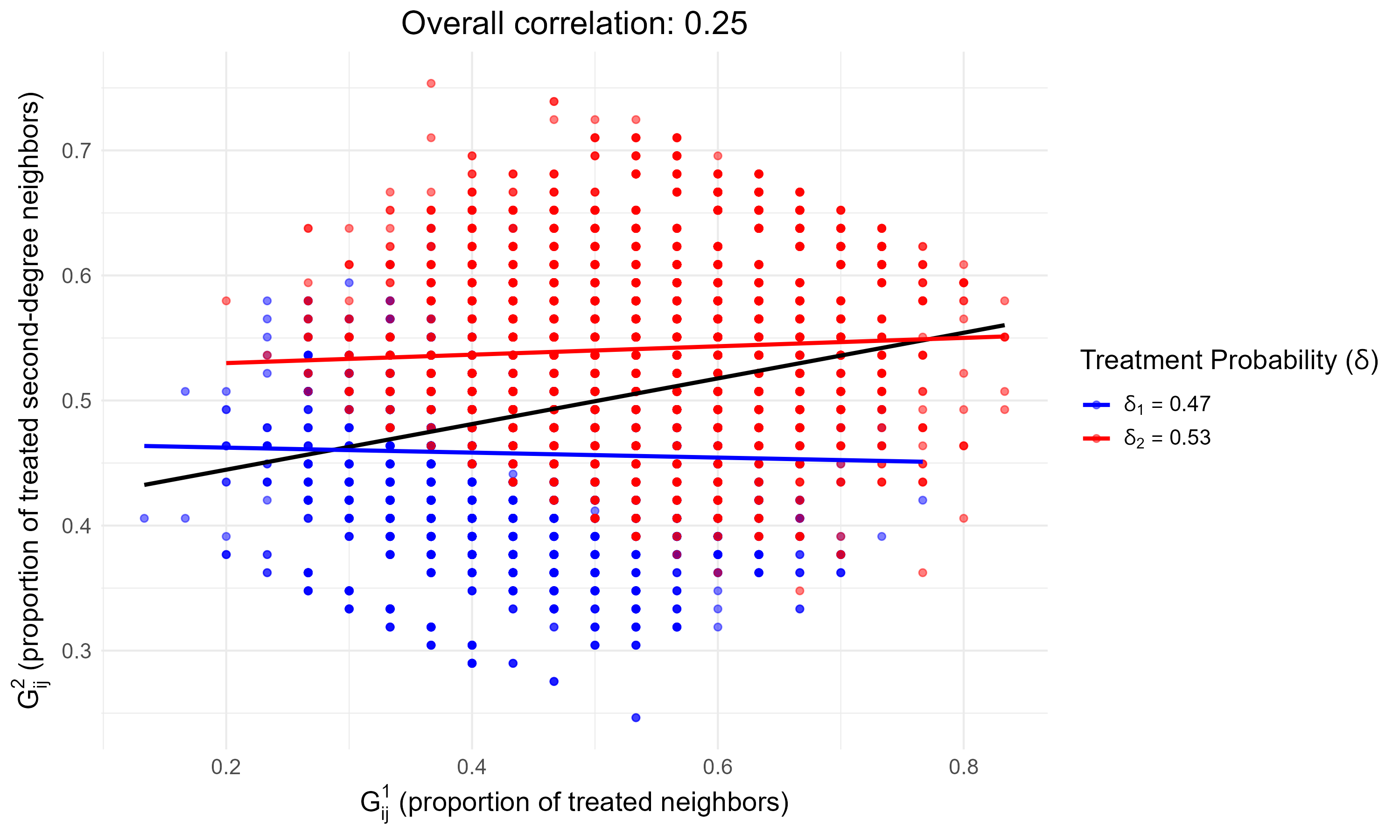}
        \caption{low-separation case: $\{\delta_1, \delta_2\} = \{0.47, 0.53\}$}
        \label{fig:simpson_047}
    \end{subfigure}

    \caption{Empirical correlation between exposures $G_{ij,1}$ and $G_{ij,2}$ under different first-stage probability separation settings ($\delta_1, \delta_2$). The plots illustrate how the separation between $\delta_1$ and $\delta_2$ induces a marginal correlation (Simpson's Paradox), which diminishes as the probabilities converge.}
    \label{fig:simpson_paradox}
\end{figure*}

As shown in Figure \ref{fig:simpson_paradox}, the impact of Simpson's Paradox weakens as the first-stage probabilities converge. Figure \ref{fig:simpson_02} shows the high-separation case ($\{\delta_1, \delta_2\}=\{0.2,0.8\}$), where distant data clouds induce a strong positive correlation. In contrast, Figures \ref{fig:simpson_04} and \ref{fig:simpson_047} show that as the probabilities approach each other (from $\{0.4,0.6\}$ to $\{0.47,0.53\}$), the distinct groups merge, substantially diminishing the correlation.

Consequently, the bias in the naive $\text{OLS}(G_1)$ estimator in Scenario 3, caused by omitting $G_{ij,2}$, is most severe in the setting of Figure \ref{fig:simpson_02} and is mitigated as the first-stage probabilities approach each other.

\section{Supplementary Materials for Empirical Application}
\label{app:empirical_supp}
\setcounter{equation}{0}
\setcounter{figure}{0}
\setcounter{table}{0}

\subsection{Household-level Assignment and Probability Calculations}
\label{app:household_formulas}

In the empirical application, treatment is assigned at the household level. Let $\mathcal{H}_i$ denote the set of households in cluster $i$. For each unit $j$ in cluster $i$, let $h(ij) \in \mathcal{H}_i$ denote the household to which unit $ij$ belongs. Therefore, we have $A_{ij} = A_{ik}$ if $h(ij) = h(ik)$. Let $\tilde{A}_{u}$ denote the binary treatment indicator for household $u \in \mathcal{H}_i$.
Specifically, for any unit $i k$ in cluster $i$, its individual treatment status $A_{i k}$ is determined by its household $h(i k)$:
$$
A_{i k}=\tilde{A}_{h(i k)}
$$
To adapt our estimators to this setting, we define the relevant sets at the household level. For a unit $ij$, let $\mathcal{I}_{ij}$ be its individual-level interference set. We define $\mathcal{U}_{ij}$ as the set of unique households that contain at least one member of $\mathcal{I}_{ij}$:
$$
\mathcal{U}_{ij} = \left\{ u \in \mathcal{H}_i : \exists k \text{ s.t. } ik \in \mathcal{I}_{ij} \text{ and } h(ik) = u \right\}.
$$
Similarly, we partition this set of households based on the $h$-order neighborhood. Let $\mathcal{U}^{h}_{ij} \subseteq \mathcal{U}_{ij}$ be the set of households containing at least one unit from the $h$-order neighborhood $\mathcal{C}^{h}_{ij}$, and let $\mathcal{U}^{rest}_{ij} = \mathcal{U}_{ij} \setminus \mathcal{U}^{h}_{ij}$ be the set of remaining households in the interference set.

Under the household-level Bernoulli approximation discussed in Section 2.2, treatment is assigned at the household level, then the probability of observing a specific treatment vector is non-zero only if all individuals in the same household share the same treatment. Any treatment vector $\mathbf{a}$ that assigns different values to members of the same household has zero probability  by design. Therefore, the probabilities are calculated based on the number of treated households.
Conditional on the first-stage cluster assignment $S_i = \delta_m$, the probability of observing the treatment vector for the interference set $\mathbf{A}_{\mathcal{I}_{ij}}$ (which corresponds to household treatments $\tilde{\mathbf{A}}_{\mathcal{U}_{ij}}$) is:
$$
P_{\Delta}\left(\mathbf{A}_{\mathcal{I}_{ij}} = \mathbf{a}_{\mathcal{I}_{ij}} \mid S_i = \delta_m \right) = \prod_{u \in \mathcal{U}_{ij} \cup \{h(ij)\}} \delta_m^{\tilde{a}_{u}} (1 - \delta_m)^{1 - \tilde{a}_{u}}.
$$
Note that this probability considers the union of the unit's own household and households in its interference set.
Therefore, the denominator of the HT weights $w_{ij}^h(\alpha, \alpha_h)$ (defined in Section 4.2.1) is
$$
P_{\Delta}\left(\mathbf{A}_{\mathcal{I}_{i j}}=\mathbf{a}_{\mathcal{I}_{i j}}\right)=\sum_{m=1}^M P\left(S_i=\delta_m\right) P\left(\mathbf{A}_{\mathcal{I}_{i j}}=\mathbf{a} \mid S_i=\delta_m\right)
$$

The numerator of the weight $w_{ij}^h(\alpha, \alpha_h)$ is defined by two hypothetical Bernoulli trials at the household level: one for households in the $h$-order neighborhood ($\mathcal{U}^{h}_{ij}$) with probability $\alpha_h$, and one for the remaining households ($\mathcal{U}^{rest}_{ij}$) with probability $\alpha$. The corresponding probabilities are:
$$
\begin{aligned}
P_{\alpha_h}\left(\mathbf{A}_{\mathcal{C}^{h}_{ij}} = \mathbf{a}_{\mathcal{C}^{h}_{ij}}\right) &= \prod_{u \in \mathcal{U}^{h}_{ij}} \alpha_h^{\tilde{a}_{u}} (1 - \alpha_h)^{1 - \tilde{a}_{u}}, \\
P_{\alpha}\left(\mathbf{A}_{\mathcal{C}_{i\backslash(j, \mathcal{C}^{h}_{ij})}} = \mathbf{a}_{\mathcal{C}_{i\backslash(j, \mathcal{C}^{h}_{ij})}}\right) &= \prod_{u \in \mathcal{U}^{rest}_{ij}} \alpha^{\tilde{a}_{u}} (1 - \alpha)^{1 - \tilde{a}_{u}}.
\end{aligned}
$$
These household-level probabilities are then substituted into the weight formulas for the Horvitz-Thompson, Hajek, and WLS estimators.

\subsection{Spillover Effects using the Horvitz-Thompson estimator}
Figures~\ref{fig:ht_a0h1}--\ref{fig:ht_a1h2} present spillover effect estimates using the Horvitz-Thompson estimator under the finite distance neighbor interference set. Compared to the Hajek estimator results reported in the main text, the Horvitz-Thompson estimator produces wider confidence intervals, consistent with our simulation findings regarding the lower efficiency of the Horvitz-Thompson estimator under a large cluster size.

\begin{figure*}[!htbp]
    \centering
    \begin{subfigure}{0.42\linewidth}
    \hspace*{-0.5cm}
        \includegraphics[height=5.5cm]{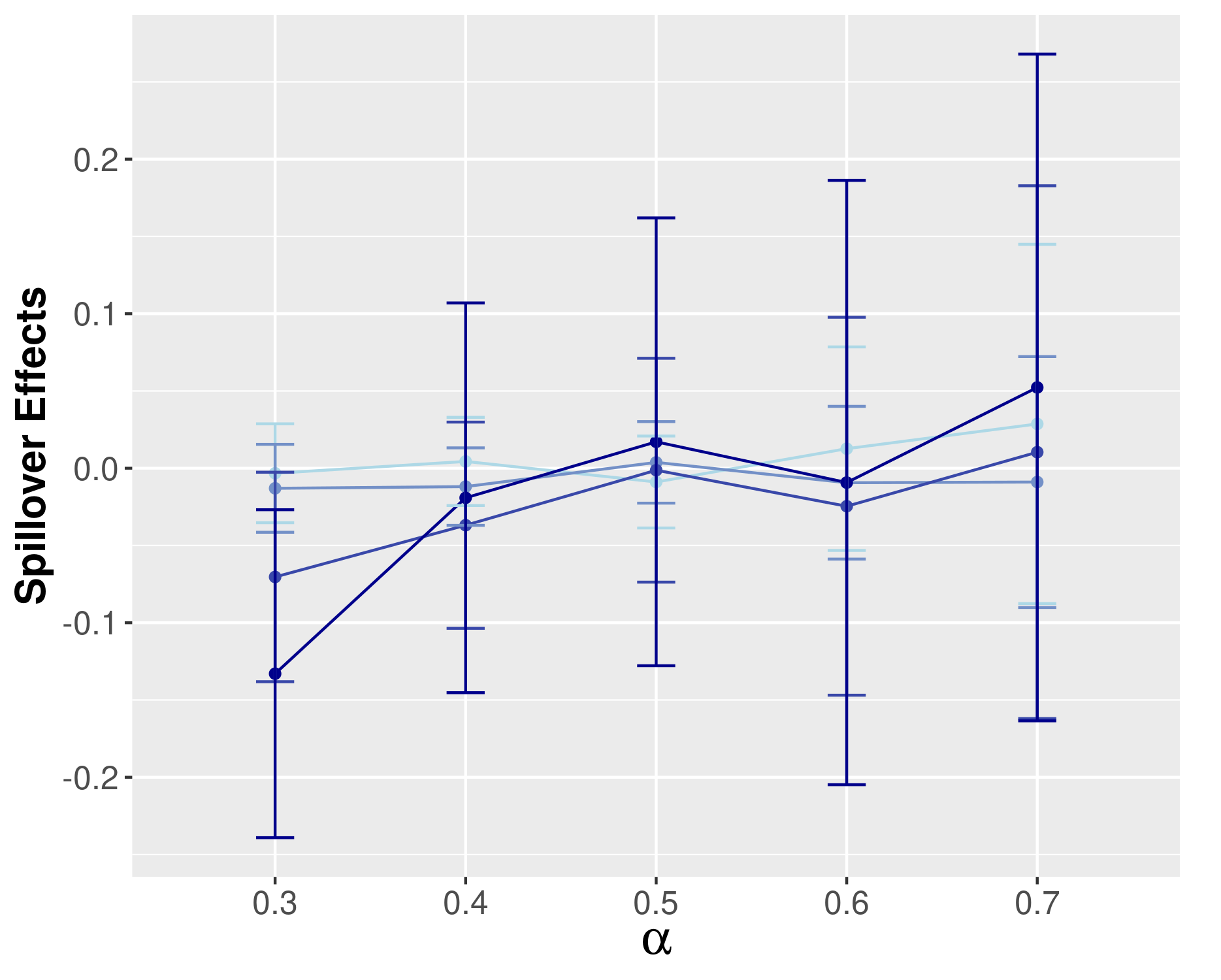}
        \caption{$\rma=0, h=1$}
        \label{fig:ht_a0h1}
    \end{subfigure}
    \hfill
    \begin{subfigure}{0.52\linewidth}
        \includegraphics[height=5.5cm]{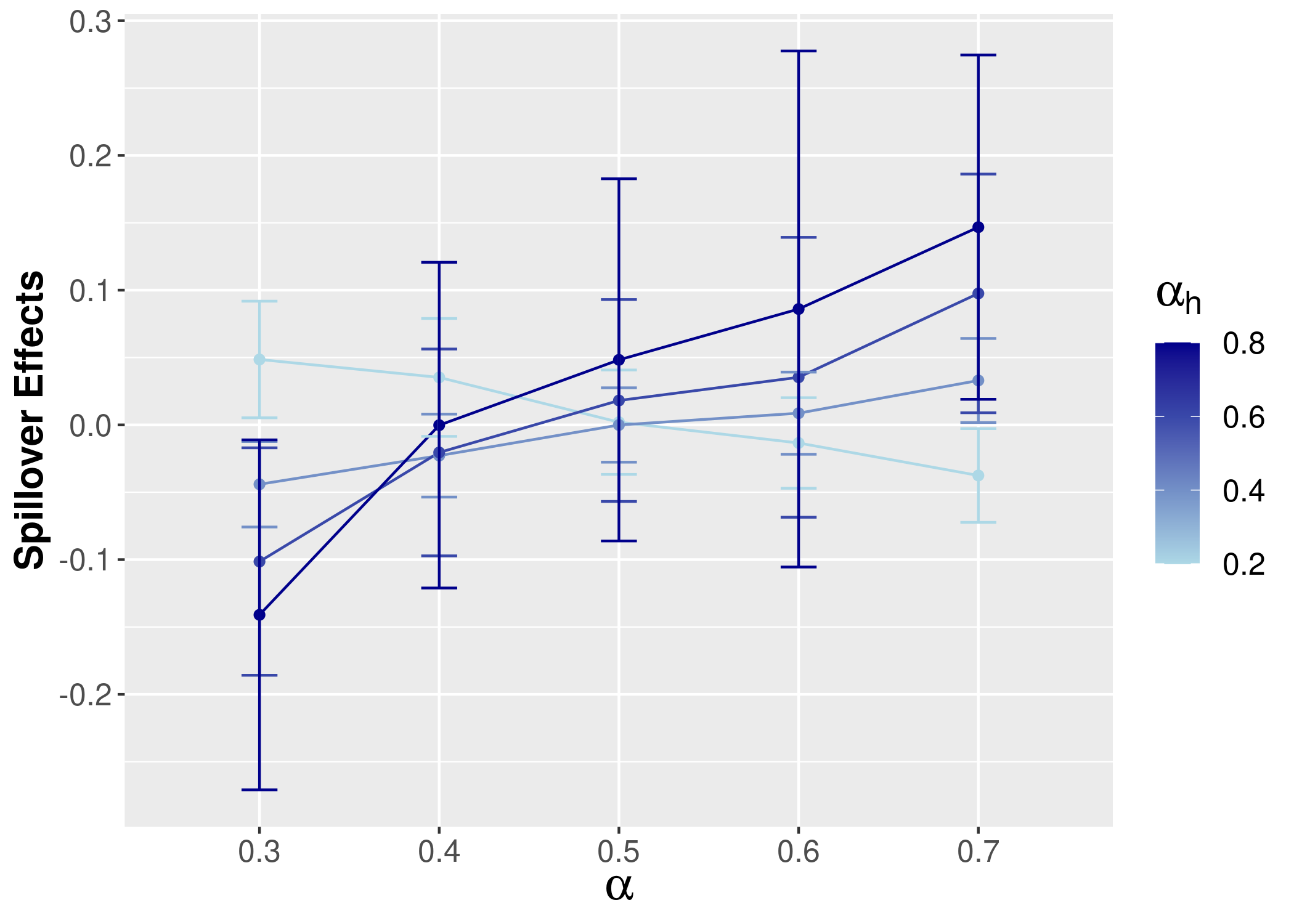}
        \caption{$\rma=1, h=1$}
        \label{fig:ht_a1h1}
    \end{subfigure}

    \vspace{0.5cm}

    \begin{subfigure}{0.42\linewidth}
    \hspace*{-0.5cm}
        \includegraphics[height=5.5cm]{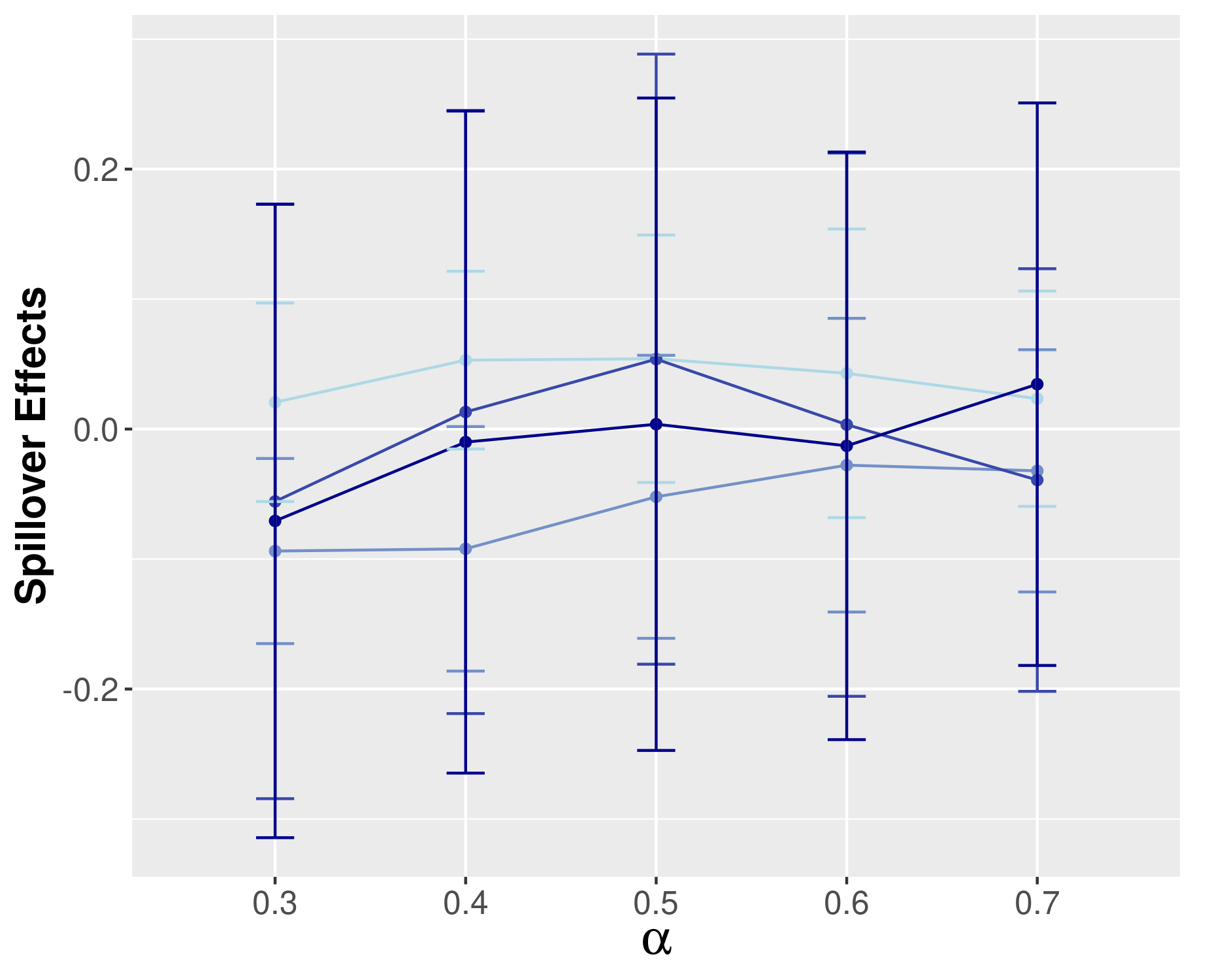}
        \caption{$\rma=0, h=2$}
        \label{fig:ht_a0h2}
    \end{subfigure}
    \hfill
    \begin{subfigure}{0.52\linewidth}
        \includegraphics[height=5.5cm]{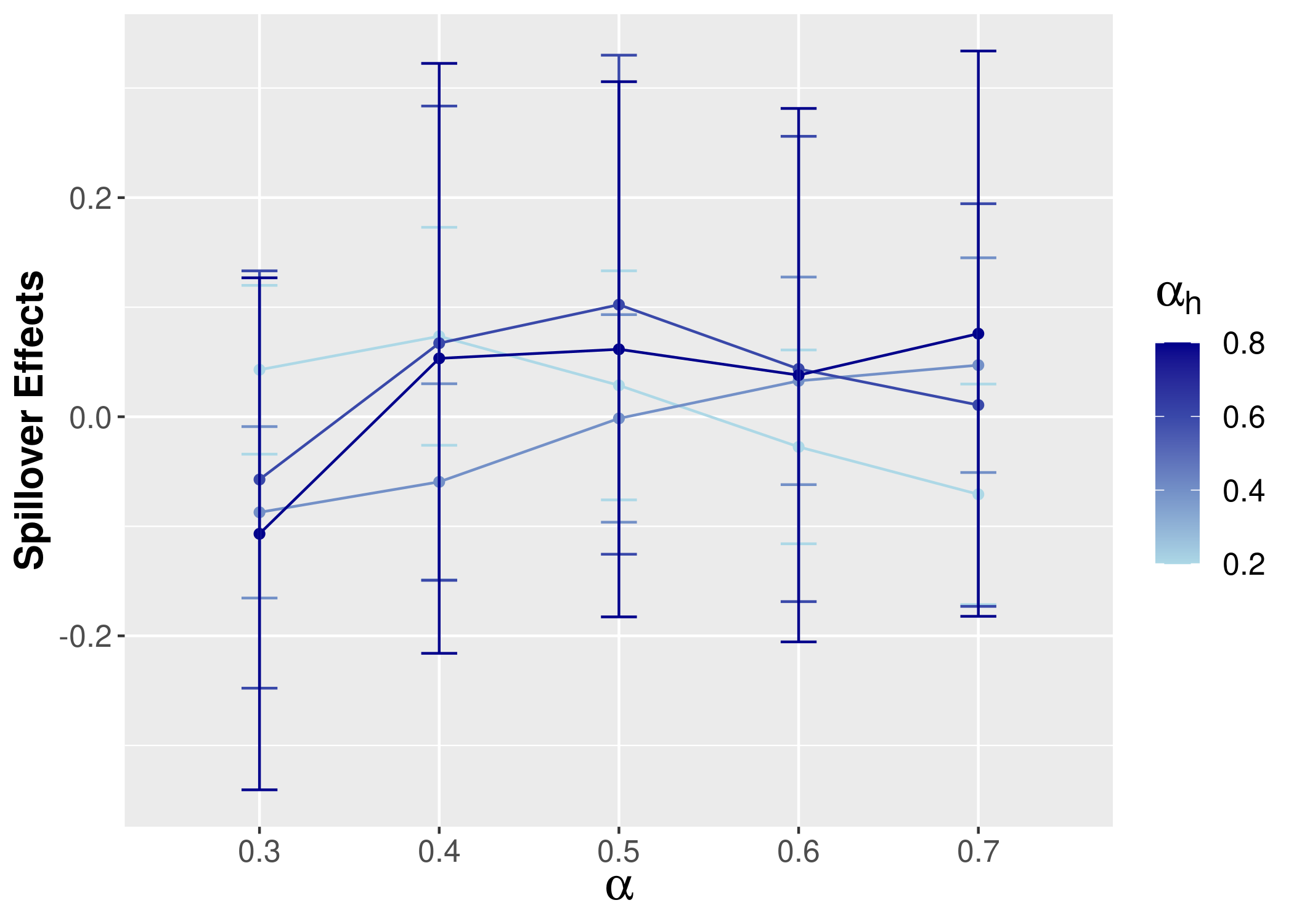}
        \caption{$\rma=1, h=2$}
        \label{fig:ht_a1h2}
    \end{subfigure}

    \caption{Horvitz-Thompson estimates of first- and second-order spillover effects. The top row displays first-order effects ($h=1$) and the bottom row displays second-order effects ($h=2$). Point estimates and 95\% confidence intervals are reported for the untreated (left) and for the treated (right), and for different values of $\alpha$ (x-axis) and $\alpha_h$ (colors), assuming the interference set as the subset of individuals reachable through a finite path. }
    \label{fig:ht_spillover_combined}
\end{figure*}

\subsection{Results Under Different Interference Set Definitions}
\label{app:interference}

To examine the robustness of our findings to the specification of the interference set,  we report results under two additional interference set definitions.

\subsubsection{Distance-\texorpdfstring{$3$}{3} Neighbor Interference Set}

First, we consider the $k^{\text{th}}$-order neighborhood interference assumption (Assumption 4.C) with $k=3$.
That is, for each unit $ij$ we set
$$
\mathcal{I}_{ij} = \{ik \in \mathcal{N}_i \setminus \{ij\}: d_{ij,ik}\le 3\}
= \bigcup_{h=1}^{3}\mathcal{N}^{h}_{ij}.
$$

Figures~\ref{fig:hajek_d3_a0h1}--\ref{fig:wr_d3_a0h2} present spillover effect estimates using the Hajek and WLS estimators under this distance-$3$ interference set, yielding similar patterns to the main analysis using the finite-path interference set (Figure 5--8). This is expected because, in this empirical network, the distance-$3$ neighborhood already contains the majority of reachable individuals for most units, so restricting interference to distance $3$ makes little practical difference.

\begin{figure*}[!htbp]
\centering

\begin{subfigure}{0.42\linewidth}
\hspace*{-0.5cm}
\centering
\includegraphics[height=5.5cm]{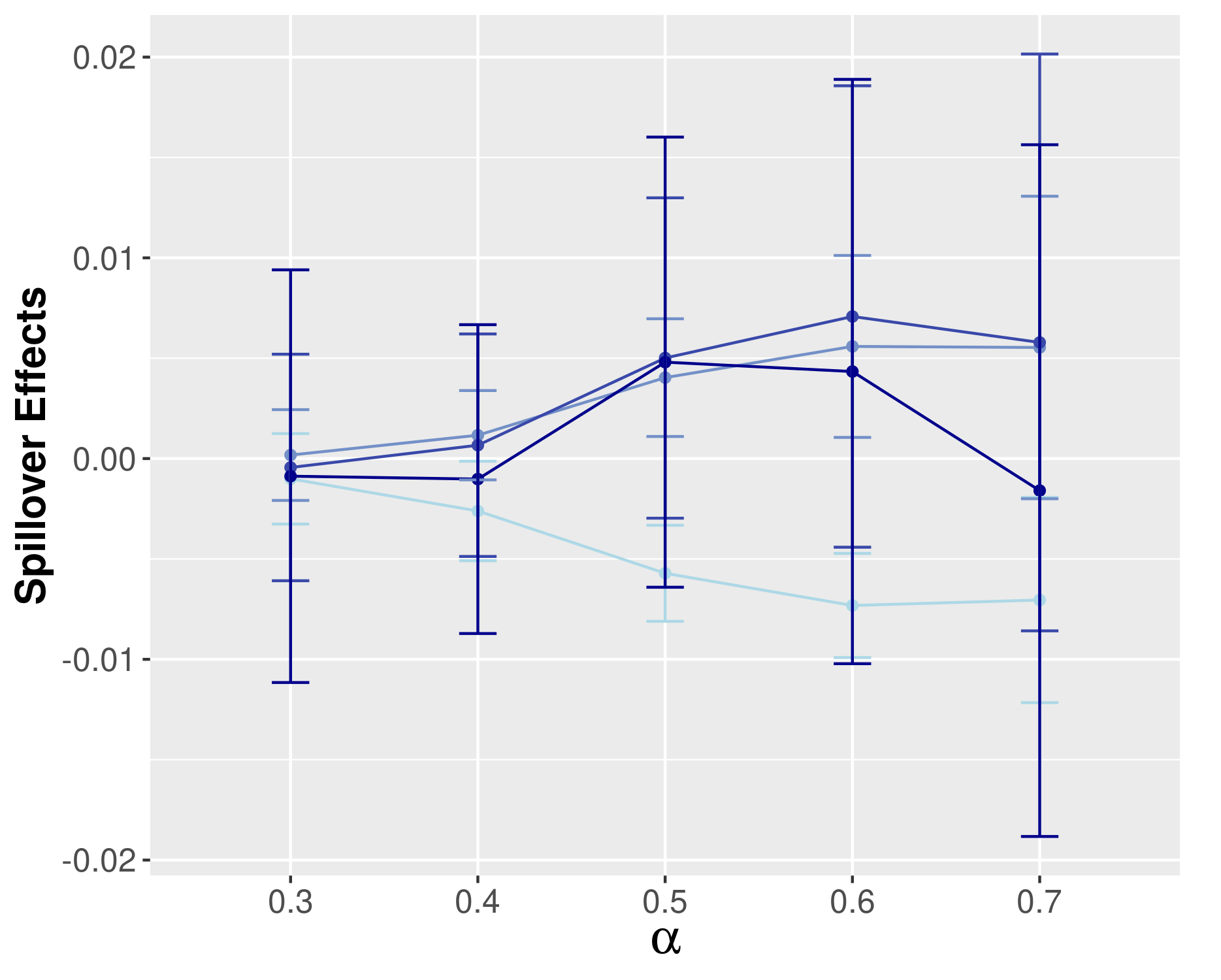}
\caption{$\rma=0,\ h=1$}
\label{fig:hajek_d3_a0h1}
\end{subfigure}
    \hfill
\begin{subfigure}{0.52\linewidth}
\centering
\includegraphics[height=5.5cm]{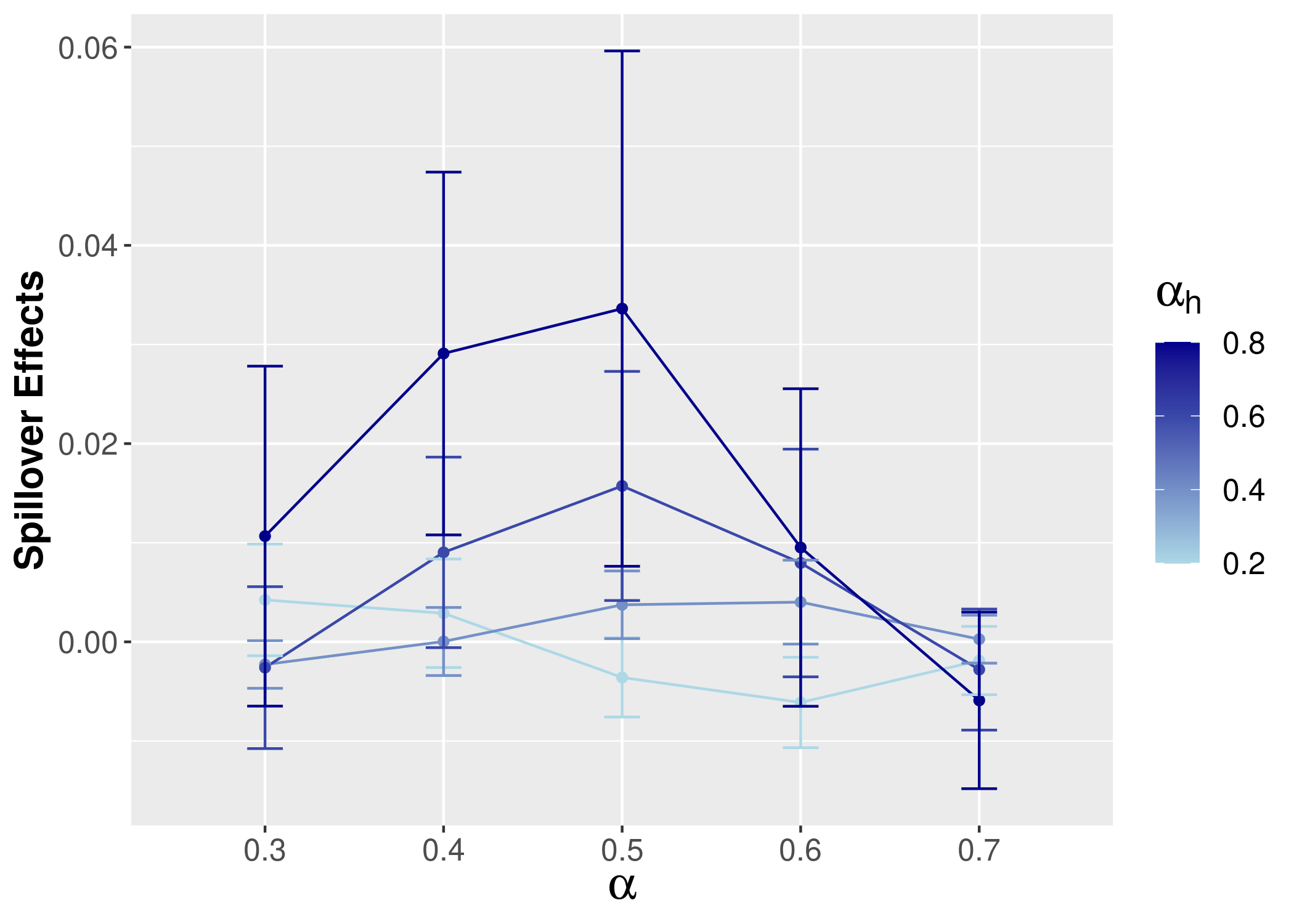}
\caption{$\rma=1,\ h=1$}
\label{fig:hajek_d3_a1h1}
\end{subfigure}

\vspace{0.4cm}

\begin{subfigure}{0.42\linewidth}
\hspace*{-0.5cm}
\centering
\includegraphics[height=5.5cm]{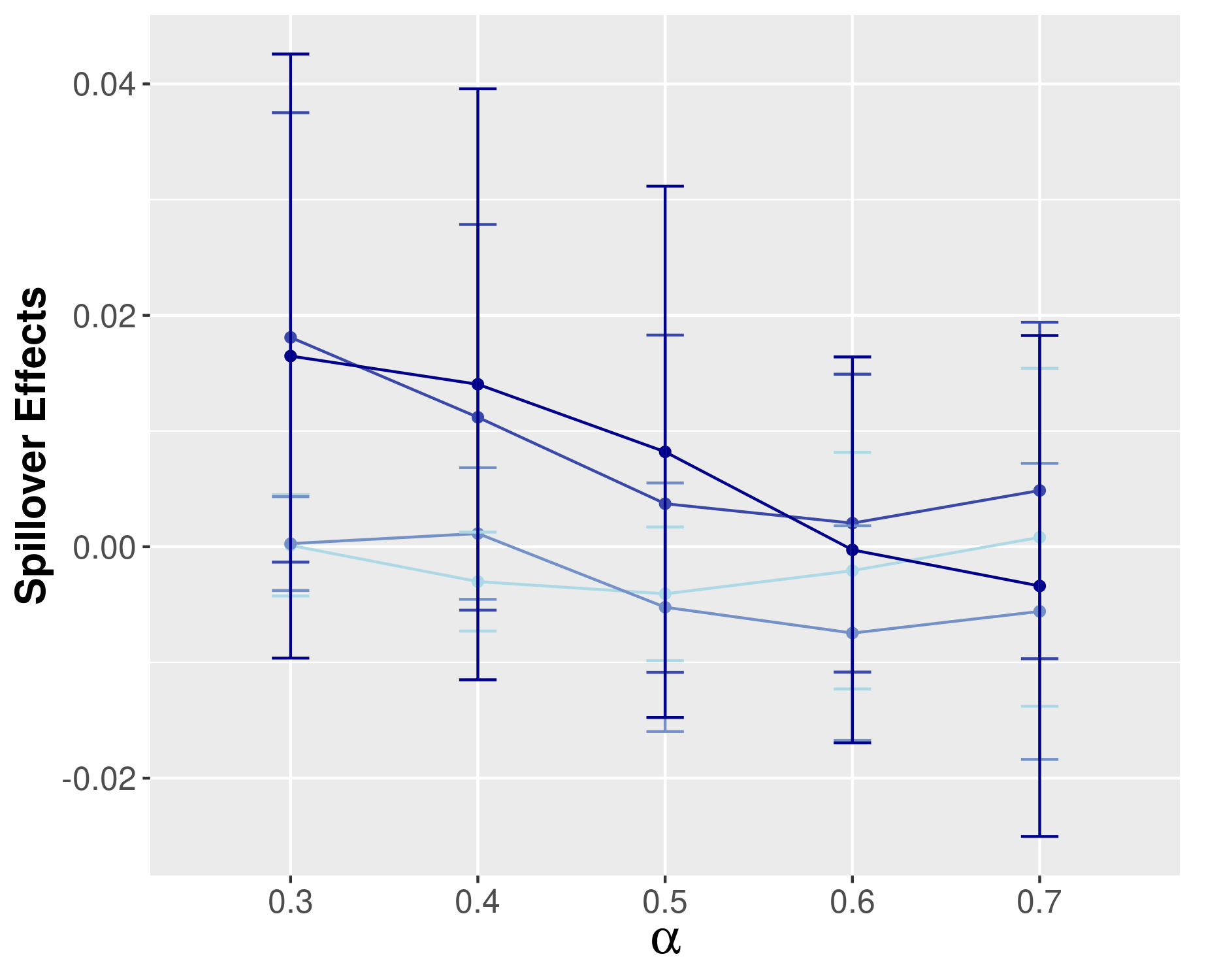}
\caption{$\rma=0,\ h=2$}
\label{fig:hajek_d3_a0h2}
\end{subfigure}
    \hfill
\begin{subfigure}{0.52\linewidth}
\centering
\includegraphics[height=5.5cm]{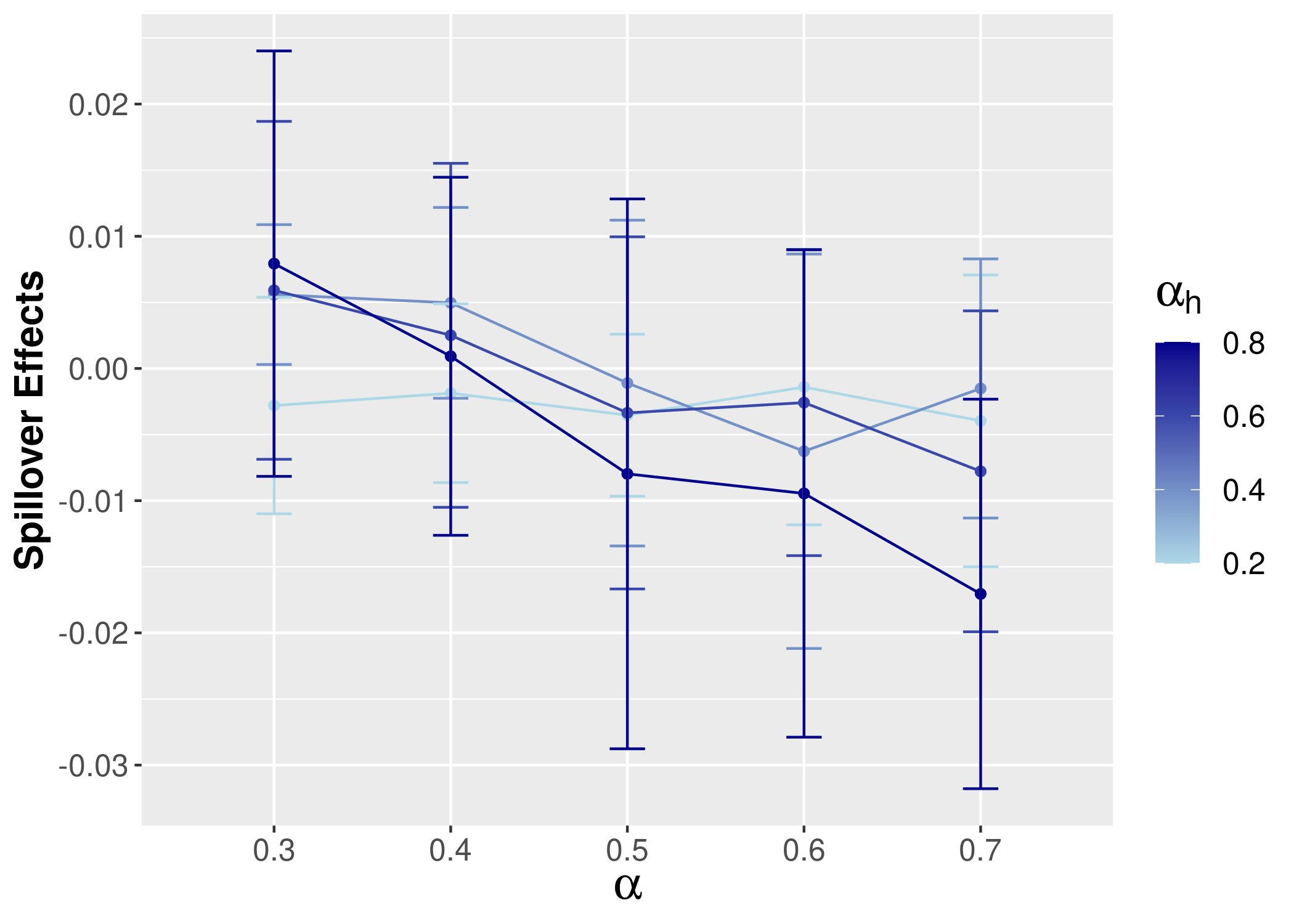}
\caption{$\rma=1,\ h=2$}
\label{fig:hajek_d3_a1h2}
\end{subfigure}

\caption{Spillover effect estimates using the Hajek estimator under the distance-3 interference set. Point estimates and 95\% confidence intervals are reported as a function of $\alpha$ (x-axis) and $\alpha_h$ (colors).}
\label{fig:d3_hajek}
\end{figure*}

\begin{figure*}[!htbp]
\centering

\begin{subfigure}{0.42\linewidth}
\centering
\hspace*{-0.5cm}
\includegraphics[height=5.5cm]{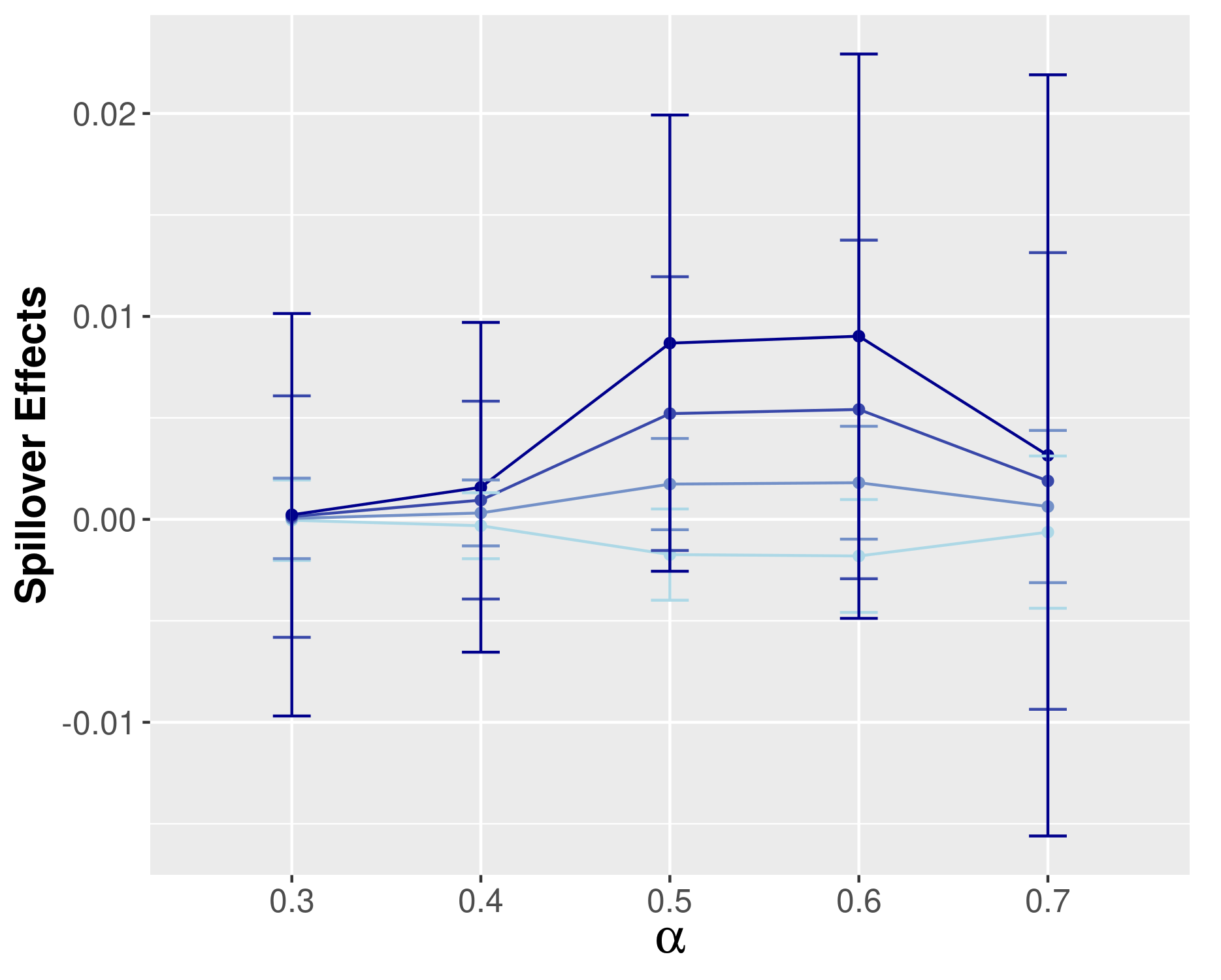}
\caption{$\rma=0$, $h=1$}
         \label{fig:wr_d3_a0h1}
\end{subfigure}
    \hfill
\begin{subfigure}{0.52\linewidth}
\centering
\includegraphics[height=5.5cm]{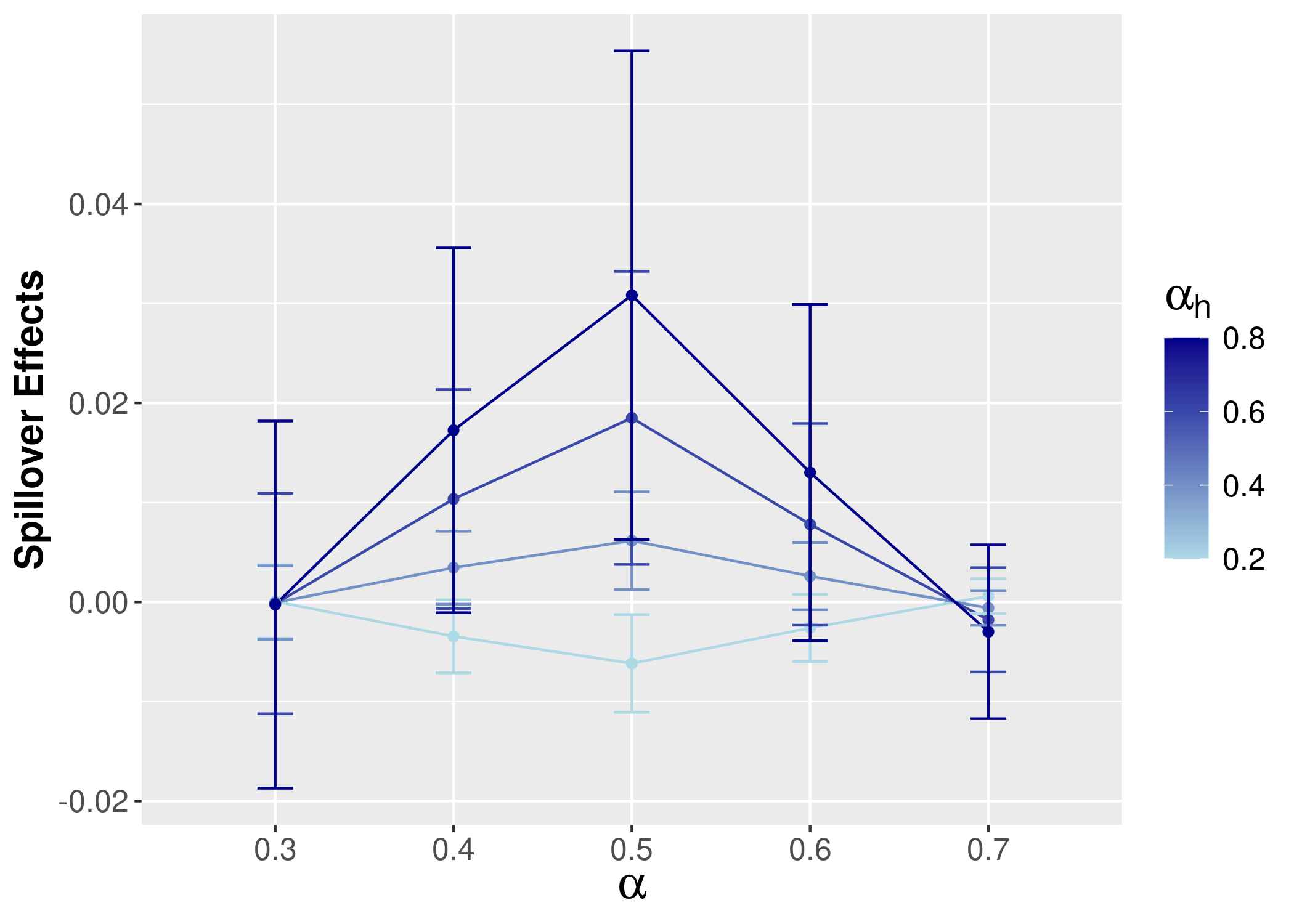}
\caption{$\rma=1$, $h=1$}
       \label{fig:wr_d3_a1h1}
\end{subfigure}

\vspace{0.4cm}

\begin{subfigure}{0.42\linewidth}
\centering
\hspace*{-0.5cm}
\includegraphics[height=5.5cm]{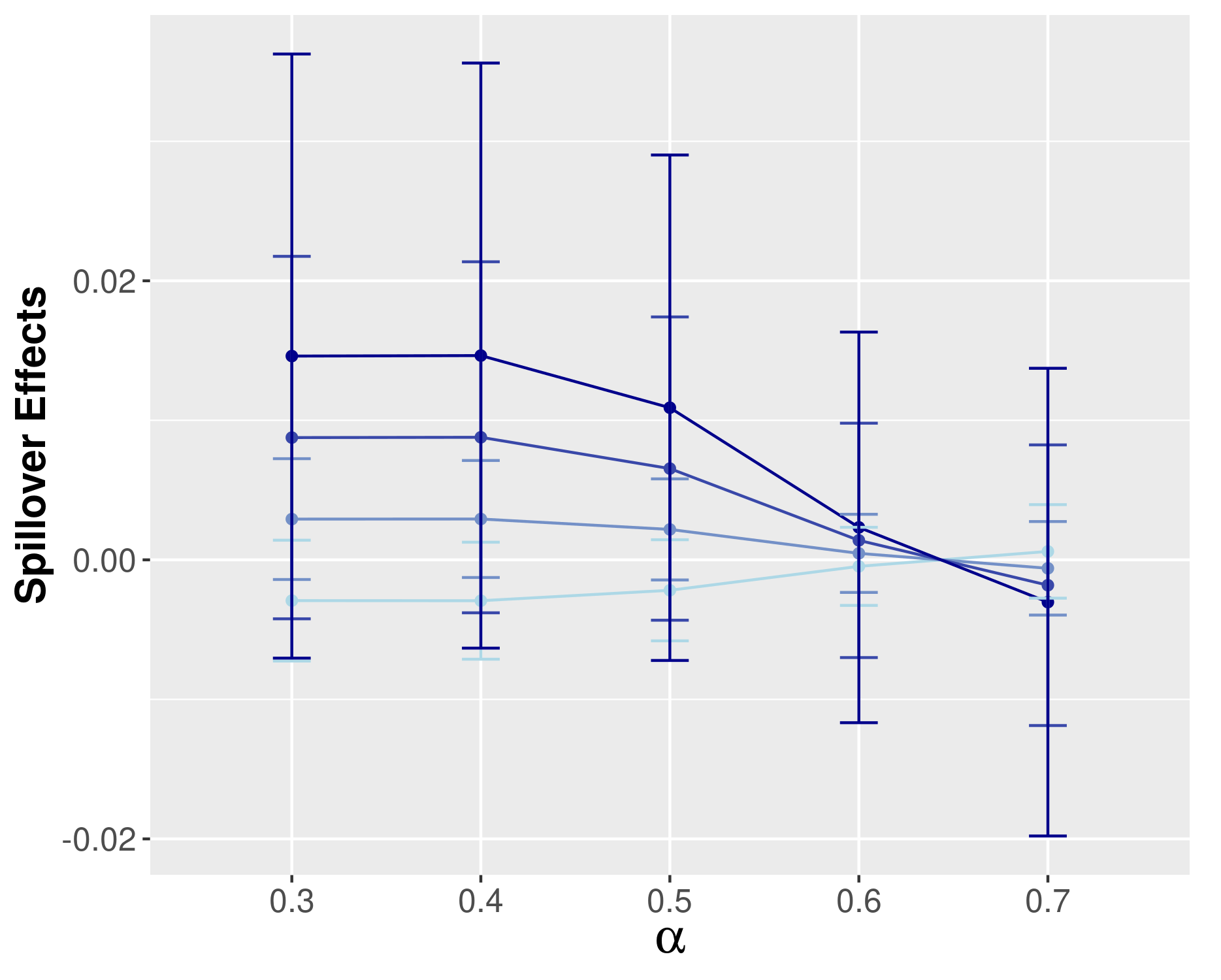}
\caption{$\rma=0$, $h=2$}
         \label{fig:wr_d3_a0h2}
\end{subfigure}
    \hfill
\begin{subfigure}{0.52\linewidth}
\centering
\includegraphics[height=5.5cm]{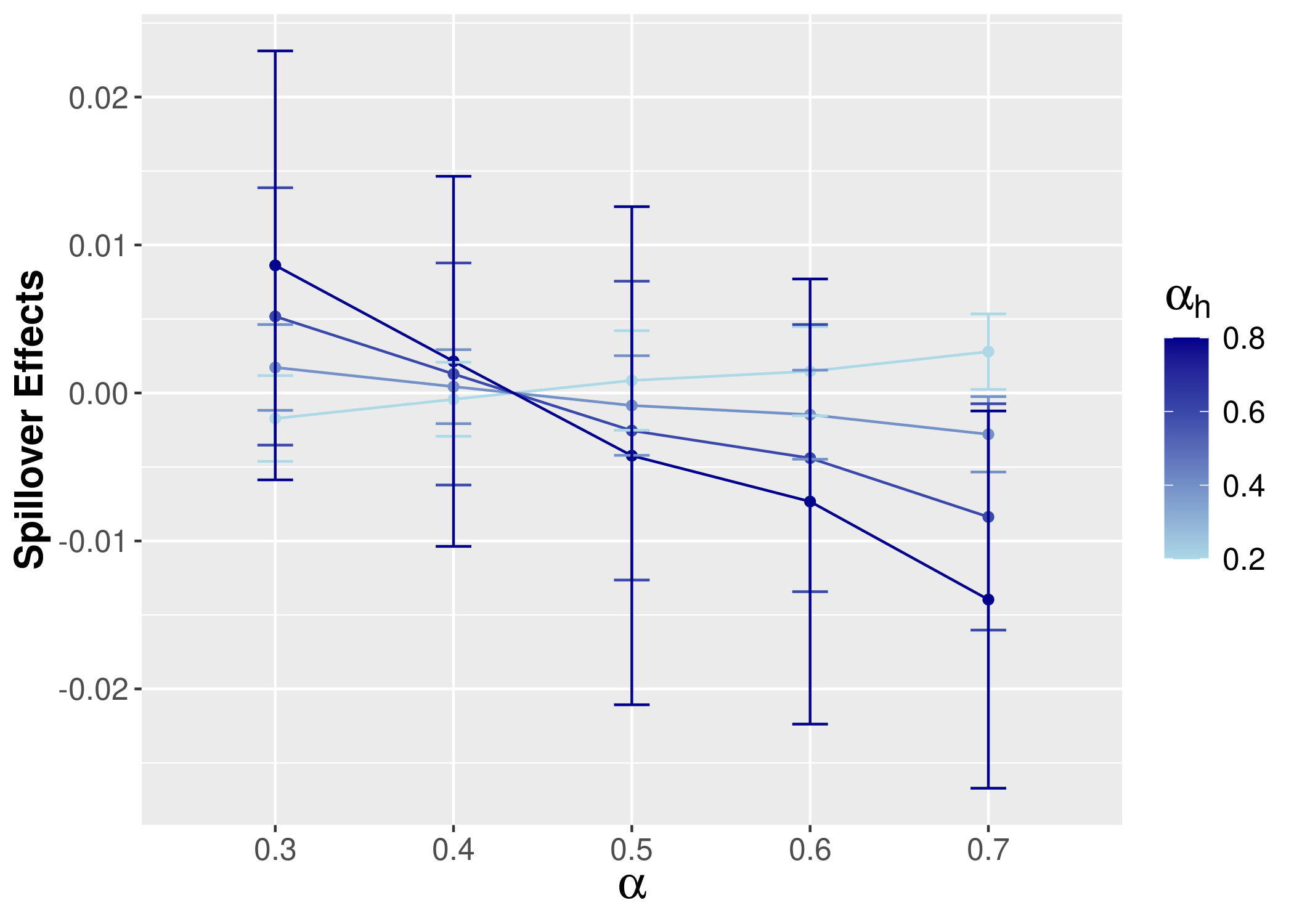}
\caption{$\rma=1$, $h=2$}
       \label{fig:wr_d3_a1h2}
\end{subfigure}

\caption{Spillover effect estimates using the WLS estimator under the distance-3 interference set. Point estimates and 95\% confidence intervals are reported as a function of $\alpha$ (x-axis) and $\alpha_h$ (colors).}
\label{fig:d3_WLS}
\end{figure*}

\subsubsection{Community Detection Algorithm-Based Interference Set}
We also consider an alternative interference set definition based on community detection. Specifically, within each cluster $i$, we apply Newman’s leading-eigenvector community detection algorithm \citep{newman2006finding} to partition the network into communities. Let $\mathcal{C}_i(ij)$ denote the community (within cluster $i$) to which unit $ij$ belongs. We exclude communities with fewer than five members, and define
$$
\mathcal{I}_{ij}
= \mathcal{C}_i(ij)\setminus \{ij\}.
$$
Figures~\ref{fig:comm_hajek}--\ref{fig:comm_wls} report the corresponding spillover effect estimates under this community-based interference set definition. We observe that the Hajek and WLS estimators yield highly consistent patterns.

For first-order spillover effects ($h=1$), the results are consistent with the main analysis based on finite paths (Figure 5 and 7): treated individuals ($\rma=1$) exhibit the inverted U-shaped trend, peaking at medium levels of $\alpha$, while untreated individuals ($\rma=0$) show negligible effects with slightly increasing trend. However, both effect sizes are notably reduced.

For second-order spillover effects ($h=2$), while the estimates for untreated individuals ($\rma=0$) remain similarly negligible, the trend for treated individuals shifts from a decreasing trend in the main analysis (Figure 6 and 8) to an increasing trend in the community-based analysis (Figure~\ref{fig:comm_hajek}, \ref{fig:comm_wls}); however, most point estimates remain statistically insignificant.

The differences in trend are expected because the community detection changes the network structure and interference sets. First, the algorithm excludes communities with fewer than five members, thereby removing a subset of the data included in the main analysis and yielding a truncated version of the original study population. Second, by partitioning the network into distinct communities, the algorithm eliminates a significant number of first- and second-order neighbors, specifically those forming "bridges" between communities, which can lead to the observed changes in the estimated patterns.

\begin{figure*}[!htbp]
\centering

\begin{subfigure}{0.42\linewidth}
\hspace*{-0.5cm}
\centering
\includegraphics[height=5.5cm]{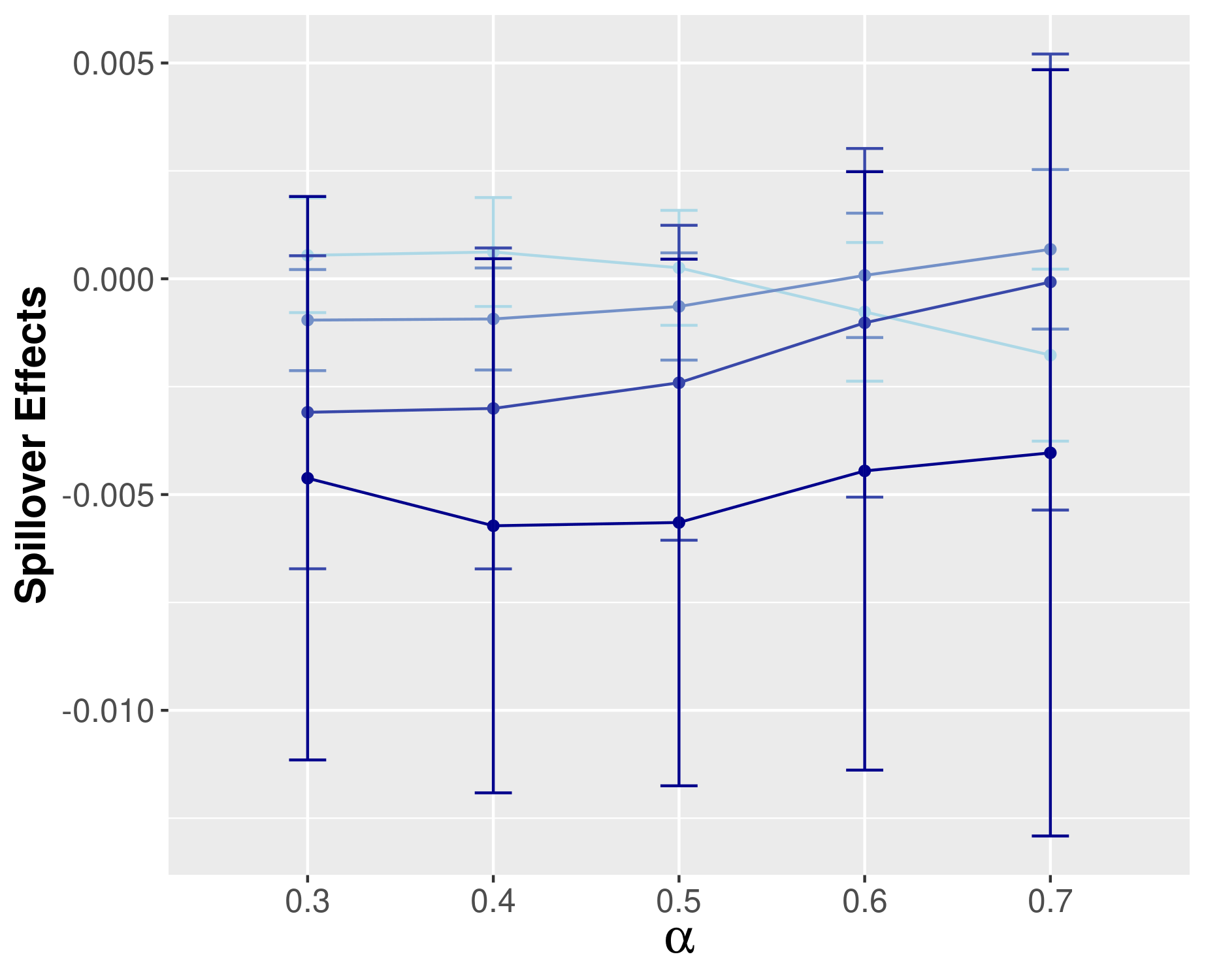}
\caption{$\rma=0,\ h=1$}
\label{fig:comm_hajek_a0h1}
\end{subfigure}
    \hfill
\begin{subfigure}{0.52\linewidth}
\centering
\includegraphics[height=5.5cm]{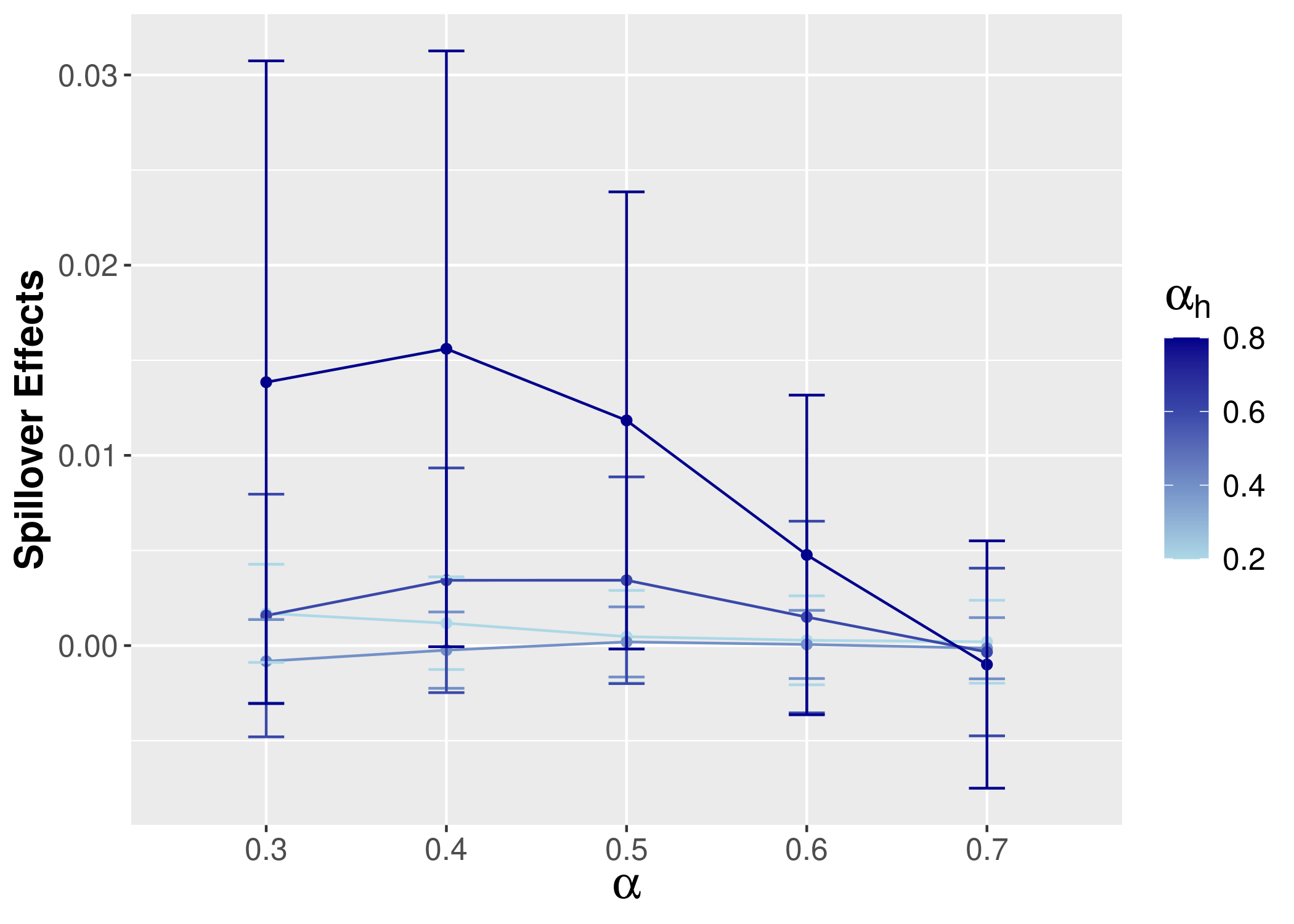}
\caption{$\rma=1,\ h=1$}
\label{fig:comm_hajek_a1h1}
\end{subfigure}

\vspace{0.4cm}

\begin{subfigure}{0.42\linewidth}
\hspace*{-0.5cm}
\centering
\includegraphics[height=5.5cm]{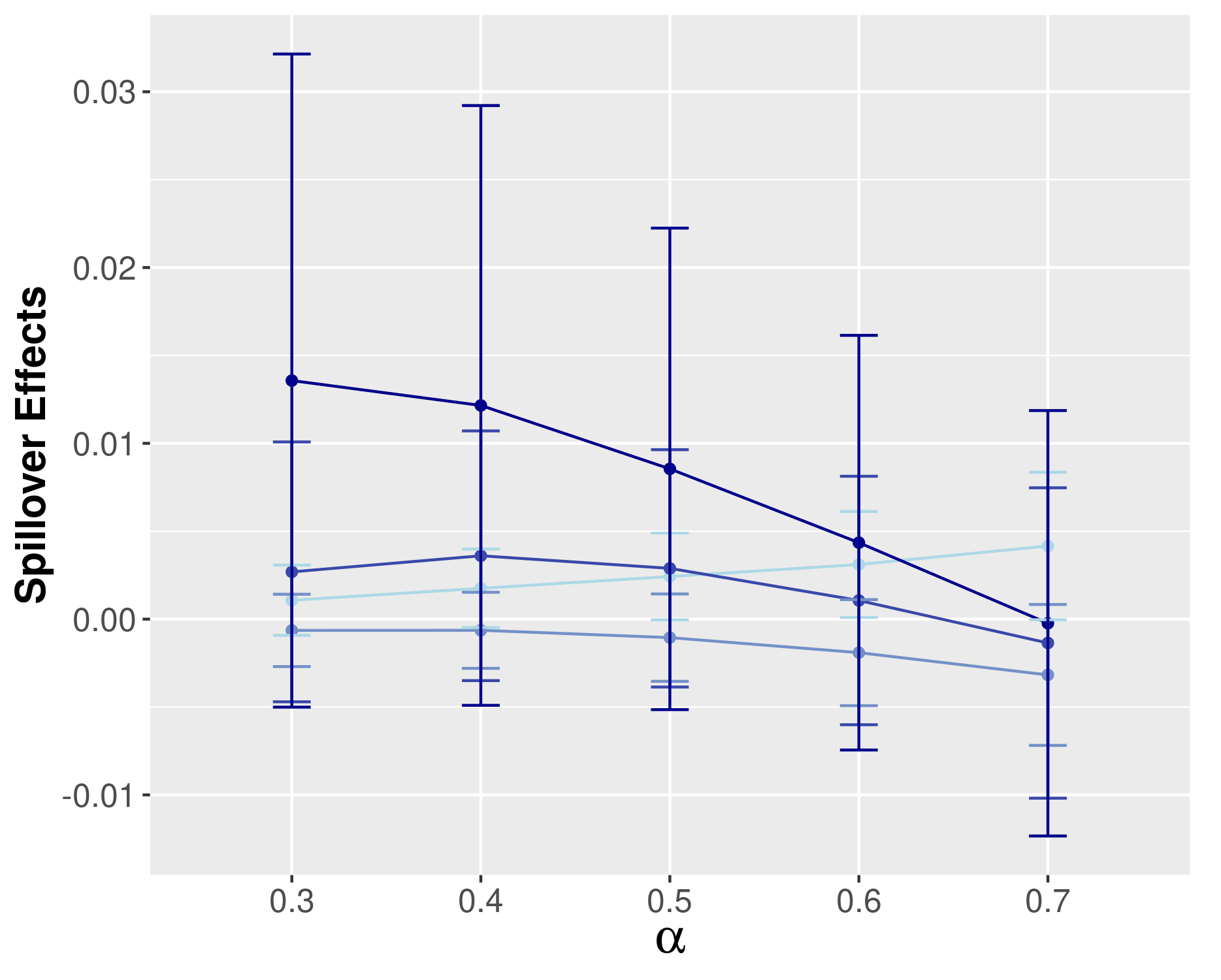}
\caption{$\rma=0,\ h=2$}
\label{fig:comm_hajek_a0h2}
\end{subfigure}
    \hfill
\begin{subfigure}{0.52\linewidth}
\centering
\includegraphics[height=5.5cm]{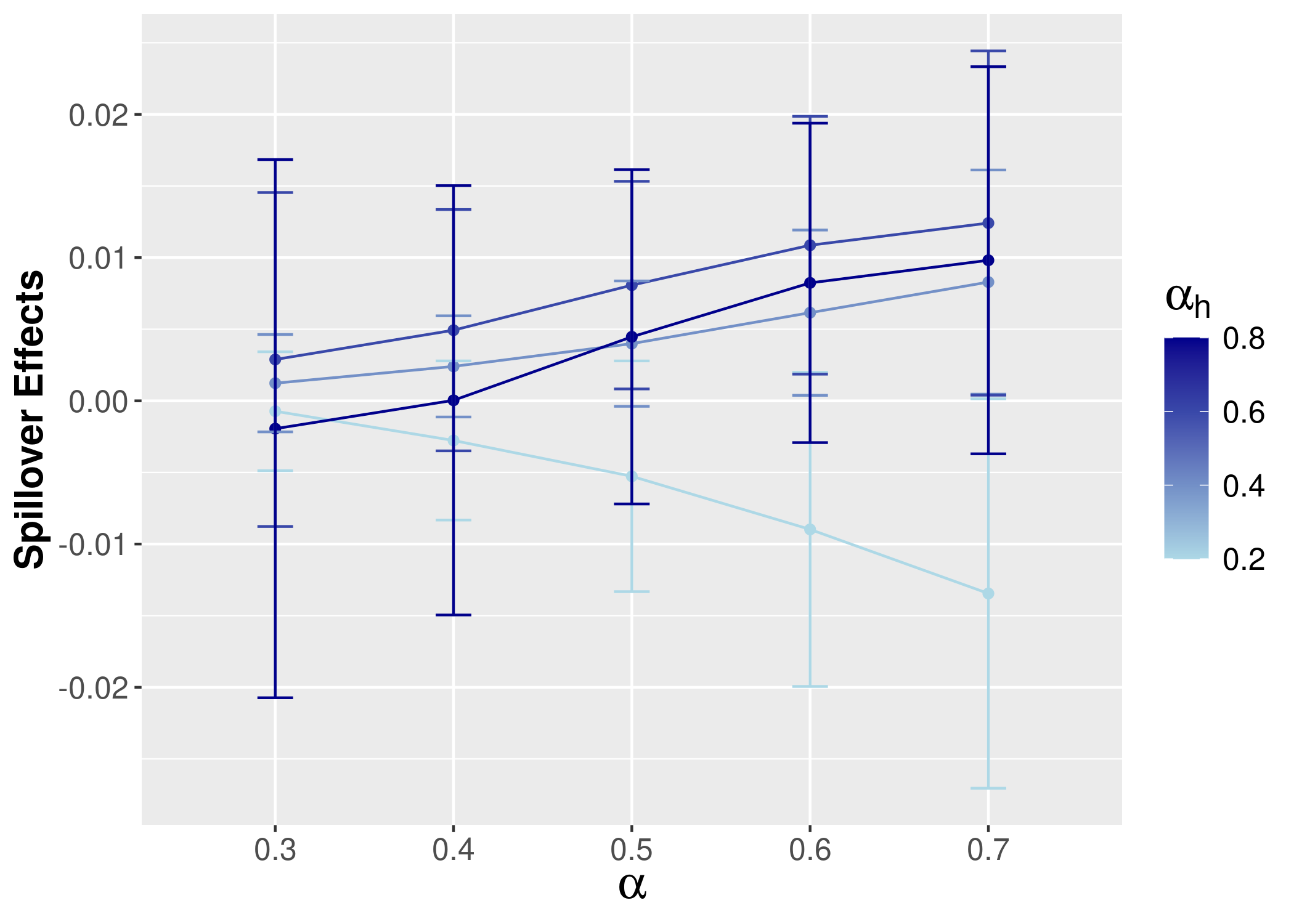}
\caption{$\rma=1,\ h=2$}
\label{fig:comm_hajek_a1h2}
\end{subfigure}

\caption{Spillover effect estimates using the Hajek estimator under the community-based interference set. Point estimates and 95\% confidence intervals are reported as a function of $\alpha$ (x-axis) and $\alpha_h$ (colors).}
\label{fig:comm_hajek}
\end{figure*}

\begin{figure*}[!htbp]
\centering

\begin{subfigure}{0.42\linewidth}
\hspace*{-0.5cm}
\centering
\includegraphics[height=5.5cm]{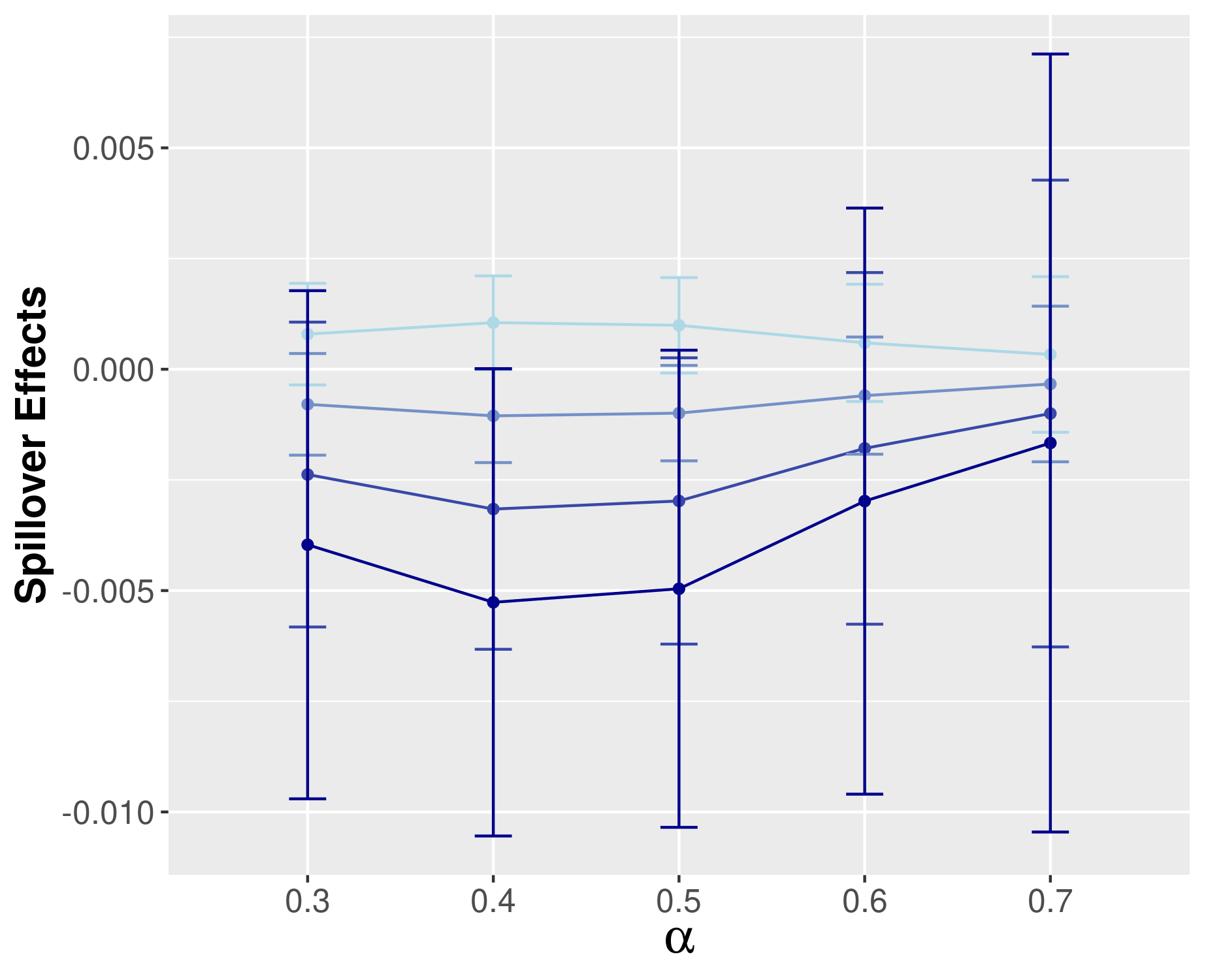}
\caption{$\rma=0$, $h=1$}
         \label{fig:comm_wr_a0h1}
\end{subfigure}
    \hfill
\begin{subfigure}{0.52\linewidth}
\centering
\includegraphics[height=5.5cm]{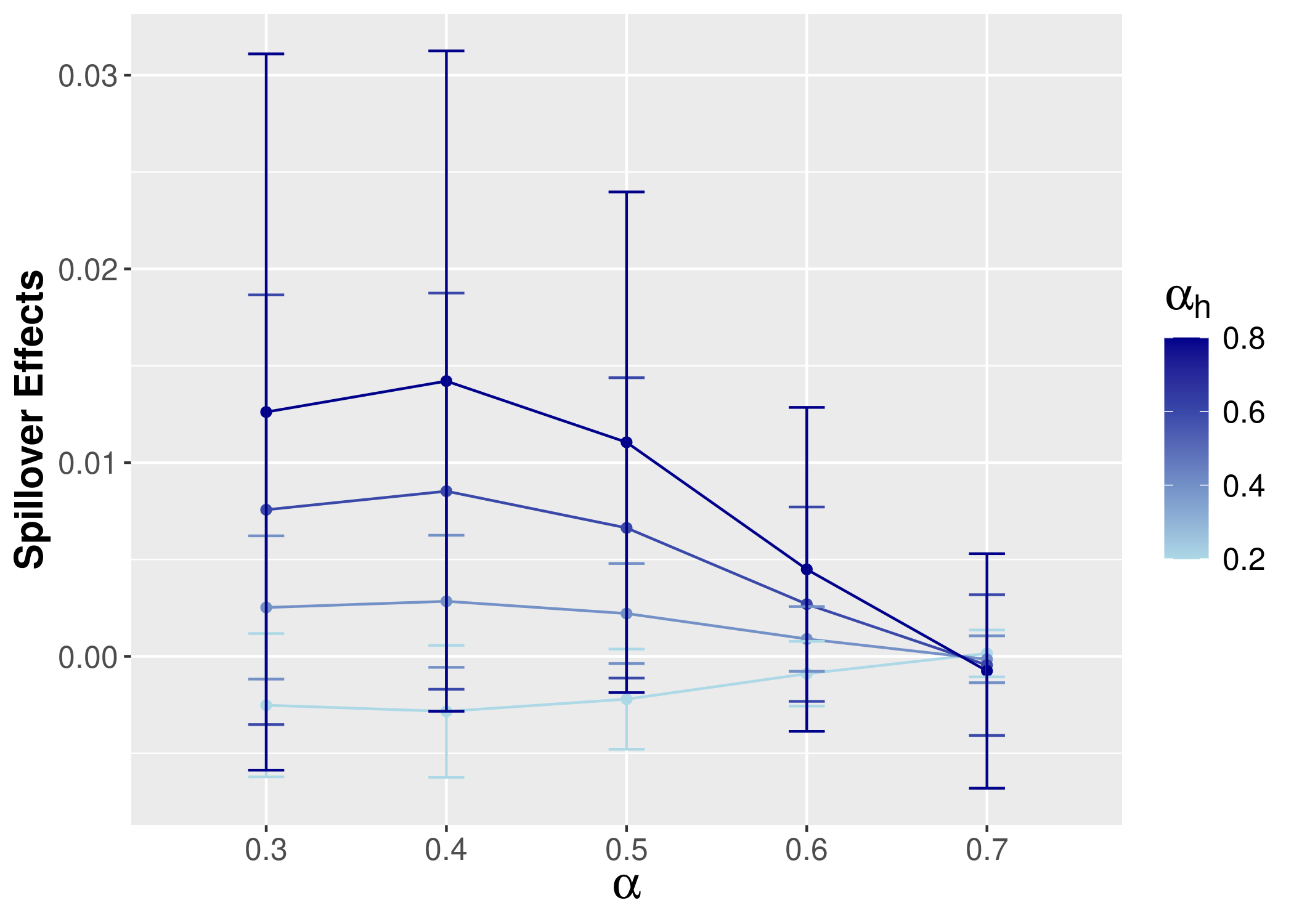}
\caption{$\rma=1$, $h=1$}
       \label{fig:comm_wr_a1h1}
\end{subfigure}

\vspace{0.4cm}

\begin{subfigure}{0.42\linewidth}
\hspace*{-0.5cm}
\centering
\includegraphics[height=5.5cm]{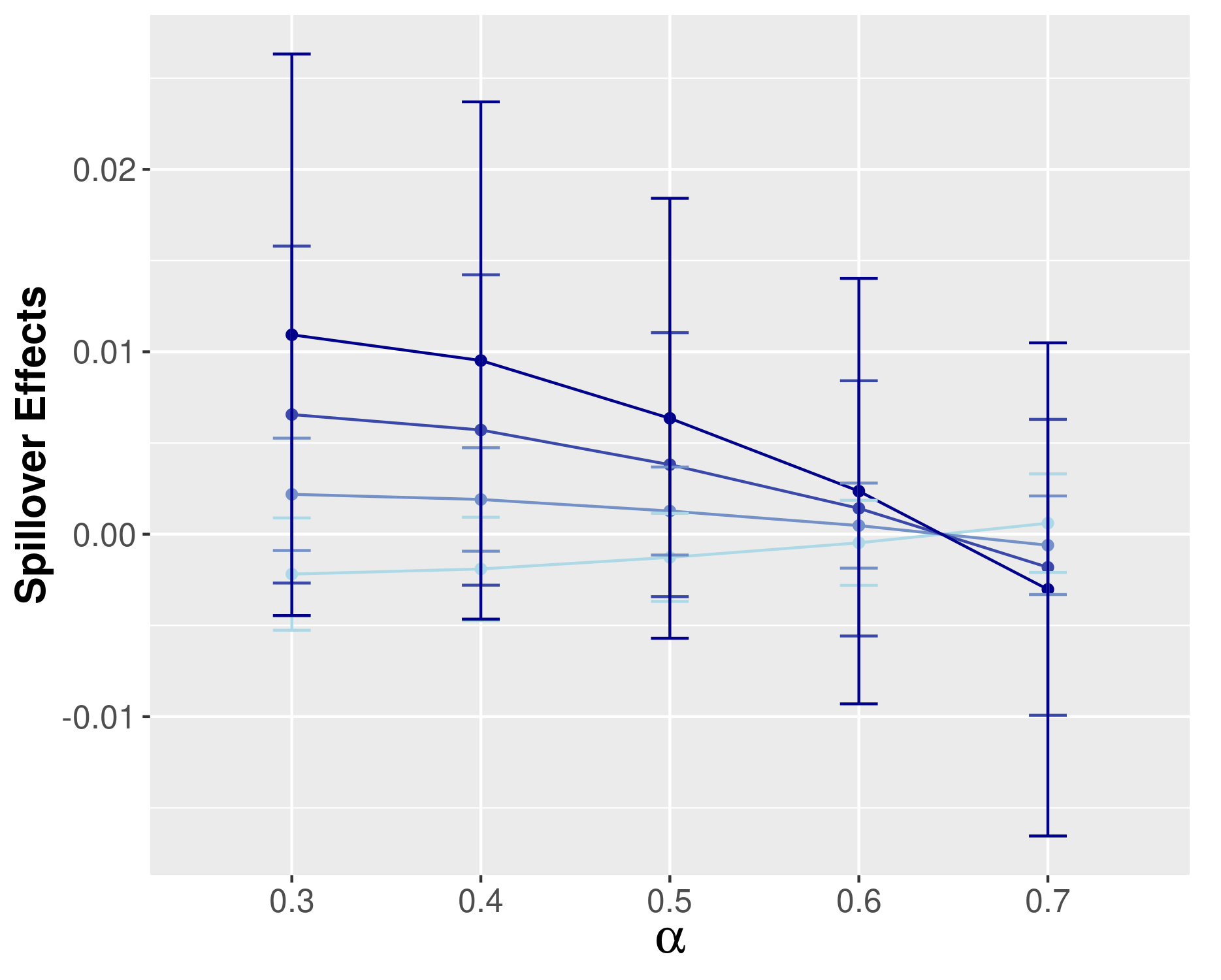}
\caption{$\rma=0$, $h=2$}
         \label{fig:comm_wr_a0h2}
\end{subfigure}
    \hfill
\begin{subfigure}{0.52\linewidth}
\centering
\includegraphics[height=5.5cm]{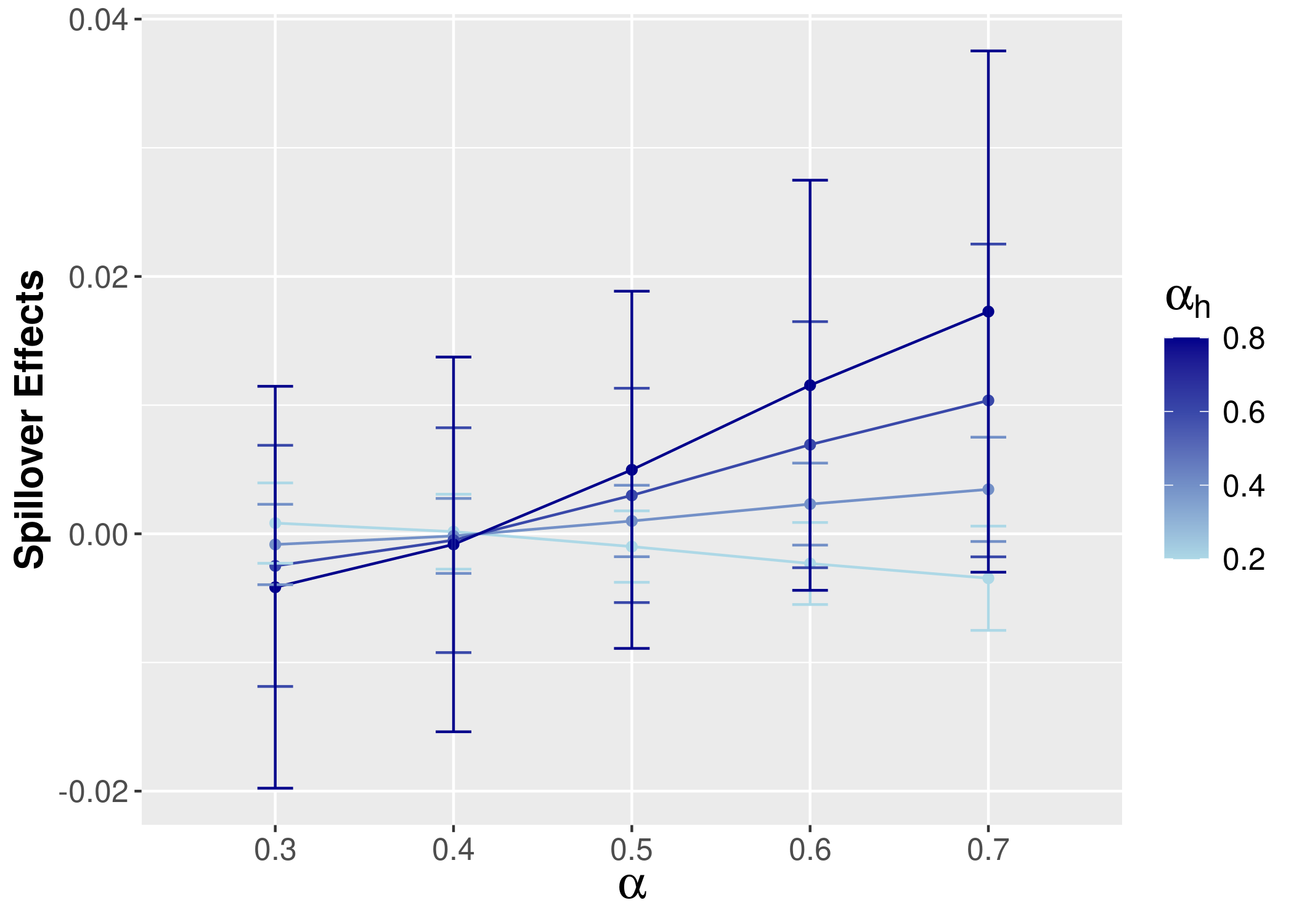}
\caption{$\rma=1$, $h=2$}
       \label{fig:comm_wr_a1h2}
\end{subfigure}

\caption{Spillover effect estimates using the WLS estimator under  the community-based interference set. Point estimates and 95\% confidence intervals are reported as a function of $\alpha$ (x-axis) and $\alpha_h$ (colors).}
\label{fig:comm_wls}
\end{figure*}

\bibliographystyle{plainnat}
\bibliography{citation}

\end{document}